\documentclass[aos]{imsart}

\RequirePackage{amsthm,amsmath,amsfonts,amssymb}
\RequirePackage[authoryear]{natbib}
\RequirePackage[colorlinks,citecolor=blue,urlcolor=blue]{hyperref}
\RequirePackage{graphicx}

\RequirePackage{algorithm}
\RequirePackage[noend]{algpseudocode}
\RequirePackage{cancel}
\RequirePackage{mathtools}

\startlocaldefs
\theoremstyle{plain}
\newtheorem{mythm}{Theorem}
\newtheorem{myprop}{Proposition}[section]

\newtheorem{mylemma}[myprop]{Lemma}
\theoremstyle{remark}
\newtheorem{mydef}[myprop]{Definition}
\newtheorem{myexample}[myprop]{Example}
\newtheorem*{myremark}{Remark}
\def\ci{\perp\!\!\!\perp}
\newcommand{\G}{\mathcal{G}}
\newcommand{\V}{\mathbf{V}}
\newcommand{\Ovar}{\mathbf{O}}
\newcommand{\La}{\mathbf{L}}
\newcommand{\E}{\mathbf{E}}
\newcommand{\D}{\mathcal{D}}
\newcommand{\M}{\mathcal{M}}
\newcommand{\PAG}{\mathcal{P}}
\newcommand{\Mtaumax}{\mathcal{M}^{\taumax}}
\newcommand{\Mtaumaxtilde}{\mathcal{M}^{\taumaxtilde}}
\newcommand{\Mtaumaxlim}{\mathcal{M}^{\taumax}_{lim}}
\newcommand{\stat}{stat}
\newcommand{\statsubscript}{st}
\newcommand{\Mtaumaxstat}{\mathcal{M}_{\statsubscript}^{\taumax}}
\newcommand{\PAGtaumax}{\mathcal{P}^{\taumax}}
\newcommand{\PAGtaumaxtilde}{\mathcal{P}^{\taumaxtilde}}
\newcommand{\PAGtaumaxlim}{\mathcal{P}^{\taumax}_{lim}}
\newcommand{\Dc}{\D_c}
\newcommand{\pa}{pa}
\newcommand{\an}{an}

\newcommand{\Iindex}{\mathbf{I}}
\newcommand{\IindexO}{\Iindex_{\Ovar}}
\newcommand{\IindexL}{\Iindex_{\La}}
\newcommand{\Tindex}{\mathbf{T}}
\newcommand{\TindexO}{\Tindex_{\Ovar}}
\newcommand{\tailhead}{{\rightarrow}}
\newcommand{\headtail}{{\leftarrow}}
\newcommand{\headhead}{{\leftrightarrow}}
\newcommand{\ohead}{{\circ\!{\rightarrow}}}

\newcommand{\oo}{{\circ\!{\--}\!\circ}}
\newcommand{\asthead}{{\ast\!\!{\rightarrow}}}
\newcommand{\headast}{{{\leftarrow}\!\ast}}
\newcommand{\astast}{{\ast\!{\--}\!\ast}}
\newcommand{\oast}{{\circ\!{\--}}\!\!\!\ast}
\newcommand{\porder}{p_{ts}}
\newcommand{\taumax}{p}

\newcommand{\taumaxtilde}{\tilde{\taumax}}

\newcommand{\BR}{\mathcal{A}}
\newcommand{\BRstat}{\BR^{\stat}}
\newcommand{\BRtoro}{\BR_{to}}
\newcommand{\BRtora}{\BR_{ta}}
\newcommand{\BRtsDAG}{\BR_{\D}}
\newcommand{\BRtsDAGstat}{\BRstat_{\D}}

\endlocaldefs

\begin{document}

\begin{frontmatter}
\title{Characterization of causal ancestral graphs for time series with latent confounders}
\runtitle{Causal ancestral graphs for time series}

\begin{aug}
\author[A]{\fnms{Andreas}~\snm{Gerhardus}\ead[label=e1]{andreas.gerhardus@dlr.de}}
\address[A]{German Aerospace Center, Institute of Data Science, \printead{e1}}
\end{aug}

\begin{abstract}
In this paper, we introduce a novel class of graphical models for representing time lag specific causal relationships and independencies of multivariate time series with unobserved confounders. We completely characterize these graphs and show that they constitute proper subsets of the currently employed model classes. As we show, from the novel graphs one can thus draw stronger causal inferences---without additional assumptions. We further introduce a graphical representation of Markov equivalence classes of the novel graphs. This graphical representation contains more causal knowledge than what current state-of-the-art causal discovery algorithms learn.
\end{abstract}



\end{frontmatter}


\section{Introduction}\label{sec:introduction}
In recent decades causal graphical models have become a standard tool for reasoning about causal relationships, e.g.~\citet{Pearl2009}, \citet{Spirtes2000}, \citet{koller2009probabilistic}. The most basic and popular class of models are directed acyclic graphs (DAGs). In their interpretation as causal Bayesian networks these graphs specify interventional distributions and causal effects in terms of the observational distribution, e.g.~\citet{Spirtes1993}, \citet{pearl1995causal}, \citet{Pearl2000}. DAGs can only model acyclic causal relationships among variables that are not subject to latent confounding, i.e., such that there are no unobserved common causes of observed variables. The latter assumption is known as \emph{causal sufficiency} and intuitively means that all variables relevant for describing the system's causal relationships are modeled explicitly. If causal sufficiency cannot be asserted, as is often the case, then one approach is to instead work with maximal ancestral graphs (MAGs), see \citet{richardson2002}, \citet{zhang2008causal}. This larger class of graphs retains a well-defined causal interpretation in presence of latent confounding.

MAGs can even represent selection variables, i.e., unobserved variables that determine which sample points belong to the observed population. In this paper, we rule out selection variables by assumption. It is then sufficient to work with a subclass of MAGs that, following \citet{FCI_cyclic}, are called \emph{directed maximal ancestral graphs (DMAGs)}. Assuming the absence of selection variables is common both in the literature on causal effect estimation and causal discovery, e.g.~\citet{zhang2006causal}, \citet{perkovic2018complete} and \citet{Entner2010}, \citet{malinsky2018causal}, \citet{LPCMCI}. As an advantage, DMAGs convey significantly stronger inferences about the presence of causal ancestral relationships than MAGs. Moreover, for time series there is exactly one sample point per time step and hence potential selection bias would at least not go unnoticed.

To use any of these model classes for causal reasoning one needs to already know the system's causal structure in form of the respective graph. If this knowledge is not available and experiments are infeasible, then one must rely on observational causal discovery, e.g.~\citet{Spirtes2000}, \citet{Peters2018}, which refers to learning causal relationships from observational data under suitable enabling assumptions. So-called independence-based methods, also called constraint-based methods, attempt to learn the causal graph from independencies in the observed probability distribution. In general, learning the graph from independencies is an under-determined problem since distinct graphs may describe the same set of independencies. This non-uniqueness is known as \emph{Markov equivalence}. Without more assumptions it is then only possible to learn those features of the causal graph that it shares with all its Markov equivalent graphs. These shared features can in turn be represented by certain graphs, which for the case of MAGs are \emph{partial ancestral graphs (PAGs)}, see \citet{ali2009markov}, \citet{Zhang2008}. There are sound and complete causal discovery algorithms for learning PAGs, e.g.~the FCI algorithm, see \citet{Spirtes1995}, \citet{Spirtes2000}, \citet{Zhang2008}. Here, \emph{sound} refers to correctness of the method and \emph{complete} to it learning all shared features. The refinement of PAGs obtained by restricting from MAGs to DMAGs are called \emph{directed partial ancestral graphs (DPAGs)} in \citet{FCI_cyclic}.

The causal graphical model framework outlined above does not inherently rely on temporal information, and the non-temporal setting so far is its major domain of application. However, dynamical systems and time series data are ubiquitous and of great interest to science and beyond. In this setting, Granger causality (see \citet{Granger1969}) is a widely-used framework for causal analyses. This framework employs a predictive notion of causality, according to which a time series $X$ has a causal influence on time series $Y$ if the past of $X$ helps in predicting the present of $Y$ given that the pasts of all time series other than $X$ are already known. Granger causality has two central limitations: First, it requires the absence of latent confounders, i.e., unobserved time series that are a common cause of two observed time series. Second, it cannot in general deal with contemporaneous causal influences, i.e., causal influences on time scales below the sampling interval. For an in-depth discussion of these limitations see, e.g., \citet[chapter 10]{Peters2018}.

Since the causal graphical model framework is not subject to these two limitations, in recent years there has been a growing interest in adapting it to the time series setting. Generally, there are three ways to do this. The \underline{first} approach, e.g.~\citet{eichler2007causal}, \citet{eichler2010graphical}, \citet{eichler2010granger} and \citet{didelez2008graphical}, \citet{mogensen2020markov}, uses a graph in which there is one vertex per component time series. The edges then summarize the causal influences at all time lags, thus giving a conveniently compressed graphical representation of the causal relationships. However, the information about time lags of individual cause-and-effect relationships is lost. The \underline{second} approach uses graphs with one vertex per component time series and time step, thus resolving the time lags. There are various causal discovery methods that implement this approach, e.g.~\citet{chu2008search}, \citet{hyvarinen2010estimation}, \citet{Entner2010}, \citet{malinsky2018causal}, \citet{Runge2020a}, \citet{pamfil2020Dynotears}, \citet{LPCMCI}, and application works from diverse domains, e.g.~\citet{kretschmer2016using}, \citet{huckins2020causal}, \citet{saetia2021constructing}. By resolving time lags it becomes possible to obtain a data-driven process understanding and to study the effect of interventions on particular time steps of variables. However, learning a time-resolved graph is statistically more challenging than learning a time-collapsed graph and one might need to compromise on the number of resolved time steps. \citet{assaad2022discovery} proposes a \underline{third}, intermediate approach with two vertices per component time series (one for the present time step and one for the entire past).

We follow the second approach. In this case, the temporal information inherent to time series restricts the connectivity pattern (i.e., absence and presence of edges, edge orientations) of the resulting time-resolved graphs. Namely, since we here consider graphical models in which directed edges signify causal influences (DAGs, DMAGs and DPAGs), the directed edges must not point backwards in time. In addition, we assume time invariant causal relationships. This invariance, known as \emph{causal stationarity}, implies that the graph's edges are repetitive in time. For DAGs that represent time series \emph{without} latent confounders, which we call \emph{time series DAGs (ts-DAGs)}, these are the only restrictions on the connectivity pattern.

For DMAGs that represent time series \emph{with} latent confounders, the corresponding restrictions on the connectivity pattern have, however, not yet been worked out. Although there are works on independence-based time series causal discovery with latent confounding, see \citet{Entner2010}, \citet{malinsky2018causal}, \citet{LPCMCI}, no characterization of the associated class of graphical models has been given. This is the conceptual gap that we close in the present work, i.e., we completely characterize which DMAGs are obtained by marginalizing ts-DAGs and hence can serve as causal graphical model for causally stationary time series with latent confounders. We call the novel graphs defined by this characterization \emph{time series DMAGs (ts-DMAGs)} and show that these novel graphs constitute a strictly smaller model class than the previously considered model classes. We further show that, without imposing additional assumptions, one can draw stronger causal inferences from ts-DMAGs than from the previously considered graphs. We also introduce \emph{time series DPAGs (ts-DPAGs)} as representations of Markov equivalence classes of ts-DMAGs. Time series DPAGs are more informative than the graphs learned by current latent time series causal discovery algorithms. As a remark, since contemporaneous causal interactions are allowed without restrictions other than acyclicity, the time series case considered here formally subsumes and hence is more general than the (acyclic) non-temporal case. 

The structure of this paper is as follows: In \textbf{Sec.~\ref{sec:notation}} we summarize basic graphical concepts and introduce our notation. In \textbf{Sec.~\ref{sec:class_graphical_models}} we first specify the considered type of causally stationary time series processes. We then introduce ts-DMAGs, a class of causal graphical models for representing the causal relationships and independencies among only the observed variables of such processes at finitely many regularly spaced observed time steps. In \textbf{Sec.~\ref{sec:characterization_full_section}} we analyze ts-DMAGs and first derive several properties that they necessarily have. With Theorems~\ref{thm:complete_characterization} and \ref{thm:complete_characterization_2} we then completely characterize ts-DMAGs by a single necessary and sufficient condition. We further show that ts-DMAGs are a strict subset of the classes of graphical models that have previously been considered in the literature (see Sec.~\ref{sec:previous_model_classes}). For this reason, and as we demonstrate with examples, one can draw stronger causal inferences from ts-DMAGs than from the previously considered graphs. We further introduce the concept of \emph{stationarification} in order to illuminate various discussions. In \textbf{Sec.~\ref{sec:equivalence_classes_and_causal_discovery}} we put these developments to use in the context of causal discovery by defining ts-DPAGs as representations of the Markov equivalence classes of ts-DMAGs. We show that these graphs contain more causal information than the output of current causal discovery algorithms. Moreover, we point out an incorrect claim in the literature that, as we argue, has misguided recent developments (see the discussion below Theorem~\ref{thm:causal_discovery_non_stat_is_better}). We also present an algorithm that learns ts-DPAGs from data. We give further theoretical results and all proofs in the \textbf{Supplementary Material} \citep{gerhardus2022_supplement}.

\section{Basic graphical concepts and notation}\label{sec:notation}
Our notation and terminology is a mixture of those used in \citet{maathuis2015generalized}, \citet{Complete_characterization_adjustment} and \citet{FCI_cyclic} as well as some idiosyncratic notation.

A \textbf{graph} $\G = (\V, \E)$ consists of a set of vertices $\V$ together with a set of edges $\E \subseteq \V \times \V$. The vertices $i,\, j \in \V$ are \emph{adjacent} if $(i, j) \in \E$ or $(j, i) \in \E$. We then say that there is an \emph{edge between $i$ and $j$} and that \emph{$i$ is an adjacency of $j$}, and similiary for $i$ and $j$ interchanged.

Throughout this paper we only consider \textbf{directed partial mixed graphs}. These are graphs that satisfy three conditions: First, there is at most one edge between any pair of vertices. Second, no vertex is adjacent to itself. Third, there are at most four types of edges: \emph{directed edges} ($\tailhead$), \emph{bidirected edges} ($\headhead$), \emph{partially directed edges} ($\ohead$), and \emph{non-directed edges} ($\oo$). The third condition is formalized by a decomposition of $\E$ as $\E = \E_{\tailhead} \,\dot{\cup}\, \E_{\headhead} \,\dot{\cup}\, \E_{\ohead} \,\dot{\cup}\, \E_{\oo}$ that specifies the \emph{edge types} (also called \emph{edge orientations}). This decomposition is considered part of the specification of a concrete graph. A \textbf{directed mixed graph} is a partial mixed graph without partially directed and non-directed edges, and a \textbf{directed graph} is a directed mixed graph without bidirected edges. The \textbf{skeleton} of a graph is the object obtained when disregarding the information about the decomposition of $\E$ into $\E_{\tailhead} \,\dot{\cup}\, \E_{\headhead} \,\dot{\cup}\,\E_{\ohead} \,\dot{\cup}\,\E_{\oo}$.

Given directed partial mixed graphs $\G = (\V, \E)$ and $\G^\prime = (\V^\prime, \E^\prime)$, we say that \emph{$\G^\prime$ is a \textbf{subgraph} of $\G$} and that \emph{$\G$ is a \textbf{supergraph} of $\G^\prime$}, denoted as $\G^\prime \subseteq \G$ or $\G \supseteq \G^\prime$, if $\V^\prime \subseteq \V$ and $(i, j) \in \E^\prime_{\bullet}$ with $\bullet \in \{\tailhead, \, \headhead, \, \ohead, \, \oo\}$ implies $(i, j) \in \E_{\bullet}$. Given a directed partial mixed graph $\G = (\V, \E)$, its \emph{\textbf{induced subgraph} on $\V^\prime \subseteq \V$} is the graph $\G^\prime = (\V^\prime, \E^\prime)$ such that $(i, j) \in \E^\prime_{\bullet}$ with $\bullet \in \{\tailhead, \, \headhead, \, \ohead, \, \oo\}$ if and only if $i, j \in \V^\prime$ and $(i, j) \in \E_{\bullet}$.

We denote a directed edge $(i, j) \in \E_{\tailhead}$ as $i \tailhead j$ or $j \headtail i$ and say \emph{$i \tailhead j$ ($j \headtail i$) is in $\G$} if $(i, j) \in \E_{\tailhead}$; similarly for the other edge types. We view edges as composite objects of the symbols at their ends---the \textbf{edge marks}---which are \emph{tails}, \emph{heads}, or \emph{circles}. For example, $i \ohead j$ has a circle-mark at $i$ and a head mark at $j$, and $i \tailhead j$ has a tail mark at $i$. Tails and heads are \emph{non-circle marks} and \emph{unambiguous} orientations. Circle marks are \emph{ambiguous} orientations. The symbol `$\ast$' is a wildcard for all three marks. For example, $\asthead$ may be $\tailhead$, $\headhead$, or $\ohead$.

A \textbf{walk} in $\G$ is an ordered sequence $\pi = (i_1, i_2, \ldots, i_n)$ of vertices such that $i_k$ and $i_{k+1}$ are adjacent in $\G$ for all $k = 1, \dots, n-1$. The integer $n \geq 1$ is the \emph{length of $\pi$} and a vertex in this sequence is said to \emph{be on $\pi$}. A \textbf{path} is a walk on which every vertex occurs at most once. For a path $\pi = (i_1, i_2, \dots, i_n)$ the vertices $i_1$ and $i_n$ are the \emph{end-point vertices of $\pi$}, all other vertices on $\pi$ are the \emph{non end-point vertices of $\pi$}. We refer to $\pi$ as a \emph{path between $i_1$ and $i_n$} and graphically represent it by $i_1 \astast i_2 \astast \dots \astast i_n$ where $i_{k} \astast i_{k+1}$ is the unique edge between $i_k$ and $i_{k+1}$. Such a graphical representation can also specify a path. We say that \emph{$\pi$ is out of $i_1$} if $i_1 \tailhead i_2$ in $\G$ and that \emph{$\pi$ is into $i_1$} if $i_1 \headast i_2$ in $\G$; similarly for the other end-point vertex. For $1 \leq a < b \leq n$ we write $\pi(i_a, i_b)$ for the path $(i_a, i_{a+1}, \dots i_b)$ and $\pi(i_b, i_a)$ for the path $(i_b, i_{b-1}, \dots i_a)$. Both of these are \emph{subpaths of $\pi$}. Given walks $\pi_1 = (i_1, i_2, \dots, i_n)$ and $\pi_2 = (j_1, j_2, \dots, j_m)$ with $i_n = j_1$ we write $\pi_1 \oplus \pi_2$ for the walk $(i_1, i_2, \dots, i_n, j_2, \dots, j_m)$. A vertex $i_k$ on path $\pi$ is a \emph{collider on $\pi$} if it is a non end-point vertex of $\pi$ and $\pi(i_{k-1}, i_{k+1})$ is $i_{k-1} \asthead i_k \headast i_{k+1}$, else it is a \emph{non-collider on $\pi$}. If the vertices $i$ and $k$ are non-adjacent, then the path $i \astast j \astast k$ is an \emph{unshielded triple} and the path $i \asthead j \headast k$ an \emph{unshielded collider}. A \emph{path} of length $n=1$ is called \emph{trivial}. The path $\pi = (i_1, i_2, \dots, i_n)$ is a \textbf{directed path} if $i_k \tailhead i_{k+1}$ in $\G$ for all $1 \leq k \leq n-1$ or $i_k \headtail i_{k+1}$ in $\G$ for all $1 \leq k \leq n-1$. In the former case we speak of a \emph{directed path from $i_1$ to $i_n$}, in the latter case of a \emph{directed path from $i_n$ to $i_1$}.

If the edge $i \tailhead j$ is in $\G$, then $i$ is a \textbf{parent} of $j$ and $j$ is a \textbf{child} of $i$. The vertex $i$ is an \textbf{ancestor} of $j$ and $j$ is a \textbf{descendant} of $i$ if $i = j$ or if there is a directed path from $i$ to $j$. The set of parents and ancestors of a vertex $i$ in $\G$ are respectively denoted as $\pa(i, \G)$ and $\an(i, \G)$. We say \emph{vertex $i$ is an ancestor of a set $\mathbf{S}$ of vertices} and \emph{$\mathbf{S}$ is a descendant of $i$} if at least one element of $\mathbf{S}$ is a descendant of $i$. Similary, \emph{vertex $i$ is a descendant of a set $\mathbf{S}$ of vertices} and \emph{$\mathbf{S}$ is an ancestor of $i$} if at least one element of $\mathbf{S}$ is an ancestor of $i$.

A directed partial mixed graph $\G$ has a \emph{directed cycle} if there are distinct vertices $i$ and $j$ with $i \in \an(j, \G)$ and $j \in \an(i, \G)$. A \textbf{directed acyclic graph (DAG)} $\D$ is a directed graph without directed cycles. A directed partial mixed graph $\G$ has an \emph{almost directed cycle} if the edge $i \headhead j$ is in $\G$ and $i \in \an(j, \G)$. A \emph{directed ancestral graph} is a directed mixed graph without directed cycles and almost directed cycles. An \emph{inducing path between $i$ and $j$} is a path $\pi$ between $i$ and $j$ such that all non end-point vertices of $\pi$ are colliders on $\pi$ and ancestors of $i$ or $j$. A \textbf{directed maximal ancestral graph (DMAG)} $\M$ is a directed ancestral graph that has no inducing paths between non-adjacent vertices. Every DAG is a DMAG. \textbf{Directed partial ancestral graphs (DPAGs)} $\PAG$ are directed partial mixed graphs that represent Markov equivalence classes of DMAGs, see Def.~\ref{def:pags_background_knowledge} in Sec.~\ref{sec:background} for a formal definition.

\section{A class of causal graphical models for time series with latent confounders}\label{sec:class_graphical_models}
In this section we first formally specify the considered type of time series processes, see Sec.~\ref{sec:process}. We then explain how, if there are no unobserved variables, certain DAGs with an infinite number of vertices (Def.~\ref{def:tsDAG}) can model these processes as causal Bayesian networks, see Secs.~\ref{sec:ts-DAGs} and \ref{sec:ts-DAG_causal}. Importantly, Def.~\ref{def:ts_dmag_implied} in Sec.~\ref{sec:implied_DMAGs} introduces so-called \emph{time series DMAGs (ts-DMAGs)}. These graphs are projections of the infinite DAGs and represent the causal relationships and independencies among only a subset of observed variables at a finite number of regularly sampled or regularly subsampled observed time steps. Time series DMAGs constitute the novel class of causal graphical models which is the central topic of this paper.

\subsection{Structural vector autoregressive processes}\label{sec:process}
We consider multivariate time series $\{\mathbf{V}_t\}_{t \in \mathbb{Z}}$, where $\mathbf{V}_t = (V^1_t, \dots, V^{n_V}_t)$ with the component time series $V^i = \{V^i_t\}_{t \in \mathbb{Z}}$ for $1 \leq i \leq n_V$, that are generated by an \emph{acyclic structural vector autoregressive process with contemporaneous influences}, e.g.~\citet{malinsky2018causal}. That is to say, for all $t \in \mathbb{Z}$ (time index) and $1 \leq i \leq n_V$ (variable index) the value of $V^i_t$ is determined as 
\begin{equation}\label{eq:process}
V^i_t \coloneqq f^i(P\!A^i_t, \epsilon^i_t) \, ,
\end{equation}
where $f^i$ is a measurable function that depends on all its arguments, the random variables $\epsilon^i_t$ (so-called ``noise'' variables) are jointly independent (with respect to both indices) and have a distribution that may depend on $i$ but not on $t$, and $P\!A^{i}_t \subseteq \{V^k_{t-\tau} ~|~ 1 \leq k \leq n_V, \, 0 \leq \tau \leq \porder\} \setminus \{V^i_t\}$. Here, the \emph{order} $\porder$ of the process is the smallest integer for which the set inclusion in the previous sentence holds (for all $i$ and $t$). We demand that $0 \leq \porder < \infty$.

We allow contemporaneous causal influences (i.e., $V^k_{t-\tau} \in P\!A^{i}_t $ with $\tau = 0$). Further, for all $\Delta t \in \mathbb Z$ we assume the sets $P\!A^i_{t}$ and $P\!A^i_{t-\Delta t}$ to be consistent in the sense that $V^k_{t-\tau} \in P\!A^i_{t}$ if and only if $V^k_{t-\tau-\Delta t} \in P\!A^i_{t-\Delta t}$. Acyclicity means the system of equations is recursive. The attribute $\emph{structural}$ asserts that eq.~\eqref{eq:process} is a \emph{structural causal model (SCM)}, e.g.~\citet{bollen1989structural}, \citet{Pearl2009}, \citet{Peters2018}, which we indicate by the `$\coloneqq$' symbol. Because of this causal interpretation we refer to the variables $P\!A^i_t$ as \emph{causal parents} of $V^i_t$ and to the consistency of $P\!A^i_{t}$ and $P\!A^i_{t-\Delta t}$ as \emph{causal stationarity}. The restriction of $P\!A^i_t$ to variables $V^k_{t-\tau}$ with $\tau \geq 0$ ensures that there is no causal influence backward in time.

\subsection{Time series DAGs}\label{sec:ts-DAGs}
The causal parentships specified by an SCM are graphically represented by the SCM's \emph{causal graph}, e.g.~\citet{Spirtes2000}, \citet{Pearl2009}, \citet{Peters2018}. The causal graph is a directed graph with one vertex per variable, typically excluding the noise variables, and directed edges from each variable to all variables of which it is a causal parent. The same construction applies to structural processes as in eq.~\eqref{eq:process}. However, as we capture by the below three notions, the resulting ``temporal causal graphs'' carry more structure than their non-temporal counterparts.

First, the random variable $V^i_t$ corresponds to a particular time step $t$ of a particular component time series $V^i$. This correspondence is captured by the following notion.

\begin{mydef}[Time series structure]\label{def:time_structure}
A graph $\G = (\V, \E)$ has a \emph{time series structure} if $\V = \Iindex \times \Tindex$, where $\Iindex = \{1, 2, \dots, n\}$ with $n \geq 1$ is the \emph{variable index set} and $\Tindex = \{t \in \mathbb{Z} ~|~ t_s \leq t \leq t_e\}$ with $t_s \in \mathbb{Z} \cup \{-\infty\}$ and $t_e \in \mathbb{Z} \cup \{+\infty\}$ and $t_s \leq t_e$ is the \emph{time index set}.
\end{mydef}

We say that a vertex $(i, t) \in \V$ is \emph{at time $t$} and, if $t_a \leq t \leq t_b$, to be \emph{in the time window $[t_a, t_b]$}. We further say $(i, t_i) \in \V$ is \emph{before} $(j, t_j) \in \V$ and $(j, t_j) \in \V$ is \emph{after} $(i, t_i) \in \V$ if $t_i < t_j$. An edge $((i, t_i), (j, t_j)) \in \E$ has \emph{length} or \emph{lag} $|t_i - t_j|$. We call edges of length zero \emph{contemporaneous} and call all other edges \emph{lagged}.

Second, below eq.~\eqref{eq:process} we explicitly restricted the causal parents $P\!A^{i}_t$ to only contain vertices that are before or at time $t$. This restriction is captured by the following notion.

\begin{mydef}[Time order]\label{def:time_order}
A directed partial mixed graph $\G =  (\V, \E)$ with time series structure is \emph{time ordered} if $((i, t_i), (j, t_j)) \in \E_{\tailhead}$ implies $t_i \leq t_j$.
\end{mydef}

In a time ordered graph $\G$ the ancestral relationship $(i, t_i) \in \an((j, t_j), \G)$ implies $t_i \leq t_j$. This fact shows that also indirect causal influences are correctly restricted to not go backwards in time as soon as this restriction is imposed on direct causal influences.

Third, the property of causal stationarity (see Sec.~\ref{sec:process}) restricts the edges to be repetitive in time. This restriction is captured by the following notion.

\begin{mydef}[Repeating edges]\label{def:repeating_edges}
A directed partial mixed graph $\G =  (\V, \E)$ with time series structure has \emph{repeating edges} if the following holds: If $((i, t_i), (j, t_j)) \in \E_{\bullet}$ with $\bullet \in \{\tailhead, \headhead, \ohead, \oo\}$ and $(i, t_i + \Delta t), (j, t_j + \Delta t) \in \V$, then $((i, t_i + \Delta t), (j, t_j + \Delta t)) \in \E_{\bullet}$.
\end{mydef}

\begin{myremark}[on Def.~\ref{def:repeating_edges}]
Section~\ref{sec:class_graphical_models} is concerned with DAGs and DMAGs only. In these graphs there are by definition no edges of the types $\ohead$ or $\oo$. However, in Sec.~\ref{sec:equivalence_classes_and_causal_discovery} we will apply the concept of repeating edges also to DPAGs. Since these graphs (DPAGs) can contain edges $\ohead$ or $\oo$, we already here formulate Def.~\ref{def:repeating_edges} in sufficient generality.
\end{myremark}

By combining the three notions introduced in Defs.~\ref{def:time_structure}, \ref{def:time_order} and \ref{def:repeating_edges} we define the following class of graphical models, which plays an important role throughout the paper.

\begin{mydef}[Time series DAG]\label{def:tsDAG}
A \emph{time series DAG} (\emph{ts-DAG}) is a DAG $\D = (\V, \E)$ with time series structure $\V = \Iindex \times \Tindex$ with $\Tindex = \mathbb{Z}$ that is time ordered and has repeating edges.
\end{mydef}

Due to time order and repeating edges, a ts-DAG $\D$ is fully specified by its variable index set together with its edges that point to a vertex at time $t$. Hence, if the longest edge of $\D$ is of finite length $p_{\D} < \infty$ then one unambiguously specifies $\D$ by drawing all vertices within the time window $[t-p_{\D}, t]$ and the edges between these; see Fig.~\ref{fig:example_tsDAG} for an example. In slight abuse of notation we sometimes denote vertices by the random variable that they represent.

\begin{figure}[tb]
\centering
\includegraphics[width=0.7\linewidth, page = 1]{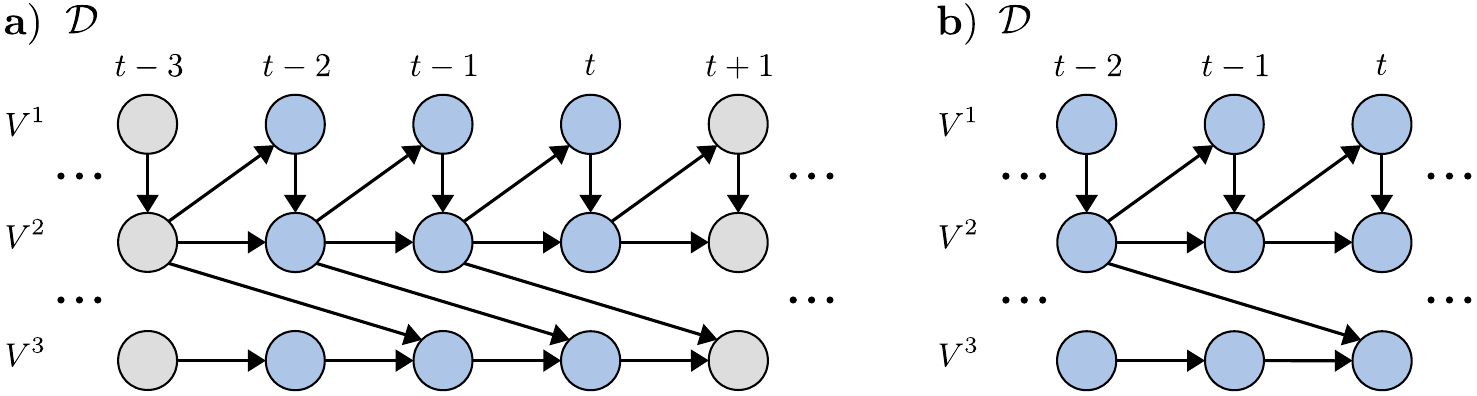}
\caption{
Two illustrations of the same ts-DAG $\D$ that represents a structural process as in eq.~\eqref{eq:process} of order $p_{\D} = \porder = 2$ with three component time series $V^1$, $V^2$, and $V^3$. Given the implicit assertion that there is no edge of length larger than those depicted, the ts-DAG is uniquely specified by showing a segment of $\porder +1$ successive time steps. The horizontal dots indicate that the structure is repeated into the infinite past and infinite future.}
\label{fig:example_tsDAG}
\end{figure}

\subsection{Time series DAGs as causal graphs for structural vector autoregressive processes}\label{sec:ts-DAG_causal}
Since the structural process in eq.~\eqref{eq:process} is acyclic by assumption, i.e., since the system of equations is recursive, its causal graph is acyclic (hence the terminology). In combination with the discussions in the previous subsection we thus get the following result.

\begin{mylemma}\label{lemma:why_ts_DAG}
The causal graph of an acyclic and causally stationary structural vector autoregressive process as in eq.~\eqref{eq:process} is a ts-DAG.
\end{mylemma}

This observation has been made before, for example in \citet{runge2012escaping} and \citet{Peters2013} where ts-DAGs have respectively been called \emph{time series graphs} and \emph{full time graphs}. We note that, because of time order, the assumption of acyclicity restricts only the ts-DAG's contemporaneous edges.

In the non-time series setting, an acyclic SCM defines a unique distribution over the SCM's variables (pushforward of the noise distribution by the structural assignments). According to the \emph{causal Markov condition}, see \citet{Spirtes2000}, the SCM's causal graph is a Bayesian network for this so-called \emph{entailed distribution} \citep{Pearl2009}, which in turn implies that $d$-separations (see \citet{Pearl1988}, denoted by `$\ci$') in the causal graph imply the corresponding independencies in the distribution \citep{VERMA199069, geiger1990identifying}. The \emph{causal faithfulness condition}, see \citet{Spirtes2000}, assumes the reverse implication, i.e., that all independencies imply the corresponding $d$-separations. Then, $d$-separations and independencies are in one-to-one correspondence.

Although acyclic, the time series setting specified by eq.~\eqref{eq:process} is more complicated: Since time is indexed by $t \in \mathbb Z$ (as opposed to, e.g., $t \in \mathbb N$), there is no initial ``starting'' distribution that can be pushforwarded to explicitly define a unique entailed distribution. Instead, we need to ask whether eq.~\eqref{eq:process} \emph{implicitly} defines a distribution; and if yes, how many. Following the terminology in \citet{bongers2018causal}, this question asks for \emph{solutions} to eq.~\eqref{eq:process}, that is, for stochastic processes which satisfy eq.~\eqref{eq:process} almost surely. The existence of such solutions as well as their uniqueness (up to almost sure equality) and properties are non-trivial and not considered here. Rather, for the purpose of this paper we \emph{assume} that eq.~\eqref{eq:process} is solved by a (not necessarily unique) strictly stationary stochastic process whose finite-dimensional distributions satisfy the causal Markov and causal faithfulness condition with respect to its ts-DAG. This assumption is common in the literature, cf.~\citet{Entner2010}, \citet{malinsky2018causal}, \citet{LPCMCI}, and is here only needed for the connection to causality. The results of the present paper are, technically, about marginalizing the independence (i.e., $d$-separation) models of ts-DAGs and remain valid also without that additional assumption. The issue of solving eq.~\eqref{eq:process} is an important aspect to consider in future work.

\subsection{Time series DMAGs}\label{sec:implied_DMAGs}
In most real-world scenarios, unobserved common causes cannot be excluded. As mentioned in Sec.~\ref{sec:introduction} for the non-time series setting, directed maximal ancestral graphs (DMAGs) are often used for causal modeling in the presence of unobserved variables. This use of DMAGs as causal graphical models was pioneered in \citet{richardson2002}, which defines a \emph{marginalization / projection} procedure that from a DAG $\G$ over vertices $\V$, of which only a subset $\Ovar \subseteq \V$ is observed, constructs a DMAG $\M_{\Ovar}(\D)$ over the observed variables $\Ovar$ only (see also \citet{zhang2008causal}). The projection of $\D$ to $\M_{\Ovar}(\D)$ has two properties: First, both graphs have the same ancestral relationships among vertices in $\Ovar$. Second, $d$-separations in $\D$ among vertices in $\Ovar$ are in one-to-one correspondence to the similar concept of $m$-separation in $\M_{\Ovar}(\D)$ (also denoted by `$\ci$'). These two properties ensure that if $\D$ is a causal graph then also $\M_{\Ovar}(\D)$ carries causal meaning and can be used for causal reasoning as explained in \citet{zhang2008causal}.

Below, we generalize the construction of such ``causal'' DMAGs to the time series setting. To begin, we first note that for time series there are two types of unobserved variables:
\begin{itemize}
\item \textbf{Unobservable variables}:
Some component time series $L^1, \dots, L^{n_L}$ with $0 \leq n_L < n_V$ may be unobserved entirely by the experimental setup. We call these $L^i$ \emph{unobservable} and call the other component time series $O^1, \dots, O^{n_O}$ with $n_O = n_V - n_l \leq \infty$ \emph{observable}. The variable index set $\Iindex$ of the ts-DAG $\D$ accordingly decomposes as $\Iindex = \IindexO \,\dot\cup\, \IindexL$. This first type of unobserved variables is similar to the case of unobserved variables in the non-time series setting.
\item \textbf{Temporally unobserved variables}:
In addition, throughout the paper we will treat only a finite number of time steps $\TindexO$ as observed. This construction is specific to the time series setting and means that at times $\Tindex \setminus \TindexO$ also the observable time series are treated as unobserved. The rational for doing so is that in practice only finitely many observations are available and hence one can only reason about DMAGs of finite temporal extension.
\end{itemize}
Throughout the paper we restrict the set $\TindexO$ of observed time steps to take one of the following two forms:
\begin{itemize}
\item \textbf{Regular sampling}: All time steps within a time interval $[t-\taumax, t]$ for some non-negative integer $\taumax < \infty$ and a reference time step $t$ are observed, i.e., $\TindexO = \{t-\tau ~|~ 0 \leq \tau \leq \taumax\}$.
\item \textbf{Regular subsampling}: Every $n$-th time step, for $n \geq 2$ an integer, within $[t-\taumax, t]$ with $\taumax < \infty$ is observed, i.e., $\TindexO = \{t-\tau ~|~ 0 \leq \tau \leq \taumax , \, \tau \!\! \mod n = 0\}$.
\end{itemize}
The time window length $\taumax$ is not restricted relative to the order $\porder$ of the data-generating process, i.e., we allow all of $\taumax < \porder$ and $\taumax = \porder$ and $\taumax > \porder$. The reference time step $t$ is arbitrary since the ts-DAG $\D$ has repeating edges. We are led to the following definition.

\begin{mydef}[Time series DMAG]\label{def:ts_dmag_implied}
Let $\D = (\V, \E)$ be a ts-DAG with variable index set $\Iindex$, let $\IindexO \subseteq\Iindex$, and let $\TindexO \subsetneq \mathbb{Z}$ be regularly sampled or regularly subsampled. The \emph{time series DMAG implied by $\D$ over $\Ovar = \IindexO \times \TindexO$}, denoted as $\M_{\Ovar}(\D)$ or $\M_{\IindexO \times \TindexO}(\D)$ and also referred to as a \emph{ts-DMAG}, is the DMAG on the vertex set $\Ovar$ that is obtained by applying the MAG latent projection defined in \citet[pp.~1442--3]{zhang2008causal} to $\D$ with $\La = \V \setminus \Ovar$ being the set of latent vertices.
\end{mydef}

Figure \ref{fig:example_tsDMAG} illustrates the construction of ts-DMAGs as projections of ts-DAGs. We stress that \emph{all} vertices prior to the observed time window (i.e., before time $t-\taumax$) are treated as unobserved, even if they are observable and hence would be observed for a larger value of $\taumax$.

\begin{figure}[tb]
\centering
\includegraphics[width=0.8\linewidth, page = 1]{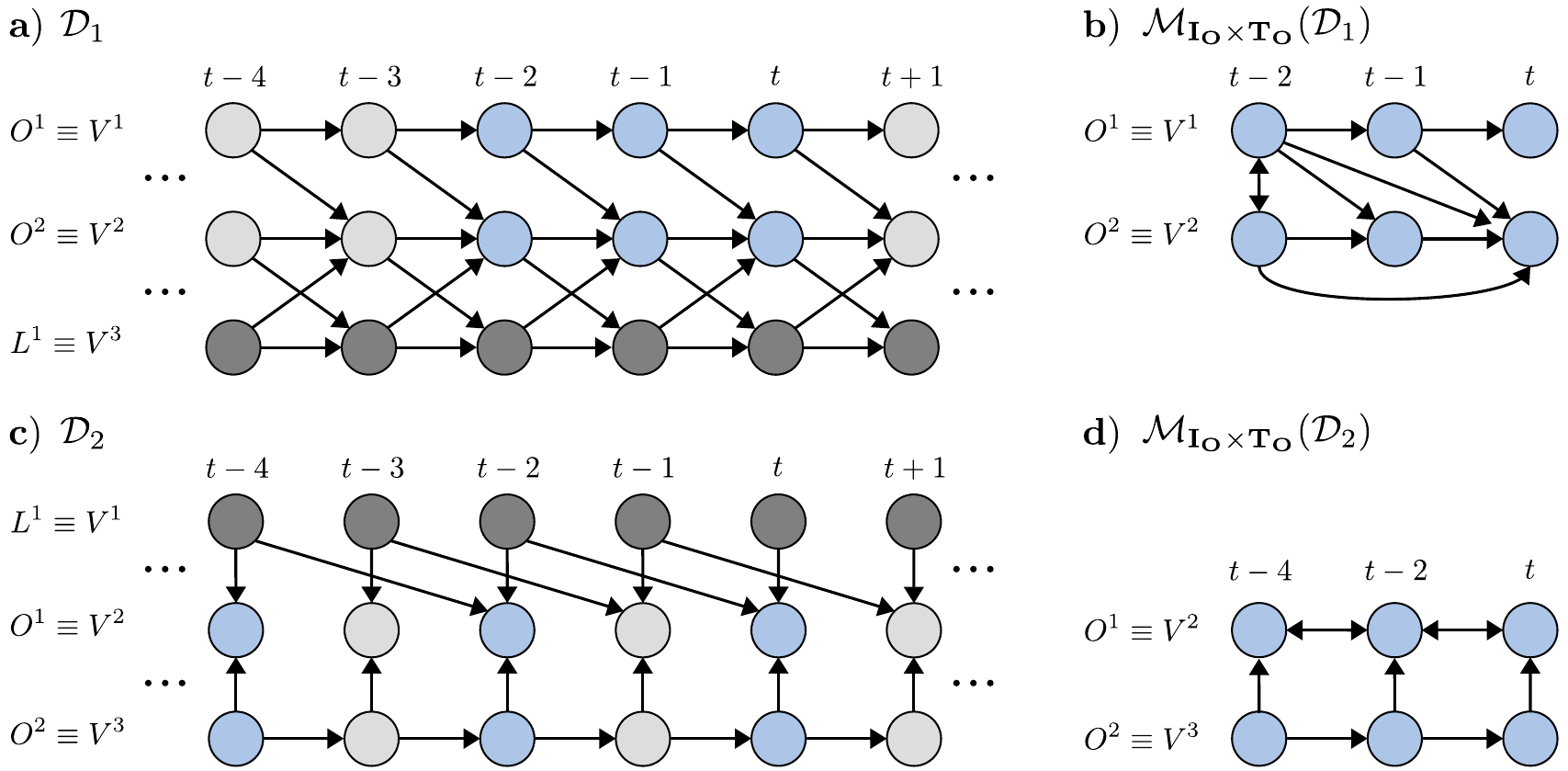}
\caption{The ts-DAG $\D_1$ in part a) implies the ts-DMAG $\M_{\IindexO \times \TindexO}(\D_1)$ in part b) for $\IindexO = \{1, 2\}$ and $\TindexO = \{t-2, t-1, t\}$ (regular sampling). The ts-DAG $\D_2$ in part c) implies the ts-DMAG $\M_{\IindexO \times \TindexO}(\D_2)$ in part d) for $\IindexO = \{2, 3\}$ and $\TindexO = \{t-4, t-2, t\}$ (regular subsampling). Color coding: Observed vertices are light blue, unobservable vertices are dark gray, temporally unobserved observable vertices are light gray.}
\label{fig:example_tsDMAG}
\end{figure}

\begin{myremark}[on Def.~\ref{def:ts_dmag_implied}]
The time series DMAG $\M_{\Ovar}(\D)$ is defined as the MAG latent projection of an \emph{infinite} object, namely of the ts-DAG $\D$. An implementation of this projection in a procedure that always terminates in finite time is possible but non-trivial. Such a procedure will be discussed in \citet{gerhardus_OracleCI}. For the present paper, however, this procedure is not needed because all theoretical results and examples either do not require the explicit construction of ts-DMAGs or one can carry out the required projections by hand.
\end{myremark}

Time series DMAGs are the central objects of interest in this paper and a significant part of the paper deals with deriving their properties. We will see that the repeating edges property of ts-DAGs $\D$ plays an essential role in this regard. As a first step, the following lemma notes which of the defining properties of ts-DAGs carry over to ts-DMAGs.

\begin{mylemma}\label{lemma:DMAG_has_time_series_structure_and_time_order}
Let $\M_{\Ovar}(\D)$ be a ts-DMAG. Then:
\begin{enumerate}
  \item $\M_{\Ovar}(\D)$ has a time series structure.
  \item $\M_{\Ovar}(\D)$ is time ordered.
  \item There are cases in which $\M_{\Ovar}(\D)$ does not have repeating edges.
\end{enumerate}
\end{mylemma}

While according to part 1 of Lemma~\ref{lemma:DMAG_has_time_series_structure_and_time_order} every ts-DMAG is a DMAG with time series structure, part 2~implies that the reverse is not true. Namely, DMAGs with time series structure that are not time ordered cannot be ts-DMAGs. We thus see that ts-DMAGs are a proper subclass of DMAGs with time series structure. The following example shows that ts-DMAGs $\M_{\Ovar}(\D)$ do not in general have repeating edges.

\begin{myexample}
The ts-DMAG in part b) of Fig.~\ref{fig:example_tsDMAG} does not have repeating edges because there is the edge $O^1_{t-2} \headhead O^2_{t-2}$ although $O^1_{t-1}$ and $O^2_{t-1}$ (and $O^1_{t}$ and $O^2_{t}$) are non-adjacent.
\end{myexample}

Despite this fact, the repeating edges property of the ts-DAG $\D$ strongly restricts the connectivity pattern of the ts-DMAG $\M_{\Ovar}(\D)$. We will work out these restrictions in Sec.~\ref{sec:characterization_full_section}.

\section{Characterization of ts-DMAGs}\label{sec:characterization_full_section}
The main goal of this section is to characterize the space of ts-DMAGs, i.e., to find conditions that specify exactly which DMAGs with time series structure are ts-DMAGs. Theorem~\ref{thm:complete_characterization} in Sec.~\ref{sec:complete_characterization_subsection} achieves this goal by providing a single condition that is both necessary and sufficient. The theorem uses the notion of \emph{canonical ts-DAGs}, see Def.~\ref{def:canonical_DAG} in Sec.~\ref{sec:canonical_ts-DAGs}. In Sec.~\ref{sec:stationarification} we introduce \emph{stationarified ts-DMAGs} and, more generally, the concept of \emph{stationarification}. This concept simplifies the definition of canonical ts-DAGs and is useful to describe the output of two recent time series causal discovery algorithms (see Sec.~\ref{sec:existing_algorithms}). In Sec.~\ref{sec:previous_model_classes} we show that ts-DMAGs constitute a strict subset of the classes of graphical models that have so far been used in the literature for describing time lag specific causal relationships and independencies in time series with latent confounders. Section~\ref{sec:properties_ts_DMAGs} discusses several properties that ts-DMAGs necessarily have, but which can also be obeyed by DMAGs that are not ts-DMAGs. These properties are useful for the discussions in Secs.~\ref{sec:previous_model_classes} and \ref{sec:equivalence_classes_and_causal_discovery}. In Sec.~\ref{sec:both_sampling_equivalent} we show that regular sampling and regular subsampling are equivalent from a graphical point of view. Section~\ref{sec:characterization_stat_DMAG} gives a characterization of the space of stationarified ts-DMAGs. At first, however, we spell out the motivation for the analysis.

\subsection{Motivation}\label{sec:motivation_for_characterization}
When using a class of graphs to represent causal knowledge, it is desirable to know which graphs belong to this class and which do not. Otherwise, it is impossible to fully characterize which causal claims a given graph of that class conveys. Another, a posteriori motivation has been mentioned in the previous paragraph: In Sec.~\ref{sec:previous_model_classes} we will see that ts-DMAGs are a strict subset of the previously employed model classes. Thus, when using ts-DMAGs as targets of inference in causal discovery or to reason about causal effects, it is, respectively, possible to learn more qualitative causal relationships (see Secs.~\ref{sec:tsDPAGs} and \ref{sec:existing_algorithms} for an in-depth discussion) and to identify more causal effects (see Example~\ref{ex:tsDAG_better_than_tora}) from data without having imposed any additional assumption or restriction.

\subsection{Equivalence of regular subsampling and regular sampling}\label{sec:both_sampling_equivalent}
In Sec.~\ref{sec:implied_DMAGs} we restricted the set of observed time steps $\TindexO$ to regular sampling or regular subsampling. While different at first sight, these two cases are equivalent in the following sense.

\begin{mylemma}\label{lemma:both_sampling_equivalent}
Let $\D$ be a ts-DAG and $1\leq n_{\text{steps}} \in \mathbb{N}$. For $1 \leq n \in \mathbb Z$ define the set $\TindexO^{n} = \{t - m \cdot n ~|~ 0 \leq m \leq n_{\text{steps}}- 1 \}$. Then, with equality up to relabeling vertices:
\begin{enumerate}
  \item For every $n > 1$ there is a ts-DAG $\D^{\prime}$ such that $\M_{\IindexO \times \TindexO^{n}}(\D) = \M_{\IindexO \times \TindexO^{1}}(\D^\prime)$.
  \item For every $n > 1$ there is a ts-DAG $\D^{\prime}$ such that $\M_{\IindexO \times \TindexO^{1}}(\D) = \M_{\IindexO \times \TindexO^{n}}(\D^\prime)$.
\end{enumerate}
\end{mylemma}
Lemma~\ref{lemma:both_sampling_equivalent} implies: Every property that ts-DMAGs necessarily have in case of regular sampling is also necessarily obeyed in case of regular subsampling (part 1) and vice versa (part 2). Moreover, every set of additional properties that, when imposed on a DMAG $\M$ with time series structure, is sufficient for $\M$ to be a ts-DMAG in case of regular sampling is also sufficient in case of regular subsampling (part 2) and vice versa (part 1).

Due to this equivalence we from here on restrict to regular sampling, without losing generality, and write $\Mtaumax(\D)$ for $\M_{\Ovar}(\D)$ where $\Ovar = \IindexO \times \TindexO$ and $\TindexO = \{t - \tau ~|~ 0 \leq \tau \leq \taumax \}$.

\subsection{Properties of ts-DMAGs}\label{sec:properties_ts_DMAGs}
In this subsection, we discuss several properties that ts-DMAGs $\Mtaumax(\D)$ necessarily have. These properties are such that a certain graphical property persists when the involved vertices are shifted in time. We use the following definitions.

\begin{mydef}[Time-shift persistent graphical properties]\label{def:time_shift_persistent}
A partial mixed graph $\G = (\V, \E) $ with time series structure has$\ldots$
\begin{enumerate}
  \item[1.] $\ldots$ \emph{repeating adjacencies} if the following holds: If $((i, t_i), (j, t_j)) \in \E$ and $(i, t_i + \Delta t), (j, t_j + \Delta t) \in \V$ then $((i, t_i + \Delta t), (j, t_j + \Delta t)) \in \E$.
  \item[2.] $\ldots$ \emph{past-repeating adjacencies} if the following holds: If $((i, t_i), (j, t_j)) \in \E$ and $(i, t_i + \Delta t), (j, t_j + \Delta t)\in \V$ with $\Delta t < 0$ then $((i, t_i + \Delta t), (j, t_j + \Delta t)) \in \E$.
  \item[3.] $\ldots$ \emph{repeating orientations} if the following holds: If $((i, t_i), (j, t_j)) \in \E_{\bullet}$ with $\bullet \in \{\tailhead, \headhead, \ohead,$ $\oo\}$ and $((i, t_i + \Delta t), (j, t_j + \Delta t)) \in \E$ then $((i, t_i + \Delta t), (j, t_j + \Delta t)) \in \E_{\bullet}$.
\end{enumerate}
A DMAG $\M = (\V, \E) $ with time series structure has
\begin{enumerate}
  \item[4.] $\ldots$ \emph{repeating ancestral relationships} if the following holds: If $(i, t_i) \in \an((j, t_j), \M)$ and $(i, t_i + \Delta t), (j, t_j + \Delta t) \in \V$ then $(i, t_i + \Delta t) \in \an((j, t_j + \Delta t), \M)$.
  \item[5.] $\ldots$ \emph{repeating separating sets} if the following holds: If $(i, t_i) \ci (j, t_j) ~|~ \mathbf{S}$ and $\{(i, t_i + \Delta t), (j, t_j + \Delta t)\} \cup \mathbf{S}_{\Delta t} \subseteq \V$, where $\mathbf{S}_{\Delta t}$ is obtained by shifting every vertex in $\mathbf{S}$ by $\Delta t$ time steps, then $(i, t_i+ \Delta t) \ci (j, t_j + \Delta t) ~|~ \mathbf{S}_{\Delta t}$.
\end{enumerate}
\end{mydef}

\begin{myremark}[on Def.~\ref{def:time_shift_persistent}]
Section~\ref{sec:characterization_full_section} is concerned with DAGs and DMAGs only. However, in Sec.~\ref{sec:equivalence_classes_and_causal_discovery} we will apply the concepts of repeating adjacencies, past-repeating adjacencies and repeating orientations also to DPAGs (which are a special case of partial mixed graphs). Hence, we already here formulate the definition in sufficient generality.
\end{myremark}

\begin{figure}[tb]
\centering
\includegraphics[width=0.94\linewidth, page = 1]{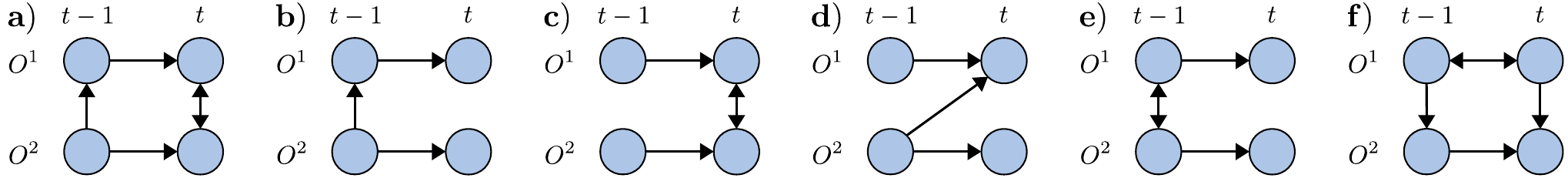}
\caption{Examples of time ordered DMAGs with time series structure for illustrating the properties from Def.~\ref{def:time_shift_persistent} and the repeating edges property from Def.~\ref{def:repeating_edges}. In each case we state which of these properties apply.
\textbf{a)} Repeating adjacencies, repeating separating sets, past-repeating adjacencies.
\textbf{b)} Repeating orientations, repeating separating sets, past-repeating adjacencies.
\textbf{c)} Repeating orientations, repeating ancestral relationships.
\textbf{d)} All but repeating separating sets.
\textbf{e)} All but repeating repeating edges and repeating adjacencies.
\textbf{f)} All.}
\label{fig:repeating_properties}
\end{figure}

Figure~\ref{fig:repeating_properties} illustrates the five properties introduced by Def.~\ref{def:time_shift_persistent} as well as their distinctions. Below we will make frequent use of the implications expressed by the following lemma.

\begin{mylemma}\label{lemma:interrelation_repeating_properties}
\begin{enumerate}
  \item Repeating edges is equivalent to the combination of repeating adjacencies and repeating orientations.
  \item Repeating adjacencies implies past-repeating adjacencies.
  \item Repeating ancestral relationships implies repeating orientations.
  \item In graphs with time index set $\Tindex = \mathbb{Z}$, repeating edges implies repeating ancestral relationships and repeating separating sets.
\end{enumerate}
\end{mylemma}

These implications further show that the combination of repeating adjacencies and repeating ancestral relationships implies repeating edges. Importantly, repeating orientations does \emph{not} imply repeating ancestral relationships, see part b) of Fig.~\ref{fig:repeating_properties} for an example.

Since ts-DAGs have repeating edges, according to Lemma~\ref{lemma:interrelation_repeating_properties} they in fact also have all five properties given in Def.~\ref{def:time_shift_persistent}. How about ts-DMAGs? While these in general \emph{do not} inherit repeating edges from the underlying ts-DAG, see part 3~of Lemma~\ref{lemma:DMAG_has_time_series_structure_and_time_order}, the following lemma shows that ts-DMAGs do feature some of the weaker time-shift persistent properties.

\begin{mylemma}\label{lemma:necessary_properites_time_series_DMAG}
\begin{enumerate}
  \item Time series DMAGs $\Mtaumax(\D)$ have repeating ancestral relationships.
  \item Time series DMAGs $\Mtaumax(\D)$ have repeating orientations.
  \item Time series DMAGs $\Mtaumax(\D)$ have repeating separating sets.
  \item Time series DMAGs $\Mtaumax(\D)$ have past-repeating adjacencies.
  \item There are cases in which a ts-DMAG $\Mtaumax(\D)$ does not have repeating adjacencies.
\end{enumerate}
\end{mylemma}

The ts-DMAGs in parts b) and d) of Fig.~\ref{fig:example_tsDMAG} indeed satisfy the properties asserted by parts 1 through 4 of Lemma~\ref{lemma:necessary_properites_time_series_DMAG}. Moreover, part 5 of Lemma~\ref{lemma:necessary_properites_time_series_DMAG} clarifies why ts-DMAGs may fail to have repeating edges: They do not necessarily have repeating adjacencies but only the weaker property of past-repeating adjacencies. The following example illustrates this fact.

\begin{myexample}\label{example:motivation-for-stationarification}
Consider the ts-DAG $\D_1$ in part a) of Fig.~\ref{fig:example_tsDMAG}. In this graph the $d$-separation $O^1_{t + \Delta t} \ci O^2_{t + \Delta t} ~|~ O^1_{t + \Delta t- 1}$ holds for all $\Delta t \in \mathbb Z$. Hence, the vertices $O^1_t$ and $O^2_t$ (and, similarly, $O^1_{t-1}$ and $O^2_{t-1}$) are non-adjacent in the ts-DMAG $\M^{2}(\D_1)$ in part b) of the figure. However, since $O^1_{t-3}$ is temporally unobserved and the $d$-separation $O^1_{t + \Delta t} \ci O^2_{t + \Delta t} ~|~ \mathbf{S}$ requires that $O^1_{t + \Delta t - 1} \in \mathbf{S}$, the vertices $O^1_{t-2}$ and $O^1_{t-2}$ are adjacent in $\M^{2}(\D_1)$.
\end{myexample}

That ts-DMAGs have repeating orientations and repeating separating sets has already been found and used in \citet{Entner2010}.

\subsection{Stationarified ts-DMAGs}\label{sec:stationarification}
Example~\ref{example:motivation-for-stationarification} shows that in a ts-DMAG $\Mtaumax(\D)$ there may be an edge $(i, t_i + \Delta t) \astast (j, t_j + \Delta t)$ with $\Delta t < 0$ even if the vertices $(i, t_i)$ and $(j, t_j)$ are non-adjacent in $\Mtaumax(\D)$. This is the case even though one then knows that $(i, t_i + \Delta t)$ and $(j, t_j + \Delta t)$ \emph{can} be $d$-separated in underlying ts-DAG $\D$, just not by a set of vertices that is within the observed time window. One might thus view such an edge $(i, t_i + \Delta t) \astast (j, t_j + \Delta t)$ in $\Mtaumax(\D)$ as an artifact of the chosen time window and hence prefer to manually remove the edge by subjecting the ts-DMAG to the following operation.

\begin{mydef}[Stationarification]\label{def:stationarification}
Let $\G = (\V, \E)$ be a directed partial mixed graph with time series structure. The \emph{stationarification} of $\G$, denoted as $\stat(\G)$, is the graph $\stat(\G) = (\V^\prime, \E^\prime)$ defined as follows:
\begin{enumerate}
	\item It has the same set of vertices as $\G$, i.e., $\V^\prime = \V$.
	\item There is an edge $((i, t_i), (j, t_j)) \in \E^\prime_{\bullet}$ with $\bullet \in \{\tailhead, \, \headhead, \, \ohead, \, \oo\}$ if and only if $((i, t_i + \Delta t), (j, t_j + \Delta t)) \in \E_{\bullet}$ in $\G$ for all $\Delta t$ with $(i, t_i + \Delta t), (j, t_j + \Delta t) \in \V$.
\end{enumerate}
\end{mydef}

\begin{myremark}[on Def.~\ref{def:stationarification}]
Section~\ref{sec:characterization_full_section} is concerned with DAGs and DMAGs only. In these graphs there are by definition no edges of the types $\ohead$ or $\oo$. However, in Sec.~\ref{sec:equivalence_classes_and_causal_discovery} we will apply the concept of sationarification also to DPAGs. Since these graphs (DPAGs) can contain edges $\ohead$ or $\oo$, we already here formulate the definition in sufficient generality.
\end{myremark}

To see that the process of stationarification indeed achieves what it is supposed to do, consider the ts-DMAG $\mathcal{M}_1$ in part a) of Fig.~\ref{fig:stationarifications}. In this graph there is the edge $O^1_{t-2} \headhead O^2_{t-2} \in \E_{\headhead}$ while the vertices $O^1_{t-1}$ and $O^2_{t-1}$ (and, similarly, $O^1_{t}$ and $O^2_t$) are non-adjacent. According to part 2 of Def.~\ref{def:stationarification} (note the ``for all $\Delta t$''), the vertices $O^1_{t-2}$ and $O^2_{t-2}$ are therefore non-adjacent in the stationarification $\stat(\mathcal{M}_1)$ of $\mathcal{M}_1$ as shown in part b) of Fig.~\ref{fig:stationarifications}.

\begin{figure}[tb]
\centering
\includegraphics[width=0.8\linewidth, page = 1]{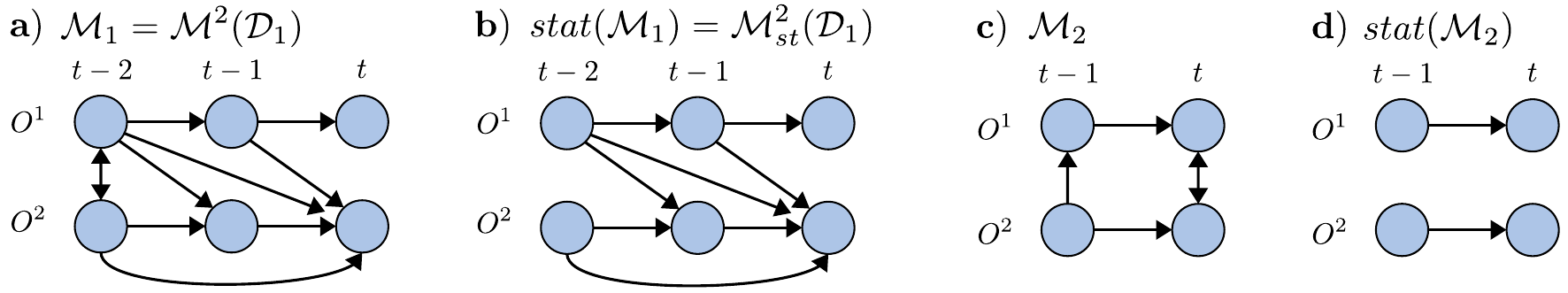}
\caption{A ts-DMAG $\M_1 = \M^2(\D_1)$ (the same as in part b) of Fig.~\ref{fig:example_tsDMAG}) and a DMAG with time series structure $\M_2$ (the same as in part a) of Fig.~\ref{fig:repeating_properties}) together with their stationarifications. Note that although $\M_2$ has repeating adjacencies its contemporaneous edges are not in $\stat(\M_2)$ because these edges do not have the same orientation.}
\label{fig:stationarifications}
\end{figure}

Stationarification removes an edge $(i, t_i + \Delta t) \astast (j, t_j + \Delta t)$ also if $(i, t_i)$ and $(j, t_j)$ are adjacent but if the edges $(i, t_i) \astast (j, t_j)$ and $(i, t_i + \Delta t) \astast (j, t_j + \Delta t)$ have different orientations (note the ``$\bullet$'' subscripts on $\E^\prime_{\bullet}$ and $\E_{\bullet}$ in part 2 of Def.~\ref{def:stationarification}). This effect, illustrated by parts c) and d) of Fig.~\ref{fig:stationarifications}, ensures that $\stat(\G)$ is the unique largest subgraph of $\G$ with repeating edges. For graphs with repeating orientations (as, e.g., ts-DMAGs) this effect does not occur and stationarification only concerns adjacencies (as, e.g., in parts a) and b) of Fig.~\ref{fig:stationarifications}).

Since ts-DMAGs $\Mtaumax(\D)$ have repeating orientation and past-repeating adjacencies, their stationarifications $\stat(\Mtaumax(\D))$ can be characterized with the following simpler condition.

\begin{mylemma}\label{lemma:stat_of_implied_DMAG}
The stationarification $\stat(\Mtaumax(\D))$ of a ts-DMAG $\Mtaumax(\D)$ is the unique subgraph of $\Mtaumax(\D)$ in which the vertices $(i, t_j - \tau)$ and $(j, t_j)$ with $\tau \geq 0$ are adjacent if and only if the vertices $(i, t - \tau)$ and $(j, t)$ are adjacent in $\Mtaumax(\D)$.
\end{mylemma}

Because the stationarification $\stat(\G)$ is a subgraph of $\G$, a time series structure and time order naturally carry over from $\G$ to $\stat(\G)$. Moreover, we can prove the following.

\begin{mylemma}\label{lemma:stat_of_implied_DMAG_is_DMAG}
The stationarification $\stat(\Mtaumax(\D))$ of a ts-DMAG $\Mtaumax(\D)$ is a DMAG.
\end{mylemma}

We thus refer to $\stat(\Mtaumax(\D))$ as a \emph{stationarified ts-DMAG} and abbreviate $\stat(\Mtaumax(\D))$ as $\Mtaumaxstat(\D)$. However, as the following example shows, a stationarified ts-DMAG $\Mtaumaxstat(\D)$ may not be the MAG latent projection of any ts-DAG, i.e., may not be a ts-DMAG.

\begin{myexample}
The stationarified ts-DMAG $\M^2_{\statsubscript}(\D_1)$ in part b) of Fig.~\ref{fig:stationarifications} implies the $d$-separation $O^1_{t-2} \ci O^2_{t-2}$ and the $d$-connections $O^1_{t-1} {\cancel{\,\,\ci}\,} O^2_{t-1}$ and $O^1_{t} {\cancel{\,\,\ci}\,} O^2_{t}$. The graph $\M^2_{\statsubscript}(\D_1)$ does thus not have repeating separating sets and can, by means of Lemma~\ref{lemma:necessary_properites_time_series_DMAG}, not be a ts-DMAG. Also note that in the underlying ts-DAG $\D_1$, shown in part a) of Fig.~\ref{fig:example_tsDMAG}, the $d$-connection $O^1_{t-2} {\cancel{\,\,\ci}\,} O^2_{t-2}$ holds. From this observation we learn that $(i, t_i) \ci (j, t_j) ~|~ \mathbf{S}$ in $\Mtaumaxstat(\D)$ does not necessarily imply $(i, t_i) \ci (j, t_j) ~|~ \mathbf{S}$ in $\D$.
\end{myexample}

The vertices $(i, t - \tau_i)$ and $(j, t - \tau_j)$ with $0 \leq \tau_j \leq \tau_i \leq \taumax$ are adjacent in a stationarified ts-DMAG $\Mtaumaxstat(\D)$ if and only if they can not be $d$-separated by any set of observable vertices within $[t - \taumax - \tau_j, t]$ in the underlying ts-DAG $\D$ (instead of $[t - \taumax, t]$, which is what a ts-DMAG would assert). The orientation of edges, however, retains the standard meaning: Tail and head marks respectively convey (non-)ancestorship according to the ts-DAG $\D$. The following lemma says that stationarification does not change ancestral relationships.

\begin{mylemma}\label{lemma:stat_DMAG_same_anc}
The ts-DMAG $\Mtaumax(\D)$ and its stationarification $\Mtaumaxstat(\D)$ agree on ancestral relationships, i.e., $(i, t_i) \in \an((j, t_j), \Mtaumax(\D))$ if and only if $(i, t_i) \in \an((j, t_j), \Mtaumaxstat(\D))$.
\end{mylemma}

Since $\Mtaumax(\D)$ and $\D$ by construction of the MAG latent projection agree on ancestral relationships, Lemma~\ref{lemma:stat_DMAG_same_anc} implies that also the stationarified ts-DMAG $\Mtaumaxstat(\D)$ agrees with the ancestral relationships of $\D$. Thus, $\Mtaumaxstat(\D)$ has repeating ancestral relationships.

In summary, edges in the ts-DMAG $\Mtaumax(\D)$ that are not also in $\Mtaumaxstat(\D)$ are due to marginalizing over observable vertices before $t-\taumax$. Such edges disappear when $\taumax$ is sufficiently increased, see also \citet[Sec.~B.8]{gerhardus2022_supplement}. However, as we will show in Sec.~\ref{sec:why_DPAG_of_tsDPAG}, these additional edges play a useful role in causal discovery. In Sec.~\ref{sec:existing_algorithms} we will further use the concept of stationarification to describe the SVAR-FCI causal discovery algorithm from \citet{malinsky2018causal} and the LPCMCI causal discovery algorithm from \citet{LPCMCI}.

\subsection{Canonical ts-DAGs}\label{sec:canonical_ts-DAGs}
In the current subsection we return to the goal of characterizing the space of ts-DMAGs. To this end, we first recall the concept of \emph{canonical DAGs}.

\begin{mydef}[Canonical DAG. From Sec.~6.1 of \citet{richardson2002}, specialized to the case of \emph{directed} ancestral graphs]\label{def:canonical_DAG}
Let $\G = (\V, \E)$ be a directed ancestral graph. The canonical DAG $\Dc(\G)$ of $\G$ is the graph $\Dc(\G) = (\V^{ca}, \E^{ca})$ defined as follows:
\begin{enumerate}
	\item Its vertex set is $\V^{ca} = \V \cup \La$, where $\La = \{l_{ij} ~|~ (i, j) \in \E_{\headhead}\}$.
	\item Its edge set $\E^{ca} = \E^{ca}_{\rightarrow}$ consists of the edges
	\begin{itemize}
		\item $i \tailhead j$ for all $(i, j) \in \E_{\tailhead}$ and
		\item $l_{ij} \tailhead i$ for all $l_{ij} \in \La$ and
		\item $l_{ij} \tailhead j$ for all $l_{ij} \in \La$.
	\end{itemize}
\end{enumerate}
\end{mydef}

Intuitively, the canonical DAG $\Dc(\G)$ of a directed ancestral graph $\G$ is obtained by replacing each bidirected edge $i \headhead j$ in $\G$ with $i \headtail l_{ij} \tailhead j$ where $l_{ij}$ is an additionally inserted, unobserved vertex. The canonical DAG $\Dc(\G)$ is a DAG and has the convenient property that there are no edges pointing into unobserved vertices and hence that there are also no edges between two unobserved vertices. Despite this simple structure of unobserved vertices, the following result shows that canonical DAGs are expressive enough to generate all DMAGs. 

\begin{mylemma}[Theorem~6.4 in \citet{richardson2002}, specialized to \emph{directed} ancestral graphs]\label{lemma:theorem_richardson}
If $\M$ is a DMAG over vertex set $\Ovar$, then the MAG latent projection $\M_{\Ovar}(\Dc(\M))$ of the canonical DAG $\Dc(\M)$ of $\M$ equals $\M$, i.e., $\M = \M_{\Ovar}(\Dc(\M))$.
\end{mylemma}

Lemma~\ref{lemma:theorem_richardson} means that every DMAG is the MAG latent projection of some DAG. Moreoever, the condition $\M = \M_{\Ovar}(\Dc(\M))$ yields a characterization of DMAGs in the sense that a directed ancestral graph $\G$ is a DMAG if and only if it meets the condition $\G = \M_{\Ovar}(\Dc(\G))$. Because DMAGs are already characterized by definition,\footnote{As directed ancestral graphs without inducing paths between non-adjacent vertices, see Sec.~\ref{sec:notation}.} the alternative characterization by the condition $\G = \M_{\Ovar}(\Dc(\G))$ is of limited use in this case.

For ts-DMAGs, however, there is no definitional characterization. In addition, because not every DMAG with time series structure is a ts-DMAG (see the explanation below Lemma~\ref{lemma:DMAG_has_time_series_structure_and_time_order}), characterizing ts-DMAGs is a non-trivial task. In the remaining parts of the current subsection and Sec.~\ref{sec:complete_characterization_subsection}, we show that ts-DMAGs can be characterized by an appropriate generalization of the condition $\G = \M_{\Ovar}(\Dc(\G))$. The first step of such a generalization is to find an appropriate generalization of canonical DAGs.

The generalization of canonical DAGs to the time series setting is non-trivial for the following reason. Consider an edge $(i, t_i) \astast (j, t_j)$ in a DMAG $\M$ with time series structure that is not in the DMAG's stationarification $\stat(\M)$. If, depending on the orientation of the edge $(i, t_i) \astast (j, t_j)$ in $\M$, either $(i, t_i) \tailhead (j, t_j)$ or $(i, t_i) \headtail (j, t_j)$ or $(i, t_i) \headtail (l_{ij}, t_{ij}) \tailhead (j, t_j)$ with $(l_{ij}, t_{ij})$ unobserved were included in a ``canonical ts-DAG'' $\Dc(\M)$, then the repeating edges property of ts-DAGs would require the same structure to be present at all other time steps too. Hence, in $\Dc(\M)$ there would be $(i, t_i + \Delta t) \tailhead (j, t_j+ \Delta t)$ or $(i, t_i+ \Delta t) \headtail (j, t_j+ \Delta t)$ or $(i, t_i+ \Delta t) \headtail (l_{ij}, t_{ij}+ \Delta t) \tailhead (j, t_j+ \Delta t)$ for all $\Delta t \in \mathbb Z$. Consequently, in the MAG latent projection $\M_{\Ovar}(\Dc(\M))$ of $\Dc(\M)$ there would be an edge $(i, t_i + \Delta t) \astast (j, t_j + \Delta t)$ of the same type for all $\Delta t$. But then also in the stationarification $\stat(\M_{\Ovar}(\Dc(\M)))$ of $\M_{\Ovar}(\Dc(\M))$ there would be the edge $(i, t_i + \Delta t) \astast (j, t_j + \Delta t)$ for all $\Delta t$. Hence, $\M$ could not equal $\M_{\Ovar}(\Dc(\M))$.

Given these considerations, the canonical ts-DAG $\Dc(\M)$ of a ts-DMAG $\M$ should instead only take into account the edges in the stationarification $\stat(\M)$ of $\M$. We are thus lead to the following definition, which for use further below is not restricted to ts-DMAGs but more generally applies to acyclic directed mixed graphs.

\begin{mydef}[Canonical ts-DAG]\label{def:canonical_ts-DAG}
Let $\G$ be an acyclic directed mixed graph with time series structure and let $\V = \Iindex \times \Tindex$ with $\Tindex = \{t - \tau ~|~ 0 \leq \tau \leq \taumax\}$ be its set of vertices. Denote with $\E^{\stat}$ the set of edges of $\stat(\G)$. The canonical ts-DAG associated to $\G$, denoted as $\Dc(\G)$, is the graph $\Dc(\G) = (\V^{ca}, \E^{ca})$ defined as follows:
\begin{enumerate}
	\item Its set of vertices is $\V^{ca} = \left(\Iindex \cup \mathbf{J}\right) \times \mathbb{Z}$, where $\mathbf{J} = \{(i, j, \tau) ~|~ ((i, t-\tau), (j, t)) \in \E^{\stat}_{\headhead}\}$. The variable index set is $\Iindex \cup \mathbf{J}$ and the time index set is $\mathbb{Z}$.
	\item Its set of edges $\E^{ca} = \E^{ca}_{\rightarrow}$ are for all $\Delta t \in \mathbb{Z}$
	\begin{itemize}
		\item $(i, t-\tau + \Delta t) \tailhead (j, t + \Delta t)$ for all $((i, t-\tau) , (j, t)) \in \E^{\stat}_{\tailhead}$ and
		\item $((i, j, \tau), t + \Delta t) \tailhead (i, t + \Delta t)$ for all $(i, j, \tau) \in \mathbf{J}$ and
		\item $((i, j, \tau), t -\tau + \Delta t) \tailhead (j, t + \Delta t)$ for all $(i, j, \tau) \in \mathbf{J}$.
	\end{itemize}
\end{enumerate}
\end{mydef}

\begin{figure}[tb]
\centering
\includegraphics[width=0.65\linewidth, page = 1]{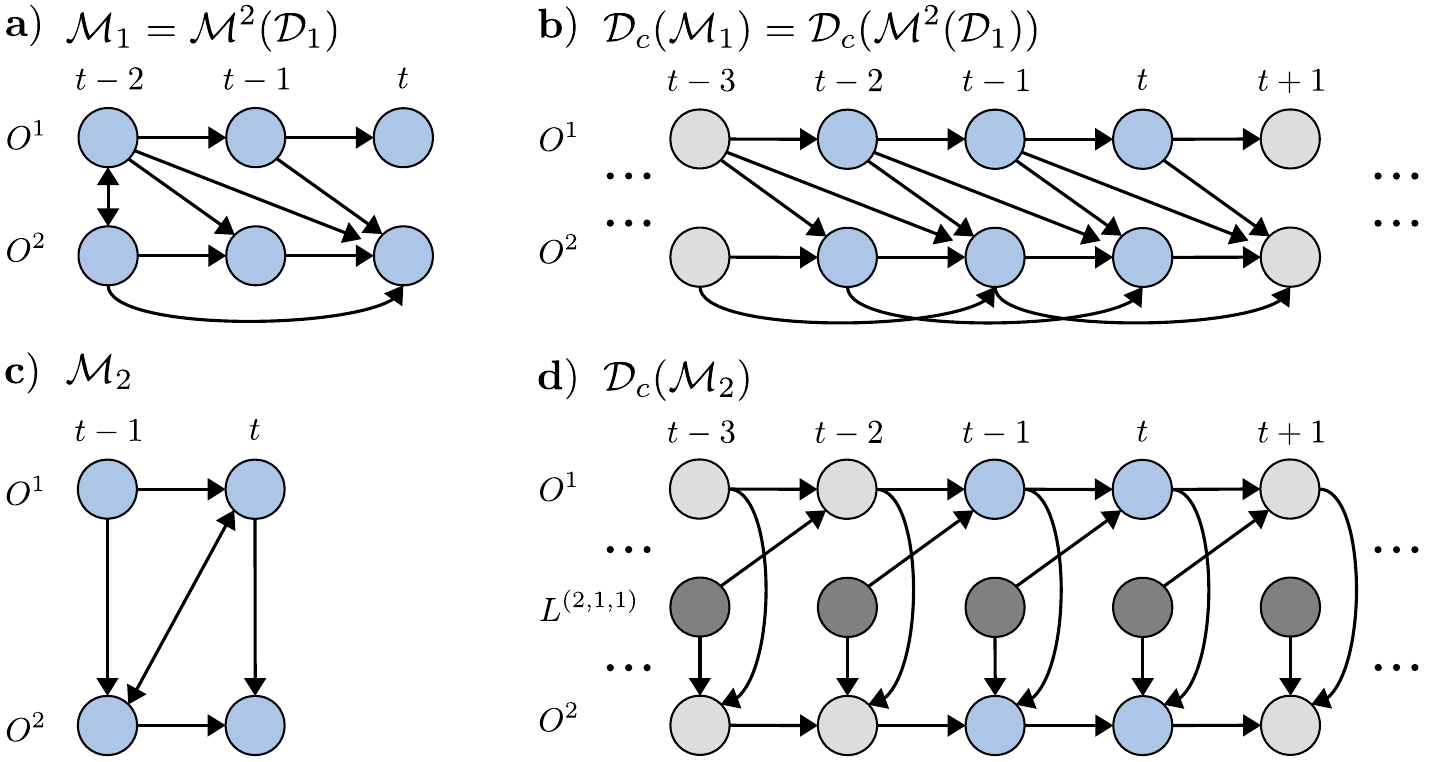}
\caption{%
A ts-DMAG $\M_1 = \M^2(\D_1)$ (the same as in part b) of Fig.~\ref{fig:example_tsDMAG} and part a) of Fig.~\ref{fig:stationarifications}) and a DMAG with time series structure $\M_2$ together with their canonical ts-DAGs. In $\Dc(\M_1)$ there is no unobservable time series because in $\M_1$ there is no bidirected edge that is repetitive in time and hence there is no bidirected edge in $\stat(\M_1)$. The unobservable time series $L^{(2, 1, 1)}$ in $\Dc(\M_2)$ in the notation of Def.~\ref{def:canonical_ts-DAG} corresponds to $(2, 1, 1) \in \mathbf{J}$ and results from the edge $O^2_{t-1} \headhead O^1_{t}$ in $\stat(\M_2) = \M_2$.
}
\label{fig:canonical_ts-DAG}
\end{figure}

Figure \ref{fig:canonical_ts-DAG} illustrates canonical ts-DAGs. Intuitively, the canonical ts-DAG $\Dc(\G)$ of $\G$ is obtained in three steps: First, replace $\G$ by its stationarification $\stat(\G)$. Second, in $\stat(\G)$ replace every bidirected edge $(i, t_j - \tau) \headhead (j, t_j)$ with $(i, t_j -\tau) \headtail ((i, j, \tau), t_j - \tau) \tailhead (j, t_j)$ where $((i, j, \tau), t_j - \tau)$ is an additionally inserted, unobserved vertex. Third, repeat this structure into the infinite past and future according to the repeating edges property. This intuition identifies the vertices $((i, j, \tau), s)$ with $(i, j, \tau) \in \mathbf{J}$ and $s \in \mathbb{Z}$ as analogs of the unobserved vertices $l_{ij} \in \mathbf{L}$ in standard canonical DAGs (see Def.~\ref{def:canonical_DAG} above) and, in addition, means that the time series indexed by $\mathbf{J}$ are treated as unobservable. The key difference between standard canonical DAGs and canonical ts-DAGs is the first of the three steps, i.e., the application of stationarification. A similarity is that also in canonical ts-DAGs there are no edges into unobservable vertices and hence no edges between two unobservable vertices.

Canonical ts-DAGs are indeed ts-DAGs and, by means of the following result, yield the desired generalization of Lemma~\ref{lemma:theorem_richardson}.

\begin{mylemma}\label{lemma:mtaumax_of_gc_is_id}
Let $\D$ be a ts-DAG with variable index set $\Iindex$. Let $\IindexO \subseteq \Iindex$ and $\TindexO = \{t - \tau ~|~ 0 \leq \tau \leq \taumax \}$ with $\taumax \geq 0$. Then, $\M_{\Ovar}(\D) = \M_{\Ovar}(\Dc(\M_{\Ovar}(\D)))$ with $\Ovar = \IindexO \times \TindexO$.
\end{mylemma}

\begin{myremark}[on Lemma~\ref{lemma:mtaumax_of_gc_is_id}]
The lemma involves two different MAG latent projections: First, the projection of the ts-DAG $\D$ to the ts-DMAG $\M_{\Ovar}(\D)$. Second, the projection of the canonical ts-DAG $\Dc(\M_{\Ovar}(\D))$ of $\M_{\Ovar}(\D)$ to $\M_{\Ovar}(\Dc(\M_{\Ovar}(\D)))$. In the first projection, the time series indexed by $\Iindex \setminus \IindexO$ are unobservable. In the second projection, the time series indexed by the set $\mathbf{J}$ are unobservable. In both projections, all vertices before $t - \taumax$ and after $t$ are temporally unobserved. However, since the set of observed variables is the same in both projections (namely $\Ovar$), no confusion arises when writing $\Mtaumax(\D) = \Mtaumax(\Dc(\Mtaumax(\D)))$ instead of $\M_{\Ovar}(\D) = \M_{\Ovar}(\Dc(\M_{\Ovar}(\D)))$. From here on we adopt this notation.
\end{myremark}

\begin{figure}[tb]
\centering
\includegraphics[width=0.5\linewidth, page = 1]{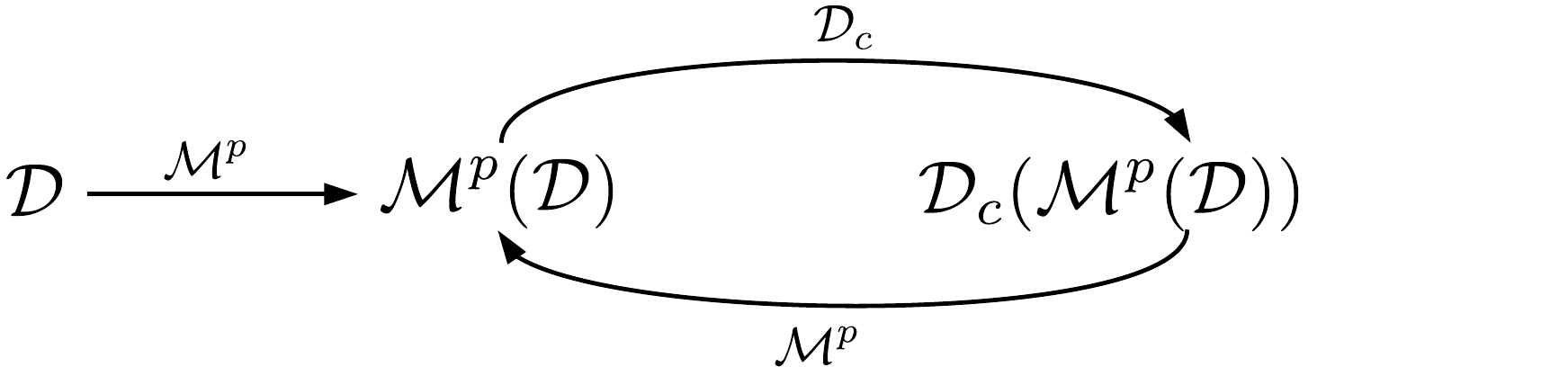}
\caption{%
Conceptual illustration of Lemma~\ref{lemma:mtaumax_of_gc_is_id}.
}
\label{fig:illustration_central_proposition_conceptual}
\end{figure}

Lemma~\ref{lemma:mtaumax_of_gc_is_id} says that the composition of creating the canonical ts-DAG and then projecting back to the original vertices is the identity operation on the space of ts-DMAGs, see Fig.~\ref{fig:illustration_central_proposition_conceptual}. This result is far from obvious for two reasons: First, if an edge $(i, t_i) \astast (j, t_j)$ in a ts-DMAG $\Mtaumax(\D)$ is not in the stationarified ts-DMAG $\Mtaumaxstat(\D)$ then in the canonical ts-DAG $\Dc(\Mtaumax(\D))$ there is neither $(i, t_i) \tailhead (j, t_j)$ nor $(i, t_i) \headtail (j, t_j)$ nor $(i, t_i) \headtail (l_{ij}, t_{ij}) \tailhead (j, t_j)$ with $(l_{ij}, t_{ij})$ unobservable. Hence, the edge $(i, t_i) \astast (j, t_j)$ needs to appear in the MAG latent projection $\Mtaumax(\Dc(\Mtaumax(\D)))$ of $\Dc(\Mtaumax(\D))$ in a non-trivial way, namely because of marginalizing over the temporally unobserved vertices. Second, this marginalization over the vertices before $t-\taumax$ must not create superfluous edges.

\begin{myexample}
The example in Fig.~\ref{fig:illustration_central_proposition} illustrates Lemma~\ref{lemma:mtaumax_of_gc_is_id}. This example also shows that the original ts-DAG $\D$ and the canonical ts-DAG $\Dc(\Mtaumax(\D))$ need not be equal.
\end{myexample}

\begin{figure}[tb]
\centering
\includegraphics[width=0.85\linewidth, page = 1]{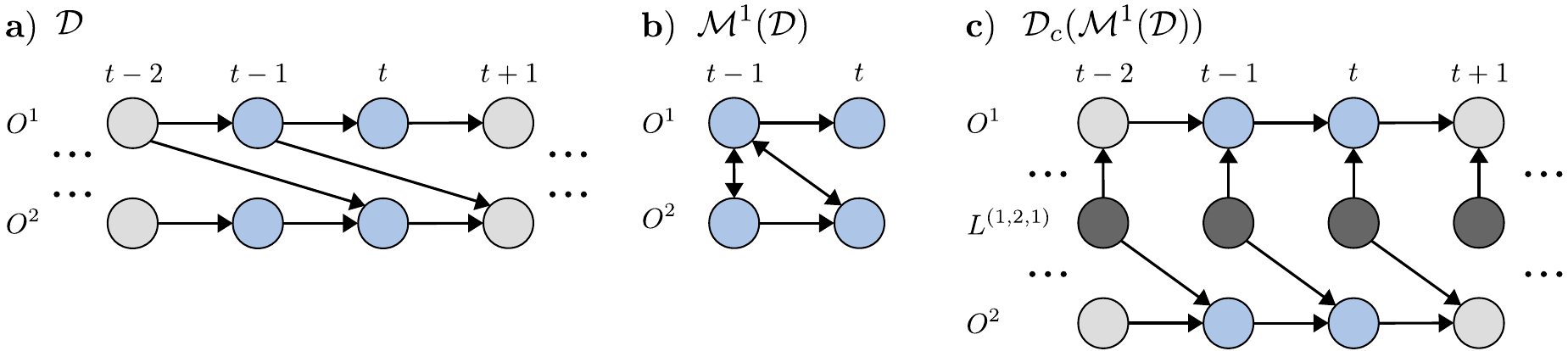}
\caption{%
A ts-DAG $\D$ together with the ts-DMAG $\M^1(\D)$ and the canonical ts-DAG $\Dc(\M^1(\D))$ of the ts-DMAG. Marginalizing $\Dc(\M^1(\D))$ to the observed vertices gives back $\M^1(\D)$. Same color coding as in Fig.~\ref{fig:example_tsDMAG}.
}
\label{fig:illustration_central_proposition}
\end{figure}

We stress that Lemma~\ref{lemma:mtaumax_of_gc_is_id} holds for \emph{arbitrary} ts-DAGs $\D$. In particular, in $\D$ there may be what in \citet{malinsky2018causal} is referred to as ``auto-lag confounders'', namely unobservable autocorrelated component time series $L$, i.e., $L_{t-\tau} \tailhead L_{t}$ with $L$ unobservable.

\subsection{A necessary and sufficient condition that characterizes ts-DMAGs}\label{sec:complete_characterization_subsection}
Lemma~\ref{lemma:mtaumax_of_gc_is_id} readily implies the following characterization of ts-DMAGs as a subclass of DMAGs with time series structure by a single necessary and sufficient condition.

\begin{mythm}\label{thm:complete_characterization}
Let $\M$ be a DMAG with time series structure and time index set $\Tindex = \{t - \tau ~|~ 0 \leq t \leq \taumax\}$. Then $\M$ is a ts-DMAG, i.e., there is a ts-DAG $\D$ such that $\M = \Mtaumax(\D)$ if and only if the MAG latent projection $\Mtaumax(\Dc(\M))$ of the canonical ts-DAG $\Dc(\M)$ of $\M$ equals $\M$, i.e., if and only if $\M = \Mtaumax(\Dc(\M))$.
\end{mythm}

Theorem~\ref{thm:complete_characterization} is one of the central results of this paper. The following four examples are included for its illustration.

\begin{myexample}
The DMAG $\M_2$ in part c) of Fig.~\ref{fig:canonical_ts-DAG} is a ts-DMAG. This conclusion follows because the canonical ts-DAG $\Dc(\M_2)$ in part d) of the figure projects to $\M_2$.
\end{myexample}

\begin{myexample}
One may use Theorem~\ref{thm:complete_characterization} to confirm that none of the four DMAGs in parts a) - d) of Fig.~\ref{fig:repeating_properties} is a ts-DMAG. In these cases this conclusion also follows because each of these four graphs violates at least one of the necessary conditions in Lemmas~\ref{lemma:DMAG_has_time_series_structure_and_time_order} and \ref{lemma:necessary_properites_time_series_DMAG}.
\end{myexample}

\begin{myexample}\label{ex:not_tsDMAG_1}
The DMAG $\M_1$ in part a) of Fig.~\ref{DMAG_with_time_series_structure_not_tsDMAG} is not a ts-DMAG because its canonical ts-DAG $\Dc(\M_1)$ in part b) projects to the ts-DMAG $\M^1(\Dc(\M_1))$ in part c), which is a proper subgraph of $\M_1$. This example also demonstrates that the equality $\stat(\M) = \Mtaumaxstat(\Dc(\M))$ is not sufficient for $\M$ to be a ts-DMAG.
\end{myexample}

\begin{myexample}\label{ex:not_tsDMAG_2}
The DMAG $\M_2$ in part d) of Fig.~\ref{DMAG_with_time_series_structure_not_tsDMAG} is not a ts-DMAG since its canonical ts-DAG $\Dc(\M_2)$ in part e) projects to the ts-DMAG $\M^1(\Dc(\M_2))$ in part f), which is a proper supergraph of $\M_2$. The edge $O^2_{t-1} \headhead O^1_t$ in $\M^1(\Dc(\M_2))$ is due to the green colored inducing path $O^2_{t-1} \headtail O^2_{t-2} \headtail O^1_{t-2} \headtail L^{(1,1,1)}_{t-2} \tailhead O^1_{t-1} \headtail L^{(1,1,1)}_{t-1} \tailhead O^1_t$ in $\Dc(\M_2)$.
\end{myexample}

\begin{figure}[tb]
\centering
\includegraphics[width=0.72\linewidth, page = 1]{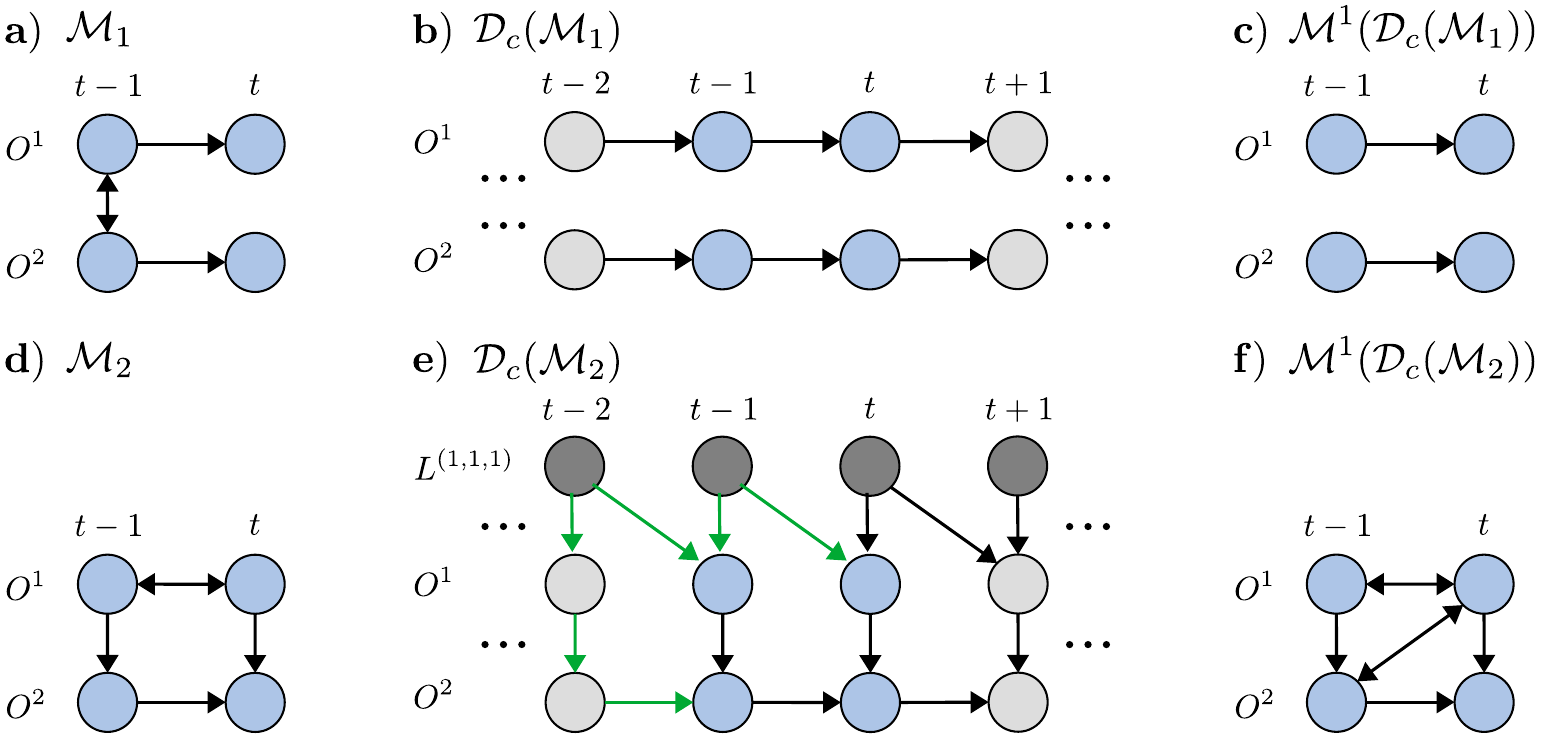}
\caption{%
Two examples of DMAGs with time series structure, $\M_1$ (the same as in part part e) of Fig.~\ref{fig:repeating_properties}) and $\M_2$ (the same as in part f) of Fig.~\ref{fig:repeating_properties}), that are not ts-DMAGs although they obey all necessary conditions in Lemmas~\ref{lemma:DMAG_has_time_series_structure_and_time_order} and \ref{lemma:necessary_properites_time_series_DMAG}. See also the discussions in Examples~\ref{ex:not_tsDMAG_1} and \ref{ex:not_tsDMAG_2}.
}
\label{DMAG_with_time_series_structure_not_tsDMAG}
\end{figure}

Importantly, the DMAGs considered in Examples~\ref{ex:not_tsDMAG_1} and \ref{ex:not_tsDMAG_2} obey all necessary conditions given in Lemmas~\ref{lemma:DMAG_has_time_series_structure_and_time_order} and \ref{lemma:necessary_properites_time_series_DMAG}. The condition $\M = \Mtaumax(\Dc(\M))$ is thus strictly stronger than even the combination of all these necessary conditions. This observation clearly demonstrates the significance and non-triviality of Theorem~\ref{thm:complete_characterization}.

As an alternative to Theorem~\ref{thm:complete_characterization}, we also characterize ts-DMAGs as a subclass of directed mixed graphs with time series structure.

\begin{mythm}\label{thm:complete_characterization_2}
Let $\G$ be a directed mixed graph with time series structure and time index set $\Tindex = \{t - \tau ~|~ 0 \leq t \leq \taumax\}$. Then $\G$ is a ts-DMAG, i.e., there is a ts-DAG $\D$ such that $\G = \Mtaumax(\D)$ if and only if $\G$ is acyclic and $\G = \Mtaumax(\Dc(\G))$.
\end{mythm}

Theorem~\ref{thm:complete_characterization_2} is even stronger than Theorem~\ref{thm:complete_characterization} because Theorem~\ref{thm:complete_characterization_2} does not require the graph $\G$ to be ancestral and/or maximal. Acyclicity, however, is needed because the definition of canonical ts-DAGs $\Dc(\G)$ requires $\G$ to be acyclic, as does the notion of $d$-separation.

\subsection{Implications for stationarified ts-DMAGs}\label{sec:characterization_stat_DMAG}
A ts-DMAG $\Mtaumax(\D)$ by definition uniquely determines its stationarification $\Mtaumaxstat(\D)$. How about the opposite? That is, can a ts-DMAG $\Mtaumax(\D)$ be uniquely determined from its stationarification $\Mtaumaxstat(\D)$? At first it seems perfectly conceivable that different ts-DMAGs have the same stationarification, which would make it impossible to uniquely determine $\Mtaumaxstat(\D)$ from $\Mtaumax(\D)$. However, as a corollary to the observation $\Dc(\G) = \Dc(\stat(\G))$ and Lemma~\ref{lemma:mtaumax_of_gc_is_id} we get the following result.

\begin{mylemma}\label{lemma:mtaumax_from_mtaumaxstat}
Let $\D$ be a ts-DAG. Then, the ts-DMAG $\Mtaumax(\D)$ equals the MAG latent projection $\Mtaumax(\Dc(\Mtaumaxstat(\D)))$ of the canonical ts-DAG $\Dc(\Mtaumaxstat(\D))$ of the stationarification $\Mtaumaxstat(\D) = \stat(\Mtaumax(\D))$ of $\Mtaumax(\D)$, i.e., $\Mtaumax(\D) = \Mtaumax(\Dc(\Mtaumaxstat(\D)))$.
\end{mylemma}

According to Lemma~\ref{lemma:mtaumax_from_mtaumaxstat} one can always uniquely determine $\Mtaumax(\D)$ from $\Mtaumaxstat(\D)$. A ts-DMAG $\Mtaumax(\D)$ and its stationarification $\Mtaumaxstat(\D)$ thus carry the exact same information about the underlying ts-DAG $\D$. In this sense $\Mtaumax(\D)$ and $\Mtaumaxstat(\D)$ are, if interpreted in the correct way, equivalent descriptions.

Lastly, we also arrive at two characterizations of stationarified ts-DMAGs.

\begin{mylemma}\label{lemma:complete_characterization_stat}
Let $\M$ be a DMAG with time series structure and time index set $\Tindex = \{t - \tau ~|~ 0 \leq t \leq \taumax\}$. Then, $\M$ is a stationarified ts-DMAG, i.e., there is a ts-DAG $\D$ such that $\M = \Mtaumaxstat(\D)$ if and only if $\M = \Mtaumaxstat(\Dc(\M))$.
\end{mylemma}

\begin{mylemma}\label{lemma:complete_characterization_stat_2}
Let $\G$ be a directed mixed graph with time series structure and time index set $\Tindex = \{t - \tau ~|~ 0 \leq t \leq \taumax\}$. Then, $\G$ is a stationarified ts-DMAG, i.e., there is a ts-DAG $\D$ such that $\G = \Mtaumaxstat(\D)$ if and only if $\G$ is acyclic and $\G = \Mtaumaxstat(\Dc(\G))$.
\end{mylemma}

\subsection{Comparison with previously considered model classes}\label{sec:previous_model_classes}
The author is aware of two distinct classes of graphical models based on DMAGs that have so far been used to represent time-lag specific causal relationships in time series with latent confounders. Here, we show that both these model classes are strictly larger than the class of ts-DMAGs.

The first previously used model class, employed by the tsFCI algorithm from \citet{Entner2010}, are DMAGs with time series structure that are time ordered and have repeating orientations as well as past-repeating adjacencies. Lemmas~\ref{lemma:DMAG_has_time_series_structure_and_time_order} and \ref{lemma:necessary_properites_time_series_DMAG} show that ts-DMAGs fall into this model class. The reverse, however, is not true: The graphs in part b) of Fig.~\ref{fig:repeating_properties} and parts a) and d) of Fig.~\ref{DMAG_with_time_series_structure_not_tsDMAG} fall into the model class used by tsFCI but are not ts-DMAGs.

The second previously used model class, employed by the SVAR-FCI algorithm from \citet{malinsky2018causal} and LPCMCI from \citet{LPCMCI}, are DMAGs with time series structure that are time ordered and have repeating edges. From Lemma~\ref{lemma:DMAG_has_time_series_structure_and_time_order} and Def.~\ref{def:stationarification} we see that each ts-DMAG $\Mtaumax(\D)$ is associated to a graph in this model, namely to the stationarified ts-DMAG $\Mtaumaxstat(\D) = \stat(\Mtaumax(\D))$. Lemma~\ref{lemma:mtaumax_from_mtaumaxstat} further implies that the mapping $\iota: \Mtaumax(\D) \mapsto \Mtaumaxstat(\D)$ is injective. Conversely, not all graphs in the model class used by SVAR-FCI and LPCMCI are ts-DMAGs: The graph in part d) of Fig.~\ref{DMAG_with_time_series_structure_not_tsDMAG} is an example.

\section{Markov equivalence classes of ts-DMAGs and causal discovery}\label{sec:equivalence_classes_and_causal_discovery}
This section discusses the implications of the concepts and results of Sec.~\ref{sec:characterization_full_section} for causal discovery. To this end, Def.~\ref{def:tsDPAGs} in Sec.~\ref{sec:tsDPAGs} introduces \emph{time series DPAGs (ts-DPAGs)} as graphs that represent Markov equivalence classes of ts-DMAGs. Time series DPAGs are refinements of DPAGs obtained by incorporating our background knowledge about the data generating process---namely that the data are generated by a process as in eq.~\eqref{eq:process} and that the observed time steps are regularly (sub-)sampled. We further introduce several alternative refinements of DPAGs, see Secs.~\ref{sec:background} and \ref{sec:types_background_knowledge}, concretely DPAGs which represent Markov equivalence classes of stationarified ts-DMAGs and DPAGs which incorporate only some of the necessary properties of ts-DMAGs as background knowledge. As we show, these alternative DPAGs carry less information about the underlying ts-DAG than ts-DPAGs do. Using the introduced terminology, in Sec.~\ref{sec:existing_algorithms} we discuss and compare three algorithms for independence-based causal discovery in time series with latent confounders and show that none of them learns ts-DPAGs. That is, all of these algorithms are conceptually suboptimal as they fail to learn causal properties of the underlying ts-DAG that in principle can be learned. As opposed to that, Algorithm~\ref{algo:complete_algo} in Sec.~\ref{sec:complete_algo_weaker} does learn ts-DPAGs and in this sense is \emph{complete}. Another important result is Theorem~\ref{thm:causal_discovery_non_stat_is_better} in Sec.~\ref{sec:why_DPAG_of_tsDPAG}, according to which DPAGs based on stationarified DMAGs carry less causal information than DPAGs based on non-stationarified DMAGs. Theorem~\ref{thm:causal_discovery_non_stat_is_better} corrects an erroneous claim that has appeared in the literature, see the explanation below Theorem~\ref{thm:causal_discovery_non_stat_is_better} in Sec.~\ref{sec:why_DPAG_of_tsDPAG} and the discussion of the SVAR-FCI algorithm in Sec.~\ref{sec:existing_algorithms} for more details.

\subsection{Background knowledge and DPAGs}\label{sec:background}
Markov equivalent DMAGs by definition have the same $m$-separations and thus cannot be distinguished by statistical independencies. They might, however, be distinguished if additional assumptions are made. One type of such assumptions is \emph{background knowledge}, i.e., the assertion that DMAGs with certain properties can be excluded as these are in conflict with a priori knowledge about the system of study.

\begin{mydef}[Background knowledge, cf.~\citet{FCI_cyclic}]\label{def:background_knowledge}
A \emph{background knowledge} $\BR$ is a Boolean function on the set of all DMAGs. If $\BR(\M) = 1$, then $\M$ is said to be \emph{consistent} with $\BR$, else it is said to be \emph{inconsistent} with $\BR$.
\end{mydef}

Combining Definition 2 in \citet{FCI_cyclic} with the definition of PAGs in \citet{andrews2020completeness}, we refine DPAGs by background knowledge as follows.

\begin{mydef}[DPAGs refined by background knowledge]\label{def:pags_background_knowledge}
Let $\M$ be a DMAG, let $[\M]$ be its Markov equivalence class, and for a background knowledge $\BR$ let $[\M]_{\BR}$ be the subset of $[\M]$ that is consistent with $\BR$, i.e., $[\M]_{\BR} = \{\M \in [\M] ~|~ \BR(\M) = 1\}$. Then:
\begin{enumerate}
\item A directed partial mixed graph $\PAG$ is a \emph{DPAG for $\M$} if
\begin{itemize}
	\item $\PAG$ has the same skeleton (i.e., the same set of adjacencies) as $\M$ and
	\item every non-circle mark in $\PAG$ is also in $\M$.
\end{itemize}
\item A DPAG $\PAG$ for $\M$ is called \emph{maximally informative (m.i.) with respect to $[\M]^\prime \subseteq [\M]$} if
\begin{itemize}
	\item every non-circle mark in $\PAG$ is in every element of $[\M]^\prime$ and
	\item for every circle mark in $\PAG$ there are $\M_1, \M_2 \in [\M]^\prime$ such that in $\M_1$ there is a tail mark and in $\M_2 \in [\M]^\prime$ there is a head mark instead of the circle mark.
\end{itemize}
\item The \emph{maximally informative (m.i.) DPAG with respect to $\BR$}, denoted as $\PAG(\M, \BR)$, is the m.i.~DPAG of $\M$ with respect to $[\M]_{\BR}$.
\item The \emph{conventional m.i.~DPAG for $\M$} is the m.i.~DPAG $\PAG(\M) = \PAG(\M, \BR_{\emptyset})$, where $\BR_{\emptyset}$ is the ``empty'' background knowledge for which $\BR_{\emptyset} = 1$ constant.
\end{enumerate}
\end{mydef}

To compare different background knowledges and the accordingly refined DPAGs, we employ the following terminology.

\begin{mydef}[Stronger/weaker background knowledge, more/less informative DPAG]\label{def:stronger_weaker_background_knowledge}
Let $\BR_1$ and $\BR_2$ be background knowledges, and let $\PAG_1$ and $\PAG_2$ be DPAGs for $\M$. We say$\ldots$
\begin{itemize}
\item $\ldots$ $\BR_1$ is \emph{stronger} than $\BR_2$ and $\BR_2$ is \emph{weaker} than $\BR_1$ if $\BR_1(\M) = 1$ implies $\BR_2(\M) = 1$.
\item $\ldots$ $\PAG_1$ is \emph{more informative} than $\PAG_2$ and $\PAG_2$ is \emph{less informative} than $\PAG_1$ if every circle mark in $\PAG_1$ is also in $\PAG_2$.
\end{itemize}
\end{mydef}

It follows that $\PAG(\M, \BR_1)$ is more informative than $\PAG(\M, \BR_2)$ if $\BR_1$ is stronger than $\BR_2$. By construction $\PAG(\M, \BR)$ is the most informative DPAG for $\M$ that can be learned from statistical independencies together with the background knowledge $\BR$. 

\subsection{Considered background knowledges}\label{sec:types_background_knowledge}
In the below discussions we are interested in the following background knowledges.

\begin{mydef}[Specific background knowledges]\label{def:specific_background_knowledges} The \emph{background knowledge of$\ldots$}
\begin{itemize}
\item $\ldots$ \emph{an underlying ts-DAG} is the background knowledge $\BRtsDAG$ for which $\BRtsDAG(\M) = 1$ if and only if $\M$ is a ts-DMAG, i.e., $\BRtsDAG(\M) = 1$ if and only if there is a ts-DAG $\D$ with $\M = \Mtaumax(\D)$.
\item $\ldots$ \emph{an underlying ts-DAG for stationarifications} is the background knowledge $\BRtsDAGstat$ for which $\BRtsDAGstat(\M) = 1$ if and only if $\M$ is a stationarified ts-DMAG, i.e., $\BRtsDAGstat(\M) = 1$ if and only if there is a ts-DAG $\D$ with $\M = \Mtaumaxstat(\D)$.
\item $\ldots$ \emph{time order and repeating ancestral relationships} is the background knowledge $\BRtora$ for which $\BRtora(\M) = 1$ if and only if $\M$ is time ordered and has repeating ancestral relationships.
\item $\ldots$ \emph{time order and repeating orientations} is the background knowledge $\BRtoro$ for which $\BRtoro(\M) = 1$ if and only if $\M$ is time ordered and has repeating orientations.
\end{itemize}
\end{mydef}

The first background knowledge $\BRtsDAG$ is as much background knowledge as is available in the time series setting defined in Sec.~\ref{sec:process}. In Sec.~\ref{sec:tsDPAGs} we will use $\BRtsDAG$ to define ts-DPAGs. The second background knowledge $\BRtsDAGstat$ is the equivalent background knowledge when working with stationarified ts-DMAGs $\Mtaumaxstat(\D)$ instead of ts-DMAGs $\Mtaumax(\D)$. We will use $\BRtsDAGstat$ to compare causal discovery based on $\Mtaumax(\D)$ with causal discovery based on $\Mtaumaxstat(\D)$. Given that a ts-DMAG $\Mtaumax(\D)$ and its stationarification $\Mtaumaxstat(\D)$ are in one-to-one correspondence, see Sec.~\ref{sec:characterization_stat_DMAG}, one might also expect the corresponding DPAGs to carry the same information. Interestingly, as we will show in Sec.~\ref{sec:why_DPAG_of_tsDPAG}, this expectation is incorrect. The third and fourth background knowledges $\BRtora$ and $\BRtoro$ equally apply to both standard and stationarified ts-DMAGs. They are included for comparison with existing causal discovery algorithms.

The four specified background knowledges compare as follows: Since both ts-DMAGs $\Mtaumax(\D)$ and stationarified ts-DMAGs $\Mtaumaxstat(\D)$ are time ordered and have repeating ancestral relationships, both $\BRtsDAG$ and $\BRtsDAGstat$ are stronger than $\BRtora$. Since repeating ancestral relationships imply repeating orientations, $\BRtora$ is stronger than $\BRtoro$. For stationarified ts-DMAGs $\Mtaumaxstat(\D)$, however, $\BRtora$ and $\BRtoro$ are equivalent (as follows from Lemma~\ref{lemma:interrelation_repeating_properties}). In our notation this equivalence is expressed as $\PAG(\Mtaumaxstat(\D), \BRtora) = \PAG(\Mtaumaxstat(\D), \BRtoro)$.

\subsection{DPAGs of ts-DMAGs $\Mtaumax(\D)$ carry more information than DPAGs of stationarified ts-DMAGs $\Mtaumaxstat(\D)$}\label{sec:why_DPAG_of_tsDPAG}
In this subsection we show that, when working with the background knowledges specified in Def.~\ref{def:specific_background_knowledges}, DPAGs of ts-DMAGs can never carry less but may carry more information about the underlying ts-DAG than DPAGs of stationarified ts-DMAGs. This is so despite the fact that, as explained in Sec.~\ref{sec:characterization_stat_DMAG}, a ts-DMAG and its stationarification are in one-to-one correspondence. Towards proving the claim we first note the following.

\begin{mylemma}\label{lemma:stat_of_DPAG}
Let $\D$ be a ts-DAG and let $\BR \in \{\BRtsDAG, \BRtora, \BRtoro\}$. Then, the graph $\stat(\PAG(\Mtaumax(\D), \BR))$ is a DPAG for $\Mtaumaxstat(\D)$.
\end{mylemma}

In particular, both DPAGs $\PAG(\Mtaumaxstat(\D), \BRstat)$ and $\stat(\PAG(\Mtaumax(\D), \BR))$ have the same adjacencies. Moreover, it is well-defined to ask whether one of the two DPAGs is more informative than the other. The following result answers this question.

\begin{mythm}\label{thm:causal_discovery_non_stat_is_better}
Let $\D$ be a ts-DAG and let $(\BR, \BRstat)$ either be $(\BRtoro, \BRtoro)$ or $(\BRtora, \BRtora)$ or $(\BRtsDAG, \BRtsDAGstat)$. Then:
\begin{enumerate}
  \item Every non-circle mark (head or tail) in $\PAG(\Mtaumaxstat(\D), \BRstat)$ is also in $\stat(\PAG(\Mtaumax(\D), \BR))$.
  \item Every non-circle mark in $\PAG(\Mtaumaxstat(\D), \BRstat)$ is also in $\PAG(\Mtaumax(\D), \BR)$.
  \item There are cases in which a non-circle mark that is in $\stat(\PAG(\Mtaumax(\D), \BR))$ is not also in $\PAG(\Mtaumaxstat(\D), \BRstat)$.
  \item There are cases in which a non-circle mark that is in $\PAG(\Mtaumax(\D), \BR)$ is not also in $\PAG(\Mtaumaxstat(\D), \BRstat)$, even regarding adjacencies that are shared by both graphs.
\end{enumerate}
\end{mythm}

Theorem~\ref{thm:causal_discovery_non_stat_is_better} contradicts the opposite claim in \citet{malinsky2018causal} according to which more unambiguous edge orientations (heads or tails) may be inferred if, as licensed by the assumption of causal stationarity, the property of repeating adjacencies is enforced in causal discovery; see Sec.~\ref{sec:existing_algorithms} for more details. The following example illustrates Theorem~\ref{thm:causal_discovery_non_stat_is_better}.

\begin{myexample}\label{example:causal_discovery_non_stat_is_better}
Part a) and b) of Fig.~\ref{fig:causal_discovery_non_stat_is_better} respectively show a ts-DMAG $\Mtaumax(\D)$ and its conventional m.i.~DPAG $\PAG(\Mtaumax(\D))$. To derive $\PAG(\Mtaumax(\D))$ one may, for example, apply the FCI orientation rules, see \citet{Zhang2008}, to the skeleton of $\Mtaumax(\D)$. Part c) of the same figure shows $\PAG(\Mtaumax(\D), \BRtoro)$, where the head at $O^2_t$ on $O^2_{t-1} \headhead O^2_{t}$ follows by time order. Repeating orientations does not help in orienting the last remaining circle mark on $O^1_{t-1} \ohead O^2_{t-1}$ because $O^1_{t}$ and $O^2_{t}$ are non-adjacent. The stronger background knowledge $\BRtora$ is, however, sufficient to do so: Vertex $O^1_{t-1}$ cannot be an ancestor of $O^2_{t-1}$ because $O^1_{t}$ is not an ancestor of $O^2_{t}$, which in turn follows because there is no possibly directed path from $O^1_{t}$ to $O^2_{t}$, see \citet[pp.~81f]{zhang2006causal}. We hence get the DPAG $\PAG(\Mtaumax(\D), \BRtora)$ shown in part d). Since there are no circle marks left, $\PAG(\Mtaumax(\D), \BRtora)$ here equals the DPAG $\PAG(\Mtaumax(\D), \BRtsDAG)$ in part e).\footnote{In general, the DPAGs $\PAG(\Mtaumax(\D), \BRtsDAG)$ and $\PAG(\Mtaumax(\D), \BRtora)$ are not equal and can contain circle marks.}
The graphs $\stat(\PAG(\Mtaumax(\D), \BRtoro))$, $\stat(\PAG(\Mtaumax(\D), \BRtora))$ and $\stat(\PAG(\Mtaumax(\D), \BRtsDAG))$ are respectively obtained by removing the edge between $O^1_{t-1}$ and $O^2_{t-2}$ from the graphs in parts c), d) and e).
The stationarified ts-DMAG $\Mtaumaxstat(\D)$ and its conventional m.i.~DPAG $\PAG(\Mtaumaxstat(\D))$ are shown in part f) and g). Part h) shows $\PAG(\Mtaumaxstat(\D), \BRtoro)$, where there is a head mark at $O^2_t$ on $O^2_{t-1} \tailhead O^2_t$ due to time order. As explained in Sec.~\ref{sec:types_background_knowledge}, the equality $\PAG(\Mtaumaxstat(\D), \BRtoro) = \PAG(\Mtaumaxstat(\D), \BRtora)$ always holds. With the characterization of stationarified ts-DMAGs in Lemma~\ref{lemma:complete_characterization_stat} (or Lemma~\ref{lemma:complete_characterization_stat_2}) we can further show that in this example the graph in h) equals $\PAG(\Mtaumaxstat(\D), \BRtsDAG)$ in part i). Note that the ts-DMAG $\M^1(\D)$ in part a) is indeed a ts-DMAG. For example, its canonical ts-DAG $\Dc(\M^1(\D))$ projects to $\M^1(\D)$.
\end{myexample}

\begin{figure}[tb]
\centering
\includegraphics[width=0.85\linewidth, page = 1]{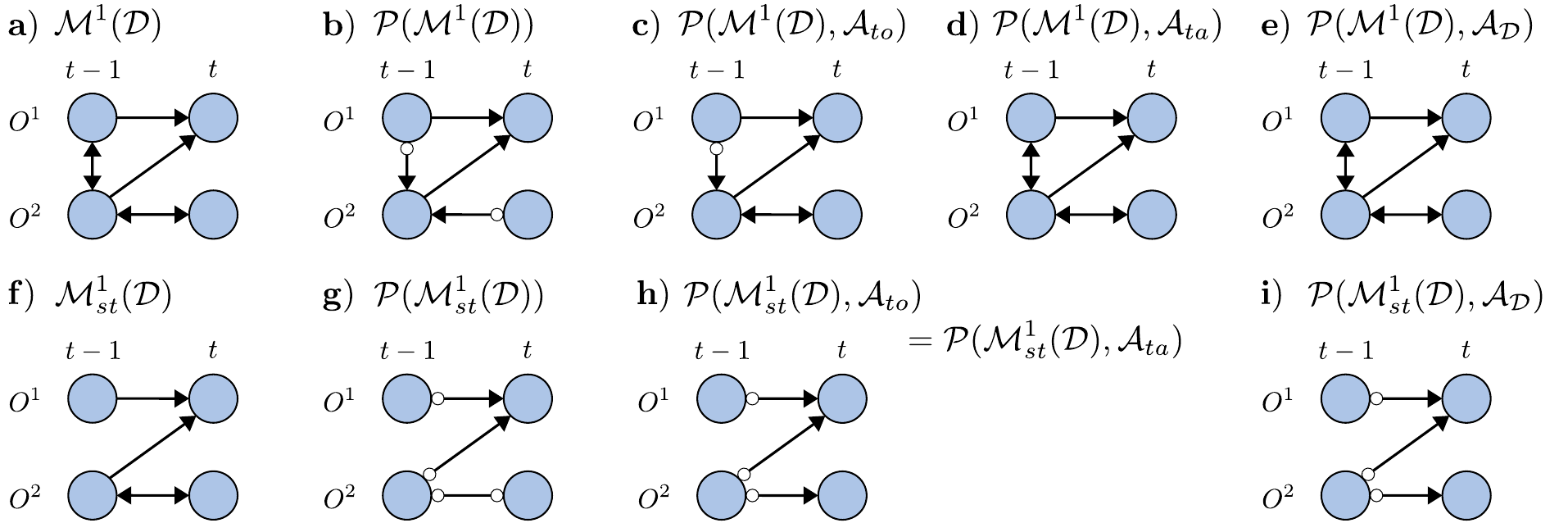}
\caption{An example for illustrating Theorem~\ref{thm:causal_discovery_non_stat_is_better}, see also the discussion in Example~\ref{example:causal_discovery_non_stat_is_better}.}
\label{fig:causal_discovery_non_stat_is_better}
\end{figure}

Theorem~\ref{thm:causal_discovery_non_stat_is_better} and Example~\ref{example:causal_discovery_non_stat_is_better} show that $\PAG(\Mtaumax(\D), \BR)$ and $\stat(\PAG(\Mtaumax(\D), \BR))$ have more unambiguous edge marks than $\PAG(\Mtaumaxstat(\D), \BRstat)$. It thus is \emph{conceptually} advantageous to work with DPAGs $\PAG(\Mtaumax(\D), \BR)$ of ts-DMAGs---or with their stationarifications $\stat(\PAG(\Mtaumax(\D), \BR))$, if one prefers graphs with repeating edges---rather than with DPAGs $\PAG(\Mtaumaxstat(\D), \BRstat)$ of stationarified ts-DMAGs. One might argue, though, that the additional ambiguous orientations (i.e., circle marks) which $\PAG(\Mtaumaxstat(\D), \BRstat)$ has as compared to $\PAG(\Mtaumax(\D), \BR)$ might turn into unambiguous orientations (i.e., head or tail marks) in $\PAG(\M^{\taumaxtilde}_{\statsubscript}(\D), \BRstat)$ for an increased length $\taumaxtilde > \taumax$ of the observed time window.\footnote{There are examples with this property, but it is unknown to the author whether this property is a general fact.} However, increasing $\taumax$ to $\taumaxtilde$ also increases the number of observed vertices and thus yields a higher-dimensional causal discovery problem. Having more observed vertices typically hurts finite-sample performance of causal discovery, see, e.g., the simulation studies in \citet{LPCMCI}. On the other hand, algorithms that work with stationarified ts-DMAGs rather than ts-DMAGs may scale more favorably with the length $\taumax$ of the observed time window because they remove the edges $O^i_{t-\Delta t-\tau} \astast O^j_{t-\Delta t}$ for all $\Delta t$ as soon as the edge $O^i_{t-\tau} \astast O^j_{t}$ is removed and therefore typically make fewer independence tests. From a \emph{practical} perspective there thus is a trade-off between working with ts-DMAGs vs.~working with stationarified ts-DMAGs, which calls for empirical evaluation in future work.

\subsection{Time series DPAGs}\label{sec:tsDPAGs}
In Sec.~\ref{sec:why_DPAG_of_tsDPAG} we showed that DPAGs of ts-DMAGs always carry more information about the underlying ts-DAG than DPAGs of stationarified ts-DMAGs. Because of this fact we choose to define \emph{time series DPAGs} as the former type of DPAGs. 

\begin{mydef}[Time series DPAG]\label{def:tsDPAGs}
Let $\D$ be a ts-DAG with variable index set $\Iindex$, let $\IindexO \subseteq \Iindex$, and let $\TindexO \subsetneq \mathbb{Z}$ be regularly sampled or regularly subsampled. The \emph{time series DPAG implied by $\D$ over $\Ovar = \IindexO \times \TindexO$}, denoted as $\PAG_{\Ovar}(\D)$ or $\PAG_{\IindexO \times \TindexO}(\D)$ and also referred to as a \emph{ts-DPAG}, is the m.i.~DPAG $\PAG(\M_{\Ovar}(\D), \BRtsDAG)$.
\end{mydef}

\begin{myremark}[on Def.~\ref{def:tsDPAGs}]
The equivalence of regular sampling and regular subsampling, see Sec.~\ref{sec:both_sampling_equivalent}, carries over to ts-DPAGs. We hence restrict to regular sampling without loss of generality and write $\PAGtaumax(\D)$ for $\PAG_{\Ovar}(\D)$ with $\Ovar = \IindexO \times \TindexO$ and $\TindexO = \{t - \tau ~|~ 0 \leq \tau \leq \taumax \}$.
\end{myremark}

The following example discusses a case in which the use of the strongest background knowledge $\BRtsDAG$ leads to strictly more unambiguous edge orientations than $\BRtora$. We thus cannot replace $\BRtsDAG$ with $\BRtora$ in the definition of ts-DPAGs without losing information.

\begin{myexample}\label{ex:tsDAG_better_than_tora}
The ts-DMAG $\M^1(\D)$ in part a) of Fig.~\ref{fig:tsDAG_better_than_tora} gives rise to $\PAG(\M^1(\D), \BRtora)$ in part b) with a circle mark at $O^1_t$ on $O^1_{t-1} \ohead O^1_t$. According to the stronger background knowledge $\BRtsDAG$ one can orient this edge as $O^1_{t-1} \tailhead O^1_t$ because the opposite hypothesis gives the graph in part d) in Fig.~\ref{DMAG_with_time_series_structure_not_tsDMAG}, which by means of Theorem~\ref{thm:complete_characterization} was shown to not be a ts-DMAG, see Example~\ref{ex:not_tsDMAG_2}. Thus, from the ts-DPAG $\PAG^1(\D)$ we can conclude that $O^1_{t-1}$ has a causal influence on $O^1_t$ whereas from $\PAG(\M^1(\D), \BRtora)$ we can only conclude that this causal influence might but also might not exist.

Furthermore, see \citet[Lemma~B.8]{gerhardus2022_supplement}, in the ts-DMAG $\M^1(\D)$ the pair $(O^1_{t-1}, O^1_{t})$ cannot suffer from latent confounding.\footnote{Interestingly, we can draw this conclusion although $O^1_{t-1} \tailhead O^1_{t}$ is not \emph{visible}, thereby suggesting that the notion of visibility from \citet{zhang2008causal} needs refinement for ts-DMAGs and ts-DPAGs.} Thus, the causal effect of $O^1_{t-1}$ on $O^1_t$ is identifiable and can be estimated from observations by adjusting for the empty set. Importantly, if we interpret $\M^1(\D)$ not as a ts-DMAG but as a ``standard'' DMAG, then the causal effect of $O^1_{t-1}$ on $O^1_t$ would be unidentifiable as follows from Lemma~10 in \citet{zhang2008causal}.
\end{myexample}

\begin{figure}[tb]
\centering
\includegraphics[width=0.50\linewidth, page = 1]{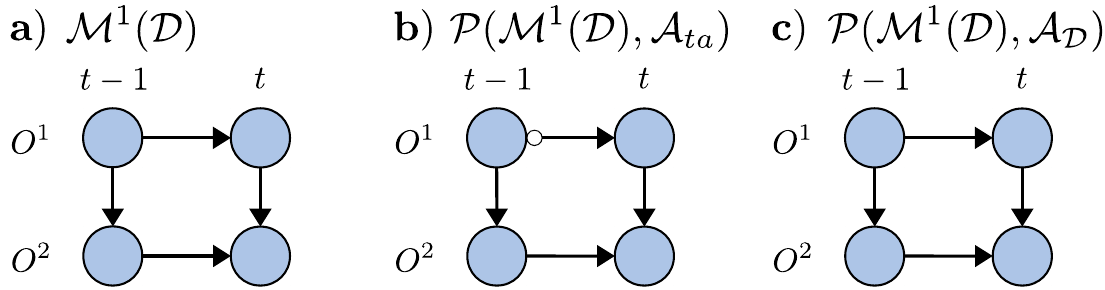}
\caption{A case in which $\PAGtaumax(\D)$ has a non-circle mark that is not in $\PAG(\Mtaumax(\D), \BRtora)$.}
\label{fig:tsDAG_better_than_tora}
\end{figure}

Example~\ref{ex:tsDAG_better_than_tora} clearly demonstrates the importance of our characterization of ts-DMAGs for the tasks of causal discovery and causal inference. Moreover, the following result shows that ts-DPAGs are complete with respect to ancestral relationships.
\begin{mylemma}\label{lemma:why_tsDPAGs}
If in a ts-DPAG $\PAGtaumax(\D)$ there is an edge $(i, t_i) \oast (j, t_j)$, then there are ts-DAGs $\D_1$ and $\D_2$ such that both ts-DMAGs $\Mtaumax(\D_1)$ and $\Mtaumax(\D_2)$ are Markov equivalent to the ts-DMAG $\Mtaumax(\D)$ and $(i, t_i) \in \an((j, t_j), \D_1)$ and $(i, t_i) \notin \an((j, t_j), \D_2)$. 
\end{mylemma}

\subsection{Existing causal discovery algorithms do not learn ts-DPAGs}\label{sec:existing_algorithms}
To the best of the author's knowledge, so far there is no causal discovery algorithm that learns ts-DPAGs $\PAGtaumax(\D)$. Hence, all existing causal discovery algorithms fail to learn some causal relationships that can be learned. This failure also applies to the independence-based algorithms tsFCI from \citet{Entner2010}, SVAR-FCI from \citet{malinsky2018causal} and LPCMCI from \citet{LPCMCI}.\footnote{The same is true for the score-based and hybrid algorithms from \citet{gao2010latent} and \citet{malinsky2018causal}.} Below, we discuss and compare these three algorithms \emph{conceptually} (but also note the \emph{practical} considerations discussed at the end of Sec.~\ref{sec:why_DPAG_of_tsDPAG}).

The \textbf{tsFCI} algorithm from \citet{Entner2010} refines the well-known FCI algorithm, see \citet{Spirtes1995}, \citet{Spirtes2000}, \citet{Zhang2008}, to structural processes as in eq.~\eqref{eq:process}. To this end, see the blue colored instructions in parts 2(a) and 2(b) of Algorithm 1~in \citet{Entner2010}, tsFCI imposes time order from the start and enforces repeating orientations at all steps. In addition, see the blue colored instructions in parts 1(b) and 1(c) of Algorithm 1~in \citet{Entner2010}, tsFCI excludes future vertices from conditioning sets and uses repeating separating sets to avoid unnecessary independence tests (these latter two modifications are, however, only relevant computationally and statistically but not conceptually). Importantly, \citet{Entner2010} introduces two variants of the algorithm. The first variant, which we call \textbf{tsFCI$^{l}$} (with `$l$' for `lagged'), assumes that in the data-generating ts-DAG there are no contemporaneous edges and hence orients all contemporaneous edges in the DPAG as bidirected. This first variant is as specified by Algorithm 1~in \citet{Entner2010}. However, in section~6 of \citet{Entner2010} (see, in particular, their footnote~3) the authors explain the minor modifications that have to be done when not making the additional assumption of no contemporaneous causation. Moreover, there they also show an application of the resulting more general variant. We refer to this second variant as \textbf{tsFCI$^{l+c}$} (with `$l+c$' for `lagged plus contemporaneous'). To summarize, in our terminology tsFCI$^{l+c}$ attempts to learn the DPAG $\PAG(\Mtaumax(\D), \BRtoro)$ of the ts-DMAG $\Mtaumax(\D)$. Since $\PAG(\Mtaumax(\D), \BRtoro)$ may contain circle marks that are not in the ts-DPAG $\PAGtaumax(\D) = \PAG(\Mtaumax(\D), \BRtsDAG)$, see Examples~\ref{example:causal_discovery_non_stat_is_better} and \ref{ex:tsDAG_better_than_tora}, tsFCI$^{l+c}$ does not learn all ancestral relationships that can be learned when using the available background knowledge $\BRtsDAG$. From Example~\ref{example:causal_discovery_non_stat_is_better} we even conclude that tsFCI$^{l+c}$ learns fewer orientations as can be learned with the weaker background knowledge $\BRtora$.

As compared to tsFCI$^{l+c}$, the more recent \textbf{SVAR-FCI} algorithm from \citet{malinsky2018causal} enforces repeating adjacencies by removing the edges $O^i_{t-\Delta t-\tau} \astast O^j_{t-\Delta t}$ for all $\Delta t$ as soon as the edge $O^i_{t-\tau} \astast O^j_{t}$ is removed---even in cases where there is no associated separating set in the observed time window.\footnote{To clarify: \citet{malinsky2018causal} does not mention tsFCI$^{l+c}$ but refer to tsFCI$^l$ when writing `tsFCI'.} This modification is achieved by the respective second lines in the ``then''-clauses in steps 5 and 11 of Algorithm 3.1~in \citet{malinsky2018causal}. Consequently, SVAR-FCI finds a skeleton which has repeating adjacencies, i.e., SVAR-FCI finds the skeleton of the stationarified ts-DMAG $\Mtaumaxstat(\D)$ rather than the skeleton of the ts-DMAG $\Mtaumax(\D)$. On the skeleton of $\Mtaumaxstat(\D)$ the algorithm then applies the FCI orientation rules, augmented with the background knowledge of time order and repeating orientations. In our terminology SVAR-FCI hence attempts to learn the DPAG $\PAG(\Mtaumaxstat(\D), \BRtoro)$ of the stationarified ts-DMAG $\Mtaumaxstat(\D)$. Now recall Theorem~\ref{thm:causal_discovery_non_stat_is_better}, which says that all unambiguous edge orientations in this DPAG $\PAG(\Mtaumaxstat(\D), \BRtoro)$ are also in $\PAG(\Mtaumax(\D), \BRtoro)$---the one learned by tsFCI$^{l+c}$---while there are cases in which the opposite is not true. Thus, if ground-truth knowledge of (conditional) independencies is given, SVAR-FCI can never learn more unambiguous edge orientations than tsFCI$^{l+c}$ while there are cases in which it learns strictly fewer. The additional edge removals thus actually have the opposite effect of what was intended in \citet{malinsky2018causal}. Moreover, also SVAR-FCI fails to learn all identifiable ancestral relationships of the underlying ts-DAG.

\begin{myexample}\label{ex:svar_fci_worse}
Assume ground-truth knowledge about (conditional) independencies. When applied to the ts-DMAG in part a) of Fig.~\ref{fig:causal_discovery_non_stat_is_better}, tsFCI$^{l+c}$ returns the graph in part c) whereas SVAR-FCI returns the graph in part h) with strictly fewer unambiguous edge marks. This difference is relevant: From tsFCI$^{l+c}$'s output we can conclude that $O^2_{t-1}$ generically has a causal effect on $O^1_t$ whereas from SVAR-FCI's output we can only conclude that $O^2_{t-1}$ might but also might not generically have a causal effect on $O^1_t$. As another difference, tsFCI$^{l+c}$ gives the edge $O^1_{t-1} \ohead O^2_{t-1}$ whereas SVAR-FCI gives that $O^1_{t-1}$ and $O^2_{t-1}$ are non-adjacent. At first one might think that SVAR-FCI is at the advantage in this regard, because the absence of an edge between $O^1_{t-1}$ and $O^2_{t-1}$  correctly conveys that the pair $(O^1_{t-1}, O^2_{t-1})$ is not confounded by unobser\underline{vable} variables. Note, however, that tsFCI$^{l+c}$ conveys the same conclusion by means of having learned that $O^1_{t}$ and $O^2_{t}$ are non-adjacent (cf.~last paragraph in Sec.~\ref{sec:stationarification}). In fact, see Theorem~\ref{thm:causal_discovery_non_stat_is_better}, one can always post-process the output of tsFCI$^{l+c}$ by stationarification $\stat(\cdot)$ to obtain a graph that, compared to the graph learned by SVAR-FCI, has the same adjacencies and the same or more unambiguous edge orientations. In the current example, this post-processing step amounts to removing the edge $O^1_{t-1} \ohead O^2_{t-1}$ from the graph in part c) of Fig.~\ref{fig:causal_discovery_non_stat_is_better}.
\end{myexample} 

The \textbf{LPCMCI} algorithm from \citet{LPCMCI} applies several modifications to SVAR-FCI that significantly improve the finite-sample performance. The infinite sample properties are unchanged, however. Thus also LPCMCI in general learns fewer orientations than tsFCI$^{l+c}$ and fails to learn all ancestral relationships that can be learned.

\subsection{ts-DPAGs can be learned from data}\label{sec:complete_algo_weaker}
In this subsection we show that ts-DPAGs can, at least in principle, be learned from data. In fact, using the characterization of ts-DMAGs by Theorem~\ref{thm:complete_characterization}, we can immediately write down Algorithm~\ref{algo:complete_algo} for this purpose.

Practically, however, finding the set of candidate DMAGs $\mathbf{M}^{\ast}$ in step 2 is expected to become computationally infeasible for large graphs $\PAG^{\ast}$. This expectation is based on the empirical finding in \citet{LVIDA} according to which the Zhang MAG listing algorithm (there used not for causal discovery but for causal effect estimation and in a non-temporal setting) becomes too slow for graphs with about $15$ to $20$ vertices. On the contrary, when using tsFCI$^{l+c}$ in step 1, the DPAG $\PAG^{\ast}$ already incorporates the background knowledge $\BRtoro$ of time order and repeating orientations. Hence, $\PAG^{\ast}$ will tend to have fewer circle marks than a typical DPAG in the non-temporal setting. Therefore, Algorithm~\ref{algo:complete_algo} might be feasible for yet larger graphs. Nevertheless, it would be desirable to instead derive orientation rules that impose the background knowledge $\BRtsDAG$ directly on $\PAG^{\ast}$ and thus entirely circumvent the need to determine the set $\mathbf{M}^{\ast}$. Moreover, recalling from the remark on Def.~\ref{def:ts_dmag_implied}, an implementation of the projection procedure required for step 3 is possible but non-trivial and will be left to future work, see \citet{gerhardus_OracleCI}. The following example illustrates Algorithm~\ref{algo:complete_algo}.

\begin{algorithm*}[tb]
\caption{Learning ts-DPAGs}
\begin{algorithmic}[1]
\State Apply any causal discovery algorithm on the time window $[t-\taumax, t]$ that determines a DPAG $\PAG^{\ast}$ for $\Mtaumax(\D)$ which is at least as informative as the conventional m.i.~DPAG $\PAG(\Mtaumax(\D))$. The tsFCI$^{l+c}$ algorithm meets this requirement, whereas SVAR-FCI and LPCMCI do not meet this requirement.
\State Let $\mathbf{M}^{\ast}$ be the set of all DMAGs that are represented by and are Markov equivalent to the DPAG $\PAG^{\ast}$. This step can, e.g., be done with the Zhang MAG listing algorithm described in \citet{LVIDA}.
\State Let $\mathbf{M}$ be the set of all DMAGs $\M \in \mathbf{M}^{\ast}$ for which $\M = \Mtaumax(\Dc(\M))$. This step can be executed as follows: For each $\M \in \mathbf{M}^{\ast}$, first, construct the canonical ts-DAG $\Dc(\M)$ by applying Def.~\ref{def:canonical_ts-DAG} and, second, apply the MAG latent projection to determine the ts-DMAG $\Mtaumax(\Dc(\M))$ and, third, check for equality of the graphs $\M$ and $\Mtaumax(\Dc(\M))$.
\State Let $\PAG$ be the m.i.~DPAG with respect to the set $\mathbf{M}$. Note that parts 1 and 2 of Def.~\ref{def:pags_background_knowledge} specify how $\PAG$ is determined from $\mathbf{M}$.
\State \Return ts-DPAG $\PAG = \PAGtaumax(\D)$
\end{algorithmic} \label{algo:complete_algo}
\end{algorithm*}

\begin{myexample}\label{myexample:complete_algo}
Consider the graph $\PAG(\M^1(\D), \BRtora)$ in part b) of Fig.~\ref{fig:tsDAG_better_than_tora}, which in this example equals $\PAG(\M^1(\D), \BRtoro)$. Given ground-truth knowledge of (conditional) independencies, this graph is the output of tsFCI$^{l+c}$ (and of SVAR-FCI and LPCMCI) on any ts-DAG $\D$ that projects to $\M^1(\D)$ in part a) of the figure. Such $\D$ exists, e.g.~the canonical ts-DAG $\Dc(\M^1(\D))$ of $\M^1(\D)$. There is exactly one circle mark in $\PAG^{\ast} = \PAG(\M^1(\D), \BRtora)$, namely on $O^1_{t-1} \ohead O^1_{t}$. This circle mark can be oriented either as a tail ($O^1_{t-1} \tailhead O^1_{t}$), giving rise to a DMAG $\M_1$, or as a head ($O^1_{t-1} \headhead O^1_{t}$), giving rise to a DMAG $\M_2$. Both of these candidates are represented by and Markov equivalent to $\PAG^{\ast}$, hence $\mathbf{M}^{\ast} = \{\M_1, \M_2\}$. Moving to step 3, the first candidate $\M_1$ passes the check $\M_1= \M^1(\Dc(\M_1))$, whereas (see Example~\ref{ex:tsDAG_better_than_tora}) $\M_2 \neq \M^1(\Dc(\M_2))$ for the second candidate $\M_2$. Thus $\mathbf{M} = \{\M_1\}$. Since there is only a single element $\M_1$ in $\mathbf{M}$, this DMAG $\M_1$ according to parts 1 and 2 of Def.~\ref{def:pags_background_knowledge} equals the ts-DPAG $\PAG^1(\D)$. Noting that $\PAG^1(\D)$ (learned by Algorithm~\ref{algo:complete_algo}) has an additional unambiguous edge mark as compared to $\PAG(\M^1(\D), \BRtoro)$ (learned by tsFCI$^{l+c}$, SVAR-FCI and LPCMCI), we see that Algorithm~\ref{algo:complete_algo} is indeed more informative than the existing algorithms.
\end{myexample}

\begin{myremark}[on Example~\ref{myexample:complete_algo}]
The example has two non-generic properties. First, see \citet[Fig.~C and Example~B.15]{gerhardus2022_supplement}, in general there can be circle marks in the ts-DPAG $\PAGtaumax(\D)$. Second, the ts-DMAG $\Mtaumax(\D) = \M^1(\D)$ here has repeating edges and thus equals $\Mtaumaxstat(\D) = \M^1_{\statsubscript}(\D)$. Only if the equality $\Mtaumax(\D) = \Mtaumaxstat(\D)$ holds, then also SVAR-FCI and LPCMCI learn the DPAG $\PAG^{\ast} = \PAG(\Mtaumax(\D), \BRtoro)$, which then also necessarily equals $\PAG(\Mtaumax(\D), \BRtora)$. In general, however, $\Mtaumax(\D)$ and $\Mtaumaxstat(\D)$ are not equal and neither SVAR-FCI nor LPCMCI may be used for step 1 of Algorithm~\ref{algo:complete_algo}, see Fig.~\ref{fig:causal_discovery_non_stat_is_better} and Example~\ref{example:causal_discovery_non_stat_is_better}.
\end{myremark}

\section{Discussion}
In this paper, we developed and analyzed ts-DMAGs, a class of graphical models for representing time-lag specific causal relationships and independencies among finitely many regularly (sub-)sampled time steps of causally stationary multivariate time series with unobserved components. As a central result, Theorems~\ref{thm:complete_characterization} and \ref{thm:complete_characterization_2} completely characterize ts-DMAGs. Examples demonstrated that ts-DMAGs constitute a strictly smaller class of graphical models than the graphs that have previously been employed in the literature, see Sec.~\ref{sec:previous_model_classes} for details. At the same time, using ts-DMAGs does not require additional assumptions or restrictions on the data-generating process. From ts-DMAGs one can thus draw stronger causal inferences than from the previously employed model classes, both in causal discovery and causal effect estimation. In addition, we defined ts-DPAGs as representations of Markov equivalence classes of ts-DMAGs. Time series DPAGs contain as much information about the ancestral relationships as can in principle be learned from observational data under the standard assumptions of independence-based causal discovery. We then showed that current time series causal discovery algorithms do not learn ts-DPAGs, i.e., they fail to learn some causal relationships that can be learned. As opposed to that, Algorithm~\ref{algo:complete_algo} does learn ts-DPAGs. With Theorem~\ref{thm:causal_discovery_non_stat_is_better} we corrected the incorrect claim from the literature that causal discovery on stationarified DMAGs gives more unambiguous edge orientations than causal discovery on non-stationarified DMAGs---in fact, the opposite is true. We envision that these results will be used to improve time series causal inference methods that resolve time lags, which in turn can have applications in diverse scientific and technical domains.

The results presented here point to various directions of future research. First, it would be valuable to consider the causal discovery problem in more detail. In particular, it is desirable to develop orientation rules that impose the background knowledge of an underlying ts-DAG $\BRtsDAG$ without the need for listing all DMAGs consistent with $\BRtsDAG$. Second, it remains open to characterize the causal inferences that can be drawn from ts-DMAGs and ts-DPAGs. As shown by Example~\ref{ex:tsDAG_better_than_tora}, deriving such a characterization is a non-trivial task that goes beyond the corresponding task in the non-temporal setting. Third, one may analyze the additional restrictions and causal inferences that follow when, as opposed to this work, assumptions on the connectivity pattern of the underlying ts-DAG are imposed. Lastly, it is desirable to generalize our results to cases with cyclic causal relationships and selection variables.

%
%

\begin{acks}[Acknowledgments]
I thank Jakob Runge for helpful discussions and suggestions. I thank Tom Hochsprung and Wiebke Günther for careful proofreading and suggestions on how to make the paper more accessible. I thank two anonymous reviewers and two anonymous Associate Editors for suggestions and questions that helped me to improve the paper.
\end{acks}
%

\begin{supplement}
\stitle{Supplement to “Characterization of causal ancestral graphs for time series with latent confounders”}
\sdescription{This Supplementary Material contains: First, a glossary of abbreviations and frequently used symbols. Second, theoretical results that were omitted from the main text due to space constraints. Third, proofs of all theoretical results presented in the main text together with various auxiliary results that are used in these proofs.}
\end{supplement}


\bibliographystyle{imsart-nameyear} 
\bibliography{library_annals}       

\begin{thebibliography}{46}

\bibitem[\protect\citeauthoryear{Ali, Richardson and
  Spirtes}{2009}]{ali2009markov}
\begin{barticle}[author]
\bauthor{\bsnm{Ali},~\bfnm{R.~Ayesha}\binits{R.~A.}},
  \bauthor{\bsnm{Richardson},~\bfnm{Thomas~S.}\binits{T.~S.}} \AND
  \bauthor{\bsnm{Spirtes},~\bfnm{Peter}\binits{P.}}
(\byear{2009}).
\btitle{Markov Equivalence for Ancestral Graphs}.
\bjournal{The Annals of Statistics}
\bvolume{37}
\bpages{2808--2837}.
\end{barticle}
\endbibitem

\bibitem[\protect\citeauthoryear{Andrews, Spirtes and
  Cooper}{2020}]{andrews2020completeness}
\begin{binproceedings}[author]
\bauthor{\bsnm{Andrews},~\bfnm{Bryan}\binits{B.}},
  \bauthor{\bsnm{Spirtes},~\bfnm{Peter}\binits{P.}} \AND
  \bauthor{\bsnm{Cooper},~\bfnm{Gregory~F.}\binits{G.~F.}}
(\byear{2020}).
\btitle{On the Completeness of Causal Discovery in the Presence of Latent
  Confounding with Tiered Background Knowledge}.
In \bbooktitle{Proceedings of the Twenty Third International Conference on
  Artificial Intelligence and Statistics}
(\beditor{\bfnm{Silvia}\binits{S.}~\bsnm{Chiappa}} \AND
  \beditor{\bfnm{Roberto}\binits{R.}~\bsnm{Calandra}}, eds.).
\bseries{Proceedings of Machine Learning Research}
\bvolume{108}
\bpages{4002--4011}.
\bpublisher{PMLR}.
\end{binproceedings}
\endbibitem

\bibitem[\protect\citeauthoryear{Assaad, Devijver and
  Gaussier}{2022}]{assaad2022discovery}
\begin{binproceedings}[author]
\bauthor{\bsnm{Assaad},~\bfnm{Charles~K.}\binits{C.~K.}},
  \bauthor{\bsnm{Devijver},~\bfnm{Emilie}\binits{E.}} \AND
  \bauthor{\bsnm{Gaussier},~\bfnm{Eric}\binits{E.}}
(\byear{2022}).
\btitle{Discovery of extended summary graphs in time series}.
In \bbooktitle{Proceedings of the Thirty-Eighth Conference on Uncertainty in
  Artificial Intelligence}
(\beditor{\bfnm{James}\binits{J.}~\bsnm{Cussens}} \AND
  \beditor{\bfnm{Kun}\binits{K.}~\bsnm{Zhang}}, eds.).
\bseries{Proceedings of Machine Learning Research}
\bvolume{180}
\bpages{96--106}.
\bpublisher{PMLR}.
\end{binproceedings}
\endbibitem

\bibitem[\protect\citeauthoryear{Bollen}{1989}]{bollen1989structural}
\begin{bbook}[author]
\bauthor{\bsnm{Bollen},~\bfnm{Kenneth~A}\binits{K.~A.}}
(\byear{1989}).
\btitle{{Structural Equations with Latent Variables}}.
\bpublisher{John Wiley \& Sons}, \baddress{New York, NY, USA}.
\end{bbook}
\endbibitem

\bibitem[\protect\citeauthoryear{Bongers, Blom and
  Mooij}{2018}]{bongers2018causal}
\begin{barticle}[author]
\bauthor{\bsnm{Bongers},~\bfnm{Stephan}\binits{S.}},
  \bauthor{\bsnm{Blom},~\bfnm{Tineke}\binits{T.}} \AND
  \bauthor{\bsnm{Mooij},~\bfnm{Joris~M}\binits{J.~M.}}
(\byear{2018}).
\btitle{Causal modeling of dynamical systems}.
\bjournal{arXiv preprint arXiv:1803.08784}.
\end{barticle}
\endbibitem

\bibitem[\protect\citeauthoryear{Chu and Glymour}{2008}]{chu2008search}
\begin{barticle}[author]
\bauthor{\bsnm{Chu},~\bfnm{Tianjiao}\binits{T.}} \AND
  \bauthor{\bsnm{Glymour},~\bfnm{Clark}\binits{C.}}
(\byear{2008}).
\btitle{Search for Additive Nonlinear Time Series Causal Models}.
\bjournal{Journal of Machine Learning Research}
\bvolume{9}
\bpages{967--991}.
\end{barticle}
\endbibitem

\bibitem[\protect\citeauthoryear{Didelez}{2008}]{didelez2008graphical}
\begin{barticle}[author]
\bauthor{\bsnm{Didelez},~\bfnm{Vanessa}\binits{V.}}
(\byear{2008}).
\btitle{Graphical models for marked point processes based on local
  independence}.
\bjournal{Journal of the Royal Statistical Society: Series B (Statistical
  Methodology)}
\bvolume{70}
\bpages{245--264}.
\end{barticle}
\endbibitem

\bibitem[\protect\citeauthoryear{Eichler}{2010}]{eichler2010graphical}
\begin{binproceedings}[author]
\bauthor{\bsnm{Eichler},~\bfnm{Michael}\binits{M.}}
(\byear{2010}).
\btitle{Graphical Gaussian modelling of multivariate time series with latent
  variables}.
In \bbooktitle{Proceedings of the Thirteenth International Conference on
  Artificial Intelligence and Statistics}
\bpages{193--200}.
\bpublisher{JMLR Workshop and Conference Proceedings}.
\end{binproceedings}
\endbibitem

\bibitem[\protect\citeauthoryear{Eichler and Didelez}{2007}]{eichler2007causal}
\begin{binproceedings}[author]
\bauthor{\bsnm{Eichler},~\bfnm{Michael}\binits{M.}} \AND
  \bauthor{\bsnm{Didelez},~\bfnm{Vanessa}\binits{V.}}
(\byear{2007}).
\btitle{Causal Reasoning in Graphical Time Series Models}.
In \bbooktitle{Proceedings of the Twenty-Third Conference on Uncertainty in
  Artificial Intelligence}
(\beditor{\bfnm{Ron}\binits{R.}~\bsnm{Parr}} \AND
  \beditor{\bfnm{Linda}\binits{L.}~\bparticle{van~der} \bsnm{Gaag}}, eds.).
\bseries{UAI'07}
\bpages{109–116}.
\bpublisher{AUAI Press}, \baddress{Arlington, Virginia, USA}.
\end{binproceedings}
\endbibitem

\bibitem[\protect\citeauthoryear{Eichler and
  Didelez}{2010}]{eichler2010granger}
\begin{barticle}[author]
\bauthor{\bsnm{Eichler},~\bfnm{Michael}\binits{M.}} \AND
  \bauthor{\bsnm{Didelez},~\bfnm{Vanessa}\binits{V.}}
(\byear{2010}).
\btitle{On Granger causality and the effect of interventions in time series}.
\bjournal{Lifetime data analysis}
\bvolume{16}
\bpages{3--32}.
\end{barticle}
\endbibitem

\bibitem[\protect\citeauthoryear{Entner and Hoyer}{2010}]{Entner2010}
\begin{binproceedings}[author]
\bauthor{\bsnm{Entner},~\bfnm{Doris}\binits{D.}} \AND
  \bauthor{\bsnm{Hoyer},~\bfnm{Patrik~O.}\binits{P.~O.}}
(\byear{2010}).
\btitle{On Causal Discovery from Time Series Data using FCI}.
In \bbooktitle{Proceedings of the 5th European Workshop on Probabilistic
  Graphical Models}
(\beditor{\bfnm{P.}\binits{P.}~\bsnm{Myllym{\"a}ki}},
  \beditor{\bfnm{T.}\binits{T.}~\bsnm{Roos}} \AND
  \beditor{\bfnm{T.}\binits{T.}~\bsnm{Jaakkola}}, eds.)
\bpages{121--128}.
\bpublisher{Helsinki Institute for Information Technology HIIT},
  \baddress{Helsinki, FI}.
\end{binproceedings}
\endbibitem

\bibitem[\protect\citeauthoryear{Gao and Tian}{2010}]{gao2010latent}
\begin{barticle}[author]
\bauthor{\bsnm{Gao},~\bfnm{Wei}\binits{W.}} \AND
  \bauthor{\bsnm{Tian},~\bfnm{Zheng}\binits{Z.}}
(\byear{2010}).
\btitle{Latent ancestral graph of structure vector autoregressive models}.
\bjournal{Journal of Systems Engineering and Electronics}
\bvolume{21}
\bpages{233--238}.
\end{barticle}
\endbibitem

\bibitem[\protect\citeauthoryear{Geiger, Verma and
  Pearl}{1990}]{geiger1990identifying}
\begin{barticle}[author]
\bauthor{\bsnm{Geiger},~\bfnm{Dan}\binits{D.}},
  \bauthor{\bsnm{Verma},~\bfnm{Thomas}\binits{T.}} \AND
  \bauthor{\bsnm{Pearl},~\bfnm{Judea}\binits{J.}}
(\byear{1990}).
\btitle{Identifying independence in Bayesian networks}.
\bjournal{Networks}
\bvolume{20}
\bpages{507--534}.
\end{barticle}
\endbibitem

\bibitem[\protect\citeauthoryear{Gerhardus}{2023}]{gerhardus2022_supplement}
\begin{barticle}[author]
\bauthor{\bsnm{Gerhardus},~\bfnm{Andreas}\binits{A.}}
(\byear{2023}).
\btitle{Supplement to “Characterization of causal ancestral graphs for time
  series with latent confounders”}.
\bdoi{xx.xxx/[provided by typesetter]}
\end{barticle}
\endbibitem

\bibitem[\protect\citeauthoryear{Gerhardus and Runge}{2020}]{LPCMCI}
\begin{binproceedings}[author]
\bauthor{\bsnm{Gerhardus},~\bfnm{Andreas}\binits{A.}} \AND
  \bauthor{\bsnm{Runge},~\bfnm{Jakob}\binits{J.}}
(\byear{2020}).
\btitle{High-recall causal discovery for autocorrelated time series with latent
  confounders}.
In \bbooktitle{Advances in Neural Information Processing Systems}
(\beditor{\bfnm{H.}\binits{H.}~\bsnm{Larochelle}},
  \beditor{\bfnm{M.}\binits{M.}~\bsnm{Ranzato}},
  \beditor{\bfnm{R.}\binits{R.}~\bsnm{Hadsell}},
  \beditor{\bfnm{M.~F.}\binits{M.~F.}~\bsnm{Balcan}} \AND
  \beditor{\bfnm{H.}\binits{H.}~\bsnm{Lin}}, eds.)
\bvolume{33}
\bpages{12615--12625}.
\bpublisher{Curran Associates, Inc.}
\end{binproceedings}
\endbibitem

\bibitem[\protect\citeauthoryear{Gerhardus et~al.}{2023}]{gerhardus_OracleCI}
\begin{barticle}[author]
\bauthor{\bsnm{Gerhardus},~\bfnm{Andreas}\binits{A.}},
  \bauthor{\bsnm{Wahl},~\bfnm{Jonas}\binits{J.}},
  \bauthor{\bsnm{Faltenbacher},~\bfnm{Sofia}\binits{S.}},
  \bauthor{\bsnm{Ninad},~\bfnm{Urmi}\binits{U.}} \AND
  \bauthor{\bsnm{Runge},~\bfnm{Jakob}\binits{J.}}
(\byear{2023}).
\btitle{Projecting infinite time series graphs to finite marginal graphs using
  number theory}.
\bjournal{In preparation}.
\end{barticle}
\endbibitem

\bibitem[\protect\citeauthoryear{Granger}{1969}]{Granger1969}
\begin{barticle}[author]
\bauthor{\bsnm{Granger},~\bfnm{C~W~J}\binits{C.~W.~J.}}
(\byear{1969}).
\btitle{Investigating causal relations by econometric models and cross-spectral
  methods}.
\bjournal{Econometrica}
\bvolume{37}
\bpages{424--438}.
\end{barticle}
\endbibitem

\bibitem[\protect\citeauthoryear{Huckins et~al.}{2020}]{huckins2020causal}
\begin{barticle}[author]
\bauthor{\bsnm{Huckins},~\bfnm{Jeremy~F}\binits{J.~F.}},
  \bauthor{\bsnm{DaSilva},~\bfnm{Alex~W}\binits{A.~W.}},
  \bauthor{\bsnm{Hedlund},~\bfnm{Elin~L}\binits{E.~L.}},
  \bauthor{\bsnm{Murphy},~\bfnm{Eilis~I}\binits{E.~I.}},
  \bauthor{\bsnm{Rogers},~\bfnm{Courtney}\binits{C.}},
  \bauthor{\bsnm{Wang},~\bfnm{Weichen}\binits{W.}},
  \bauthor{\bsnm{Obuchi},~\bfnm{Mikio}\binits{M.}},
  \bauthor{\bsnm{Holtzheimer},~\bfnm{Paul~E}\binits{P.~E.}},
  \bauthor{\bsnm{Wagner},~\bfnm{Dylan~D}\binits{D.~D.}} \AND
  \bauthor{\bsnm{Campbell},~\bfnm{Andrew~T}\binits{A.~T.}}
(\byear{2020}).
\btitle{Causal factors of anxiety and depression in college students:
  longitudinal ecological momentary assessment and causal analysis using Peter
  and Clark momentary conditional independence}.
\bjournal{JMIR mental health}
\bvolume{7}
\bpages{e16684}.
\end{barticle}
\endbibitem

\bibitem[\protect\citeauthoryear{Hyv{{\"a}}rinen
  et~al.}{2010}]{hyvarinen2010estimation}
\begin{barticle}[author]
\bauthor{\bsnm{Hyv{{\"a}}rinen},~\bfnm{Aapo}\binits{A.}},
  \bauthor{\bsnm{Zhang},~\bfnm{Kun}\binits{K.}},
  \bauthor{\bsnm{Shimizu},~\bfnm{Shohei}\binits{S.}} \AND
  \bauthor{\bsnm{Hoyer},~\bfnm{Patrik~O.}\binits{P.~O.}}
(\byear{2010}).
\btitle{Estimation of a Structural Vector Autoregression Model Using
  Non-Gaussianity}.
\bjournal{Journal of Machine Learning Research}
\bvolume{11}
\bpages{1709--1731}.
\end{barticle}
\endbibitem

\bibitem[\protect\citeauthoryear{Koller and
  Friedman}{2009}]{koller2009probabilistic}
\begin{bbook}[author]
\bauthor{\bsnm{Koller},~\bfnm{Daphne}\binits{D.}} \AND
  \bauthor{\bsnm{Friedman},~\bfnm{Nir}\binits{N.}}
(\byear{2009}).
\btitle{{Probabilistic Graphical Models: Principles and Techniques}}.
\bpublisher{MIT Press}, \baddress{Cambridge, MA, USA}.
\end{bbook}
\endbibitem

\bibitem[\protect\citeauthoryear{Kretschmer et~al.}{2016}]{kretschmer2016using}
\begin{barticle}[author]
\bauthor{\bsnm{Kretschmer},~\bfnm{Marlene}\binits{M.}},
  \bauthor{\bsnm{Coumou},~\bfnm{Dim}\binits{D.}},
  \bauthor{\bsnm{Donges},~\bfnm{Jonathan~F.}\binits{J.~F.}} \AND
  \bauthor{\bsnm{Runge},~\bfnm{Jakob}\binits{J.}}
(\byear{2016}).
\btitle{Using Causal Effect Networks to Analyze Different Arctic Drivers of
  Midlatitude Winter Circulation}.
\bjournal{Journal of Climate}
\bvolume{29}
\bpages{4069 - 4081}.
\end{barticle}
\endbibitem

\bibitem[\protect\citeauthoryear{Maathuis and
  Colombo}{2015}]{maathuis2015generalized}
\begin{barticle}[author]
\bauthor{\bsnm{Maathuis},~\bfnm{Marloes~H.}\binits{M.~H.}} \AND
  \bauthor{\bsnm{Colombo},~\bfnm{Diego}\binits{D.}}
(\byear{2015}).
\btitle{A generalized back-door criterion}.
\bjournal{The Annals of Statistics}
\bvolume{43}
\bpages{1060--1088}.
\end{barticle}
\endbibitem

\bibitem[\protect\citeauthoryear{Malinsky and Spirtes}{2016}]{LVIDA}
\begin{binproceedings}[author]
\bauthor{\bsnm{Malinsky},~\bfnm{Daniel}\binits{D.}} \AND
  \bauthor{\bsnm{Spirtes},~\bfnm{Peter}\binits{P.}}
(\byear{2016}).
\btitle{Estimating Causal Effects with Ancestral Graph Markov Models}.
In \bbooktitle{Proceedings of the Eighth International Conference on
  Probabilistic Graphical Models}
(\beditor{\bfnm{Alessandro}\binits{A.}~\bsnm{Antonucci}},
  \beditor{\bfnm{Giorgio}\binits{G.}~\bsnm{Corani}} \AND
  \beditor{\bfnm{Cassio~Polpo}\binits{C.~P.}~\bsnm{Campos}}, eds.).
\bseries{Proceedings of Machine Learning Research}
\bvolume{52}
\bpages{299--309}.
\bpublisher{PMLR}, \baddress{Lugano, Switzerland}.
\end{binproceedings}
\endbibitem

\bibitem[\protect\citeauthoryear{Malinsky and
  Spirtes}{2018}]{malinsky2018causal}
\begin{binproceedings}[author]
\bauthor{\bsnm{Malinsky},~\bfnm{Daniel}\binits{D.}} \AND
  \bauthor{\bsnm{Spirtes},~\bfnm{Peter}\binits{P.}}
(\byear{2018}).
\btitle{Causal Structure Learning from Multivariate Time Series in Settings
  with Unmeasured Confounding}.
In \bbooktitle{Proceedings of 2018 ACM SIGKDD Workshop on Causal Disocvery}
(\beditor{\bfnm{T~D}\binits{T.~D.}~\bsnm{Le}},
  \beditor{\bfnm{K}\binits{K.}~\bsnm{Zhang}},
  \beditor{\bfnm{E}\binits{E.}~\bsnm{K{\i}c{\i}man}},
  \beditor{\bfnm{A}\binits{A.}~\bsnm{Hyv\"{a}rinen}} \AND
  \beditor{\bfnm{Lin}\binits{L.}~\bsnm{Liu}}, eds.).
\bseries{Proceedings of Machine Learning Research}
\bvolume{92}
\bpages{23--47}.
\bpublisher{PMLR}, \baddress{London, UK}.
\end{binproceedings}
\endbibitem

\bibitem[\protect\citeauthoryear{Mogensen and
  Hansen}{2020}]{mogensen2020markov}
\begin{barticle}[author]
\bauthor{\bsnm{Mogensen},~\bfnm{S{\o}ren~Wengel}\binits{S.~W.}} \AND
  \bauthor{\bsnm{Hansen},~\bfnm{Niels~Richard}\binits{N.~R.}}
(\byear{2020}).
\btitle{Markov equivalence of marginalized local independence graphs}.
\bjournal{The Annals of Statistics}
\bvolume{48}
\bpages{539--559}.
\end{barticle}
\endbibitem

\bibitem[\protect\citeauthoryear{Mooij and Claassen}{2020}]{FCI_cyclic}
\begin{binproceedings}[author]
\bauthor{\bsnm{Mooij},~\bfnm{Joris~M.}\binits{J.~M.}} \AND
  \bauthor{\bsnm{Claassen},~\bfnm{Tom}\binits{T.}}
(\byear{2020}).
\btitle{Constraint-Based Causal Discovery using Partial Ancestral Graphs in the
  presence of Cycles}.
In \bbooktitle{Proceedings of the 36th Conference on Uncertainty in Artificial
  Intelligence (UAI)}
(\beditor{\bfnm{Jonas}\binits{J.}~\bsnm{Peters}} \AND
  \beditor{\bfnm{David}\binits{D.}~\bsnm{Sontag}}, eds.).
\bseries{Proceedings of Machine Learning Research}
\bvolume{124}
\bpages{1159--1168}.
\bpublisher{PMLR}.
\end{binproceedings}
\endbibitem

\bibitem[\protect\citeauthoryear{Pamfil et~al.}{2020}]{pamfil2020Dynotears}
\begin{binproceedings}[author]
\bauthor{\bsnm{Pamfil},~\bfnm{Roxana}\binits{R.}},
  \bauthor{\bsnm{Sriwattanaworachai},~\bfnm{Nisara}\binits{N.}},
  \bauthor{\bsnm{Desai},~\bfnm{Shaan}\binits{S.}},
  \bauthor{\bsnm{Pilgerstorfer},~\bfnm{Philip}\binits{P.}},
  \bauthor{\bsnm{Georgatzis},~\bfnm{Konstantinos}\binits{K.}},
  \bauthor{\bsnm{Beaumont},~\bfnm{Paul}\binits{P.}} \AND
  \bauthor{\bsnm{Aragam},~\bfnm{Bryon}\binits{B.}}
(\byear{2020}).
\btitle{DYNOTEARS: Structure Learning from Time-Series Data}.
In \bbooktitle{Proceedings of the Twenty Third International Conference on
  Artificial Intelligence and Statistics}
(\beditor{\bfnm{Silvia}\binits{S.}~\bsnm{Chiappa}} \AND
  \beditor{\bfnm{Roberto}\binits{R.}~\bsnm{Calandra}}, eds.).
\bseries{Proceedings of Machine Learning Research}
\bvolume{108}
\bpages{1595--1605}.
\bpublisher{PMLR}.
\end{binproceedings}
\endbibitem

\bibitem[\protect\citeauthoryear{Pearl}{1988}]{Pearl1988}
\begin{bbook}[author]
\bauthor{\bsnm{Pearl},~\bfnm{Judea}\binits{J.}}
(\byear{1988}).
\btitle{{Probabilistic Reasoning in Intelligent Systems: Networks of Plausible
  Inference}}.
\bpublisher{Morgan Kaufmann Publishers Inc.}, \baddress{San Francisco, CA,
  USA}.
\end{bbook}
\endbibitem

\bibitem[\protect\citeauthoryear{Pearl}{1995}]{pearl1995causal}
\begin{barticle}[author]
\bauthor{\bsnm{Pearl},~\bfnm{Judea}\binits{J.}}
(\byear{1995}).
\btitle{Causal diagrams for empirical research}.
\bjournal{Biometrika}
\bvolume{82}
\bpages{669--688}.
\end{barticle}
\endbibitem

\bibitem[\protect\citeauthoryear{Pearl}{2000}]{Pearl2000}
\begin{bbook}[author]
\bauthor{\bsnm{Pearl},~\bfnm{Judea}\binits{J.}}
(\byear{2000}).
\btitle{{Causality: Models, Reasoning, and Inference}}.
\bpublisher{Cambridge University Press}, \baddress{New York, NY, USA}.
\end{bbook}
\endbibitem

\bibitem[\protect\citeauthoryear{Pearl}{2009}]{Pearl2009}
\begin{bbook}[author]
\bauthor{\bsnm{Pearl},~\bfnm{Judea}\binits{J.}}
(\byear{2009}).
\btitle{{Causality: Models, Reasoning, and Inference}},
\bedition{2nd} ed.
\bpublisher{Cambridge University Press}, \baddress{Cambridge, UK}.
\end{bbook}
\endbibitem

\bibitem[\protect\citeauthoryear{Perkovi\'c
  et~al.}{2018a}]{perkovic2018complete}
\begin{barticle}[author]
\bauthor{\bsnm{Perkovi\'c},~\bfnm{Emilija}\binits{E.}},
  \bauthor{\bsnm{Textor},~\bfnm{Johannes}\binits{J.}},
  \bauthor{\bsnm{Kalisch},~\bfnm{Markus}\binits{M.}} \AND
  \bauthor{\bsnm{Maathuis},~\bfnm{Marloes~H.}\binits{M.~H.}}
(\byear{2018}a).
\btitle{Complete Graphical Characterization and Construction of Adjustment Sets
  in Markov Equivalence Classes of Ancestral Graphs}.
\bjournal{Journal of Machine Learning Research}
\bvolume{18}
\bpages{1--62}.
\end{barticle}
\endbibitem

\bibitem[\protect\citeauthoryear{Perkovi\'c
  et~al.}{2018b}]{Complete_characterization_adjustment}
\begin{barticle}[author]
\bauthor{\bsnm{Perkovi\'c},~\bfnm{Emilija}\binits{E.}},
  \bauthor{\bsnm{Textor},~\bfnm{Johannes}\binits{J.}},
  \bauthor{\bsnm{Kalisch},~\bfnm{Markus}\binits{M.}} \AND
  \bauthor{\bsnm{Maathuis},~\bfnm{Marloes~H.}\binits{M.~H.}}
(\byear{2018}b).
\btitle{Complete Graphical Characterization and Construction of Adjustment Sets
  in Markov Equivalence Classes of Ancestral Graphs}.
\bjournal{Journal of Machine Learning Research}
\bvolume{18}
\bpages{1-62}.
\end{barticle}
\endbibitem

\bibitem[\protect\citeauthoryear{Peters, Janzing and
  Sch\"{o}lkopf}{2013}]{Peters2013}
\begin{binproceedings}[author]
\bauthor{\bsnm{Peters},~\bfnm{Jonas}\binits{J.}},
  \bauthor{\bsnm{Janzing},~\bfnm{Dominik}\binits{D.}} \AND
  \bauthor{\bsnm{Sch\"{o}lkopf},~\bfnm{Bernhard}\binits{B.}}
(\byear{2013}).
\btitle{Causal Inference on Time Series using Restricted Structural Equation
  Models}.
In \bbooktitle{Advances in Neural Information Processing Systems}
(\beditor{\bfnm{C.~J.~C.}\binits{C.~J.~C.}~\bsnm{Burges}},
  \beditor{\bfnm{L.}\binits{L.}~\bsnm{Bottou}},
  \beditor{\bfnm{M.}\binits{M.}~\bsnm{Welling}},
  \beditor{\bfnm{Z.}\binits{Z.}~\bsnm{Ghahramani}} \AND
  \beditor{\bfnm{K.~Q.}\binits{K.~Q.}~\bsnm{Weinberger}}, eds.)
\bvolume{26}.
\bpublisher{Curran Associates, Inc.}
\end{binproceedings}
\endbibitem

\bibitem[\protect\citeauthoryear{Peters, Janzing and
  Sch{\"{o}}lkopf}{2017}]{Peters2018}
\begin{bbook}[author]
\bauthor{\bsnm{Peters},~\bfnm{Jonas}\binits{J.}},
  \bauthor{\bsnm{Janzing},~\bfnm{Dominik}\binits{D.}} \AND
  \bauthor{\bsnm{Sch{\"{o}}lkopf},~\bfnm{Bernhard}\binits{B.}}
(\byear{2017}).
\btitle{{Elements of Causal Inference: Foundations and Learning Algorithms}}.
\bpublisher{MIT Press}, \baddress{Cambridge, MA, USA}.
\end{bbook}
\endbibitem

\bibitem[\protect\citeauthoryear{Richardson and Spirtes}{2002}]{richardson2002}
\begin{barticle}[author]
\bauthor{\bsnm{Richardson},~\bfnm{Thomas}\binits{T.}} \AND
  \bauthor{\bsnm{Spirtes},~\bfnm{Peter}\binits{P.}}
(\byear{2002}).
\btitle{Ancestral Graph Markov Models}.
\bjournal{The Annals of Statistics}
\bvolume{30}
\bpages{962--1030}.
\end{barticle}
\endbibitem

\bibitem[\protect\citeauthoryear{Runge}{2020}]{Runge2020a}
\begin{binproceedings}[author]
\bauthor{\bsnm{Runge},~\bfnm{Jakob}\binits{J.}}
(\byear{2020}).
\btitle{Discovering contemporaneous and lagged causal relations in
  autocorrelated nonlinear time series datasets}.
In \bbooktitle{Proceedings of the 36th Conference on Uncertainty in Artificial
  Intelligence (UAI)}
(\beditor{\bfnm{Jonas}\binits{J.}~\bsnm{Peters}} \AND
  \beditor{\bfnm{David}\binits{D.}~\bsnm{Sontag}}, eds.).
\bseries{Proceedings of Machine Learning Research}
\bvolume{124}
\bpages{1388--1397}.
\bpublisher{PMLR}.
\end{binproceedings}
\endbibitem

\bibitem[\protect\citeauthoryear{Runge et~al.}{2012}]{runge2012escaping}
\begin{barticle}[author]
\bauthor{\bsnm{Runge},~\bfnm{Jakob}\binits{J.}},
  \bauthor{\bsnm{Heitzig},~\bfnm{Jobst}\binits{J.}},
  \bauthor{\bsnm{Petoukhov},~\bfnm{Vladimir}\binits{V.}} \AND
  \bauthor{\bsnm{Kurths},~\bfnm{J\"urgen}\binits{J.}}
(\byear{2012}).
\btitle{Escaping the Curse of Dimensionality in Estimating Multivariate
  Transfer Entropy}.
\bjournal{Physical Review Letters}
\bvolume{108}
\bpages{258701}.
\end{barticle}
\endbibitem

\bibitem[\protect\citeauthoryear{Saetia, Yoshimura and
  Koike}{2021}]{saetia2021constructing}
\begin{barticle}[author]
\bauthor{\bsnm{Saetia},~\bfnm{Supat}\binits{S.}},
  \bauthor{\bsnm{Yoshimura},~\bfnm{Natsue}\binits{N.}} \AND
  \bauthor{\bsnm{Koike},~\bfnm{Yasuharu}\binits{Y.}}
(\byear{2021}).
\btitle{Constructing brain connectivity model using causal network
  reconstruction approach}.
\bjournal{Frontiers in Neuroinformatics}
\bvolume{15}
\bpages{619557}.
\end{barticle}
\endbibitem

\bibitem[\protect\citeauthoryear{Spirtes, Glymour and
  Scheines}{1993}]{Spirtes1993}
\begin{bbook}[author]
\bauthor{\bsnm{Spirtes},~\bfnm{P.}\binits{P.}},
  \bauthor{\bsnm{Glymour},~\bfnm{C.}\binits{C.}} \AND
  \bauthor{\bsnm{Scheines},~\bfnm{R.}\binits{R.}}
(\byear{1993}).
\btitle{{Causation, Prediction, and Search}}.
\bseries{{Lecture Notes in Statistics}}
\bvolume{81}.
\bpublisher{Springer-Verlag}, \baddress{New York, NY, USA}.
\end{bbook}
\endbibitem

\bibitem[\protect\citeauthoryear{Spirtes, Glymour and
  Scheines}{2000}]{Spirtes2000}
\begin{bbook}[author]
\bauthor{\bsnm{Spirtes},~\bfnm{P.}\binits{P.}},
  \bauthor{\bsnm{Glymour},~\bfnm{C.}\binits{C.}} \AND
  \bauthor{\bsnm{Scheines},~\bfnm{R.}\binits{R.}}
(\byear{2000}).
\btitle{{Causation, Prediction, and Search}},
\bedition{2nd} ed.
\bpublisher{MIT Press}, \baddress{Cambridge, MA, USA}.
\end{bbook}
\endbibitem

\bibitem[\protect\citeauthoryear{Spirtes, Meek and
  Richardson}{1995}]{Spirtes1995}
\begin{binproceedings}[author]
\bauthor{\bsnm{Spirtes},~\bfnm{Peter}\binits{P.}},
  \bauthor{\bsnm{Meek},~\bfnm{Christopher}\binits{C.}} \AND
  \bauthor{\bsnm{Richardson},~\bfnm{Thomas}\binits{T.}}
(\byear{1995}).
\btitle{Causal Inference in the Presence of Latent Variables and Selection
  Bias}.
In \bbooktitle{Proceedings of the Eleventh Conference on Uncertainty in
  Artificial Intelligence}
(\beditor{\bfnm{P}\binits{P.}~\bsnm{Besnard}} \AND
  \beditor{\bfnm{S}\binits{S.}~\bsnm{Hanks}}, eds.).
\bseries{UAI'95}
\bpages{499–506}.
\bpublisher{Morgan Kaufmann Publishers Inc.}, \baddress{San Francisco, CA,
  USA}.
\end{binproceedings}
\endbibitem

\bibitem[\protect\citeauthoryear{Verma and Pearl}{1990}]{VERMA199069}
\begin{bincollection}[author]
\bauthor{\bsnm{Verma},~\bfnm{Thomas}\binits{T.}} \AND
  \bauthor{\bsnm{Pearl},~\bfnm{Judea}\binits{J.}}
(\byear{1990}).
\btitle{Causal Networks: Semantics and Expressiveness}.
In \bbooktitle{Uncertainty in Artificial Intelligence},
(\beditor{\bfnm{Ross~D.}\binits{R.~D.}~\bsnm{Shachter}},
  \beditor{\bfnm{Tod~S.}\binits{T.~S.}~\bsnm{Levitt}},
  \beditor{\bfnm{Laveen~N.}\binits{L.~N.}~\bsnm{Kanal}} \AND
  \beditor{\bfnm{John~F.}\binits{J.~F.}~\bsnm{Lemmer}}, eds.).
\bseries{Machine Intelligence and Pattern Recognition}
\bvolume{9}
\bpages{69-76}.
\bpublisher{North-Holland}.
\end{bincollection}
\endbibitem

\bibitem[\protect\citeauthoryear{Zhang}{2006}]{zhang2006causal}
\begin{bphdthesis}[author]
\bauthor{\bsnm{Zhang},~\bfnm{Jiji}\binits{J.}}
(\byear{2006}).
\btitle{Causal inference and reasoning in causally insufficient systems},
\btype{PhD thesis},
\bpublisher{Department of Philosophy, Carnegie Mellon University}.
\end{bphdthesis}
\endbibitem

\bibitem[\protect\citeauthoryear{Zhang}{2008a}]{zhang2008causal}
\begin{barticle}[author]
\bauthor{\bsnm{Zhang},~\bfnm{Jiji}\binits{J.}}
(\byear{2008}a).
\btitle{Causal Reasoning with Ancestral Graphs}.
\bjournal{Journal of Machine Learning Research}
\bvolume{9}
\bpages{1437-1474}.
\end{barticle}
\endbibitem

\bibitem[\protect\citeauthoryear{Zhang}{2008b}]{Zhang2008}
\begin{barticle}[author]
\bauthor{\bsnm{Zhang},~\bfnm{Jiji}\binits{J.}}
(\byear{2008}b).
\btitle{On the completeness of orientation rules for causal discovery in the
  presence of latent confounders and selection bias}.
\bjournal{Artificial Intelligence}
\bvolume{172}
\bpages{1873-1896}.
\end{barticle}
\endbibitem

\end{thebibliography}


\begin{thebibliography}{9}

\bibitem[\protect\citeauthoryear{Gerhardus and Runge}{2020}]{LPCMCI}
\begin{binproceedings}[author]
\bauthor{\bsnm{Gerhardus},~\bfnm{Andreas}\binits{A.}} \AND
  \bauthor{\bsnm{Runge},~\bfnm{Jakob}\binits{J.}}
(\byear{2020}).
\btitle{High-recall causal discovery for autocorrelated time series with latent
  confounders}.
In \bbooktitle{Advances in Neural Information Processing Systems}
(\beditor{\bfnm{H.}\binits{H.}~\bsnm{Larochelle}},
  \beditor{\bfnm{M.}\binits{M.}~\bsnm{Ranzato}},
  \beditor{\bfnm{R.}\binits{R.}~\bsnm{Hadsell}},
  \beditor{\bfnm{M.~F.}\binits{M.~F.}~\bsnm{Balcan}} \AND
  \beditor{\bfnm{H.}\binits{H.}~\bsnm{Lin}}, eds.)
\bvolume{33}
\bpages{12615--12625}.
\bpublisher{Curran Associates, Inc.}
\end{binproceedings}
\endbibitem

\bibitem[\protect\citeauthoryear{Malinsky and
  Spirtes}{2018}]{malinsky2018causal}
\begin{binproceedings}[author]
\bauthor{\bsnm{Malinsky},~\bfnm{Daniel}\binits{D.}} \AND
  \bauthor{\bsnm{Spirtes},~\bfnm{Peter}\binits{P.}}
(\byear{2018}).
\btitle{Causal Structure Learning from Multivariate Time Series in Settings
  with Unmeasured Confounding}.
In \bbooktitle{Proceedings of 2018 ACM SIGKDD Workshop on Causal Disocvery}
(\beditor{\bfnm{T~D}\binits{T.~D.}~\bsnm{Le}},
  \beditor{\bfnm{K}\binits{K.}~\bsnm{Zhang}},
  \beditor{\bfnm{E}\binits{E.}~\bsnm{K{\i}c{\i}man}},
  \beditor{\bfnm{A}\binits{A.}~\bsnm{Hyv\"{a}rinen}} \AND
  \beditor{\bfnm{Lin}\binits{L.}~\bsnm{Liu}}, eds.).
\bseries{Proceedings of Machine Learning Research}
\bvolume{92}
\bpages{23--47}.
\bpublisher{PMLR}, \baddress{London, UK}.
\end{binproceedings}
\endbibitem

\bibitem[\protect\citeauthoryear{Pearl}{2009}]{Pearl2009}
\begin{bbook}[author]
\bauthor{\bsnm{Pearl},~\bfnm{Judea}\binits{J.}}
(\byear{2009}).
\btitle{{Causality: Models, Reasoning, and Inference}},
\bedition{2nd} ed.
\bpublisher{Cambridge University Press}, \baddress{Cambridge, UK}.
\end{bbook}
\endbibitem

\bibitem[\protect\citeauthoryear{Richardson and Spirtes}{2002}]{richardson2002}
\begin{barticle}[author]
\bauthor{\bsnm{Richardson},~\bfnm{Thomas}\binits{T.}} \AND
  \bauthor{\bsnm{Spirtes},~\bfnm{Peter}\binits{P.}}
(\byear{2002}).
\btitle{Ancestral Graph Markov Models}.
\bjournal{The Annals of Statistics}
\bvolume{30}
\bpages{962--1030}.
\end{barticle}
\endbibitem

\bibitem[\protect\citeauthoryear{Spirtes, Glymour and
  Scheines}{2000}]{Spirtes2000}
\begin{bbook}[author]
\bauthor{\bsnm{Spirtes},~\bfnm{P.}\binits{P.}},
  \bauthor{\bsnm{Glymour},~\bfnm{C.}\binits{C.}} \AND
  \bauthor{\bsnm{Scheines},~\bfnm{R.}\binits{R.}}
(\byear{2000}).
\btitle{{Causation, Prediction, and Search}},
\bedition{2nd} ed.
\bpublisher{MIT Press}, \baddress{Cambridge, MA, USA}.
\end{bbook}
\endbibitem

\bibitem[\protect\citeauthoryear{Spirtes and
  Richardson}{1997}]{MAGs_equivalence}
\begin{binproceedings}[author]
\bauthor{\bsnm{Spirtes},~\bfnm{Peter}\binits{P.}} \AND
  \bauthor{\bsnm{Richardson},~\bfnm{Thomas~S.}\binits{T.~S.}}
(\byear{1997}).
\btitle{A Polynomial Time Algorithm for Determining DAG Equivalence in the
  Presence of Latent Variables and Selection Bias}.
In \bbooktitle{Proceedings of the Sixth International Workshop on Artificial
  Intelligence and Statistics}
(\beditor{\bfnm{David}\binits{D.}~\bsnm{Madigan}} \AND
  \beditor{\bfnm{Padhraic}\binits{P.}~\bsnm{Smyth}}, eds.).
\bseries{Proceedings of Machine Learning Research}
\bvolume{R1}
\bpages{489--500}.
\bpublisher{PMLR}
\bnote{Reissued by PMLR on 30 March 2021.}
\end{binproceedings}
\endbibitem

\bibitem[\protect\citeauthoryear{Verma and Pearl}{1990}]{verma1990equivalence}
\begin{binproceedings}[author]
\bauthor{\bsnm{Verma},~\bfnm{Thomas}\binits{T.}} \AND
  \bauthor{\bsnm{Pearl},~\bfnm{Judea}\binits{J.}}
(\byear{1990}).
\btitle{Equivalence and Synthesis of Causal Models}.
In \bbooktitle{Proceedings of the Sixth Annual Conference on Uncertainty in
  Artificial Intelligence}.
\bseries{UAI '90}
\bpages{255–270}.
\bpublisher{Elsevier Science Inc.}, \baddress{New York, NY, USA}.
\end{binproceedings}
\endbibitem

\bibitem[\protect\citeauthoryear{Zhang}{2008a}]{zhang2008causal}
\begin{barticle}[author]
\bauthor{\bsnm{Zhang},~\bfnm{Jiji}\binits{J.}}
(\byear{2008}a).
\btitle{Causal Reasoning with Ancestral Graphs}.
\bjournal{Journal of Machine Learning Research}
\bvolume{9}
\bpages{1437-1474}.
\end{barticle}
\endbibitem

\bibitem[\protect\citeauthoryear{Zhang}{2008b}]{Zhang2008}
\begin{barticle}[author]
\bauthor{\bsnm{Zhang},~\bfnm{Jiji}\binits{J.}}
(\byear{2008}b).
\btitle{On the completeness of orientation rules for causal discovery in the
  presence of latent confounders and selection bias}.
\bjournal{Artificial Intelligence}
\bvolume{172}
\bpages{1873-1896}.
\end{barticle}
\endbibitem

\end{thebibliography}


\end{document}


\begin{frontmatter}
\title{Supplement to “Characterization of causal ancestral graphs for time series with latent confounders”}

\begin{aug}
\author[A]{\fnms{Andreas}~\snm{Gerhardus}\ead[label=e1]{andreas.gerhardus@dlr.de}}
\address[A]{German Aerospace Center, Institute of Data Science, \printead{e1}}
\end{aug}

\begin{abstract}
This Supplementary Material contains: First, a glossary of abbreviations and frequently used symbols. Second, theoretical results that were omitted from the main text because of space constraints. Third, proofs of all theoretical results presented in the main text together with various auxiliary results that are used in these proofs.
\end{abstract}



\end{frontmatter}


\section{Glossary of abbreviations and frequently used symbols}
The following glossary might be helpful for reading the main paper and this Supplementary Material.

\begin{table}[h]
\begin{tabular}{l|l|l}
\textbf{Term / symbol} & \textbf{Meaning} & \textbf{Comment}\\
\hline
DAG & directed acyclic graph & \\
MAG / DMAG & (directed) maximal ancestral graph & \\
PAG / DPAG & (directed) partial ancestral graph & \\
ts-DAG  & time series DAG  & see Def.~3.4 \\
ts-DMAG / ts-DPAG & time series DMAG / DPAG & see Defs.~3.6 and 5.7\\
$\D$ & DAG or ts-DAG & \\
$\M$ & DMAG or ts-DMAG & \\
$\PAG$ & partial ancestral graph or DPAG or ts-DPAG \\
$\Ovar$ & set of all observed vertices in a graph & \\
$\M_{\Ovar}(\D)$ & MAG latent projection of the DAG or ts-DAG $\D$ to & see pp.~1442-3 in\\
& the vertices $\Ovar$ & \citet{zhang2008causal} \\
$\IindexO$ & variable indices of observed component time series &  \\
$\TindexO$ & time indices of observed time steps &  \\
$\M_{\IindexO \times \TindexO}(\D)$ & ts-DMAG of the ts-DAG $\D$ with observed vertices & see Def.~3.6\\
& $\Ovar = \IindexO \times \TindexO$ & \\
$\M_{\Ovar}$ & synonymous to $\M_{\IindexO \times \TindexO}(\D)$ with $\Ovar = \IindexO \times \TindexO$ & see Def.~3.6\\
$\taumax$ & length of the observed time window, for regular & \\
& sampling related to $\TindexO$ by $\TindexO = \{t-\tau ~|~ 0 \leq \tau \leq \taumax\}$ & \\
$\Mtaumax(\D)$ & synonymous to $\M_{\IindexO \times \TindexO}(\D)$ with & see Sec.~4.2\\
& $\TindexO = \{t-\tau ~|~ 0 \leq \tau \leq \taumax\}$ & \\
$\stat(\cdot)$ & stationarification, can be applied to graphs with time & see Def.~4.6 \\
& series structure &\\
$\Mtaumaxstat(\D)$ & stationarified ts-DMAG, synonymous to $\stat(\Mtaumax(\D))$ \\
$\Dc(\cdot)$ & canonical DAG or canonical ts-DAG & see Defs.~4.11 and 4.13 \\
$\PAG(\cdot,\, \BR)$ & m.i. DPAG wrt. to background knowledge $\BR$ & see Def.~5.2 \\
$\BRtsDAG$ & background knowledge of an underlying ts-DAG & see Def.~5.4 \\
$\BRtsDAGstat$ & background knowledge of an underlying ts-DAG & see Def.~5.4 \\
& for stationarifications & \\
$\BRtora$ & background knowledge of time order and & see Def.~5.4 \\
& repeating ancestral relationships & \\
$\BRtoro$ & background knowledge of time order and & see Def.~5.4 \\
& repeating  orientations& \\
$\PAGtaumax(\D)$ & time series DPAG, synonymous to $\PAGtaumax(\D,\, \BRtsDAG)$ & see Def.~5.7
\end{tabular}
\caption{Glossary of abbreviations and frequently used symbols.}
\end{table}

\section{Omitted results}
This section presents several theoretical results that were omitted from the main text due to space constraints.

\subsection{ts-DMAGs are a generalization of DMAGs}\label{secapp:tsDMAGs_more_general_than_DMAGs}
Consider an arbitrary DMAG $\M_{nt}$ with vertex set $\Ovar_{nt}$ and without time series structure (the subscript ``\emph{nt}'' stands for ``\textit{n}on-\textit{t}emporal''). As proven in \citet{richardson2002}, there is DAG $\D_{nt}$ over some vertex set $\V_{nt} \supset \Ovar_{nt}$ such that $\M_{nt} = \M_{\Ovar_{nt}}(\D_{nt})$. Now define $\D$ as the ts-DAG that consists of disconnected copies of $\D_{nt}$ at every time step $s \in \mathbb{Z}$, i.e., there are no lagged edges in $\D$ and for all $s \in \mathbb{Z}$ its induced subgraph on $\V_{nt} \times \{s\}$ is $\D_{nt}$. It then immediately follows that the ts-DMAG $\M_{\IindexO \times \TindexO}$ with $\IindexO = \V_{nt}$ and $\TindexO = \{t\}$ equals $\M_{nt}$.

This consideration identifies ts-DMAGs as a proper generalization of DMAGs and thereby shows that all statements about ts-DMAGs also to apply to DMAGs as a special case.

\subsection{Future vertices are not relevant for determining ts-DMAGs}\label{sec:future_irrelevant}
In the MAG latent projection of $\D$ to $\M_{\IindexO \times \TindexO}(\D)$ the vertices $\La = \V \setminus (\IindexO \times \TindexO)$ are unobserved, see Def.~4.6 in the main text. Since $t$ is the upper bound of the set of observed time steps $\TindexO$, this form of $\La$ means that in particular all vertices after $t$, i.e., in $[t+1, +\infty)$ are unobserved. However, for determining $\M_{\IindexO \times \TindexO}(\D)$ these vertices are irrelevant:

\begin{mylemma}\label{lemmaapp:future_irrelevant}
Let $\D$ be a ts-DAG with vertex set $\V = \Iindex \times \mathbb Z$. Denote with $\D_{\leq t}$ the subgraph of $\D$ induced on $\V_{\leq t} = \Iindex \times \Tindex_{\leq t}$ with $\Tindex_{\leq t} = \{s \in \mathbb Z ~|~ s \leq t\}$, i.e., the graph obtained by removing all vertices after $t$ and all edges involving these vertices from $\D$. Then, $\M_{\IindexO \times \TindexO}(\D) = \M_{\IindexO \times \TindexO}(\D_{\leq t})$, i.e., before applying the MAG latent projection one may simply ignore the part of $\D$ that is after $t$.
\end{mylemma}

\begin{proof}[Proof of Lemma~\ref{lemmaapp:future_irrelevant}]
Let $(i, t_i)$ and $(j, t_j)$ with $t_i, t_j \leq t$ be distinct vertices in $\D$. Then, $(i, t_i)$ and $(j, t_j)$ are non-adjacent in $\D$ if and only if $(i, t_i) \ci (j, t_j) ~|~ \mathbf{S}$ in $\D$ with $\mathbf{S} = \pa((i, t_i), \D) \cup \pa((j, t_j), \D) \setminus \{(i, t_i), \, (j, t_j)\}$, see \citet[Lemma~1]{verma1990equivalence}. Moreover, all vertices in $\mathbf{S}$ are before or at $t$ due to time order of $\D$, i.e., $\mathbf{S}$ is a subset of $\V_{\leq t} = \Iindex \times \Tindex_{\leq t}$. Consequently, $(i, t_i)$ and $(j, t_j)$ are non-adjacent in $\D$ if and only if there is a subset $\mathbf{S}^\prime \subseteq \V_{\leq t}$ such that $(i, t_i) \ci (j, t_j) ~|~ \mathbf{S}^\prime$ in $\D$. This observation implies that the graphs $\M_{\V_{\leq t}}(\D)$ and $\D_{\leq t}$ have the same skeleton, and the equality $\M_{\V_{\leq t}}(\D) = \D_{\leq t}$ follows because both $\M_{\V_{\leq t}}(\D)$ and $\D_{\leq t}$ have the same ancestral relationships among vertices in $\V_{\leq t}$ as $\D$. From $\M_{\V_{\leq t}}(\D) = \D_{\leq t}$ the statement follows with the commutativity of the marginalization process as stated by Theorem 4.20 in \citet{richardson2002}.
\end{proof}

This result, which follows from time order and $d$-separation, has the intuitive interpretation that the future need not be known in order to reason about the past and present.

\subsection{Temporal confounding}\label{secapp:no_unobservable}
As explained in Sec.~3.4 of the main text, in the construction of $\M_{\IindexO \times \TindexO}(\D)$ all vertices before $t-\taumax$, i.e., in $(-\infty, t-\taumax-1]$ are on purpose treated as unobserved---even if they are observable and hence become observed for some $\taumaxtilde > \taumax$. As the following example shows, such temporally unobserved observable vertices before $t-\taumax$ can act as latent confounders of observed vertices.

\begin{myexample}
In the ts-DAG $\D_1$ shown in part a) of Fig.~2 of the main text, the temporally unobserved vertex $O^1_{t-3}$ confounds the observed vertices $O^1_{t-2}$ and $O^2_{t-2}$ through the path $O^1_{t-2} \leftarrow O^1_{t-3} \rightarrow O^2_{t-2}$. This argument remains valid even without the unobservable time series $L^1$.
\end{myexample}

In the time series setting one thus effectively always deals with the case of latent confounding, even if component time series are observable. This observation further demonstrates the importance of conceptually understanding the latent variable setting as approached in this paper.

We also note that it is precisely this type of confounding that gives rise to edges which are in $\Mtaumax(\D)$ but not in $\Mtaumaxstat(\D)$.

\subsection{An axiomatic characterization of stationarifications}
As we have already noted below Def.~4.6 in the main text, the definition of stationarification implies the following result.

\begin{mylemma}\label{lemmaapp:properties_stationarification_2}
$\stat(\G)$ is the unique largest subgraph of $\G$ that has repeating edges.
\end{mylemma}

\begin{proof}[Proof of Lemma~\ref{lemmaapp:properties_stationarification_2}]
Combine both parts of Lemma~\ref{lemmaapp:properties_stationarification}.
\end{proof}

Indeed, we could alternatively have defined stationarifications by this property and then derived that stationarifications fulfill the properties as given in Def.~4.6.

\begin{mylemma}\label{lemmaapp:properties_stationarification}
\begin{enumerate}
	\item $\stat(\G)$ has repeating edges.
	\item If $\G^\prime$ is a subgraph of $\G$ and has repeating edges, then $\G^\prime$ is a subgraph of $\stat(\G)$.
\end{enumerate}
\end{mylemma}

\begin{proof}[Proof of Lemma~\ref{lemmaapp:properties_stationarification}]
\textbf{1.}
Consider an edge $((i, t_i), (j, t_j)) \in E_{\bullet}$ in $\stat(\G)$ and let $\Delta t$ be such that $(i, t_i + \Delta t), (j, t_j + \Delta t) \in \V$. Using the second point in Def.~4.6 twice, we first get $((i, t_i + \Delta t^\prime), (j, t_j + \Delta t^\prime)) \in E_{\bullet}$ in $\G$ for all $\Delta t^\prime$ for which $(i, t_i + \Delta t^\prime), (j, t_j + \Delta t^\prime) \in \V$ and thus $((i, t_i + \Delta t), (j, t_j + \Delta t)) \in E_{\bullet}$ in $\stat(\G)$.

\textbf{2.}
Let $\G^\prime \subseteq \G$ have repeating edges and assume $\G^\prime \cancel \subseteq \stat(\G)$. Since both $\G^\prime$ and $\stat(\G)$ are subgraphs of $\G$, adjacencies that are shared by $\G^\prime$ and $\stat(\G)$ correspond to edges of the same type. Thus, $\G^\prime \cancel \subseteq \stat(\G)$ implies that there is an adjacency in $\G^\prime$ which is not in $\stat(\G)$, i.e., $((i, t_i), (j, t_j)) \in \E$ in $\G^\prime$ and $((i, t_i), (j, t_j)) \notin \E$ in $\stat(\G)$. Since $\G^\prime$ has repeating edges by assumption, $((i, t_i), (j, t_j)) \in \E$ in $\G^\prime$ implies that $((i, t_i + \Delta t), (j, t_j + \Delta t)) \in \E$ in $\G$ for all $\Delta t$ for which $(i, t_i + \Delta t), (j, t_j + \Delta t) \in \V$. But then the second point in Def.~4.6 gives $((i, t_i), (j, t_j)) \in \E$ in $\stat(\G)$. Contradiction.
\end{proof}

\subsection{Why the case of no unobservable vertices remains special}\label{sec:no_unobservable_2}
As discussed in Sec.~\ref{secapp:no_unobservable}, even if there are no unobservable time series one in general still is in the setting of latent confounding. It is worth noting, though, that the case of no unobservable vertices remains special:

\begin{mylemma}\label{lemmaapp:no_unobservables}
Let $\D$ be a ts-DAG with variable index set $\Iindex$, let $\IindexO = \Iindex$, and let $\TindexO = \{t-\tau ~|~ 0 \leq \tau \leq \taumax \}$ where $\taumax \geq \porder$ with $\porder$ the largest lag in $\D$. Then, $\stat(\M_{\IindexO \times \TindexO}(\D))$ equals the subgraph of $\D$ induced on $\IindexO \times \TindexO$.
\end{mylemma}

\begin{myremark}[on Lemma~\ref{lemmaapp:no_unobservables}]
The proof is given in Sec.~\ref{secapp:contains_proof_for_why_no_unobservable_special} below.
\end{myremark}

In other words: If all component time series are observable ($\IindexO = \Iindex$) and there are enough regularly sampled time steps to capture all direct causal influences (choice of $\TindexO$ and $\taumax \geq \porder$), the stationarified ts-DMAG $\stat(\M_{\IindexO \times \TindexO}(\D))$ equals the segment of $\D$ on $\TindexO$. For ts-DMAGs $\M_{\IindexO \times \TindexO}(\D)$ the same is not necessarily true.

\subsection{Different DMAGs with the same stationarification cannot both be ts-DMAGs}
The following result is an immediate consequence of the one-to-one correspondence between a ts-DMAG and its stationarification (see Sec.~4.7 in the main text).

\begin{mylemma}\label{lemmaapp:different_DMAGs_not_both_implied}
Let $\M_1$ and $\M_2$ be DMAGs with time series structure such that $\M_1 \neq \M_2$ and $\stat(\M_1) = \stat(\M_1)$. Then, at least one of $\M_1$ and $\M_2$ is not a ts-DMAG.
\end{mylemma}

\begin{myremark}[on Lemma~\ref{lemmaapp:different_DMAGs_not_both_implied}]
The proof is given in Sec.~\ref{secapp:contains_proof_not_same_stationarification} below.
\end{myremark}

\begin{myexample}
The two DMAGs with time series structure in parts a) and b) of Fig.~\ref{figapp:same_stationarification} have the same stationarification (namely the graph in part c) of the figure). Therefore, at most one of them can be a ts-DMAG. Indeed, using Theorem~1 (or Theorem~2) from the main text we confirm that $\M_2$ is not a ts-DMAG and that $\M_1$ is a ts-DMAG.
\end{myexample}

\renewcommand{\thefigure}{A}
\begin{figure}[tb]
\centering
\includegraphics[width=0.78\linewidth, page = 1]{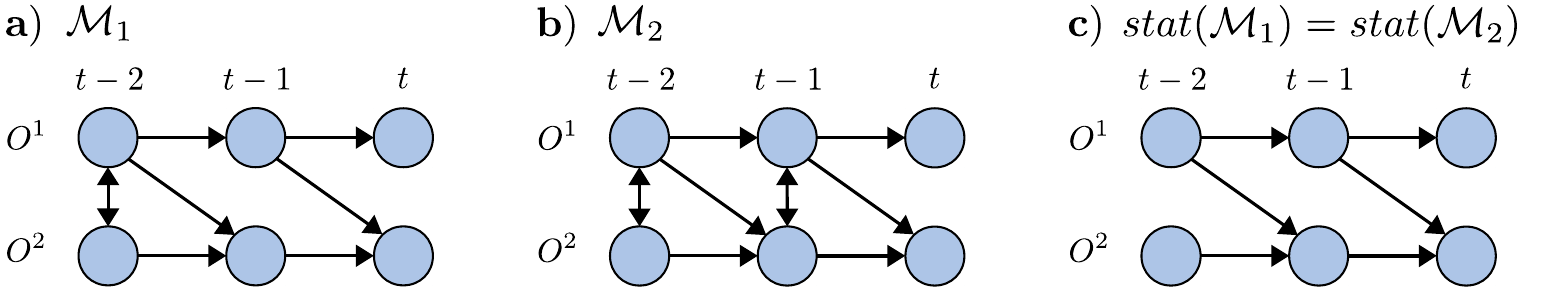}
\caption{Two different DMAGs with time series structure that have the same stationarification.
}
\label{figapp:same_stationarification}
\end{figure}

\subsection{Additional results on Example~5.8}
Here, we formalize and prove the following claim made in Example~5.8 in the main text.

\begin{mylemma}\label{lemmaapp:additional-result-more-identifiability}
Let $\D^\prime$ be a ts-DAG such that its ts-DMAG $\M^1(\D^\prime)$ equals the ts-DMAG $\M^1(\D)$ in part a) of Fig.~10 in the main text. Then, in $\D^\prime$ the pair of observed vertices $(O^1_{t-1}, O^1_t)$ is not subject to unobserved confounding, that is, in $\D^\prime$ there is no inducing path (relative to the set of observed vertices) between $O^1_{t-1}$ and $O^1_t$ that is into $O^1_{t-1}$. Consequently, from $\M^1(\D)$ we can conclude that the causal effect of $O^1_{t-1}$ on $O^1_t$ is identifiable and can be estimated from observations by adjusting for the empty set. Moreover, since the ts-DPAG $\PAG^1(\D)$ in part c) of Fig.~10 is equal to the ts-DMAG $\M^1(\D)$, we can draw this conclusion not only from $\M^1(\D)$ but also from $\PAG^1(\D)$.
\end{mylemma}

\begin{proof}[Proof of Lemma~\ref{lemmaapp:additional-result-more-identifiability}]
\underline{First}, we prove that the pair $(O^1_{t-1}, O^1_t)$ is not subject to unobserved confounding. To this end, we begin by deriving the existence of certain paths in $\D^\prime$:
\begin{enumerate}
\item Since $O^1_{t-1}$ is an ancestor of $O^1_t$ in $\D^\prime$ according to the edge $O^1_{t-1} \tailhead O^1_t$ in $\M^1(\D^\prime)$, in $\D^\prime$ there is a directed path $\pi_1$ from $O^1_{t-1}$ to $O^1_t$. This path cannot intersect $O^2_t$ because else $O^2_t$ would be an ancestor of $O^1_t$ by means of the subpath $\pi_1(O^2_t, O^1_t)$, which together with the fact that $O^1_t$ is an ancestor of $O^2_t$ according to the edge $O^1_t \tailhead O^2_t$ in $\M^1(\D^\prime)$ contradicts acyclicity. The path $\pi_1$ can also not intersect $O^2_{t-1}$ because the subpath $\pi_1(O^2_{t-1}, O^1_t)$ would then be a directed path from $O^2_{t-1}$ to $O^1_t$ such that all its non end-point vertices, if any, are unobserved. Consequently, there would need to be the edge $O^2_{t-1} \tailhead O^1_t$ in $\M^1(\D^\prime)$. Moreover, due to time order, $\pi_1$ can also not contain any vertex $O^1_{s}$ or $O^2_s$ with $s \leq t-2$. We conclude that all non end-point vertices of $\pi_1$, if any, are unobservable.
\item Since $O^1_{t-1}$ is an ancestor of $O^2_{t-1}$ in $\D^\prime$ according to the edge $O^1_{t-1} \tailhead O^2_{t-1}$ in $\M^1(\D^\prime)$, in $\D^\prime$ there is a directed path $\pi_2$ from $O^1_{t-1}$ to $O^2_{t-1}$. This path can, due to time order, neither intersect $O^1_t$ nor $O^2_t$. Moreover, also due to time order, $\pi_2$ cannot contain any vertex $O^1_{s}$ or $O^2_s$ with $s \leq t-2$. We conclude that all non end-point vertices of $\pi_2$, if any, are unobservable.
\item Since $O^2_{t-1}$ is an ancestor of $O^2_t$ in $\D^\prime$ according to the edge $O^2_{t-1} \tailhead O^2_t$ in $\M^1(\D^\prime)$, in $\D^\prime$ there is a directed path $\pi_3$ from $O^2_{t-1}$ to $O^2_t$. This path cannot intersect $O^1_{t-1}$ because else $O^2_{t-1}$ would be an ancestor of $O^1_{t-1}$ by means of the subpath $\pi_3(O^2_{t-1}, O^1_{t-1})$, which together with the fact that $O^1_{t-1}$ is an ancestor of $O^2_{t-1}$ according to the edge $O^1_{t-1} \tailhead O^2_{t-1}$ in $\M^1(\D^\prime)$ contradicts acyclicity. The path $\pi_3$ can also not intersect $O^1_{t}$ because the subpath $\pi_3(O^2_{t-1}, O^1_t)$ would then be a directed path from $O^2_{t-1}$ to $O^1_t$ such that all its non end-point vertices, if any, are unobserved. Consequently, there would need to be the edge $O^2_{t-1} \tailhead O^1_t$ in $\M^1(\D^\prime)$. Moreover, due to time, order $\pi_3$ can also not contain any vertex $O^1_{s}$ or $O^2_s$ with $s \leq t-2$. We conclude that all non end-point vertices of $\pi_3$, if any, are unobservable.
\end{enumerate}
\end{proof}

For $i = 1, 2, 3$ let $\pi_i^\prime$ be a copy of $\pi_i$ that is shifted backwards in time by one time step. These paths $\pi_i^\prime$ exist due to the repeating edges property of $\D^\prime$. Then, the concatenation $\rho_1 = \pi_1^\prime(O^1_{t-1}, O^1_{t-2}) \oplus \pi_2^\prime \oplus \pi_3^\prime$ is a collider-free path between $O^1_{t-1}$ and $O^2_{t-1}$ that is into both $O^1_{t-1}$ and $O^2_{t-1}$ such that all its non end-point vertices are unobserved. In particular, $\rho$ is an inducing path.

Now suppose that, contrary to the claim to be proven, in $\D^\prime$ there is an inducing path $\rho_2$ between $O^1_{t-1}$ and $O^1_t$ that is into $O^1_{t-1}$. Then, according to Lemma 32 in \citet{zhang2008causal} the concatenation $\rho_1 \oplus \rho_2$ has a subsequence $\rho$ which is an inducing path between $O^2_{t-1}$ and $O^1_t$ in $\D^\prime$. However, then there would need to be an edge between $O^2_{t-1}$ and $O^1_t$ in $\M^1(\D^\prime)$ (according to the ancestral relationships, this edge would be $O^2_{t-1} \headhead O^1_t)$, which is a contradiction. We conclude that there is no inducing path between $O^1_{t-1}$ and $O^1_t$ which is into $O^1_{t-1}$, that is, the pair $(O^1_{t-1}, O^1_t)$ is not subject to unobserved confounding.

\underline{Second}, we prove that the causal effect of $O^1_{t-1}$ and $O^1_t$ is identifiable and can be estimated by adjusting for the empty set. To this end, assume that $O^1_{t-1}$ and $O^1_t$ are \emph{not} $d$-separated in the graph $\G$ that is obtained by removing from $\D^\prime$ all edges out of $O^1_{t-1}$. Then, there is at least one path $\pi$ between $O^1_{t-1}$ and $O^1_t$ in $\G$ that is active given the empty set. This path
\begin{enumerate}
\item is into $O^1_{t-1}$ because in $\G$ there are no edges out of $O^1_{t-1}$,
\item is collider-free because $\pi$ is active given the empty set,
\item is also a path in $\D^\prime$ because $\G$ is a subgraph of $\D^\prime$,
\item needs to intersect at least one of $O^2_{t-1}$ and $O^2_t$ because else it would be an inducing path between $O^1_{t-1}$ and $O^1_t$ that is $O^1_{t-1}$ in $\D^\prime$,
\item is into $O^1_t$ because else it would need to be directed from $O^1_t$ to $O^1_{t-1}$, which contradicts time order,
\item does not intersect $O^2_t$ because else $O^2_t$ would need to be an ancestor of $O^1_t$ (which is not possible to due to the edge $O^1_t \tailhead O^2_t$ in $\M^1(\D^\prime)$) or of $O^1_{t-1}$ (which is not possible due to time order),
\item does not intersect $O^2_{t-1}$ because else the subpath $\pi(O^2_{t-1}, O^1_t)$ would be a collider-free path such that all its non end-point vertices, if any, are unobserved. Consequently, there would need to be an edge between $O^2_{t-1}$ and $O^1_t$ in $\M^1(\D^\prime)$.
\end{enumerate}
Since the combination of points 6 and 7 in this enumeration contradicts point 4 of the enumeration, such a path $\pi$ cannot exist. Consequently, $O^1_{t-1}$ and $O^1_t$ are $d$-separated in $\G$. The second rule of the $do$-calculus, e.g.~\citet{Pearl2009}, thus gives that the interventional distribution $P(O^1_t ~\vert~ do(O^1_{t-1} = o^1_{t-1}))$ is expressed in terms of the observational distribution as $P(O^1_t ~\vert~ do(O^1_{t-1}= o^1_{t-1})) = P(O^1_t ~\vert~ O^1_{t-1}= o^1_{t-1})$.
\hfill $\square$

\subsection{Increasing the number of observed time steps}\label{secapp:increasing_taumax}
The main text consides ts-DMAGs and ts-DPAGs on observed time windows $[t-\taumax, t]$, where $\taumax \geq 0$ is arbitrary but fixed. In Sec.~\ref{secapp:different_time_windows} we first compare ts-DMAGs and ts-DPAGs on time windows of different length. We show that, as expected, the ts-DMAGs and ts-DPAGs on the longer time window can never contain less but may contain more information about the underlying ts-DAG than the ts-DMAGs and ts-DPAGs on the shorter time window. In Sec.~\ref{secapp:limiting} we then define the notions of limiting ts-DMAGs and ts-DPAGs by allowing conditioning sets from the entire past. All proofs are given in Sec.~\ref{secapp:proofs_for_increasing}.

\subsubsection{Comparison of ts-DMAGs and ts-DPAGs on different observed time windows}\label{secapp:different_time_windows}
Since the reference time step $t$ is arbitrary and only time \emph{differences} are relevant, we need only compare ts-DMAGs and ts-DPAGs on $[t-\taumax, t]$ and $[t-\taumaxtilde, t]$ with $\taumaxtilde > \taumax$. To this end we use the following notation.

\begin{mydef}[Subgraph of a ts-DMAG / ts-DPAG induced on time window]\label{defapp:sub_ts_DMAG/DPAG}
Let $\D$ be a ts-DAG, let $\taumaxtilde \geq \taumax \geq 0$, and let $t-\taumaxtilde \leq t_1 \leq t_2 \leq t$. The induced subgraph of $\Mtaumaxtilde(\D)$ (subgraph of $\PAGtaumaxtilde(\D)$) on its subset of vertices within $[t_1, t_2]$ is denoted as $\M^{\taumaxtilde, [t_1, t_2]}(\D)$ (denoted as $\PAG^{\taumaxtilde, [t_1, t_2]}(\D)$).
\end{mydef}

The additionally observed vertices in $[t-\taumaxtilde,t-\taumax-1]$ enlarge the set of potential conditions and thus may lead to more $d$-separations among the originally observed vertices. We thus get the following result.

\begin{mylemma}\label{lemmaapp:DMAGs_different_taumax_2}
Let $\D$ be a ts-DAG and let $\taumaxtilde > \taumax \geq 0$. Then, up to relabeling vertices:
\begin{enumerate}
\item For all $0 \leq \Delta t < \taumaxtilde - \taumax$: $\M^{\taumaxtilde, [t-\taumax-\Delta t, t - \Delta t]}(\D)$ is a subgraph of $\Mtaumax(\D)$.
\item $\M^{\taumaxtilde, [t-\taumaxtilde, t - \taumaxtilde + \taumax]}(\D)$ equals $\Mtaumax(\D)$.
\item There are cases in which $\M^{\taumaxtilde, [t-\taumax, t]}(\D)$ is a proper subgraph of $\Mtaumax(\D)$.
\end{enumerate}
\end{mylemma}

Moving to a semantic level, we confirm the intuition that $\Mtaumaxtilde(\D)$ cannot contain less but may contain more information about $\D$ than $\Mtaumax(\D)$.

\begin{mylemma}\label{lemmaapp:DMAGs_different_taumax_inference}
Let $\D_1$ and $\D_2$ be ts-DAGs and let $\taumaxtilde > \taumax \geq 0$. Then:
\begin{enumerate}
\item If $\Mtaumaxtilde(\D_1) = \Mtaumaxtilde(\D_2)$, then $\Mtaumax(\D_1) = \Mtaumax(\D_2)$.
\item There are cases in which $\Mtaumaxtilde(\D_1) \neq \Mtaumaxtilde(\D_2)$ and $\Mtaumax(\D_1) = \Mtaumax(\D_2)$.
\end{enumerate}
\end{mylemma}

In other words: Every inference about $\D$ that can be drawn from $\Mtaumax(\D)$ can also be drawn from $\Mtaumaxtilde(\D)$, whereas the converse need not be true.

\begin{myexample}\label{myexampleapp:different_taumax_DMAG}
Figure~\ref{figapp:DMAG_larger_taumax} shows the ts-DMAGs $\M^1(\D)$ and $\M^2(\D)$. These graphs conform with parts 1 and 2 of Lemma~\ref{lemmaapp:DMAGs_different_taumax_2} and prove its part 3. Further, the edge $O^1_{t-2} \tailhead O^2_t$ in $\M^2(\D)$ shows that $O^1_{t-2}$ is an ancestor of $O^2_t$ in $\D$, which is a conclusion that cannot be drawn from $\M^1(\D)$.
\end{myexample}

\renewcommand{\thefigure}{B}
\begin{figure}[tb]
\centering
\includegraphics[width=0.8\linewidth, page = 1]{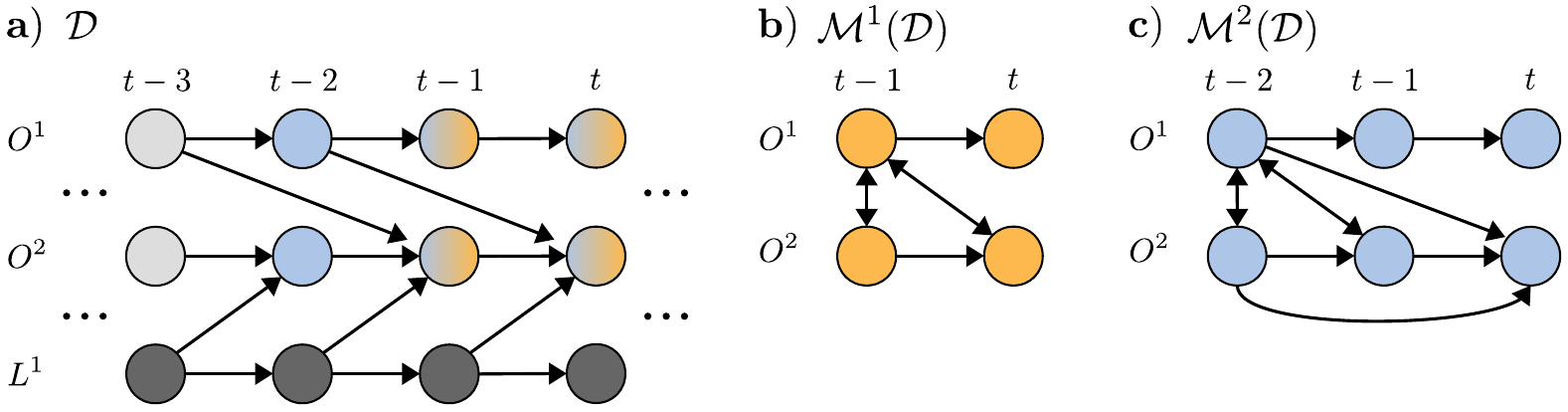}
\caption{Illustration of ts-DMAGs $\Mtaumax(\D)$ of the same ts-DAG $\D$ for different $\taumax$, see also the discussion in Example~\ref{myexampleapp:different_taumax_DMAG}. The component time series $L^1$ is unobservable.
}
\label{figapp:DMAG_larger_taumax}
\end{figure}

Since ts-DPAGs $\PAGtaumax(\D)$ by definition have the same adjacencies as the corresponding ts-DMAGs $\Mtaumax(\D)$, the effect of increasing $\taumax$ on their adjacencies is the same as for ts-DMAGs. Regarding edge orientations, Lemma~\ref{lemmaapp:DMAGs_different_taumax_inference} raises the expectation that all unambiguous edge marks in $\PAGtaumax(\D)$ should also be in $\PAGtaumaxtilde(\D)$. This is expectation is indeed correct.

\begin{mylemma}\label{lemmaapp:DPAGs_different_taumax}
Let $\D$ be a ts-DAG and let $\taumaxtilde > \taumax \geq 0$. Let $(i, t_i)$ and $(j, t_j)$ be adjacent in both $\PAGtaumax(\D)$ and $\PAGtaumaxtilde(\D)$. Then:
\begin{enumerate} 
  \item If there is a non-circle mark on $(i, t_i) \astast (j, t_j)$ in $\PAGtaumax(\D)$, then the same non-circle mark is also on $(i, t_i) \astast (j, t_j)$ in $\PAGtaumaxtilde(\D)$.
  \item There are cases in which there is a non-circle mark on $(i, t_i) \astast (j, t_j)$ in $\PAGtaumaxtilde(\D)$ that is not on  $(i, t_i) \astast (j, t_j)$ in $\PAGtaumax(\D)$.
\end{enumerate}
\end{mylemma}

\begin{mylemma}\label{lemmaapp:DPAGs_different_taumax_2}
Let $\D$ be a ts-DAG and let $\taumaxtilde > \taumax \geq 0$. Then:
\begin{enumerate}
\item Every circle edge mark in $\PAG^{\taumaxtilde, [t-\taumax, t]}(\D)$ is also in $\PAGtaumax(\D)$.
\item There are cases in which there is a non-circle edge mark in $\PAG^{\taumaxtilde, [t-\taumax, t]}(\D)$ that is not also in $\PAGtaumax(\D)$.
\end{enumerate}
\end{mylemma}

\begin{myexample}\label{myexampleapp:DPAG_larger_taumax}
Figure~\ref{fig:DPAG_larger_taumax} shows the ts-DPAGs $\PAG^1(\D)$ and $\PAG^2(\D)$ for a ts-DAG $\D$. These graphs conform with parts 1~of Lemmas~\ref{lemmaapp:DPAGs_different_taumax} and \ref{lemmaapp:DPAGs_different_taumax_2} and prove parts 2~of both these lemmas. For example, in $\PAG^1(\D)$ there is $O^1_{t-1} \ohead O^1_t$ while in $\PAG^2(\D)$ there is $O^1_{t-1} \tailhead O^1_t$ instead.
\end{myexample}

\renewcommand{\thefigure}{C}
\begin{figure}[b]
\centering
\includegraphics[width=0.8\linewidth, page = 1]{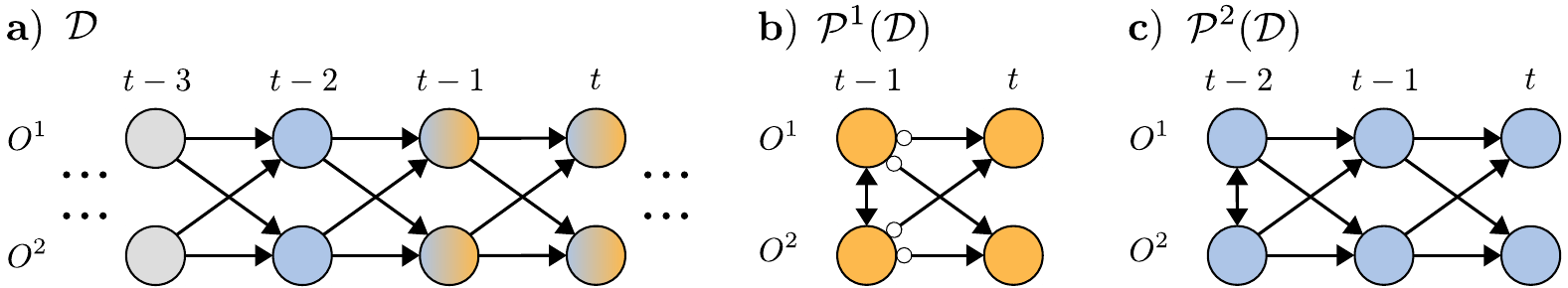}
\caption{Illustration of ts-DPAGs $\PAGtaumax(\D)$ of the same ts-DAG $\D$ for different $\taumax$.
}
\label{fig:DPAG_larger_taumax}
\end{figure}

\subsubsection{Limiting ts-DMAGs and limiting ts-DPAGs}\label{secapp:limiting}
Lemmas~\ref{lemmaapp:DMAGs_different_taumax_2} and \ref{lemmaapp:DPAGs_different_taumax} imply the following behaviour when $\taumax$ is kept fixed while $\taumaxtilde \geq \taumax$ increases beyond any bound.

\begin{mylemma}\label{lemmaapp:existence_limiting}
Let $\D$ be a ts-DAG and $\taumax \geq 0$. Then:
\begin{enumerate}
\item There is $\taumaxtilde \geq \taumax$ with $\M^{\taumaxtilde^\prime,[t-\taumax,t]}(\D) = \M^{\taumaxtilde ,[t-\taumax,t]}(\D)$ for all $ \taumaxtilde^\prime \geq \taumaxtilde$.
\item There is $\taumaxtilde \geq \taumax$ with $\PAG^{\taumaxtilde^\prime,[t-\taumax,t]}(\D) = \PAG^{\taumaxtilde ,[t-\taumax,t]}(\D)$ for all $ \taumaxtilde^\prime \geq \taumaxtilde$.
\end{enumerate}
\end{mylemma}

Lemma~\ref{lemmaapp:existence_limiting} implies that the sequence $\Delta \taumax \mapsto \M^{\taumax + \Delta \taumax, [t-\taumax,t]}(\D)$ as well as the sequence $\Delta \taumax \mapsto \PAG^{\taumax + \Delta \taumax, [t-\taumax,t]}(\D)$ convergence with respect to the discrete metric\footnote{The discrete metric $d(\cdot, \cdot)$ is defined by $d(x, y) = 1$ if $x=y$ and $d(x, y) = 0$ else.} on the space of ts-DMAGs, respectively space of ts-DPAGs.

\begin{mydef}[Limiting ts-DMAG / ts-DPAG]\label{defapp:limiting_ts_DMAG}
Let $\D$ be a ts-DAG and let $\taumax \geq 0$. The \emph{limiting ts-DMAG} $\Mtaumaxlim(\D)$, respectively \emph{limiting ts-DPAG} $\PAGtaumaxlim(\D)$, is the limit of the sequence $\Delta \taumax \mapsto \M^{\taumax + \Delta \taumax, [t-\taumax,t]}(\D)$, respectively $\Delta \taumax \mapsto \PAG^{\taumax + \Delta \taumax, [t-\taumax,t]}(\D)$, with respect to the discrete metric on the space of ts-DMAGs, respectively space of ts-DPAGs.
\end{mydef}

\renewcommand{\thefigure}{D}
\begin{figure}[tb]
\centering
\includegraphics[width=0.85\linewidth, page = 1]{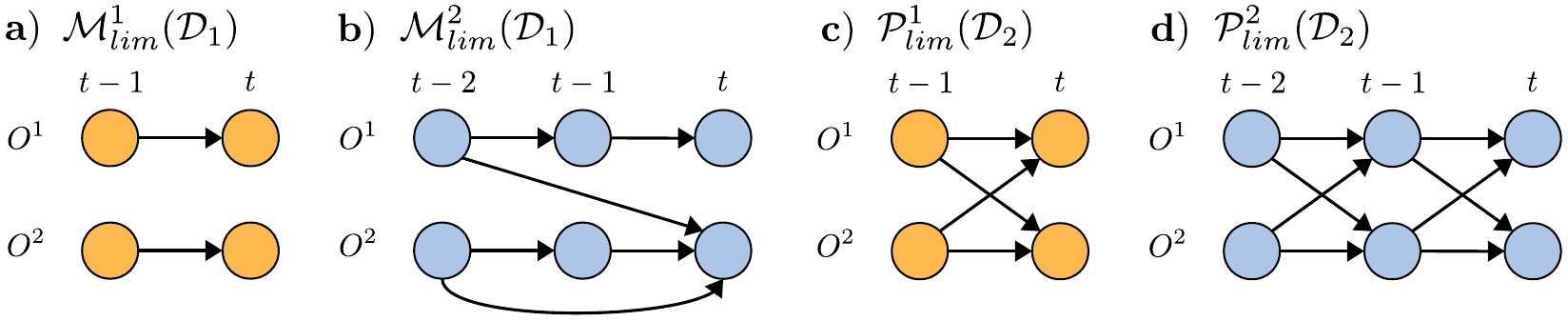}
\caption{Illustration of limiting ts-DMAGs and limiting ts-DPAGs. The underlying ts-DAGs $\D_1$ and $\D_2$ are, respectively, those shown in parts a) of Fig.~\ref{figapp:DMAG_larger_taumax} and Fig.~\ref{fig:DPAG_larger_taumax}.
}
\label{fig:limiting_objects}
\end{figure}

See Fig.~\ref{fig:limiting_objects} for examples. Similar to stationarified ts-DMAGs, limiting ts-DMAGs $\Mtaumaxlim(\D)$ are not in general DMAGs for the underlying ts-DAGs $\D$ and carry different meaning. Namely, vertices $(i, t_i)$ and $(j, t_j)$ with $t_i \leq t_j \leq t$ are adjacent in $\Mtaumaxlim(\D)$ if and only if there is no finite set of observable variables within $(-\infty , t]$ that $d$-separates $(i, t_i)$ and $(j, t_j)$ in $\D$. The same statement applies to limiting ts-DPAGs.

\begin{mylemma}\label{lemmaapp:limiting_DMAGs_properties}
\begin{enumerate}
\item $\Mtaumaxlim(\D)$ has repeating edges.
\item $\Mtaumaxlim(\D)$ is a subgraph of $\Mtaumaxstat(\D)$.
\item $\Mtaumaxlim(\D)$ is a DMAG.
\item $\PAGtaumaxlim(\D)$ has repeating edges.
\item $\PAGtaumaxlim(\D)$ is a DPAG for $\Mtaumaxlim(\D)$.
\end{enumerate}
\end{mylemma}

Unlike the examples shown in parts c) and d) of Fig.~\ref{fig:limiting_objects}, in general there may be circle marks in a limiting ts-DPAG. Lastly, given that $\Mtaumaxlim(\D)$ and $\PAGtaumaxlim(\D)$ have repeating edges, one might hope to give meaning to sending $\taumax$ to infinity in $\Mtaumaxlim(\D)$ and $\PAGtaumaxlim(\D)$ by restricting attention to edges that involve a vertex at time $t$. However, as the following example shows, such a construction is not possible in general.

\begin{myexample}
Consider the ts-DAG $\D$ in part a) of Fig.~\ref{figapp:DMAG_larger_taumax}. Since $L^1$ is unobservable and autocorrelated, for all $\taumax$ there is $O^2_{t-\taumax} \tailhead O^2_t$ in $\M^{\taumax}_{\text{lim}}(\D)$. In \citet{malinsky2018causal} this effect is discussed under the names of ``auto-lag confounders'' and ``infinite-lag associations''.
\end{myexample}

\section{Proofs for Sec.~3 and for Lemma~\ref{lemmaapp:no_unobservables} and Lemma~\ref{lemmaapp:different_DMAGs_not_both_implied}}

\subsection{Proofs for Sec.~3.3}

\begin{proof}[Proof of Lemma~3.5]
See the explanations in Secs.~3.2 and 3.3 of the main text. Formally: The set of random variables involved in the structural process defined in Sec.~3.1 is $\{V^i_t ~|~ 1 \leq i \leq n_V, \, t \in \mathbb{Z}\}$, i.e., is indexed by the set $\Iindex \times \mathbb{Z}$ where $\Iindex = \{1, \dots, n_V\}$. This form shows that the causal graph $\D$ has time series structure with time index set $\mathbb{Z}$, where vertex $(i, t) \in \Iindex \times \mathbb{Z}$ corresponds to random variable $V^i_t$. Further, $\pa((i, t), \D) = P\!A^i_t$ by definition of causal graphs and $P\!A^{i}_t \subseteq \{V^k_{t-\tau} ~|~ 1 \leq k \leq n_V, \, 0 \leq \tau \leq \porder\} \setminus \{V^i_t\}$ by definition of the data generating process. Hence, $\D$ is time ordered. The repeating edges property follows because the data generating process by definition is causally stationary, i.e., because $V^k_{t-\tau} \in P\!A^i_{t}$ if and only if $V^k_{t-\tau-\Delta t} \in P\!A^i_{t-\Delta t}$. Lastly, acyclicity of $\D$ is definitional for the data generating process.
\end{proof}

\subsection{Proofs for Sec.~3.4}

\begin{mylemma}\label{lemmaapp:dag_and_mag_same_ancestral}
Let $\D$ be a DAG with vertex set $\V = \Ovar \cup \La$. Then, for $i, j \in \Ovar$, $i \in \an(j, \D)$ if and only if $j \in \an(i, \M_{\Ovar}(\D))$.
\end{mylemma}

\begin{myremark}[on Lemma~\ref{lemmaapp:dag_and_mag_same_ancestral}]
This claim is a well-known result, see for example \citet{zhang2008causal}, \citet{Zhang2008}, which straightforwardly follows from the definition of the MAG latent projection procedure in \citet{zhang2008causal} as well as from the definitions in \citet{richardson2002}. However, since we did not find a formal proof spelled out in the literature, we here include a proof for completeness.
\end{myremark}

\begin{proof}[Proof of Lemma~\ref{lemmaapp:dag_and_mag_same_ancestral}]
\textbf{Only if.} Assume $i \in \an(j, \D)$. This assumption means that in $\D$ there is a directed path $\pi$ from $i$ to $j$. Let $(k_1, \dots, k_n)$ with $k_1 = i$ and $k_n = j$ be the ordered sequence of nodes on $\pi$ that are in $\Ovar$. Consequently, for all $1 \leq m \leq n-1$ the vertices $k_m$ and $k_{m+1}$ can in $\D$ not be $d$-separated by any set of observed variables, and hence there are the edges $k_m \tailhead k_{m+1}$ in $\M_{\Ovar}(\D)$ for all $1 \leq m \leq n-1$. These edges give a directed path from $k_1 = i$ to $k_n = j$ in $\M_{\Ovar}(\D)$, and hence $i \in \an(j, \M_{\Ovar}(\D))$.

\textbf{If.} Assume $i \in \an(j, \M_{\Ovar}(\D))$. This assumption means that in $\M_{\Ovar}(\D)$ there is a directed path $\pi$ from $i$ to $j. $Let $(k_1, \dots, k_{n^\prime})$ with $k_1 = i$ and $k_{n^\prime} = j$ be the ordered sequence of nodes on $\pi$. By definition of edge orientations in $\M$ we thus get $k_m \in \an(k_{m+1}, \D)$ for all $1 \leq m \leq n^\prime-1$, and hence $i \in \an(j, \D)$ by transitivity of ancestorship.
\end{proof}

\begin{proof}[Proof of Lemma~3.7]
\textbf{1.}
The vertex set of $\M_{\IindexO \times \TindexO}(\D)$ is $\IindexO \times \TindexO$. This decomposition defines the time series structure of $\M_{\IindexO \times \TindexO}(\D)$, namely: $\IindexO$ is its variable index set and $\TindexO$ is its time index set.

\textbf{2.}
Assume $\M_{\IindexO \times \TindexO}(\D)$ is not time ordered, i.e., assume there is $(j, t_j) \tailhead (i, t_i)$ in $\M_{\IindexO \times \TindexO}(\D)$ with $t_j > t_i$. This assumptions means $(j, t_j) \in \an((i, t_i), \M_{\IindexO \times \TindexO}(\D))$ and thus, by Lemma~\ref{lemmaapp:dag_and_mag_same_ancestral}, $(j, t_j) \in \an((i, t_i), \D)$. The latter in turn implies that in $\D$ there is directed path $\pi$ from $(j, t_j)$ to $(i, t_i)$. This path must at least contain one edge $(k, t_k) \tailhead (l, t_l)$ with $t_k > t_l$, which contradicts time order of $\D$.

\textbf{3.}
See part b) of Fig.~2 in the main text for an example.
\end{proof}

\section{Proofs for Sec.~4}

\subsection{Proofs for Sec.~4.2}

\begin{proof}[Proof of Lemma~4.1]
\textbf{1.}
The desired ts-DAG $\D^\prime$ is constructed by treating the vertices at all time steps other than $t - m \cdot n$ with $m \in \mathbb{Z}$ as members of unobservable time series in $\D^\prime$, by shifting all vertices of $\D$ within a time window $[t - m \cdot n - (n-1), t - m \cdot n]$ to time $t - m \cdot n$ in $\D^\prime$ for all $m \in \mathbb Z$, and by then relabeling the time steps according to $t - m \cdot n \mapsto t - m$. Formally: Let $\Iindex$ with $\IindexO \subseteq \Iindex$ denote the variable index set of $\D$. Define $\Iindex^\prime = \Iindex \cup \mathbf{K}$ with $\mathbf{K} = \Iindex \times \{1, \dots, n-1\}$ and consider the following map:
\begin{align*}
\phi: \qquad \Iindex \times \mathbb{Z} &\to \Iindex^\prime \times \mathbb{Z} \\
\left(i, t - \Delta t\right) &\mapsto \begin{cases}
\left(i, t - \frac{\Delta t}{n}\right) &\text{for $\Delta t\, \operatorname{mod}\, n = 0$}\\
\left((i, \Delta t\, \operatorname{mod}\, n), t - \left\lfloor \frac{\Delta t}{n} \right\rfloor\right) &\text{for $\Delta t\, \operatorname{mod}\, n \neq 0$} \ .
\end{cases}
\end{align*}
Here, $(\Delta t \, \operatorname{mod}\, n)$ with negative $\Delta t$ is defined as the smallest non-negative integer $\Delta t + n \cdot m$ with $m \in \mathbb{Z}$. Note that $\phi$ is bijective and, hence, invertible. We define $\D^\prime$ as the directed graph over the vertex set $\V^\prime = \Iindex^\prime \times \mathbb{Z}$ such that for vertices $a, b \in \V^\prime$ there is an edge $a \tailhead b$ if and only if $\phi^{-1}(a) \tailhead \phi^{-1}(b)$ in $\D$. See parts a) and b) of Fig.~\ref{fig:both_sampling_equivalent} for illustration.

This construction is such that $\D$ and $\D^\prime$ are as graphs equal up to relabeling their vertices according to $\phi$. As a consequence, $\D^\prime$ is acyclic and its $d$-separations are the same as those of $\D$. Moreover, $\D^\prime$ is indeed a ts-DAG: First, its time series structure is given by the decomposition of $\V^\prime$ into $\V^\prime = \Iindex^\prime \times \mathbb{Z}$, i.e., $\Iindex^\prime$ is its variable index set. Second, time order follows because if $b \in \V^\prime$ is after $a \in \V^\prime$, then $\phi^{-1}(b)$ is after $\phi^{-1}(a)$ together with the fact that $\D$ is time ordered. Third, if for $a = (i^\prime, t^\prime) \in \V^\prime$ we write $\phi^{-1}(a) = (i, t)$, then $\phi^{-1}((i^\prime, t^\prime - \Delta t^\prime)) = (i, t - n \cdot \Delta t^\prime)$. This observation implies that $\D^\prime$ has repeating edges. Hence, $\D^\prime$ is a ts-DAG and the statement follows since $\phi(\IindexO \times \TindexO^n) = \IindexO \times \TindexO^1$.

\textbf{2.}
The desired ts-DAG $\D^\prime$ is constructed by stretching all edges of $\D$ by a factor of $n$ in time and adding $(n-1)$ further copies of this stretched version of $\D$ to $\D^\prime$, respectively shifted by $1$ up to $(n-1)$ time steps with respect to the first copy, without any edges between the $n$ copies. Formally: $\D^\prime$ is the ts-DAG over the vertex set $\V^\prime = \Iindex \times \mathbb{Z}$, where $\Iindex$ is the variable index set of $\D$, such that $(i, t^\prime - \Delta t^\prime) \tailhead (j, t^\prime)$ in $\D^\prime$ if and only if $(\Delta t^\prime \, \operatorname{mod}\, n ) = 0$ and $(i, t^\prime - \Delta t^\prime/n) \tailhead (j, t^\prime)$ in $\D$. See parts c) and d) of Fig.~\ref{fig:both_sampling_equivalent} for illustration. The statement is apparent from this construction.
\end{proof}

\renewcommand{\thefigure}{E}
\begin{figure}[tb]
\centering
\includegraphics[width=0.99\linewidth, page = 1]{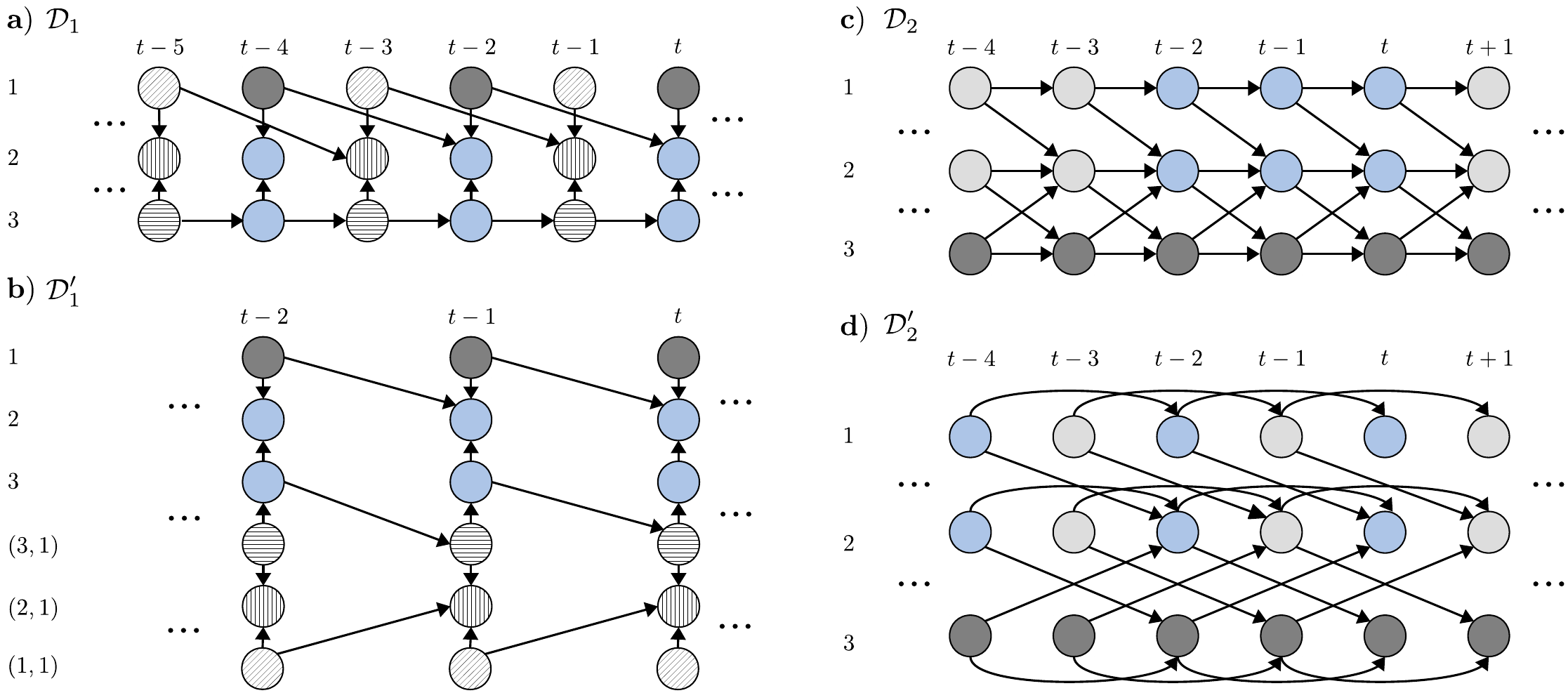}
\caption{Illustration of the constructions involved in the proof of Lemma~4.1 for $n = 2$ and $n_{steps} = 3$. The vertically arranged numbers to the left of the four ts-DAGs are their respective variable indices. \textbf{a)} The same ts-DAG as in part c) of Fig.~2, here denoted $\D_1$. Here, $\IindexO = \{2,\, 3\} \subseteq \Iindex = \{1, \, 2, \,3\}$ and $\TindexO = \{t-4, \,t-2, \,t\}$. The corresponding implied ts-DMAG $\M_{\IindexO \times \TindexO}(\D_1)$ is shown in part d) of Fig.~2. \textbf{b)} The ts-DAG $\D_1^\prime$ constructed from $\D_1$ as defined in the proof of part 1~of Lemma~4.1. Here, $\IindexO^\prime = \IindexO = \{2,\, 3\} \subseteq \Iindex^\prime = \Iindex \cup \mathbf{K}$ with $\mathbf{K} = \Iindex \times \{1\}$ and $\TindexO^\prime = \{t-2, \,t-1, \,t\}$. Note that while in $\D_1$ the hatched vertices are temporally unobserved, in $\D_1^\prime$ they are unobservable. The corresponding implied ts-DMAG $\M_{\IindexO^\prime \times \TindexO^\prime}(\D_1^\prime)$ is the same as $\M_{\IindexO \times \TindexO}(\D_1)$ up to relabeling vertices. \textbf{c)} The same ts-DAG as in part a) of Fig.~2, here denoted $\D_2$. Here, $\IindexO = \{1,\, 2\}\subseteq \Iindex = \{1, \, 2, \,3\}$ and $\TindexO = \{t-2, \,t-1, \,t\}$. The corresponding implied ts-DMAG $\M_{\IindexO \times \TindexO}(\D_2)$ is shown in part b) of Fig.~2. \textbf{d)} The ts-DAG $\D_2^\prime$ constructed from $\D_2$ as defined in the proof of part 2~of Lemma~4.1. Here, $\IindexO^\prime = \IindexO = \{1,\, 2\} \subseteq \Iindex^\prime = \Iindex$ and $\TindexO^\prime = \{t-4, \,t-2, \,t\}$. The corresponding implied ts-DMAG $\M_{\IindexO^\prime \times \TindexO^\prime}(\D_2^\prime)$ is the same as $\M_{\IindexO \times \TindexO}(\D_2)$ up to relabeling vertices.
}
\label{fig:both_sampling_equivalent}
\end{figure}

\subsection{Proofs for Sec.~4.3}

\begin{proof}[Proof of Lemma~4.3]
In all statements that involve the repeating ancestral relationships or repeating separating sets property, we implicitly assume the graph to be a DMAG (because else these properties would be undefined).

\textbf{1. \& 2.}
These statements immediately follow from the definitions of the involved properties.

\textbf{3.}
This statement follows because the ancestral relationships between an adjacent pair of vertices uniquely specifies the type of the edge between this pair of vertices.

\textbf{4.}
This statement follows because for $\Tindex = \mathbb{Z}$ repeating edges implies that the graph is invariant under time shifts, i.e., invariant under the mapping $\phi_{\Delta t}: \Iindex \times \Tindex \to \Iindex \times \Tindex$ with $\phi_{\Delta t}((i, t_i)) = (i, t_i + \Delta t)$ for all $\Delta t \in \mathbb{Z}$.
\end{proof}

\begin{proof}[Proof of Lemma~4.4]
\textbf{1.}
This statement follows because $\Mtaumax(\D)$ and $\D$ have the same ancestral relationships between vertices in $\Mtaumax(\D)$ (according to Lemma~\ref{lemmaapp:dag_and_mag_same_ancestral}) in combination with the fact that $\D$ has repeating ancestral relationships (as implied by part 4 of Lemma~4.3).

\textbf{2.}
Combine part 1 of Lemma~4.4 with part 3 of Lemma~4.3.

\textbf{3.}
Theorem 4.18 in \citet{richardson2002} implies that for sets $\mathbf{S}$ of vertices in $\Mtaumax(\D)$ the $m$-separation $(i, t_i) \ci (j, t_j) ~|~ \mathbf{S}$ holds in $\Mtaumax(\D)$ if and only if the $d$-separation $(i, t_i) \ci (j, t_j) ~|~ \mathbf{S}$ holds in $\D$. The statement now follows because $\D$ has repeating separating sets as implied by part 4 of Lemma~4.3.

\textbf{4.}
Let $(i, t_i)$ and $(j, t_j)$ be non-adjacent in $\Mtaumax(\D)$ and without loss of generality assume $t_i \leq t_j$. Consequently, there is a set of vertices $\mathbf{S}$ in $\Mtaumax(\D)$ with $\mathbf{S} \cap \{(i, t_i), (j, t_j)\} = \emptyset$ such that $(i, t_i) \ci (j, t_j) ~|~ \mathbf{S}$ in $\Mtaumax(\D)$. Due to time order of $\D$, no $(k, t_k)$ with $t_j < t_k$ can be an ancestor of $(i, t_i)$ or of $(j, t_j)$. Lemma S5 in the supplementary material of \citet{LPCMCI} then asserts that $(i, t_i) \ci (j, t_j) ~|~ \mathbf{S}^\prime$ in $\Mtaumax(\D)$ with $\mathbf{S}^\prime = \mathbf{S} \cap \{(l, t_l) ~|~ t_l \leq t_j\}$. Now take any $\Delta t \in \mathbb{Z}$ with $0 \leq \Delta t \leq t - t_j$ and let $\mathbf{S}^\prime_{\Delta t}$ be obtained by shifting all vertices in $\mathbf{S}$ forward in time by $\Delta t$ time steps. By construction of $\mathbf{S}^\prime$ all nodes in $\mathbf{S}^\prime_{\Delta t}$ are within $[t-\taumax, t]$. The repeating separating sets property of $\Mtaumax(\D)$, as asserted by part 2 of Lemma~4.4. and already proven, then implies $(i, t_i + \Delta t) \ci (j, t_j + \Delta t) ~|~ \mathbf{S}^\prime_{\Delta t}$. This fact proves the contraposition of the statement.

\textbf{5.}
Part 1 of Lemma~4.4 implies that $\Mtaumax(\D)$ has repeating orientations. Thus, part 1 of Lemma~4.3 shows that $\Mtaumax(\D)$ would also necessarily have repeating edges if $\Mtaumax(\D)$ necessarily had repeating adjacencies, thereby contradicting part 3 of Lemma~3.7.
\end{proof}

\subsection{Proofs for Sec.~4.4 and of Lemma~\ref{lemmaapp:no_unobservables}}~\label{secapp:contains_proof_for_why_no_unobservable_special}

\begin{mylemma}\label{lemmaapp:stat_of_implied_DMAG}
Let $\G$ be a directed partial mixed graph with time series structure that has repeating orientations and past-repeating adjacencies. Then, $\stat(\G)$ is the unique subgraph $\G$ in which $(i, t_i)$ and $(j, t_j)$ with $\tau = t_j - t_i \geq 0$ are adjacent if and only if $(i, t - \tau)$ and $(j, t)$ are adjacent in $\G$.
\end{mylemma}

\begin{proof}[Proof of Lemma~\ref{lemmaapp:stat_of_implied_DMAG}.]
The statement uniquely determines $\stat(\G)$ for the following reason: First, as is immediate from Def.~4.6 and asserted by the statement, $\stat(\G)$ is a subgraph of $\G$. Second, the statement specifies the set of edges that are in $\G$ but not in $\stat(\G)$. Consequently, $\stat(\G)$ is obtained by deleting a specified set of edges from $\G$.

It remains to be shown that $\stat(\G)$ has the asserted property: First, consider two vertices $(i, t_i)$ and $(j, t_j)$ with $\tau = t_j - t_i \geq 0$ that are adjacent in $\stat(\G)$. Since $\stat(\G)$ has repeating edges we thus $(i, t - \tau)$ and $(j, t)$ are adjacent in $\stat(\G)$, which in turn gives that $(i, t - \tau)$ and $(j, t)$ are adjacent in $\G$ because it is a supergraph of $\stat(\G)$. Second, consider two vertices $(i, t - \tau)$ and $(j, t)$ that are adjacent in $\G$. The past-repeating adjacencies property of $\G$ then implies that $(i, t_i + \Delta t)$ and $(j, t_j + \Delta t)$ are adjacent in $\Mtaumax(\D)$ for all $\Delta t$ with $(i, t_i + \Delta t), (j, t_j + \Delta t) \in \V$. Moreover, since $\G$ has repeating orientations, all these edges have the same orientation. By the second point in Def.~4.6 we thus get that $(i, t_i)$ and $(j, t_j)$ are adjacent in $\stat(\G)$.
\end{proof}

\begin{proof}[Proof of Lemma~4.7]
Apply Lemma~\ref{lemmaapp:stat_of_implied_DMAG} to $\G = \Mtaumax(\D)$.
\end{proof}

\begin{proof}[Proof of Lemma~\ref{lemmaapp:no_unobservables}]
Denote the subgraph of $\D$ induced on $\Iindex \times \TindexO$ as $\D_{[t-\taumax, t]}$. This graph clearly has the same set of vertices and the same time series structure as $\mathcal{M}_{\IindexO \times \TindexO}(\D)$ because $\IindexO = \Iindex$ by assumption.

First, we show that $\D_{[t-\taumax, t]}$ and $\stat(\mathcal{M}_{\IindexO \times \TindexO}(\D))$ have the same adjacencies: Let $(i, t - \tau)$ and $(j, t)$ with $0 \leq \tau = t_j - t_i$ be distinct vertices in $\D$, where without loss of generality $(j, t) \notin \an((i, t-\tau), \D)$. If $(i, t - \tau)$ and $(j, t)$ are adjacent in $\D$, then there is no set $\mathbf{S}$ with $\mathbf{S} \cap \{(i, t - \tau), \,(j, t)\} = \emptyset$ that $d$-separates them. If $(i, t - \tau)$ and $(j, t)$ are non-adjacent in $\D$, then $(i,t-\tau) \ci (j, t) ~|~ \mathbf{S}$ with $\mathbf{S} = \pa((j, t), \D) \setminus \{(i, t-\tau)\}$. By time order of $\D$ and the definition of $\porder$, all vertices in $\pa((j, t), \D) \setminus \{(i, t-\tau)\}$ are within $[t-\porder, t]$. Since $\porder \leq \taumax$ by assumption we thus get: $(i, t - \tau)$ and $(j, t)$ can be $d$-separated in $\D$ by a set of vertices in $\IindexO \times \TindexO$ if and only if they are non-adjacent in $\D$. In combination with repeating edges of $\D$ and Lemma~4.7 the desired claim follows.

Second, $\D_{[t-\taumax, t]}$ and $\mathcal{M}_{\IindexO \times \TindexO}(\D)$ also have the same edge orientations because they have the same ancestral relationships according Lemma~\ref{lemmaapp:dag_and_mag_same_ancestral}.
\end{proof}

\begin{mylemma}\label{lemmaapp:stat_of_implied_DMAG_is_DMAG}
Let $\M$ be a DMAG with time series structure that is time ordered and has repeating orientations and past-repeating adjacencies. Then, $\stat(\M)$ is a DMAG.
\end{mylemma}

\begin{proof}[Proof of Lemma~\ref{lemmaapp:stat_of_implied_DMAG_is_DMAG}]
We have to show that $\stat(\M)$ does not have directed cycles, does not have almost directed cycles, and is maximal.

\emph{No (almost) directed cycles:} Assume that $\stat(\M)$ has a directed or an almost direct cycle. Then, since $\stat(\M)$ is a subgraph of $\M$, also $\M$ has the same directed or almost directed cycle. But then $\M$ is not a DMAG. Contradiction.

\emph{Maximality:} Assume the opposite, i.e., assume in $\stat(\M)$ there are non-adjacent vertices $(i, t_i)$ and $(j, t_j)$, where without loss of generality $\tau = t_j - t_i \geq 0$, between which there is an inducing path $\pi$. We note that $\stat(\M)$ is time ordered because it is a subgraph of the time ordered graph $\M$. Since by definition of inducing paths all vertices on $\pi$ are ancestors of $(i, t_i)$ or $(j, t_j)$, we get that all vertices on $\pi$ are within the time window $[t-\taumax, t_j]$. The repeating edges property of $\stat(\M)$ now shows that $\pi_{t - t_j}$, defined as the ordered sequence of vertices obtained by shifting all vertices on $\pi$ forward in time by $t - t_j$ time steps, is a path in $\stat(\M)$ whose edges are orientated in the same way as the corresponding edges of $\pi$. Moreover, by combining part 1 of Lemma~4.3 with part 1 of Lemma~\ref{lemmaapp:properties_stationarification} we see that the stationarification $\stat(\M)$ has repeating ancestral relationships. Hence, $\pi_{t - t_j}$ is an inducing path between $(i, t - \tau)$ and $(j, t)$ in $\stat(\M)$. Since $\stat(\M) \subseteq \M$, $\pi_{t - t_j}$ is also in $\M$ an inducing path between $(i, t - \tau)$ and $(j, t)$. Maximality of $\M$ thus requires $(i, t - \tau)$ and $(j, t)$ to be adjacent in $\M$. According to Lemma~\ref{lemmaapp:stat_of_implied_DMAG} we then obtain that $(i, t - \tau)$ and $(j, t)$ are adjacent in $\stat(\M)$. Since $(i, t_i)$ and $(j, t_j)$ are non-adjacent in $\stat(\M)$ by assumption, this observation contradicts repeating edges of $\stat(\M)$.
\end{proof}

\begin{proof}[Proof of Lemma~4.8]
Apply Lemma~\ref{lemmaapp:stat_of_implied_DMAG_is_DMAG} to $\M = \Mtaumax(\D)$.
\end{proof}

\begin{proof}[Proof of Lemma~4.10]
Assume $(i, t_i) \in \an((j,t_j), \Mtaumaxstat(\D))$. This assumption means that in $\Mtaumaxstat(\D)$ there is a directed path $\pi$ from $(i, t_i)$ to $(j, t_j)$. Since $\Mtaumaxstat(\D)$ is a subgraph of $\Mtaumax(\D)$, this path $\pi$ is also in $\Mtaumax(\D)$. Hence, $(i, t_i) \in \an((j,t_j), \Mtaumax(\D))$.

Assume $(i, t_i) \in \an((j,t_j), \Mtaumax(\D))$. This assumption by Lemma~\ref{lemmaapp:dag_and_mag_same_ancestral} implies that $(i, t_i) \in \an((j,t_j), \D)$, and hence there is a directed path $\pi$ from $(i, t_i)$ to $(j, t_j)$ in $\D$. Since $\D$ is time ordered, all vertices on $\pi$ are within $[t - \taumax, t]$. Let $((k_1, t_1), \dots , (k_n, t_n))$ with $(k_1, t_1) = (i, t_i)$ and $(k_n, t_n) = (j, t_j)$ be the ordered sequence of observed vertices on $\pi$. For all $1 \leq m \leq n-1$ let $\pi_m$ be the ordered sequence of vertices obtained by shifting $\pi((k_m, t_m), (k_{m+1}, t_{m+1}))$ by $t - t_{m+1}$ time steps forward in time. These paths $\pi_m$ are directed paths from $(k_m, t - (t_{m+1} - t_m))$ to $(k_{m+1}, t)$ in $\D$ and all their non end-point vertices unobservable. Hence, the paths $\pi_m$ cannot be blocked by any set of observable variables, which implies that in $\Mtaumax(\D)$ there are the edges $(k_m, t - (t_{m+1} - t_m)) \tailhead (k_{m+1}, t)$. According to Lemma~4.7, we thus get that $(k_m, t - (t_{m+1} - t_m)) \tailhead (k_{m+1}, t)$ are in $\Mtaumaxstat(\D)$, which due to repeating edges of $\Mtaumaxstat(\D)$ in turn gives $(k_m, t_{m}) \tailhead (k_{m+1}, t_{m+1})$ in $\Mtaumaxstat(\D)$. These edges combine to a directed path from $(k_1, t_1) = (i, t_i)$ to $(k_n, t_n) = (j, t_j)$ in $\Mtaumaxstat(\D)$, hence $(i, t_i) \in \an((j,t_j), \Mtaumaxstat(\D))$.
\end{proof}

\begin{mylemma}\label{lemmaapp:stat_of_rep_edges}
Let $\G$ be a directed partial mixed graph with time structure that has repeating edges. Then, $\stat(\G) = \G$.
\end{mylemma}

\begin{proof}[Proof of Lemma~\ref{lemmaapp:stat_of_rep_edges}.]
Apply part 2 of Lemma~\ref{lemmaapp:properties_stationarification} for $(\G^\prime, \G) = (\G, \G)$ to see that $\G$ is a subgraph of $\stat(\G)$. Since $\stat(\G)$ is a subgraph of $\G$, as immediately implied by the second point in Def.~4.6, this observation shows $\stat(\G) = \G$.
\end{proof}

\begin{mylemma}\label{lemmaapp:stat_of_stat}
Let $\G$ be a directed partial mixed graph with time structure. Then $\stat(\stat(\G)) = \stat(\G)$.
\end{mylemma}

\begin{proof}[Proof of Lemma~\ref{lemmaapp:stat_of_stat}]
Combine part 1 of Lemma~\ref{lemmaapp:properties_stationarification} with Lemma~\ref{lemmaapp:stat_of_rep_edges}.
\end{proof}

\subsection{Proofs for Sec.~4.5 other than Lemma~4.14}

\begin{proof}[Proof of Lemma~4.12]
This statement is implied by Theorem 6.4 in \citet{richardson2002}.
\end{proof}

\subsection{Proof of Lemma~4.14}
We split the proof into three parts that are respectively given in Secs.~\ref{secapp:proof_part_1}, \ref{secapp:proof_part_2} and \ref{secapp:proof_part_3}. Moreover, we collect several auxiliary results and definitions in \ref{secapp:proof_prelim}. For ease of notation, in Sec.~\ref{secapp:proof_part_3} we do not always denote vertices by the tuples of their variable and time indices but sometimes just with a single character, for example $v$ instead of $(i, t)$. 

\subsubsection{Auxiliary results and definitions}\label{secapp:proof_prelim}

\begin{mylemma}\label{lemmaapp:meaning_of_edges_in_DMAG}
Let $\D$ be a DAG over vertices $\V$ and let $\Ovar \subseteq \V$ be the subset of observed vertices. Then:
\begin{enumerate}
  \item If $i \tailhead j$ in $\M_{\Ovar}(\D)$, then for every subset $\mathbf{S} \subseteq \Ovar$ that does not contain $i$ or $j$ there is path $\pi$ between $i$ and $j$ in $\D$ that is active given $\mathbf{S}$ and into $j$.
  \item If $i \headhead j$ in $\M_{\Ovar}(\D)$, then for every subset $\mathbf{S} \subseteq \Ovar$ that does not contain $i$ or $j$ there is path $\pi$ between $i$ and $j$ in $\D$ that is active given $\mathbf{S}$ and into both $i$ and $j$.
\end{enumerate}
\end{mylemma}

\begin{myremark}[on Lemma~\ref{lemmaapp:meaning_of_edges_in_DMAG}]
This result might have appeared in the literature before. Also note that the presence of $i \tailhead j$ in $\M$ does \emph{not} imply that for all $\mathbf{S}$ as above there is path between $i$ and $j$ in $\D$ that is active given $\mathbf{S}$ and out of $i$. As an example, consider the DAG over $\V = \{i, j, k, l\}$ constituted by $i \tailhead j \tailhead k$ together with $i \headtail l \tailhead k$ and choose $\Ovar = \{i, j, k\}$: Although $i \tailhead k$ in $\M$, for $\mathbf{S} = \{j\}$ the only active path in $\D$ is $i \headtail l \tailhead k$.
\end{myremark}

\begin{proof}[Proof of Lemma~\ref{lemmaapp:meaning_of_edges_in_DMAG}]
\textbf{1.}
The fact that $i \astast j$ is in $\M_{\Ovar}(\D)$ is by Theorem 4.2 in \citet{richardson2002} equivalent to the existence of an inducing path $\rho$ relative to $\Ovar$ in $\D$ between $i$ and $j$. Assume $\rho$ is out of $j$. Then, because $j$ is not an ancestor of $i$ according to $i \asthead j$ in $\M_{\Ovar}(\D)$, there is at least one collider on $\rho$. By definition of inducing paths, all colliders on $\rho$ are ancestors of $i$ or $j$. Let $k$ be the collider on $\rho$ that is closest to $j$ on $\rho$. Because $\rho$ is out of $j$, $j$ is an ancestor of $k$. Transitivity of ancestorship thus implies that $j$ is an ancestor of $i$ or $j$. Both options are a contradiction because there are no directed cycles and because $i \asthead j$ is in $\M_{\Ovar}(\D)$. Hence, $\rho$ is into $j$. The statement now follows by combining Lemmas~6.1.1. and 6.1.2 in \citet{Spirtes2000}.

\textbf{2.}
Arguments similar to those in the proof of part 1 of Lemma~\ref{lemmaapp:meaning_of_edges_in_DMAG} show that there is an inducing path $\rho$ relative to $\Ovar$ in $\D$ between $i$ and $j$ that is into both $i$ and $j$. The statement for $i \headhead j$ in $\M$ follows by Lemma~6.1.2 in \citet{Spirtes2000}.
\end{proof}

\begin{mylemma}\label{lemmaapp:meaning_of_edges_in_DMAG_2}
Let $\D$ be a DAG over vertices $\V$ and let $\Ovar \subseteq \V$ be the subset of observed vertices. Let $i \tailhead j$ be an edge in $\M_{\Ovar}(\D)$ and $\mathbf{S} \subseteq \Ovar \setminus \{i, j\}$. Then: If in $\D$ there is no path between $i$ and $j$ that is active given $\mathbf{S}$ and out of $i$, then $i$ is an ancestor of $\mathbf{S}$ in $\D$.
\end{mylemma}

\begin{proof}[Proof of Lemma~\ref{lemmaapp:meaning_of_edges_in_DMAG_2}]
We know that $i$ is an ancestor of $j$ in $\D$ because $i \tailhead j$ in $\M_{\Ovar}(\D)$. Hence, in $\D$ there is a directed path $\pi$ from $i$ to $j$. Assuming that in $\D$ there is no path between $i$ and $j$ that is active given $\mathbf{S}$ and out of $i$, $\pi$ must be blocked by $\mathbf{S}$. Consequently, $\mathbf{S}$ contains a vertex of $\pi$ and thus a descendant of $i$.
\end{proof}

\begin{mydef}[Observable vertices within a time window]
The set of observable vertices within a time window $[t_1, t_2]$, where $t_1 \leq t_2$, are denoted by $\mathbf{O}(t_1, t_2)$.
\end{mydef}

\begin{mydef}[Observable vertices within a time window not on a given path]
$\mathbf{O}(t_1, t_2)[\pi]$ is the set of all vertices in $\mathbf{O}(t_1, t_2)$ less the non end-point vertices of the path $\pi$.
\end{mydef}

\begin{mydef}[Almost adjacent]\label{defapp:almost_adjacent}
Two distinct observable vertices $(i, t_i)$ and $(j, t_j)$ in $\Dc(\Mtaumax(\D))$ are \emph{almost adjacent} if there is an unobservable vertex $(k, t_k)$ such that $(i, t_i) \headtail (k, t_k) \tailhead (j, t_j)$ in $\Dc(\Mtaumax(\D))$.
\end{mydef}

\begin{mylemma}\label{lemmaapp:almost_adjacent_unique}
If $(i, t_i)$ and $(j, t_j)$ are almost adjacent in $\Dc(\Mtaumax(\D))$, then there is a unique unobservable vertex $(k, t_k)$ such that $(i, t_i) \headtail (k, t_k) \tailhead (j, t_j)$ in $\Dc(\Mtaumax(\D))$.
\end{mylemma}

\begin{proof}[Proof of Lemma~\ref{lemmaapp:almost_adjacent_unique}]
Existence follows because $(i, t_i)$ and $(j, t_j)$ are almost adjacent, uniqueness follows in combination with the definition of canonical ts-DAGs (see Def.~4.13 in the main text).
\end{proof}

\begin{mylemma}\label{lemmaapp:at_most_almost_adjacent}
Let $(i, t_i)$ and $(j, t_j)$ with $t_i \leq t_j$ be distinct observable vertices in $\Dc(\Mtaumax(\D))$. Then, $(i, t_i)$ and $(j, t_j)$ are adjacent or almost adjacent in $\Dc(\Mtaumax(\D))$ if and only if $(i, t - (t_j - t_i))$ and $(j, t)$ are adjacent in $\Mtaumax(\D)$.
\end{mylemma}

\begin{proof}[Proof of Lemma~\ref{lemmaapp:at_most_almost_adjacent}]
\textbf{If.}
The premise is that $(i, t - (t_j - t_i)) \astast (j, t)$ in $\Mtaumax(\D)$, where $\astast$ is $\tailhead$ or $\headtail$ or $\headhead$. Past-repeating adjacencies and repeating orientations of $\Mtaumax(\D)$ then imply $(i, t - (t_j - t_i) - \Delta t) \astast (j, t - \Delta t)$ for all $0 \leq \Delta t \leq \taumax - (t_j - t_i)$, where $\astast$ is the same edge type as between $(i, t - (t_j - t_i))$ and $(j, t)$. Hence, all these edges are also in $\stat(\Mtaumax(\D))$. If $\astast$ is $\tailhead$ or $\headtail$, then the definition of canonical ts-DAGs implies $(i, t - (t_j - t_i) - \Delta t^\prime) \astast (j, t - \Delta t^\prime)$ in $\Dc(\Mtaumax(\D))$ for all $\Delta t^\prime \in \mathbb{Z}$. In particular, $(i, t_i)$ and $(j, t_j)$ are adjacent in $\Dc(\Mtaumax(\D))$. If $\astast$ is $\headhead$, then the definition of canonical ts-DAGs implies that $(i, t - (t_j - t_i) - \Delta t^\prime) \headtail ((i, j, t_j - t_i), t - \Delta t^\prime - (t_j - t_i)) \tailhead (j, t - \Delta t^\prime)$ in $\Dc(\Mtaumax(\D))$ or $t_i = t_j$ and $(i, t - \Delta t^\prime) \headtail ((j, i, 0), t - \Delta t^\prime) \tailhead (j, t - \Delta t^\prime)$ in $\Dc(\Mtaumax(\D))$ for all $\Delta t^\prime \in \mathbb{Z}$. In particular, $(i, t_i)$ and $(j, t_j)$ are almost adjacent in $\Dc(\Mtaumax(\D))$.

\textbf{Only if.}
Since the vertices $(i, t - (t_j - t_i))$ and $(j, t)$ are non-adjacent in $\Mtaumax(\D)$ they are also non-adjacent in $\Mtaumaxstat(\D)$. The statement now follows with the definition of canonical ts-DAGs.
\end{proof}

\begin{mylemma}\label{lemmaapp:at_most_almost_adjacent_2}
Let $(i, t_i)$ and $(j, t_j)$ be distinct observable vertices that are adjacent or almost adjacent in $\Dc(\Mtaumax(\D))$. Then:
\begin{enumerate}
\item $(i, t_i) \tailhead (j, t_j)$ in $\Dc(\Mtaumax(\D))$ if and only if $(i, t_i) \in \an((j, t_j), \D)$.
\item $(i, t_i) \headtail (j, t_j)$ in $\Dc(\Mtaumax(\D))$ if and only if $(j, t_j) \in \an((i, t_i), \D)$.
\item $(i, t_i)$ and $(j, t_j)$ are almost adjacent in $\Dc(\Mtaumax(\D))$ if and only if $(i, t_i) \notin \an((j, t_j), \D)$ and $(j, t_j) \notin \an((i, t_i), \D)$.
\end{enumerate}
\end{mylemma}

\begin{proof}[Proof of Lemma~\ref{lemmaapp:at_most_almost_adjacent_2}]
Assume without loss of generality that $t_i \leq t_j$, else exchange $(i, t_i)$ and $(j, t_j)$. From Lemma~\ref{lemmaapp:at_most_almost_adjacent} it then follows that $(i, t - (t_j - t_i))$ and $(j, t)$ are adjacent in $\Mtaumax(\D)$. The definition of edges in DMAGs in combination with repeating ancestral relationships of $\D$ further implies that
\begin{itemize}
\item $(i, t - (t_j - t_i)) \tailhead (j, t)$ in $\Mtaumax(\D)$ if and only if $(i, t_i) \in \an((j, t_j), \D)$,
\item $(i, t - (t_j - t_i)) \headtail (j, t)$ in $\Mtaumax(\D)$ if and only if $(j, t_j) \notin \an((i, t_i), \D)$,
\item $(i, t - (t_j - t_i)) \headhead (j, t)$ in $\Mtaumax(\D)$ if and only if $(i, t_i) \notin \an((j, t_j), \D)$ and $(j, t_j) \notin \an((i, t_i), \D)$.
\end{itemize}
Now proceed as in the proof of the \emph{if} part of Lemma~\ref{lemmaapp:at_most_almost_adjacent}.
\end{proof}

\subsubsection{Part 1: $\Mtaumax(\D)$ and $\Mtaumax(\Dc(\Mtaumax(\D)))$ have the same ancestral relationships}\label{secapp:proof_part_1}
As the first part of the proof of Lemma~4.14 we here show that $\Mtaumax(\D)$ and $\Mtaumax(\Dc(\Mtaumax(\D)))$ have the same ancestral relationships.

\begin{mylemma}\label{lemmaapp:canonical_ts-DAG_of_stat}
$\Dc(\G) = \Dc(\stat(\G))$.
\end{mylemma}

\begin{proof}[Proof of Lemma~\ref{lemmaapp:canonical_ts-DAG_of_stat}]
An inspection of Defs.~4.6 and 4.13 in the main text reveals that $\Dc(\G)$ is uniquely determined by $\stat(\G)$. The statement thus follows because $\stat(\G) = \stat(\stat(\G))$ according to Lemma~\ref{lemmaapp:stat_of_stat}.
\end{proof}

\begin{mylemma}\label{lemmaapp:stattsDMAG_same_ancestral_Dc}
The stationarified ts-DMAG $\Mtaumaxstat(\D)$ has the same ancestral relationships among vertices in $\Mtaumaxstat(\D)$ as the canonical ts-DAG $\Dc(\Mtaumax(\D))$.
\end{mylemma}

\begin{proof}[Proof of Lemma~\ref{lemmaapp:stattsDMAG_same_ancestral_Dc}]
Assume $(i, t_i) \in \an((j, t_j), \Mtaumaxstat(\D))$. This assumption means that in $\Mtaumaxstat(\D)$ there is a directed path $\pi = ((k_1, t_1), \dots, (k_n, t_n))$ from $(k_1, t_1) = (i, t_i)$ to $(k_n, t_n) = (j, t_j)$. Since $\Mtaumaxstat(\D)$ has repeating edges and is time ordered, the fact that $(k_m, t_m) \tailhead (k_{m+1}, t_{m+1})$ is in $\Mtaumaxstat(\D)$ implies $(k_m, t_m) \tailhead (k_{m+1}, t_{m+1})$ in $\Dc(\Mtaumaxstat(\D))$. Consequently, $\pi$ is also in $\Dc(\Mtaumaxstat(\D))$ and we find that $(i, t_i) \in \an((j, t_j), \Dc(\Mtaumaxstat(\D)))$.

Let $(i, t_i), (j, t_j)$ be vertices in $\Mtaumaxstat(\D)$ and assume $(i, t_i) \in \an((j, t_j), \Dc(\Mtaumaxstat(\D)))$. This assumption means in $\Dc(\Mtaumaxstat(\D))$ there is a directed path $\pi = ((k_1, t_1), \dots, (k_n, t_n))$ from $(k_1, t_1) = (i, t_i)$ to $(k_n, t_n) = (j, t_j)$. Since by definition of canonical ts-DAGs there are no edges into unobservable vertices, all vertices on $\pi$ are observed and thus also in $\Mtaumaxstat(\D)$. Moreover, again by definition of canonical ts-DAGs, any edge of $\Dc(\Mtaumaxstat(\D))$ that is between vertices in $\Mtaumaxstat(\D)$ is also in $\Mtaumaxstat(\D)$. Hence, $\pi$ is also in $\Mtaumaxstat(\D)$ and we find $(i, t_i) \in \an((j, t_j),\Mtaumaxstat(\D))$.

These considerations show that $\Mtaumaxstat(\D)$ and $\Dc(\Mtaumaxstat(\D))$ have the same ancestral relationships among vertices in $\Mtaumaxstat(\D)$. The statement follows because $\Dc(\Mtaumaxstat(\D)) = \Dc(\Mtaumax(\D))$ according to Lemma~\ref{lemmaapp:canonical_ts-DAG_of_stat}.
\end{proof}

\begin{mylemma}\label{lemmaapp:D_same_ancestral_mtaumax_as_Dc(Mtaumax(D))}
Consider a ts-DAG $\D$ and the canonical ts-DAG $\Dc(\Mtaumax(\D))$. Then:
\begin{enumerate}
  \item If $(i, t_i) \in \an((j, t_j), \D)$ and $t_j - t_i \leq \taumax$, then $(i, t_i) \in \an((j, t_j), \Dc(\Mtaumax(\D)))$.
  \item If $(i, t_i) \in \an((j, t_j), \Dc(\Mtaumax(\D)))$, then $(i, t_i) \in \an((j, t_j), \D)$.
\end{enumerate}
\end{mylemma}

\begin{proof}[Proof of Lemma~\ref{lemmaapp:D_same_ancestral_mtaumax_as_Dc(Mtaumax(D))}]
\textbf{1.}
Let $(i, t_i) \in \an((j, t_j), \D)$ with $\tau = t_j - t_i \leq \taumax$, where $\tau \geq 0$ due to time order of $\D$. The repeating ancestral relationships property of $\D$ then gives $(i, t - \tau) \in \an((j, t), \D)$, which implies $(i, t - \tau) \in \an((j, t), \Mtaumax(\D))$ by Lemma~\ref{lemmaapp:dag_and_mag_same_ancestral} and thus $(i, t - \tau) \in \an((j, t), \Mtaumaxstat(\D))$ by Lemma~4.10 and finally $(i, t - \tau) \in \an((j, t), \Dc(\Mtaumax(\D)))$ by Lemma~\ref{lemmaapp:stattsDMAG_same_ancestral_Dc}.

\textbf{2.}
Let $(i, t_i) \in \an((j, t_j), \Dc(\Mtaumax(\D)))$. This premise means that in $\Dc(\Mtaumax(\D))$ there is a directed path $\pi = ((k_1, t_1), \dots , (k_n, t_n))$ from $(k_1, t_1) = (i, t_i)$ to $(k_n, t_n) = (j, t_j)$, where $t_{m} \leq t_{m+1}$ due to time order of $\Dc(\Mtaumax(\D))$. Using repeating ancestral relationships of $\Dc(\Mtaumax(\D))$, we thus get that $(k_m, t - (t_{m+1} - t_m)) \in \an((k_{m+1}, t), \Dc(\Mtaumax(\D)))$ for all $1 \leq m \leq n-1$.  Since by definition of canonical ts-DAGs there are no edges into unobservable vertices, all vertices on $\pi$ are observable. Moreover, again due to definition of canonical ts-DAGs, $\Dc(\Mtaumax(\D))$ cannot contain edges with a lag larger than $\taumax$. These observations require $0 \leq |t_{m+1} - t_m| = t_{m+1} - t_m \leq \taumax$ and thus shows that both $(k_m, t - (t_{m+1} - t_m))$ and $(k_{m+1}, t)$ are vertices in $\Mtaumaxstat(\D)$. Using Lemma~\ref{lemmaapp:stattsDMAG_same_ancestral_Dc} we therefore get $(k_m, t - (t_{m+1} - t_m)) \in \an((k_{m+1}, t), \Mtaumaxstat(\D))$, which in turn gives $(k_m, t - (t_{m+1} - t_m)) \in \an((k_{m+1}, t), \Mtaumax(\D))$ by Lemma~4.10 and thus $(k_m, t - (t_{m+1} - t_m)) \in \an((k_{m+1}, t), \D)$ by Lemma~\ref{lemmaapp:dag_and_mag_same_ancestral} and thus $(k_m, t_m) \in \an((k_{m+1}, t_{m+1}), \D)$ by repeating ancestral relationships of $\D$. The statement now follows from transitivity of ancestorship.
\end{proof}

\begin{mylemma}\label{lemmaapp:same_ancestral_main_proposition}
$\Mtaumax(\D)$ and $\Mtaumax(\Dc(\Mtaumax(\D)))$ have the same ancestral relationships.
\end{mylemma}

\begin{proof}[Proof of Lemma~\ref{lemmaapp:same_ancestral_main_proposition}]
Combine Lemma~\ref{lemmaapp:dag_and_mag_same_ancestral} with Lemma~\ref{lemmaapp:D_same_ancestral_mtaumax_as_Dc(Mtaumax(D))}.
\end{proof}

\subsubsection{Part 2: Any adjacency in $\Mtaumax(\Dc(\Mtaumax(\D)))$ is also in $\Mtaumax(\D)$}\label{secapp:proof_part_2}
As the second part of the proof of Lemma~4.14 we here show that any adjacency in $\Mtaumax(\Dc(\Mtaumax(\D)))$ is also in $\Mtaumax(\D)$. Together with the fact that both these graphs have the same ancestral relationships, as already proven in Sec.~\ref{secapp:proof_part_1}, we then get that $\Mtaumax(\Dc(\Mtaumax(\D)))$ is a subgraph of $\Mtaumax(\D)$.

\begin{mylemma}\label{lemmaapp:MGc_subgraph_of_MG_reason}
Let $(i, t_i)$ and $(j, t_j)$ with $t - \taumax \leq t_i, t_j \leq t$ be distinct observable vertices in $\Dc(\Mtaumax(\D))$ and let $\mathbf{S} \subseteq \mathbf{O}(t-\taumax, t) \setminus\{(i, t_i), (j, t_j)\}$. Then: If $(i, t_i)$ and $(j, t_j)$ are $d$-connected given $\mathbf{S}$ in $\Dc(\Mtaumax(\D))$, then $(i, t_i)$ and $(j, t_j)$ are $d$-connected given $\mathbf{S}$ in $\D$.
\end{mylemma}

\begin{myremark}[on Lemma~\ref{lemmaapp:MGc_subgraph_of_MG_reason}]
The statement makes sense because $\D$ and $\Dc(\Mtaumax(\D))$ have the same observable time series.
\end{myremark}

\begin{proof}[Proof of Lemma~\ref{lemmaapp:MGc_subgraph_of_MG_reason}]
Let $(i, t_i)$ and $(j, t_j)$ be $d$-connected given $\mathbf{S}\subseteq \mathbf{O}(t-\taumax, t) \setminus\{(i, t_i), (j, t_j)\}$ in $\Dc(\Mtaumax(\D))$. Then, in $\Dc(\Mtaumax(\D))$ there is path $\pi$ between $(i, t_i)$ and $(j, t_j)$ that is active given $\mathbf{S}$. Since no node in $\mathbf{S}$ is after $t$, no node on $\pi$ is after $t$ because else due to time order of $\Dc(\Mtaumax(\D))$ there would be a collider on $\pi$ after $t$ that, again due to time order, cannot be unblocked by $\mathbf{S}$. Let $((k_1, t_1), \dots, (k_n, t_n))$ with $(k_1, t_1) = (i, t_i)$ and $(k_n, t_n) = (j, t_j)$ be the ordered sequence of observable vertices on $\pi$ (not necessarily temporally observed, so some of these vertices may be before $t-\taumax$). Since in $\Dc(\Mtaumax(\D))$ there are no edges into unobservable vertices and no edges with a lag larger than $\taumax$, the subpaths $\pi_m = \pi((k_m, t_m), (k_{m+1}, t_{m+1}))$ with $1 \leq m \leq n -1$ are either of the form $(k_m, t_m) \tailhead (k_{m+1}, t_{m+1})$ or $(k_m, t_m) \headtail (k_{m+1}, t_{m+1})$ or $(k_m, t_m) \headtail (l_m, t_{l_m}) \tailhead (k_{m+1}, t_{m+1})$ with $(l_m, t_{l_m})$ unobservable. In all cases $|t_m - t_{m+1}| \leq \taumax$. Now associate to each $\pi_m$ a path $\rho_m$ in $\D$ between $(k_m, t_m)$ and $(k_{m+1}, t_{m+1})$ in the following way:

\textit{Case 1:} If $\pi_m$ is $(k_m, t_m) \tailhead (k_{m+1}, t_{m+1})$, then $(k_m, t - (t_{m+1} - t_m)) \tailhead (k_{m+1}, t)$ in $\Dc(\Mtaumax(\D))$ by repeating edges of $\Dc(\Mtaumax(\D))$ and hence $(k_m, t - (t_{m+1} - t_m)) \tailhead (k_{m+1}, t)$ in $\Mtaumax(\D)$ by Lemmas~\ref{lemmaapp:at_most_almost_adjacent} and \ref{lemmaapp:at_most_almost_adjacent_2}. According to Lemma~\ref{lemmaapp:meaning_of_edges_in_DMAG} there thus is path between $(k_m, t - (t_{m+1} - t_m))$ and $(k_{m+1}, t)$ in $\D$ that is into $(k_{m+1}, t)$ and active given $\mathbf{S}_{m, t - t_{m+1}}$, where $\mathbf{S}_{m, t - t_{m+1}}$ is obtained by shifting $\mathbf{S}_m = \mathbf{S} \setminus \{(t_m, k_m), (t_{m+1}, k_{m+1})\}$ forward in time by $t - t_{m+1}$ time steps. Let $\mathbf{p}_m$ be the set of all such paths. If any path in $\mathbf{p}_m$ is out of $(k_m, t - (t_{m+1} - t_m))$, then let $\rho_{m, t - t_{m+1}}$ be any such path and let $\rho_m$ the path obtained by shifting $\rho_{m, t - t_{m+1}}$ backwards in time by $t - t_{m+1}$ time steps. If no path $\mathbf{p}_m$ is out of $(k_m, t - (t_{m+1} - t_m))$, then let $\rho_{m, t - t_{m+1}}$ be any path in $\mathbf{p}_m$ and let $\rho_m$ the path obtained by shifting $\rho_{m, t - t_{m+1}}$ backwards in time by $t - t_{m+1}$ time steps. In this latter case $(k_m, t - (t_{m+1} - t_m))$ is an ancestor of $\mathbf{S}_{m, t-t_{m+1}}$ in $\D$ according to Lemma~\ref{lemmaapp:meaning_of_edges_in_DMAG_2}. By repeating ancestral relationships of $\D$ the vertex $(k_m, t_m)$ is then an ancestor of $\mathbf{S}$.

\textit{Case 2:} If $\pi_m$ is $(k_m, t_m) \headtail (k_{m+1}, t_{m+1})$, do the same as for case 1 with the roles of $(k_m, t_m)$ and $(k_{m+1}, t_{m+1})$ exchanged.

\textit{Case 3:} If $\pi_m$ is $(k_m, t_m) \headtail (l_m, t_{l_m}) \tailhead (k_{m+1}, t_{m+1})$ and $t_m \leq t_{m+1}$, then $(k_m, t- (t_{m+1} -t_m)) \headtail (l_m, t- (t_{m+1} - t_{l_m})) \tailhead (k_{m+1}, t)$ in $\Dc(\Mtaumax(\D))$ and hence $(k_m, t - (t_{m+1} - t_m)) \headhead (k_{m+1}, t)$ in $\Mtaumax(\D)$. According to Lemma~\ref{lemmaapp:meaning_of_edges_in_DMAG} there thus is path between $(k_m, t - (t_{m+1} - t_m))$ and $(k_{m+1}, t)$ in $\D$ that is into both $(k_m, t - (t_{m+1} - t_m))$ and $(k_{m+1}, t)$ and active given $\mathbf{S}_{m, t - t_{m+1}}$. Let $\rho_{m, t - t_{m+1}}$ be any such path and let $\rho_m$ the path obtained by shifting $\rho_{m, t - t_{m+1}}$ backwards in time by $t - t_{m+1}$ time steps.

\textit{Case 4:} If $\pi_m$ is $(k_m, t_m) \headtail (l_m, t_{l_m}) \tailhead (k_{m+1}, t_{m+1})$ and $t_m > t_{m+1}$, do the same as for case 3 with the roles of $(k_m, t_m)$ and $(k_{m+1}, t_{m+1})$ exchanged.

The paths $\rho_m$ exist due to repeating adjacencies of $\D$ and they are active given $\mathbf{S}_m$ due to repeating separating sets of $\D$. Moreover, due to repeating orientations of $\D$ all edges on $\rho_m$ are oriented in the same way as the corresponding edges on $\rho_{m, t - t_{m+1}}$. Consequently: If $(k_{m+1}, t_{m+1})$ is a collider on $\pi$ and thus $\pi_m$ and $\pi_{m+1}$ meet head-to-head at $(k_{m+1}, t_{m+1})$, then, first, $(k_{m+1}, t_{m+1})$ is an ancestor of $\mathbf{S}$ because $\pi$ is active given $\mathbf{S}$ and, second, $\rho_m$ and $\rho_{m+1}$ meet head-to-head at $(k_{m+1}, t_{m+1})$. Moreover, if $(k_{m+1}, t_{m+1})$ is a non-collider on $\pi$ and thus $\pi_m$ and $\pi_{m+1}$ do not meet head-to-head at $(k_{m+1}, t_{m+1})$, then, first, $(k_{m+1}, t_{m+1})$ is not in $\mathbf{S}$ because $\pi$ is active given $\mathbf{S}$ and, second, $\rho_m$ and $\rho_{m+1}$ may or may not meet head-to-head at $(k_{m+1}, t_{m+1})$. Importantly, if they do meet head-to-head, then $(k_{m+1}, t_{m+1})$ is an ancestor of $\mathbf{S}$. By applying Lemma 3.3.1 in \citet{Spirtes2000} to the ordered sequence of paths $(\rho_1, \dots, \rho_n)$ we thus obtain a path between $(k_1, t_1) = (i, t_i)$ and $(k_n, t_n) = (j, t_j)$ in $\D$ that is active given $\mathbf{S}$, and hence $(i, t_i)$ and $(j, t_j)$ are $d$-connected given $\mathbf{S}$ in $\D$.
\end{proof}

\begin{mylemma}\label{lemmaapp:MtaumaxDc_subgraph_MtaumaxD}
$\Mtaumax(\Dc(\Mtaumax(\D)))$ is a subgraph of $\Mtaumax(\D)$.
\end{mylemma}

\begin{proof}[Proof of Lemma~\ref{lemmaapp:MtaumaxDc_subgraph_MtaumaxD}]
As an immediate consequence of Lemma~\ref{lemmaapp:MGc_subgraph_of_MG_reason}, every adjacency in $\Mtaumax(\Dc(\Mtaumax(\D)))$ is also in $\Mtaumax(\D)$. The statement then follows with Lemma~\ref{lemmaapp:same_ancestral_main_proposition} because the orientation of edges in $\Mtaumax(\D)$ and $\Mtaumax(\Dc(\Mtaumax(\D)))$ are uniquely determined by the ancestral relationships.
\end{proof}

\subsubsection{Part 3: Any adjacency in $\Mtaumax(\D)$ is also in $\Mtaumax(\Dc(\Mtaumax(\D)))$}\label{secapp:proof_part_3}
As the third and final part of the proof of Lemma~4.14 we here show that any adjacency in $\Mtaumax(\D)$ is also in $\Mtaumax(\Dc(\Mtaumax(\D)))$. We note that the proof of Lemma~\ref{lemmaapp:MGc_subgraph_of_MG_reason} crucially relies on the particular form of $\Dc(\Mtaumax(\D))$ due to which two subsequent observable vertices on a path in $\Dc(\Mtaumax(\D))$ are at most $\taumax$ time steps apart and adjacent or almost adjacent. A general ts-DAG $\D$ does, however, not necessarily have these properties, which is why this part of the proof becomes more complicated. We begin by proving the converse of Lemma~\ref{lemmaapp:MGc_subgraph_of_MG_reason} restricted to collider-free paths in $\D$.

\begin{mylemma}\label{lemmaapp:collider_free_path_in_D_foundational_sequence}
Let $\pi$ be a collider-free path in $\D$ between distinct observable vertices $(i, t_i)$ and $(j, t_j)$ with $t - \taumax \leq t_i, t_j \leq t$. Then, the ordered sequence $((k_1, t_{1}), \dots, (k_n, t_{n}))$ of observable vertices on $\pi$ with $(k_1, t_1) = (i, t_i)$ and $(k_n, t_n) = (j, t_j)$ has a unique subsequence $((l_1, s_{1}), \dots, (l_m, s_{m}))$ with the following properties:
\begin{enumerate}
    \item $(l_1, s_{1}) = (i, t_i)$ and $(l_m, s_{m}) = (j, t_j)$,
    \item $|s_{\alpha} - s_{\alpha+1}| \leq \taumax$ for all $1 \leq \alpha \leq m-1$,
    \item all non end-point vertices of $\pi((l_\alpha, s_{\alpha}), (l_{\alpha+1}, s_{\alpha+1}))$ are unobservable or are before $\operatorname{max}(s_\alpha, s_{\alpha+1}) - \taumax$ for all $1 \leq \alpha \leq m-1$.
\end{enumerate}
\end{mylemma}

\begin{myremark}[on Lemma~\ref{lemmaapp:collider_free_path_in_D_foundational_sequence}]
While the uniqueness of $((l_1, s_{1}), \dots , (l_m, s_{m}))$ is not needed for the subsequent proofs, we have included it to be able to refer to \emph{the} subsequence $((l_1, s_{1}), \dots , (l_m, s_{m}))$ instead of \emph{a} such subsequence.
\end{myremark}

\begin{proof}[Proof of Lemma~\ref{lemmaapp:collider_free_path_in_D_foundational_sequence}]
\underline{Existence:}\\
Assume without loss of generality that $t_i \leq t_j$, else exchange $(i, t_i)$ and $(j, t_j)$. Since $\pi$ is collider-free, time order of $\D$ thus implies that no vertex on $\pi$ is after $t_j$. We now prove the statement by induction over $n$, where $n$ is the number of observable vertices on $\pi$:

\textit{Induction base case: $n = 2$}\\
In this case $(i, t_i)$ and $(j, t_j)$ are the only observable vertices on $\pi$. Clearly, the sequence $((i, t_i), (j, t_j))$ has the desired properties.

\textit{Induction step: $n \mapsto n + 1$}\\
In this case $\pi$ has $n + 1 \geq 3$ observable vertices and the statement has already been proven for paths that have at most $n$ observable vertices. Let $\pi_1$ be the subpath of $\pi$ from $(i, t_i)$ to $(k_{n}, t_{n})$ (the observable vertex on $\pi$ other than $(j, t_j)$ itself that is closest to $(j, t_j)$ on $\pi$) and let $\pi_2$ be the subpath of $\pi$ from $(k_2, t_2)$ (the observable vertex on $\pi$ other than $(i, t_i)$ itself that is closest to $(i, t_i)$ on $\pi$) to $(j, t_j)$. We distinguish three collectively exhaustive cases:
\begin{itemize}
  \item Case $|t_j - t_{n}| = t_j - t_{n} \leq \taumax$:\\
  This premise implies $|t_i - t_{n}| \leq \taumax$. Hence, by assumption of induction the statement applies to $\pi_1$. The desired sequence is obtained by appending $(j, t_j)$ to the sequence obtained by applying the statement to $\pi_1$.
  \item Case $|t_j - t_2| = t_j - t_2 \leq \taumax$:\\
  By assumption of induction the statement then applies to $\pi_2$. Moreover, $|t_i - t_2| \leq \taumax$. The desired sequence is obtained by prepending $(i, t_i)$ to the sequence obtained by applying the statement to $\pi_2$.
  \item Case $|t_j - t_{n}| = t_j - t_{n} > \taumax$ and $|t_j - t_2| = t_j - t_2 > \taumax$:\\
  Since $\pi$ is collider-free and $\D$ is time ordered, this premise implies that all observable non end-point vertices on $\pi$ are before $t_j-\taumax$. Hence, the sequence $((i, t_i), (j, t_j))$ has the desired properties.
\end{itemize}

\underline{Uniqueness:}\\
Let $((l_1, s_{1}), \dots , (l_m, s_{m}))$ and $((l^\prime_1, s^\prime_{1}), \dots , (l^\prime_{m^\prime}, s_{m^\prime}))$ be two such subsequences. We proof their equality by induction over $\alpha$, where $\alpha$ is the index of the subsequences.

\textit{Induction base case: $\alpha = 1$}\\
The equality $(l_1, s_{1}) = (l^\prime_1, s^\prime_{1})$ follows due to the first property demanded in Lemma~\ref{lemmaapp:collider_free_path_in_D_foundational_sequence} applied to both sequences.

\textit{Induction step: $\alpha \mapsto \alpha+1 \leq \operatorname{min}(m, m^\prime)$}\\
By the assumption of induction $(l_q, s_{q}) = (l^\prime_q, s^\prime_{q})$ for all $1 \leq q \leq \alpha$ is given and we have to show $(l_{\alpha+1}, s_{\alpha+1}) = (l^\prime_{\alpha+1}, s^\prime_{\alpha+1})$. Assume the opposite, i.e., assume $(l_{\alpha+1}, s_{\alpha+1}) \neq (l^\prime_{\alpha+1}, s^\prime_{\alpha+1})$. Without loss of generality further assume that $(l_{\alpha+1}, s_{\alpha+1})$ is on the subpath $\pi((l^\prime_{\alpha+1}, s^\prime_{\alpha+1}), (j, t_j))$, else exchange the two subsequences. Let $r$ with $\alpha < r < m^\prime$ be such that the vertices $(l^\prime_{q}, s^\prime_{q})$ with $\alpha < q \leq r$ are non end-point vertices on $\pi((l_{\alpha}, s_{\alpha}), (l_{\alpha+1}, s_{\alpha+1}))$ and $(l^\prime_{r+1}, s^\prime_{r+1})$ is on $\pi((l_{\alpha+1}, s_{\alpha+1}), (j, t_j))$. Such an $r$ exists because both $((l_1, s_{1}), \dots , (l_m, s_{m}))$ and $((l^\prime_1, s^\prime_{1}), \dots , (l^\prime_{m^\prime}, s_{m^\prime}))$ are subsequences of $((k_1, t_{1}), \dots, (k_n, t_{n}))$. Note that the $(l^\prime_{q}, s^\prime_{q})$ with $\alpha < q \leq r$ are observable because all vertices of the subsequence $((l^\prime_1, s^\prime_{1}), \dots , (l^\prime_{m^\prime}, s_{m^\prime}))$ are observable. The third property demanded in Lemma~\ref{lemmaapp:collider_free_path_in_D_foundational_sequence} applied to the sequence $((l_1, s_{1}), \dots , (l_m, s_{m}))$ thus requires $s^\prime_{q} < \operatorname{max}(s_\alpha, s_{\alpha+1}) - \taumax$ for all $\alpha < q \leq r$. We distinguish two cases:
\begin{itemize}
\item Suppose $s_{\alpha+1} \leq s_\alpha$. Then $s^\prime_q < s_\alpha - \taumax$ for all $\alpha < q \leq r$. For $q = \alpha+1$ we thus get $s^\prime_{\alpha+1} < s_\alpha - \taumax$, which contradicts the second property demanded in Lemma~\ref{lemmaapp:collider_free_path_in_D_foundational_sequence} applied to $((l^\prime_1, s^\prime_{1}), \dots , (l^\prime_{m^\prime}, s_{m^\prime}))$.
\item Suppose $s_{\alpha+1} > s_\alpha$. Then $s^\prime_q < s_{\alpha+1} - \taumax$ for all $\alpha < q \leq r$. By time order of $\D$ in combination with the facts that $\pi$ is collider-free and that $(l_\alpha, s_\alpha)$ is on $\pi((i, t_i), (l_{\alpha+1}, s_{\alpha+1}))$, the premise $s_{\alpha+1} > s_\alpha$ implies that $(l_{\alpha+1}, s_{\alpha+1})$ is on $\pi$ between the root node of $\pi$ and $(j, t_j)$. Consequently, all vertices on $\pi((l_{\alpha+1}, s_{\alpha+1}), (j, t_j))$ are not before $s_{\alpha+1}$. Since $(l^\prime_{r+1}, s^\prime_{r+1})$ is on $\pi((l_{\alpha+1}, s_{\alpha+1}), (j, t_j))$ we thus find $s_{\alpha+1} \leq s^\prime_{r+1}$, which implies $s^\prime_q < s^\prime_{r+1} - \taumax$ for all $\alpha < q \leq r$. For $q = r$ we thus get $s_r^\prime < s^\prime_{r+1} - \taumax$, which contradicts the second property demanded in Lemma~\ref{lemmaapp:collider_free_path_in_D_foundational_sequence} applied to $((l^\prime_1, s^\prime_{1}), \dots , (l^\prime_{m^\prime}, s_{m^\prime}))$.
\end{itemize}
Consequently, the assumption $(l_{\alpha+1}, s_{\alpha+1}) \neq (l^\prime_{\alpha+1}, s^\prime_{\alpha+1})$ leads to a contradiction and $(l_{\alpha+1}, s_{\alpha+1}) = (l^\prime_{\alpha+1}, s^\prime_{\alpha+1})$ must be true.

This induction terminates when $\alpha+1 = \operatorname{min}(m, m^\prime)$. The first property demanded in Lemma~\ref{lemmaapp:collider_free_path_in_D_foundational_sequence} applied to both sequences then requires $m = m^\prime$ and $(l_m, s_{m}) = (l^\prime_m, s^\prime_{m}) = (j, t_j)$, which completes the proof
\end{proof}

\begin{myremark}[on the proof of Lemma~\ref{lemmaapp:collider_free_path_in_D_foundational_sequence}]
For the special case in which $\pi$ is a directed path we can also immediately see that the sequence $((k_1, t_{1}), \dots, (k_n, t_{n}))$ of all observable nodes on $\pi$ has the desired properties. We need the above proof to also cover the case in which $\pi$ is into both $(i, t_i)$ and $(j, t_j)$.
\end{myremark}

\begin{mylemma}\label{lemmaapp:collider_free_path_in_D}
Let $\pi$ be a collider-free path in $\D$ between distinct observable vertices $(i, t_i)$ and $(j, t_j)$ with $t - \taumax \leq t_i, t_j \leq t$ and let $((l_1, s_{1}), \dots, (l_m, s_{m}))$ be as in Lemma~\ref{lemmaapp:collider_free_path_in_D_foundational_sequence} applied to $\pi$. Then, $(l_\alpha, s_\alpha)$ and $(l_{\alpha+1}, s_{\alpha+1})$ are adjacent or almost adjacent in $\Dc(\Mtaumax(\D))$ for all $1 \leq \alpha \leq m-1$.
\end{mylemma}

\begin{proof}[Proof of Lemma~\ref{lemmaapp:collider_free_path_in_D}]
Let $\pi_\alpha$ be the path in $\D$ obtained by shifting the subpath $\pi((l_\alpha, s_{\alpha}), (l_{\alpha+1}, s_{\alpha+1}))$ forward in time by $t - \operatorname{max}(s_\alpha, s_{\alpha+1}) \geq 0$ time steps. This $\pi_\alpha$ is a path between the vertices $v_1 = (l_\alpha, t - (\operatorname{max}(s_\alpha, s_{\alpha+1})- s_\alpha))$ and $v_2 = (l_\alpha, t - (\operatorname{max}(s_\alpha, s_{\alpha+1})- s_{\alpha+1}))$, both of which are observable because all vertices in the sequence $((l_1, s_{1}), \dots, (l_m, s_{m}))$ are observable and within $[t-\taumax, t]$  because of part 2 of Lemma~\ref{lemmaapp:collider_free_path_in_D_foundational_sequence}. Moreover, $\pi_\alpha$ is collider-free because $\pi$ is collider-free. All non end-point vertices on $\pi_\alpha$ are unobserved, i.e., unobservable or temporally unobserved due to part 3 of Lemma~\ref{lemmaapp:collider_free_path_in_D_foundational_sequence}. Consequently, $\pi_\alpha$ cannot be blocked by any set of observed variables, which is why $v_1$ and $v_2$ are adjacent in $\Mtaumax(\D)$. The statement then follows from Lemma~\ref{lemmaapp:at_most_almost_adjacent}.
\end{proof}

\begin{mydef}[Canonically induced path]\label{mydef:canonical_path}
Let $\pi$ be a collider-free path in $\D$ between distinct observable vertices $(i, t_i)$ and $(j, t_j)$ with $t - \taumax \leq t_i, t_j \leq t$ and let $((l_1, s_{1}), \dots, (l_m, s_{m}))$ be as in Lemma~\ref{lemmaapp:collider_free_path_in_D_foundational_sequence} applied to $\pi$. The \emph{canonical path induced by $\pi$}, denoted $\pi_{ci}$, is the (unique) path in $\Dc(\Mtaumax(\D))$ between $(i, t_i)$ and $(j, t_j)$ with the following properties:
\begin{enumerate}
\item For all $1 \leq \alpha \leq m$ the vertex $(l_\alpha , s_\alpha )$ is on $\pi_{ci}$.
\item For all $1 \leq \alpha  \leq m -1 $ the subpath $\pi_{ci}((l_\alpha , s_\alpha )), (l_{\alpha +1}, s_{\alpha +1}))$ is
\begin{enumerate}
\item $(l_\alpha , s_\alpha) \tailhead (l_{\alpha+1}, s_{\alpha+1})$ if and only if $(l_\alpha, s_\alpha) \in \an((l_{\alpha+1}, s_{\alpha+1}), \D)$,
\item $(l_\alpha , s_\alpha) \headtail (l_{\alpha+1}, s_{\alpha+1})$ if and only if $(l_{\alpha+1}, s_{\alpha+1}) \in \an((l_{\alpha}, s_{\alpha}), \D)$,
\item $(l_\alpha , s_\alpha) \headtail (l^\prime_\alpha, s^\prime_\alpha) \tailhead (l_{\alpha+1}, s_{\alpha+1})$ with $(l^\prime_\alpha, s^\prime_\alpha)$ unobservable if and only if $(l_\alpha, s_\alpha) \notin \an((l_{\alpha+1}, s_{\alpha+1}), \D)$ and $(l_{\alpha+1}, s_{\alpha+1}) \notin \an((l_{\alpha}, s_{\alpha}), \D)$.
\end{enumerate}
\end{enumerate}
\end{mydef}

\begin{myremark}[on Def.~\ref{mydef:canonical_path}]
Existence follows from Lemma~\ref{lemmaapp:at_most_almost_adjacent_2} because according to Lemma~\ref{lemmaapp:collider_free_path_in_D} the vertices $(l_\alpha, s_\alpha)$ and $(l_{\alpha+1}, s_{\alpha+1})$ are adjacent or almost adjacent for all $1 \leq \alpha \leq m-1$. Uniqueness follows from Lemma~\ref{lemmaapp:almost_adjacent_unique} because in $\Dc(\Mtaumax(\D))$ there is at most one edge between any pair of vertices.
\end{myremark}

\begin{mylemma}\label{lemmaapp:canonical_path_properties}
Let $\pi$ be a collider-free path in $\D$ between distinct observable vertices $(i, t_i)$ and $(j, t_j)$ with $t - \taumax \leq t_i, t_j \leq t$ and let $\pi_{ci}$ be the canonical path induced by $\pi$. Then:
\begin{enumerate}
  \item All observable vertices on $\pi_{ci}$ are also on $\pi$.
  \item If $(k_1, t_1)$, $(k_2, t_2)$ and $(k_3, t_3)$ are distinct observable vertices on $\pi_{ci}$ and $(k_2, t_2)$ is on $\pi_{ci}((k_1, t_1), (k_3, t_3))$, then $(k_2, t_2)$ is on $\pi((k_1, t_1), (k_3, t_3))$.
  \item $\pi_{ci}$ is a collider-free.
  \item If $(k_1, t_1)$ and $(k_2, t_2)$ are distinct observable vertices on $\pi_{ci}$, then $\pi((k_1, t_1), (k_2, t_2))$ is an inducing path relative to $\mathbf{O}(\operatorname{max}(t_1, t_2) - \taumax, t)[\pi_{ci}((k_1, t_1), (k_2, t_2))]$.
  \item If $\pi$ is active given $\mathbf{S}$ in $\D$, then $\pi_{ci}$ is active given $\mathbf{S}$ in $\Dc(\Mtaumax(\D))$.
\end{enumerate}
\end{mylemma}

\begin{proof}[Proof of Lemma~\ref{lemmaapp:canonical_path_properties}]
Let $((l_1, s_{1}), \dots, (l_m, s_{m}))$ be as in Lemma~\ref{lemmaapp:collider_free_path_in_D_foundational_sequence} applied to $\pi$.

\textbf{1.}
The definition of canonically induced paths is such that $((l_1, s_{1}), \dots , (l_m, s_{m}))$ is the sequence of all observable vertices on $\pi_{ci}$. Moreover, all of $(l_1, s_{1}), \ldots , (l_m, s_{m})$ are on $\pi$ by construction, see Lemma~\ref{lemmaapp:collider_free_path_in_D_foundational_sequence}.

\textbf{2.}
Since $(k_2, t_2)$ is on $\pi_{ci}((k_3, t_3), (k_1, t_1))$ if it is on $\pi_{ci}((k_1, t_1), (k_3, t_3))$, we can without loss of generality assume that $(k_1, t_1)$ is closer to $(i, t_i)$ on $\pi_{ci}$ than $(k_3, t_3)$ is to $(i, t_i)$ on $\pi_{ci}$. Hence, there are $1 \leq \alpha_1 < \alpha_2 < \alpha_3 \leq m$ such that $(k_1, t_1) = (l_{\alpha_1}, s_{\alpha_1})$, $(k_2, t_2) = (l_{\alpha_2}, s_{\alpha_2})$ and $(k_3, t_3) = (l_{\alpha_3}, s_{\alpha_3})$. Moreover, by definition of $((l_1, s_{1}), \dots , (l_m, s_{m}))$ and canonically induced paths, $(l_{\alpha}, s_\alpha)$ is closer to $(i, t_i)$ on $\pi$ than $(l_q, s_q)$ with $\alpha < q$ is to $(i,t_i)$ on $\pi$.

\textbf{3.}
Assume there is a collider on $\pi_{ci}$. Since in $\Dc(\Mtaumax(\D))$ there are no edges into unobservable vertices, all colliders on $\pi_{ci}$ are observable. Hence, there must be $1 < \alpha < m$ such that $(l_{\alpha}, s_{\alpha})$ is a collider on $\pi_{ci}$, i.e., such that both $\pi_{ci}((l_{\alpha-1}, s_{\alpha-1}), (l_{\alpha}, s_{\alpha}))$ and $\pi_{ci}((l_{\alpha}, s_{\alpha}), (l_{\alpha+1}, s_{\alpha+1}))$ are into $(l_{\alpha}, s_{\alpha})$. In particular, $\pi_{ci}((l_{\alpha-1}, s_{\alpha-1}), (l_{\alpha}, s_{\alpha}))$ is not $(l_{\alpha-1}, s_{\alpha-1}) \headtail (l_{\alpha}, s_{\alpha})$ and $\pi_{ci}((l_{\alpha}, s_{\alpha}), (l_{\alpha+1}, s_{\alpha+1}))$ is not $(l_{\alpha}, s_{\alpha}) \tailhead (l_{\alpha+1}, s_{\alpha+1})$.

Assume $\pi((l_{\alpha-1}, s_{\alpha-1}), (l_{\alpha}, s_{\alpha}))$ is out of $(l_{\alpha}, s_{\alpha})$. Since $\pi$ is collider-free, it then follows that $\pi((l_{\alpha-1}, s_{\alpha-1}), (l_{\alpha}, s_{\alpha}))$ is directed from $(l_{\alpha}, s_{\alpha})$ to $(l_{\alpha-1}, s_{\alpha-1})$ and hence $(l_{\alpha}, s_{\alpha}) \in \an((l_{\alpha-1}, s_{\alpha-1}), \D)$. According to Def.~\ref{mydef:canonical_path} this ancestral relationships requires $\pi_{ci}((l_{\alpha-1}, s_{\alpha-1}), (l_{\alpha}, s_{\alpha}))$ to be $(l_{\alpha-1}, s_{\alpha-1}) \headtail (l_{\alpha}, s_{\alpha})$, which is a contradiction. Hence $\pi((l_{\alpha-1}, s_{\alpha-1}), (l_{\alpha}, s_{\alpha}))$ is into $(l_{\alpha}, s_{\alpha})$.

We similarly we find that $\pi((l_{\alpha}, s_{\alpha}), (l_{\alpha+1}, s_{\alpha+1}))$ is into $(l_{\alpha}, s_{\alpha})$ and thus that $(l_{\alpha}, s_{\alpha})$ is a collider on $\pi$, a contradiction.

\textbf{4.}
We may without loss of generality assume that $(k_1, t_1)$ is closer to $(i,t_i)$ on $\pi$ than $(k_2, t_2)$ is to $(i,t_i)$ on $\pi$. Write $t_{12} = \operatorname{max}(t_1, t_2)$.

Since $\pi$ is collider-free, also its subpath $\pi((k_1, t_1), (k_2, t_2))$ is collider-free. For showing that $\pi((k_1, t_1), (k_2, t_2))$ is an inducing path relative to $\mathbf{O}(t_{12} - \taumax, t)[\pi_{ci}((k_1, t_1), (k_2, t_2))]$ it is thus sufficient to show that none of its non end-point vertices is an element of the set $\mathbf{O}(t_{12} - \taumax, t)[\pi_{ci}((k_1, t_1), (k_2, t_2))]$. To this end, let $(k_3, t_3)$ be a non end-point vertex on $\pi((k_1, t_1), (k_2, t_2))$.

Since $((l_1, s_{1}), \dots , (l_m, s_{m}))$ is the sequence of all observable vertices on $\pi_{ci}$, there are $\alpha_1$ and $\alpha_2$ with $1 \leq \alpha_1 < \alpha_2 \leq m$ such that $(k_1, t_1) = (l_{\alpha_1}, s_{\alpha_1})$ and $(k_2, t_2) = (l_{\alpha_2}, s_{\alpha_2})$. Therefore, either $(k_3, t_3)$ equals $(l_{\alpha_3}, s_{\alpha_3})$ for some $\alpha_1 < \alpha_3 < \alpha_2$ or $(k_3, t_3)$ is a non end-point vertex on $\pi((l_{\alpha_3}, s_{\alpha_3}), (l_{\alpha_3+1}, s_{\alpha_3+1}))$ for some $\alpha_3$ with $\alpha_1 \leq \alpha_3 < \alpha_2$. In the former case, $(k_3, t_3)$ is not in $\mathbf{O}(t_{12} - \taumax, t_{12})[\pi_{ci}((k_1, t_1), (k_2, t_2))]$ because it is a non end-point vertex on $\pi_{ci}((k_1, t_1), (k_2, t_2))$. In the latter case, according to part 3 of Lemma~\ref{lemmaapp:collider_free_path_in_D_foundational_sequence}, $(k_3, t_3)$ is unobservable or before $\operatorname{max}(s_{\alpha_3}, s_{\alpha_3+1}) - \taumax$. Because $\pi$ is collider-free and both $(l_{\alpha_3}, s_{\alpha_3})$ and $(l_{\alpha_3+1}, s_{\alpha_3+1})$ are on $\pi((k_1, t_1), (k_2, t_2))$, time order of $\D$ implies $\operatorname{max}(s_{\alpha_3}, s_{\alpha_3+1}) \leq t_{12}$. Thus, $(k_3, t_3)$ is not in $\mathbf{O}(t_{12} - \taumax, t)[\pi_{ci}((k_1, t_1), (k_2, t_2))]$.

\textbf{5.}
This claim follows from parts 1 and 2 of Lemma~\ref{lemmaapp:canonical_path_properties}: Since $\pi$ is collider-free and active given $\mathbf{S}$, no vertex on $\pi$ is in $\mathbf{S}$. Thus, since all observable vertices on $\pi_{ci}$ are also on $\pi$, no vertex on $\pi_{ci}$ is in $\mathbf{S}$. Since $\pi_{ci}$ is collider-free this observation shows that $\pi_{ci}$ is active given $\mathbf{S}$.
\end{proof}

\begin{mylemma}\label{lemmaapp:MG_subgraph_of_MGc_reason_part_1}
Let $(i, t_i)$ and $(j, t_j)$ with $t - \taumax \leq t_i, t_j \leq t$ be distinct observable vertices in $\D$ and let $\mathbf{S} \subseteq \mathbf{O}(t-\taumax, t) \setminus\{(i, t_i), (j, t_j)\}$. Then: If in $\D$ there is a collider-free path $\pi$ between $(i, t_i)$ and $(j, t_j)$ that is active given $\mathbf{S}$, then $(i, t_i)$ and $(j, t_j)$ are $d$-connected given $\mathbf{S}$ in $\Dc(\Mtaumax(\D))$.
\end{mylemma}

\begin{proof}[Proof of Lemma~\ref{lemmaapp:MG_subgraph_of_MGc_reason_part_1}]
According to part 5 of Lemma~\ref{lemmaapp:canonical_path_properties} the canonically induced path $\pi_{ci}$ of $\pi$ $d$-connects $(i, t_i)$ and $(j, t_j)$ given $\mathbf{S}$ in $\Dc(\Mtaumax(\D))$.
\end{proof}

In order to show that any adjacency in $\Mtaumax(\D)$ is also in $\Mtaumax(\Dc(\Mtaumax(\D)))$---and thus to finish the proof of Lemma~4.14---it remains to prove a statement that extends Lemma~\ref{lemmaapp:MG_subgraph_of_MGc_reason_part_1} to the case in which $\pi$, the $d$-connecting path in $\D$, has $n_c \geq 1$ colliders $c_1, \dots, c_{n_c}$ (ordered starting with the collider closest to $(i, t_i)$). One might think that such a generalization follows readily now, namely by cutting $\pi$ into $n_c + 1$ collider-free paths $\pi^{a, a+1} = \pi(c_{a}, c_{a+1})$ with $0 \leq a \leq n_c$, where we let $c_0 = (i, t_i)$ and $c_{n_c+1} = (j, t_j)$, and then applying Lemma 3.3.1 in \citet{Spirtes2000} to the canonically induced paths $\pi^{a, a+1}_{ci}$ of $\pi^{a, a+1}$ with $0 \leq a \leq n_c$. While we do use a similar approach, the proof is complicated by two facts: First, the canonically induced path $\pi^{a, a+1}_{ci}$ of $\pi^{a, a+1}$ only exists if $\pi^{a, a+1}$ is between \emph{observable} vertices that are \emph{at most $\taumax$ time steps apart}, see Def.~\ref{mydef:canonical_path}. However, some of the colliders on $\pi$, i.e., some of the $c_a$ with $a \leq 1 \leq n_c$ might be unobservable and/or more than $\taumax$ time steps apart from the neighboring colliders or end-point vertices of $\pi$, i.e., from $c_{a-1}$ or $c_{a+1}$. Second, even if $\pi^{a, a+1}_{ci}$ exists, it may be \emph{out of} one of its end-point vertices although $\pi^{a, a+1}$ is into this vertex. In case this vertex is a collider on $\pi$ and an element of $\mathbf{S}$, Lemma 3.3.1 in \citet{Spirtes2000} does not apply. We will address the first of these complications by noting that, in order for $\pi$ to be active given $\mathbf{S}$, every collider on $\pi$ must be an ancestor of $\mathbf{S}$ and thus an ancestor of an observed vertex within the time window $[t-\taumax, t]$. Hence, in $\D$ there are directed paths from the $c_a$ with $a \leq 1 \leq n_c$ to some observed vertices. By joining these directed paths with the $\pi^{a, a+1}$ we get collider-free paths $\tilde{\pi}^{a, a+1}$ in $\D$ between observed vertices, the canonically induced paths $\tilde{\pi}^{a, a+1}_{ci}$ of which exist.

\begin{mydef}[Collider extension structure]\label{defapp:collider_extension}
Let $\pi$ be a (non collider-free) path between the distinct observable vertices $(i, t_i)$ and $(j, t_j)$ with $t - \taumax \leq t_i, t_j \leq t$ in $\D$ that is active given $\mathbf{S} \subseteq \mathbf{O}(t-\taumax, t) \setminus\{(i, t_i), (j, t_j)\}$. Let $c_1, \dots , c_{n_c}$ with $n_c \geq 1$ be the collider on $\pi$, ordered starting with the collider closest to $(i, t_i)$. A \emph{collider extension structure of $\pi$ with respect to $\mathbf{S}$ and $\mathbf{O}(t -\taumax, t)$} is a collection of paths $\rho^{1}, \dots , \rho^{n_c}$ such that for all $1 \leq a \leq n_c$ and $1 \leq b \leq n_c$ with $a \neq b$ all of the following holds:
\begin{enumerate}
\item One of these options holds:
\begin{enumerate}
  \item $\rho^a$ is the trivial path consisting of $c_a = d_a$ only and $c_a = d_a \in \mathbf{O}(t -\taumax, t)$.
  \item $\rho^a$ is a non-trivial directed path from $c_a$ to some vertex $d_a \in \mathbf{O}(t -\taumax, t)$.
\end{enumerate}
\item If $v$ is on $\rho^a$ and in $\mathbf{S}$, then $v = d_a$.
\item $d_a$ is an ancestor of $\mathbf{S}$.
\item $\rho^a$ intersects with $\pi$ at $c_a$ only.
\item $\rho^a$ and $\rho^b$ do not intersect.
\end{enumerate}
\end{mydef}

\begin{mylemma}\label{lemmaapp:extension_by_rho^a_paths}
Given the assumptions and notation of Def.~\ref{defapp:collider_extension}, one of the following statements holds:
\begin{enumerate}
\item There is a collider extension structure of $\pi$ with respect to $\mathbf{S}$ and $\mathbf{O}(t -\taumax, t)$.
\item In $\D$ there is a path $\pi^\prime$ between $(i, t_i)$ and $(j, t_j)$ with at most $n_c - 1$ colliders that is active given $\mathbf{S}$.
\end{enumerate}
\end{mylemma}

\begin{proof}[Proof of Lemma~\ref{lemmaapp:extension_by_rho^a_paths}]
We divide the proof into three steps.

\underline{Step 1: The first, second and third property of collider extension structures holds}\\
Let $1 \leq a \leq n_c$. Since $c_a$ is a collider on $\pi$ and $\pi$ is active given $\mathbf{S}$, there is $S_a \in \mathbf{S}$ (which may be equal to $c_a$) such that $c_a \in \an(S_a, \D)$. Hence, there is a (possibly trivial, namely if and only if $c_a = S_a$) directed path $\lambda^a$ from $c_a$ to $S_a$. On $\lambda^a$ let $d_a$ be the vertex closest to $c_a$ that is in $\mathbf{O}(t -\taumax, t)$ (which may be $c_a$ itself). This vertex $d_a$ exists because $S_a \in \mathbf{S} \subseteq \mathbf{O}(t -\taumax, t)$, such that $d_a$ is an ancestor of $\mathbf{S}$ by means of the subpath $\lambda(d_a, S_a)$. Let $\rho^a$ be the subpath $\lambda^a(c_a, d_a)$, which is the trivial path consisting of the single vertex $c_a = d_a \in \mathbf{O}(t -\taumax, t)$ or is a non-trivial directed path from $c_a$ to $d_a \in \mathbf{O}(t -\taumax, t)$. Moreover, by definition of $d_a$ together with the fact that $\mathbf{S} \subseteq \mathbf{O}(t -\taumax, t)$, no vertex on $\rho^a$ other than $d_a$ is in $\mathbf{S}$. The collection of paths $\rho^{1}, \dots , \rho^{n_c}$ thus fulfills the first three conditions of a collider extension structure of $\pi$ with respect to $\mathbf{S}$ and $\mathbf{O}(t -\taumax, t)$.

The first two of these conditions have two immediate implications that will be important later in this proof: First, if $\rho^a$ is non-trivial then $c_a \notin \mathbf{S}$. Second, if $\rho^a$ is non-trivial then it is active given $\mathbf{S} \setminus \{d_a\}$.

\underline{Step 2: The fourth property of collider extension structures or existence of $\pi^\prime$}\\
Assume there is $1 \leq a \leq n_c$ such that $\pi$ and $\rho^a$ do not intersect at $c_a$ only. Then, $\rho^a$ must be non-trivial and hence we get $c_a \notin \mathbf{S}$. Let $e_a$ be the vertex on $\rho^a$ closest to $c_a$ other than $c_a$ itself that is also on $\pi$. If $e_a$ is on $\pi((i, t_i), c_a)$, then let $v_1 = (i, t_i)$ and $v_2 = (j, t_j)$, else let $v_1 = (j, t_j)$ and $v_2 = (i, t_i)$. Let $\pi^\prime$ be the concatenation $\pi(v_1, e_a) \oplus \rho^a(e_a, c_a) \oplus \pi(c_a, v_2)$. By definition of $e_a$, $\pi^\prime$ is a path (rather than a walk) in $\D$ between $(i, t_i)$ and $(j, t_j)$. We now show that $\pi^\prime$ is active given $\mathbf{S}$ and has at most $n_c - 1$ colliders.

\textit{All colliders on $\pi^\prime$ are ancestors of $\mathbf{S}$:}
Since $\rho^{a}(e_a, c_a)$ is a non-trivial directed path from $c_a$ to $e_a$, every collider on $\pi^\prime$ is a collider on $\pi(v_1, e_a)$ or a collider on $\pi(c_a, v_2)$ or equals $e_a$. Every collider on $\pi(v_1, e_a)$ or $\pi(c_a, v_2)$ is a collider on $\pi$ and hence, because $\pi$ is active given $\mathbf{S}$, an ancestor of $\mathbf{S}$. Lastly, $e_a$ is an ancestor of $S_a \in \mathbf{S}$ by means of the path $\lambda^a(e_a, S_a)$.

\textit{No non end-point non-collider on $\pi^\prime$ is in $\mathbf{S}$:}
All non end-point non-colliders on $\pi(v_1, e_a)$ or $\pi(c_a, v_2)$ are non end-point non-colliders on $\pi$ and hence, because $\pi$ is active given $\mathbf{S}$, not in $\mathbf{S}$. All vertices on $\rho^a(e_a, c_a)$ other than, perhaps, $e_a$ are not in $\mathbf{S}$ because, as shown in step 1 of this proof, all vertices on $\rho^a$ other than $d_a$ are not in $\mathbf{S}$. Lastly, assume that $e_a$ is in $\mathbf{S}$ and a non end-point non-collider on $\pi^\prime$. Because $\rho^a(e_a, c_a)$ is into $e_a$, this assumption requires that $\pi(v_1, e_a)$ is a non-trivial path out of $e_a$. Consequently, $e_a$ is a non end-point non-collider on $\pi$, which is a contradiction because $e_a \in \mathbf{S}$ and $\pi$ is active given $\mathbf{S}$.

\textit{Number of colliders:}
There are no colliders on $\rho^a(e_a, c_a)$ because it is a directed path. If $v_1 = (i, t_i)$, then there are at most $a - 1$ colliders on $\pi(v_1, e_a)$ and exactly $n_c - a$ colliders on $\pi(c_a, v_2)$. If $v_1 = (j, t_j)$, then there are at most $n_c - a$ colliders on $\pi(v_1, e_a)$ and exactly $a - 1$ colliders on $\pi(c_a, v_2)$. The junction point $c_a$ is a non-collider on $\pi^\prime$ because $\rho^a(e_a, c_a)$ is out of $c_a$. Regarding $e_a$, there are two cases:
\begin{enumerate}
\item First, assume $e_a$ is a non-collider on $\pi^\prime$. Then, there are at most $(a-1)+(n_c - a) = n_c - 1$ colliders on $\pi^\prime$.
\item Second, assume $e_a$ is a collider on $\pi^\prime$. This assumption requires $\pi(v_1, e_a)$ to be into $e_a$. Let $r$ be that particular root node on $\pi(v_1, c_a)$ which is closest to $c_a$ on $\pi$. Then, $\pi(r, c_a)$ is non-trivial (because $c_a$ is a collider on $\pi$ and hence not a root on $\pi$) and directed from $r$ to $c_a$ (by combining the facts that $c_a$ is a collider on $\pi$, that $r$ is a root on $\pi$, and that no other root node on $\pi$ is between $r$ and $c_a$). Assume that $e_a$ is on $\pi(r, c_a)$. Then, $\pi(e_a, c_a)$ would be a non-trivial (because $e_a \neq c_a$) directed path from $e_a$ to $c_a$, which contradicts acyclicity because $c_a$ is an ancestor $e_a$ by means of $\rho^{a}(c_a, e_a)$. Hence, $e_a$ is on $\pi(v_1, r)$. Since $\pi(v_1, e_a)$ is into $e_a$, we thus see that $e_a$ is on $\pi(v_1, c_{a-1})$ if $v_1 = (i, t_i)$ and that $e_a$ is on $\pi(c_{a+1}, v_1)$ if $v_1 = (j, t_j)$. Consequently, there are at most $a - 2$ colliders on $\pi(v_1, e_a)$ if $v_1 = (i, t_i)$ and there are most $n_c - (a+1)$ colliders on $\pi(v_1, e_a)$ if $v_1 = (j, t_j)$ . In summary, there are most $(a-2) + (n_c - a) + 1= (n_c - (a+1)) + (a -1) + 1 = n_c - 1$ colliders on $\pi^\prime$. 
\end{enumerate}

Thus, if the collection of paths $\rho^{1}, \dots , \rho^{n_c}$ does not fulfill the fourth condition of a collider extension structure of $\pi$ with respect to $\mathbf{S}$ and $\mathbf{O}(t -\taumax, t)$, then there is path $\pi^\prime$ as in point~2 of this lemma. To complete this proof it is therefore sufficient to show the following statement: If the collection of paths $\rho^{1}, \dots , \rho^{n_c}$ fulfills the first four conditions of a collider extension structure of $\pi$ with respect to $\mathbf{S}$ and $\mathbf{O}(t -\taumax, t)$, then this collection of paths also fulfills the fifth condition of a collider extension structure (and hence is a collider extension structure) or there is path $\pi^\prime$ as in point 2 of the lemma.

\underline{Step 3: The fifth property of collider extension structures holds or existence of $\pi^\prime$}\\
Assume there are $1 \leq a, b \leq n_c$ with $a < b$ such that $\rho^a$ and $\rho^b$ intersect. Then, at least one of these paths must be non-trivial because $c_a \neq c_b$. If one of them is trivial and the other one is non-trivial, say $\rho^a$ is trivial and $\rho^b$ is non-trivial, then $\rho^b$ must contain $c_a$ and hence intersects with $\pi$ at a vertex other than $c_a$, namely $c_b$. Since this conclusion violates the fourth condition of a collider extension structure of $\pi$ with respect to $\mathbf{S}$ and $\mathbf{O}(t -\taumax, t)$, we do not need to consider this situation as explained in the last paragraph of step 2. Consequently, we can assume that both $\rho^a$ and $\rho^b$ are non-trivial and intersect $\pi$ at, respectively, $c_a$ and $c_b$ only.

The fact that both $\rho^a$ and $\rho^b$ are non-trivial implies $\mathbf{S} = \mathbf{S} \setminus \{c_a, c_b\}$, see step 1 of this proof. Let $f_{ab}$ be the vertex closest to $c_a$ on $\rho^a$ that is also on $\rho^b$. Because $\rho^a$ and $\rho^b$ respectively intersect $\pi$ at $c_a$ and $c_b$ only, $f_{ab}$ is neither $c_a$ nor $c_b$ and both $\rho^a(c_a, f_{ab})$ and $\rho^b(f_{ab}, c_b)$ are non-trivial paths. Moreover, $f_{ab}$ is an ancestor of $\mathbf{S}$ by means of the directed path $\lambda^a(f_{ab}, S_a)$ from $f_{ab}$ to $S_a \in \mathbf{S}$. Let $\pi^\prime$ be the concatenation $\pi((i, t_i), c_a) \oplus \rho^a(c_a, f_{ab}) \oplus \rho^b(f_{ab}, c_b) \oplus \pi(c_b, (j, t_j))$, which by the assumptions on $\rho^a$ and $\rho^b$ as well as the definition of $f_{ab}$ is a path (rather than a walk) in $\D$ between $(i, t_i)$ and $(j, t_j)$. Consider the four constituting subpaths of $\pi^\prime$:
\begin{enumerate}
\item First, $\pi((i, t_i), c_a)$ is active given $\mathbf{S} \setminus \{(i, t_i), c_a\} = \mathbf{S}$ because $\pi$ is active given $\mathbf{S}$.
\item Second and similar to the first point, $\pi(c_b, (j, t_j))$ is active given $\mathbf{S} \setminus \{c_b, (j, t_j)\}$.
\item Third, $\rho^a(c_a, f_{ab})$ is active given $\mathbf{S} \setminus \{d_a\}$ since $\rho^a$ is active given $\mathbf{S} \setminus \{d_a\}$ (see the last paragraph in part~1 of this proof). There are two cases:
\begin{enumerate}
\item If $f_{ab} = d_a$, then $\rho^a(c_a, f_{ab})$ is active given $\mathbf{S} \setminus \{c_a, f_{ab}\} = \mathbf{S} \setminus \{f_{ab}\} = \mathbf{S} \setminus \{d_{a}\}$.
\item If $f_{ab} \neq d_a$, then $f_{ab} \notin \mathbf{S}$ (because no vertex on $\rho^a$ other than $d_a$ is in $\mathbf{S}$) and $d_a$ is not on $\rho^a(c_a, f_{ab})$. Because $\rho^a(c_a, f_{ab})$ is collider-free, we thus see that $\rho^a(c_a, f_{ab})$ is active given $\mathbf{S} \setminus \{c_a, f_{ab}\} = \mathbf{S} = \{d_a\} \cup \left(\mathbf{S} \setminus \{d_a\}\right)$.
\end{enumerate}
Thus, $\rho^a(c_a, f_{ab})$ is active given $\mathbf{S} \setminus \{c_a, f_{ab}\}$.
\item Fourth and similar to the third point, $\rho^a(f_{ab}, c_{b})$ is active given $\mathbf{S} \setminus \{f_{ab}, c_b\}$.
\end{enumerate}
Since the junction points $c_a$ and $c_b$ are non-colliders on $\pi^\prime$ and not in $\mathbf{S}$, whereas the third junction point $f_{ab}$ is a collider on $\pi^\prime$ and an ancestor of $\mathbf{S}$, Lemma 3.3.1 in \citet{Spirtes2000} asserts that $\pi^\prime$ is active given $\mathbf{S}$. Lastly, there are exactly $(a-1) + 1 + (n_c - b) = n_c - (b - a) \leq n_c - 1$ colliders on $\pi^\prime$.

Thus, if $\rho^{1}, \dots , \rho^{n_c}$ fulfills the first four conditions of a collider extension structure of $\pi$ with respect to $\mathbf{S}$ and $\mathbf{O}(t -\taumax, t)$, then it also fulfills the fifth condition or there is a path $\pi^\prime$ as in point 2 of the lemma.
\end{proof}

By induction over the number of colliders $n_c$, using Lemma~\ref{lemmaapp:MG_subgraph_of_MGc_reason_part_1} as the induction base case and Lemma~\ref{lemmaapp:extension_by_rho^a_paths} for the induction step, we thus arrive at the following conclusion: For the purpose of proving Lemma~4.14 it is thus sufficient to consider $d$-connecting paths $\pi$ in $\D$ for which there is a collider extension structure of $\pi$ with respect to $\mathbf{S}$ and $\mathbf{O}(t -\taumax, t)$. This reasoning allows to overcome the first complication mentioned above in the following way (see Lemma~\ref{lemmaapp:glueing_induced_paths}).

\begin{mydef}[Notation for remaining parts of the proof]\label{defapp:collider_extended_paths}
Given the assumptions and notation of Def.~\ref{defapp:collider_extension}, let $\rho^{1}, \dots , \rho^{n_c}$ be a collider extension structure of $\pi$ with respect to $\mathbf{S}$ and $\mathbf{O}(t -\taumax, t)$. We let $\rho^0$ and $\rho^{n_c+1}$ be the trivial paths that, respectively, consist of $c_0 = d_0 = (i, t_i)$ and $c_{n_c+1} = d_{n_c+1} = (j, t_j)$ only. Moreover, for all $0 \leq a_1 < a_2 \leq n_c + 1$ we let $\pi^{a_1, a_2} = \pi(c_{a_1}, c_{a_2})$ and $\tilde{\pi}^{a_1, a_2} = \rho^{a_1}(d_{a_1}, c_{a_1}) \oplus \pi^{a_1, a_2} \oplus \rho^{a_2}$.
\end{mydef}

\begin{myremark}[on Def.~\ref{defapp:collider_extended_paths}]
The concatenation $\tilde{\pi}^{a_1, a_2}$ is a path (rather than a walk) in $\D$ according to the fourth and fifth property of collider extension structures, and both end-points of $\tilde{\pi}^{a_1, a_2}$ are in $\mathbf{O}(t-\taumax, t)$ (according to the first property of collider extension structures).
\end{myremark}

\begin{mylemma}\label{lemmaapp:collider_extended_paths_active}
Given the assumptions and notation of Def.~\ref{defapp:collider_extended_paths}, for all $0 \leq a \leq n_c$:
\begin{enumerate}
\item $\tilde{\pi}^{a, a+1}$ is active given $\mathbf{S} \setminus \{d_a, d_{a+1}\}$.
\item $\tilde{\pi}^{a, a+1}_{ci}$, the canonical induced path of $\tilde{\pi}^{a, a+1}$, is active given $\mathbf{S} \setminus \{d_a, d_{a+1}\}$.
\end{enumerate}
\end{mylemma}

\begin{myremark}[on Lemma~\ref{lemmaapp:collider_extended_paths_active}]
These two statements concern different graphs: The first statement is with respect to $\D$, the second statement is with respect to $\Dc(\Mtaumax(\D))$.
\end{myremark}

\begin{proof}[Proof of Lemma~\ref{lemmaapp:collider_extended_paths_active}]
\textbf{1.}
Since $\rho^a$ and $\rho^{a+1}$ are trivial paths or non-trivial paths out of, respectively, $c_{a}$ and $c_{a+1}$ and since $\pi^{a, a+1}$ is collider-free, $\tilde{\pi}^{a, a+1}$ is collider-free. Thus, assuming that $\tilde{\pi}^{a, a+1}$ is blocked given $\mathbf{S} \setminus \{d_a, d_{a+1}\}$, there is a non end-point non-collider $v$ on $\tilde{\pi}^{a, a+1}$ with $v \in \mathbf{S} \setminus \{d_a, d_{a+1}\}$. This vertex cannot be on $\rho^{a}(d_{a}, c_{a})$ because according to the definition of collider extension structures the opposite requires $v = d_a$. Since $v$ can similarly not be on $\rho^{a+1}$, $v$ must be a non end-point vertex on $\pi^{a_1, a_2}$. However, then $v$ must be a non end-point non-collider on $\pi$, which is a contradiction because $v \in \mathbf{S}$ and $\pi$ is active given $\mathbf{S}$.

\textbf{2.}
Since $\tilde{\pi}^{a, a+1}$ is collider-free, as shown in the proof of part~1 of Lemma~\ref{lemmaapp:collider_extended_paths_active}, and since both $d_{a}$ and $d_{a+1}$ are by definition in $\mathbf{O}(t-\taumax, t)$, the canonical induced path $\tilde{\pi}^{a, a+1}_{ci}$ exists. The statement follows by combining part~1 of Lemma~\ref{lemmaapp:collider_extended_paths_active} with part 5 of Lemma~\ref{lemmaapp:canonical_path_properties}.
\end{proof}

At this point we face the second complication mentioned above: We can now \emph{not} straight away apply Lemma 3.3.1 in \citet{Spirtes2000} to the ordered sequence of paths $\tilde{\pi}^{0, 1}_{ci}, \dots , \tilde{\pi}^{n_c, n_c+1}_{ci}$ in order to infer the existence of a path between $d_0 = (i, t_i)$ and $d_{n_c+1} = (j, t_j)$ in $\Dc(\Mtaumax(\D))$ that is active given $\mathbf{S}$. To recall, the reason is that $\tilde{\pi}^{a, a+1}_{ci}$ may be out of one of its end-point vertices although $\tilde{\pi}^{a, a+1}$ is into this vertex. To resolve this complication, we now show that such a situation requires the existence of a certain inducing path in $\D$ and hence the existence of an additional edge in $\Dc(\Mtaumax(\D))$ which can be used to bypass that vertex.

\begin{mydef}[Canonical paths]\label{defapp:canonical_paths}
Let $\pi$ be a path in $\D$ between distinct observable vertices $(i, t_i)$ and $(j, t_j)$ with $t-\taumax \leq t_i, t_j \leq \taumax$. A path $\pi_c$ between $(i, t_i)$ and $(j, t_j)$ in $\Dc(\Mtaumax(\D))$ is \emph{canonical with respect to $\pi$} if all of the following holds:
\begin{enumerate}
\item All observable vertices on $\pi_c$ are also on $\pi$.
\item If $(k_1, t_1)$, $(k_2, t_2)$ and $(k_3, t_3)$ are distinct observable vertices on $\pi_c$ and $(k_2, t_2)$ is on $\pi_c((k_1, t_1), (k_3, t_3))$, then $(k_2, t_2)$ is on $\pi((k_1, t_1), (k_3, t_3))$.
\item $\pi_c$ is collider-free.
\item If $(k_1, t_1)$ and $(k_2, t_2)$ are distinct observable vertices on $\pi_c$, then $\pi((k_1, t_1), (k_2, t_2))$ is an inducing path relative to $\mathbf{O}(\operatorname{max}(t_1, t_2) - \taumax, t)[\pi_c((k_1, t_1), (k_2, t_2))]$.
\end{enumerate}
\end{mydef}

\begin{myremark}[on Def.~\ref{defapp:canonical_paths}]
First, Lemma~\ref{lemmaapp:canonical_path_properties} implies that the canonically induced path $\pi_{ci}$ of a collider-free path $\pi$ between distinct observable vertices $(i, t_i)$ and $(j, t_j)$ with $t-\taumax \leq t_i, t_j \leq \taumax$ is canonical with respect to $\pi$. Second, while the first property in Def.~\ref{defapp:canonical_paths} is implied by the second property and would thus not be needed, we have included it in the definition for clarity.
\end{myremark}

\begin{mylemma}\label{lemmaapp:glueing_induced_paths}
Given the assumptions and notation of Def.~\ref{defapp:collider_extended_paths}, let $0 \leq a_1 < a_2 < a_3 \leq n_c+1$. Assume that $\tilde{\pi}^{a_1, a_2}_c$ is canonical with respect to $\tilde{\pi}^{a_1, a_2}$ and active given $\mathbf{S} \setminus \{d_{a_1}, d_{a_2}\}$, and that $\tilde{\pi}^{a_2, a_3}_c$ is canonical with respect to $\tilde{\pi}^{a_2, a_3}$ and active given $\mathbf{S} \setminus \{d_{a_2}, d_{a_3}\}$. Then: If at least one of $\tilde{\pi}^{a_1, a_2}_c$ and $\tilde{\pi}^{a_2, a_3}_c$ is out of $d_{a_2}$, then there is a path $\tilde{\pi}^{a_1, a_3}_c$ that is canonical with respect to $\tilde{\pi}^{a_1, a_3}$ and active given $\mathbf{S} \setminus \{d_{a_1}, d_{a_3}\}$.
\end{mylemma}

\begin{proof}[Proof of Lemma~\ref{lemmaapp:glueing_induced_paths}]
We here assume that $\tilde{\pi}^{a_2, a_3}_c$ is out of $d_{a_2}$, the case in which $\tilde{\pi}^{a_2, a_3}_c$ is into $d_{a_2}$ and $\tilde{\pi}^{a_1, a_2}_c$ out of $d_{a_2}$ follows equivalently. To simplify notation we write $t(v)$ for the time step of a vertex $v$, i.e., $v = (\cdot, t(v))$. We divide the proof into 14 steps.

\underline{Step 1: No vertex on $\tilde{\pi}^{a_1, a_2}$ or $\tilde{\pi}^{a_2, a_3}$ is after $t$}\\
If there would be a vertex on $\tilde{\pi}^{a_1, a_2}$ (on $\tilde{\pi}^{a_2, a_3}$) that is after $t$, then this path would have a collider after $t$ because of time order of $\D$ and because both $d_{a_1}$ and $d_{a_2}$ (both $d_{a_2}$ and $d_{a_3}$) are not after $t$. Again using time order, this collider could not be unblocked by $\mathbf{S} \setminus \{d_{a_1}, d_{a_2}\}$ (by $\mathbf{S} \setminus \{d_{a_2}, d_{a_3}\}$) because by definition of $\mathbf{S}$ all vertices in $\mathbf{S}$ are not after $t$ (see Def.~\ref{defapp:collider_extension}). This conclusion contradicts the assumption that $\tilde{\pi}^{a_1, a_2}$ (on $\tilde{\pi}^{a_2, a_3}$) is active given $\mathbf{S} \setminus \{d_{a_1}, d_{a_2}\}$ (given $\mathbf{S} \setminus \{d_{a_2}, d_{a_3}\}$).

\underline{Step 2: $c_{a_2}$ and $d_{a_2}$ are in $\D$ ancestors of $d_{a_3}$}\\
Because $\tilde{\pi}^{a_2, a_3}_c$ is out of $d_{a_2}$ and collider-free (the latter by means of being canonical with respect to $\tilde{\pi}^{a_2, a_3}$), the path $\tilde{\pi}^{a_2, a_3}_c$ is directed from $d_{a_2}$ to $d_{a_3}$. Hence, all vertices on $\tilde{\pi}^{a_2, a_3}_c$ are ancestors of $d_{a_3}$ and descendants of $d_{a_2}$ in $\Dc(\Mtaumax(\D))$ and thus, by part 2 of Lemma~\ref{lemmaapp:D_same_ancestral_mtaumax_as_Dc(Mtaumax(D))}, also in $\D$. Because $c_{a_2}$ is an ancestor of $d_{a_2}$ in $\D$ by means of $\rho^{a_2}$, we thus see that $c_{a_2}$ is in $\D$ an ancestor of every vertex on $\tilde{\pi}^{a_2, a_3}_c$.

\underline{Step 3: Definition and properties of $g_{a_2, a_3}$ and properties of $\tilde{\pi}^{a_2, a_3}(c_{a_2}, g_{a_2, a_3})$}\\
Let $g_{a_2, a_3}$ be the vertex next to $d_{a_2}$ on $\tilde{\pi}^{a_2, a_3}_c$ (this may be $d_{a_3}$). Since $\tilde{\pi}^{a_2, a_3}_c$ is directed from $d_{a_2}$ to $d_{a_3}$, the path $\tilde{\pi}^{a_2, a_3}_c$ is into $g_{a_2, a_3}$. Because in $\Dc(\Mtaumax(\D))$ there are no edges into unobservable vertices, we see that $g_{a_2, a_3}$ is observable. Moreover, using time order of $\Dc(\Mtaumax(\D))$ and that $\tilde{\pi}^{a_2, a_3}_c$ is directed from $d_{a_2}$ to $d_{a_3}$, we see that $g_{a_2, a_3}$ is not before $d_{a_2}$ and not after $d_{a_3}$. Hence, $g_{a_2, a_3}$ is within the observed time window $[t-\taumax, t]$.

Since $\tilde{\pi}^{a_2, a_3}_c$ is canonical with respect to $\tilde{\pi}^{a_2, a_3}$ and $g_{a_2, a_3}$ is on $\tilde{\pi}^{a_2, a_3}_c$, the vertex $g_{a_2, a_3}$ is on $\tilde{\pi}^{a_2, a_3} = \rho^{a_2}(d_{a_2}, c_{a_2}) \oplus \pi^{a_2, a_3} \oplus \rho^{a_3}$. If $g_{a_2, a_3}$ were on $\rho^{a_2}$, then $g_{a_2, a_3}$ would in $\D$ be an ancestor of $d_{a_2}$ by means of $\rho^{a_2}(g_{a_2, a_3}, d_{a_2})$. This ancestorship would violate acyclicity of $\D$ because $g_{a_2, a_3}$ is a descendant of $d_{a_2}$ according to step 2. Hence, $g_{a_2, a_3}$ is on $\tilde{\pi}^{a_2, a_3}(c_{a_2}, d_{a_3}) = \pi^{a_2, a_3} \oplus \rho^{a_3}$ excluding $c_{a_2}$. Moreover, $\tilde{\pi}^{a_2, a_3}(c_{a_2}, g_{a_2, a_3})$ is a non-trivial subpath of $\tilde{\pi}^{a_2, a_3}(c_{a_2}, d_{a_3})$ and of $\tilde{\pi}^{a_2, a_3}(d_{a_2}, g_{a_2, a_3})$. Lastly, $\tilde{\pi}^{a_2, a_3}(c_{a_2}, g_{a_2, a_3})$ is into $c_{a_2}$ because $\pi$ is into $c_{a_2}$.

Since $\tilde{\pi}^{a_2, a_3}_c$ is canonical with respect to $\tilde{\pi}^{a_2, a_3}$ and $t(d_{a_2}) \leq t(g_{a_1, a_2})$ by time order, $\tilde{\pi}^{a_2, a_3}(d_{a_2}, g_{a_2, a_3})$ is an inducing path relative to $\mathbf{O}(t(g_{a_2, a_3})-\taumax, t)[\tilde{\pi}^{a_2, a_3}_c(d_{a_2}, g_{a_2, a_3})]$. Here, the simplification $\mathbf{O}(t(g_{a_2, a_3})-\taumax, t)[\tilde{\pi}^{a_2, a_3}_c(d_{a_2}, g_{a_2, a_3})] = \mathbf{O}(t(g_{a_2, a_3})-\taumax, t)$ applies because $\tilde{\pi}^{a_2, a_3}_c(d_{a_2}, g_{a_2, a_3})$ by definition of $g_{a_2, a_3}$ consists of its end point vertices $d_{a_2}$ and $g_{a_2, a_3}$ only. Using the defining properties of inducing paths, we thus see that the path $\tilde{\pi}^{a_2, a_3}(d_{a_2}, g_{a_2, a_3})$ has the following two properties:
\begin{enumerate}
\item First, if $v$ is an observable non end-point non-collider on $\tilde{\pi}^{a_2, a_3}(d_{a_2}, g_{a_2, a_3})$, then $t(v) < t(g_{a_2, a_3})-\taumax$ or $t(v) > t$. Because step 1 excludes $t(v) > t$, in fact $t(v) < t(g_{a_2, a_3})-\taumax$.
\item Second, if $v$ is a collider on $\tilde{\pi}^{a_2, a_3}(d_{a_2}, g_{a_2, a_3})$, then $v$ is in $\D$ an ancestor of $d_{a_2}$ or $g_{a_2, a_3}$. Since $d_{a_2}$ is in $\D$ an ancestor of $g_{a_2, a_3}$, the vertex $v$ is, in fact, an ancestor of $g_{a_2, a_3}$.
\end{enumerate}
Both of these statements are also true for $\tilde{\pi}^{a_2, a_3}(c_{a_2}, g_{a_2, a_3})$ because it is a subpath of $\tilde{\pi}^{a_2, a_3}(d_{a_2}, g_{a_2, a_3})$.

\underline{Step 4: All observable vertices on $\rho^{a_2}$ other than $d_{a_2}$ are before $t(g_{a_2, a_3}) - \taumax$}\\
Let $v \neq d_{a_2}$ be an observable vertex on $\rho^{a_2}$. Since $g_{a_2, a_3}$ is not on $\rho^{a_2}$, see step 2, $v$ is then a non end-point vertex on $\tilde{\pi}^{a_2, a_3}(d_{a_2}, g_{a_2, a_3})$. Moreover, since $\rho^{a_2}$ is directed from $c_{a_2}$ to $d_{a_2}$, the vertex $v$ is a non-collider on $\tilde{\pi}^{a_2, a_3}(d_{a_2}, g_{a_2, a_3})$. From step 3 we then get that $t(v) < t(g_{a_2, a_3}) - \taumax$ or $t(v) > t$, and step 1 further excludes the case $t(v) > t$.

\underline{Step 5: Definition and properties of $h_{a_1, a_2}$ and properties of $\tilde{\pi}^{a_1, a_2}(h_{a_1, a_2}, c_{a_2})$}\\
Let $h_{a_1, a_2}$ be the observable vertex on $\tilde{\pi}^{a_1, a_2}_c$ closest to $d_{a_2}$ other than $d_{a_2}$ itself that is not more than $\taumax$ time steps before $g_{a_2, a_3}$, i.e., for which $t(g_{a_1, a_2}) - \taumax \leq t(h_{a_1, a_2})$ (note that $h_{a_1,a_2}$ may be $d_{a_1}$).

Since $\tilde{\pi}^{a_1, a_2}_c$ is canonical with respect to $\tilde{\pi}^{a_1, a_2}$, the vertex $h_{a_1, a_2}$ is on $\tilde{\pi}^{a_1, a_2} = \rho^{a_1}(d_{a_1}, c_{a_1}) \oplus \pi^{a_1, a_2} \oplus \rho^{a_2}$. Due to step 4 and $t(g_{a_1, a_2}) - \taumax \leq t(h_{a_1, a_2})$, the vertex $h_{a_1, a_2}$ cannot be on $\rho^{a_2}$ unless $h_{a_1, a_2} = d_{a_2}$, which is, however, excluded by definition. Hence, $h_{a_1, a_2}$ is on $\tilde{\pi}^{a_1, a_2}(d_{a_1}, c_{a_2}) = \rho^{a_1}(d_{a_1}, c_{a_1}) \oplus \pi^{a_1, a_2}$ excluding $c_{a_2}$ (because $c_{a_2}$ is on $\rho_{a_2}$). Moreover, $\tilde{\pi}^{a_1, a_2}(h_{a_1, a_2}, c_{a_2})$ is a non-trivial subpath of $\tilde{\pi}^{a_1, a_2}(d_{a_1}, c_{a_2})$ and of $\tilde{\pi}^{a_2, a_3}(h_{a_1, a_2}, d_{a_2})$. Lastly, $\tilde{\pi}^{a_1, a_2}(h_{a_1, a_2}, c_{a_2})$ is into $c_{a_2}$ because $\pi^{a_1, a_2}$ is into $c_{a_2}$.

Write $t_{hd} = \operatorname{max}(t(h_{a_1, a_2}), t(d_{a_2}))$ and $t_{hg} = \operatorname{max}(t(h_{a_1, a_2}), t(g_{a_2, a_3}))$. Since $\tilde{\pi}^{a_1, a_2}_c$ is canonical with respect to $\tilde{\pi}^{a_1, a_2}$, the path $\tilde{\pi}^{a_1, a_2}(h_{a_1, a_2}, d_{a_2})$ is an inducing path relative to $\mathbf{O}(t_{hd}-\taumax, t)[\tilde{\pi}^{a_1, a_2}_c(h_{a_1, a_2}, d_{a_2})]$. In particular, $\tilde{\pi}^{a_1, a_2}(h_{a_1, a_2}, d_{a_2})$ has the following two properties:
\begin{enumerate}
\item First, if $v$ is an observable non end-point non-collider on $\tilde{\pi}^{a_1, a_2}(h_{a_1, a_2}, d_{a_2})$, then $t(v) < t_{hd} - \taumax$ or $t(v) > t$ or $v$ is on $\tilde{\pi}^{a_1, a_2}_c(h_{a_1, a_2}, d_{a_2})$. The case $t(v) > t$ is excluded by step 1. If $v$ is on $\tilde{\pi}^{a_1, a_2}_c(h_{a_1, a_2}, d_{a_2})$, then $t(v) < t(g_{a_1, a_2}) - \taumax$ by definition of $h_{a_1, a_2}$. Note that $t_{hd} - \taumax \leq t_{hg} - \taumax$ and $t(g_{a_1, a_2}) - \taumax \leq t_{hg} - \taumax$. Hence, in any case, $t(v) < t_{hg} - \taumax$.
\item Second, if $v$ is a collider on $\tilde{\pi}^{a_1, a_2}(h_{a_1, a_2}, d_{a_2})$, then $v$ is in $\D$ an ancestor of $h_{a_1, a_2}$ or $d_{a_2}$. Because $d_{a_2}$ is in $\D$ an ancestor of $g_{a_2, a_3}$, the vertex $v$ is, in fact, in $\D$ an ancestor of $h_{a_1, a_2}$ or $g_{a_2, a_3}$.
\end{enumerate}
Both of these statements are also true for $\tilde{\pi}^{a_1, a_2}(h_{a_1, a_2}, c_{a_2})$ because $\tilde{\pi}^{a_1, a_2}(h_{a_1, a_2}, c_{a_2})$ is a subpath of $\tilde{\pi}^{a_1, a_2}(h_{a_1, a_2}, d_{a_2})$.

\underline{Step 6: $g_{a_2, a_3}$ and $h_{a_1, a_2}$ are adjacent or almost adjacent in $\Dc(\Mtaumax(\D))$}\\
Consider the concatentation $\psi = \tilde{\pi}^{a_1, a_2}(h_{a_1, a_2}, c_{a_2}) \oplus \tilde{\pi}^{a_2, a_3}(c_{a_2}, g_{a_2, a_3})$. This concatenation is a path (rather than a walk) in $\D$ because $\tilde{\pi}^{a_2, a_3}(c_{a_2}, g_{a_2, a_3})$ is a subpath of $\pi^{a_2, a_3} \oplus \rho^{a_3}$ and because different $\rho^a$ do not intersect (by definition of collider extension structures). The junction point $c_{a_2}$ is a collider on $\psi$ because both $\tilde{\pi}^{a_1, a_2}(h_{a_1, a_2}, c_{a_2})$ and $\tilde{\pi}^{a_2, a_3}(c_{a_2}, g_{a_2, a_3})$ are into $c_{a_2}$, see steps 5 and 3. We now show that $\psi$ is an inducing path relative to $\mathbf{O}(t_{hg} - \taumax, t_{hg})$. To this end, we separately look at the colliders and non end-point non-colliers on $\psi$.

\textit{Colliders:} According to steps 3 and 5, every collider on $\tilde{\pi}^{a_1, a_2}(h_{a_1, a_2}, c_{a_2})$ and on $\tilde{\pi}^{a_2, a_3}(c_{a_2}, g_{a_2, a_3})$ is in $\D$ an ancestor of $h_{a_1, a_2}$ or $g_{a_2, a_3}$. The junction point $c_{a_2}$ is in $\D$ an ancestor of $g_{a_2, a_3}$ according to step 1.

\textit{Non end-point non-colliders:} Let $v$ be a non end-point non-collider on $\psi$. Since $c_{a_2}$ is a collider on $\psi$, the vertex $v$ is then a non end-point non-collider on $\tilde{\pi}^{a_1, a_2}(h_{a_1, a_2}, c_{a_2})$ or a non end-point non-collider on $\tilde{\pi}^{a_2, a_3}(c_{a_2}, g_{a_2, a_3})$. With steps 3 and 5 we then get $t(v) < t_{hg} - \taumax$.

Consequently, $\psi$ is an inducing path relative to the set of observable vertices within $\mathbf{O}(t_{hg} - \taumax, t_{hg})$. By shifting this structure forward in time by $t - t_{hg}$ time steps, we see that the forward shifted copies of $g_{a_2, a_3}$ and $h_{a_1, a_2}$ are adjacent in $\Mtaumax(\D)$. Hence, $g_{a_2, a_3}$ and $h_{a_1, a_2}$ are adjacent or almost adjacent in $\Dc(\Mtaumax(\D))$ according to Lemma~\ref{lemmaapp:at_most_almost_adjacent}.

For reference further below we note that $\psi$ is also on inducing path relative to $\mathbf{O}(t_{hg} - \taumax, t)$. This statement follows because, as shown, if $v$ is a non end-point non-collider on $\psi$, then $t(v) < t_{hg} - \taumax$ and thus $v$ is not in $\mathbf{O}(t_{hg} - \taumax, t) \setminus \mathbf{O}(t_{hg} - \taumax, t_{hg})$.

\underline{Step 7: Properties of $\tilde{\pi}^{a_2, a_3}_c(g_{a_2, a_3}, d_{a_3})$}\\
Since $\tilde{\pi}^{a_2, a_3}_c$ is directed from $d_{a_2}$ to $d_{a_3}$, see step 2, $\tilde{\pi}^{a_2, a_3}_c(g_{a_2, a_3}, d_{a_3})$ is the trivial path consisting of the single vertex $g_{a_2, a_3} = d_{a_3}$ or a non-trivial directed path from $g_{a_2, a_3}$ to $d_{a_3}$ and, hence, out of $g_{a_2, a_3}$ and into $d_{a_3}$. In particular, $\tilde{\pi}^{a_2, a_3}_c(g_{a_2, a_3}, d_{a_3})$ is collider-free.

Consider any vertex $v$ on $\tilde{\pi}^{a_2, a_3}_c(g_{a_2, a_3}, d_{a_3})$. Because in $\Dc(\Mtaumax(\D))$ there are no edges into unobservable vertices and because $g_{a_2, a_3}$ is observable, $v$ is observable. Moreover, because $\tilde{\pi}^{a_2, a_3}_c$ is canonical with respect to $\tilde{\pi}^{a_2, a_3}$, the vertex $v$ is on $\tilde{\pi}^{a_2, a_3}(g_{a_2, a_3}, d_{a_3})$. Since $g_{a_2, a_3}$ is on $\pi^{a_2, a_3} \oplus \rho^{a_3}$ excluding $c_{a_2}$, see step 3, also $v$ is on $\pi^{a_2, a_3} \oplus \rho^{a_3}$ excluding $c_{a_2}$. In particular, $v$ is not on $\rho^{a_2}$.

\underline{Step 8: Properties of $\tilde{\pi}^{a_1, a_2}_c(d_{a_1}, h_{a_1, a_2})$}\\
Because $\tilde{\pi}^{a_1, a_2}_c$ is canonical with respect to $\tilde{\pi}^{a_1, a_2}$, the path $\tilde{\pi}^{a_1, a_2}_c$ is collider-free. Consequently, also $\tilde{\pi}^{a_1, a_2}_c(d_{a_1}, h_{a_1, a_2})$ is collider-free.

Consider any observable vertex $v$ on $\tilde{\pi}^{a_1, a_2}_c(d_{a_1}, h_{a_1, a_2})$. Because $\tilde{\pi}^{a_1, a_2}_c$ is canonical with respect to $\tilde{\pi}^{a_1, a_2}$, the vertex $v$ is on $\tilde{\pi}^{a_1, a_2}(d_{a_1}, h_{a_1, a_2})$. Since $h_{a_1, a_2}$ is on $\rho^{a_1}(d_{a_1}, c_{a_1}) \oplus \pi^{a_1, a_2}$ excluding $c_{a_2}$, see step 5, also $v$ is on $\rho^{a_1}(d_{a_1}, c_{a_1}) \oplus \pi^{a_1, a_2}$ excluding $c_{a_2}$. In particular, we conclude that $v$ is not on $\rho^{a_2}$.

We now show that $\tilde{\pi}^{a_1, a_2}_c(d_{a_1}, h_{a_1, a_2})$ and $\tilde{\pi}^{a_2, a_3}_c(g_{a_2, a_3}, d_{a_3})$ do not intersect. Assume the opposite, i.e., let $w$ be on both $\tilde{\pi}^{a_1, a_2}_c(d_{a_1}, h_{a_1, a_2})$ and $\tilde{\pi}^{a_2, a_3}_c(g_{a_2, a_3}, d_{a_3})$. There are two cases:
\begin{enumerate}
\item First, assume $w$ is observable. Then, according to step 7 and the previous discussion in the current step, $w$ is on $\rho^{a_1}(d_{a_1}, c_{a_1}) \oplus \pi^{a_1, a_2}$ excluding $c_{a_2}$ and on $\pi^{a_2, a_3} \oplus \rho^{a_3}$ excluding $c_{a_2}$. These observations contradict each other because $\rho^{a_1}(d_{a_1}, c_{a_1}) \oplus \pi^{a_1, a_2}$ and $\pi^{a_2, a_3} \oplus \rho^{a_3}$ intersect at $c_{a_2}$ only.
\item Second, assume $w$ is unobservable. Then, $w$ is a non end-point vertex of $\tilde{\pi}^{a_1, a_2}_c(d_{a_1}, h_{a_1, a_2})$ and of non end-point vertex of $\tilde{\pi}^{a_2, a_3}_c(g_{a_2, a_3}, d_{a_3})$. Moreover, as follows immediately from the definition of canonical ts-DAGs, every unobservable vertex in $\Dc(\Mtaumax(\D))$ is adjacent to exactly two vertices, both of which are observable. We thus find that $\tilde{\pi}^{a_1, a_2}_c(d_{a_1}, h_{a_1, a_2})$ and $\tilde{\pi}^{a_2, a_3}_c(g_{a_2, a_3}, d_{a_3})$ also intersect at an observable vertex, which has already been ruled out in the previous case and thus is a contradiction.
\end{enumerate}

\underline{Step 9: Construction of $\tilde{\pi}^{a_1, a_3}_c$}\\
The fact that $g_{a_2, a_3}$ and $h_{a_1, a_2}$ are adjacent or almost adjacent in $\Dc(\Mtaumax(\D))$, see step 6, means that in $\Dc(\Mtaumax(\D))$ there is a path $\kappa^{a_2} = h_{a_1, a_2} \tailhead g_{a_2, a_3}$ or $\kappa^{a_2} = h_{a_1, a_2} \headtail g_{a_2, a_3}$ or $\kappa^{a_2} = h_{a_1, a_2} \headtail u_{a_2} \tailhead g_{a_2, a_3}$ with $u_{a_2}$ unobservable. Let $\tilde{\pi}^{a_1, a_3}_c$ be the concatenation $\tilde{\pi}^{a_1, a_2}_c(d_{a_1}, h_{a_1, a_2}) \oplus \kappa^{a_2} \oplus \tilde{\pi}^{a_2, a_3}_c(g_{a_2, a_3}, d_{a_3})$. We now show that $\tilde{\pi}^{a_1, a_3}_c$ is a path.

In step 8 we have already shown that $\tilde{\pi}^{a_1, a_2}_c(d_{a_1}, h_{a_1, a_2})$ and $\tilde{\pi}^{a_2, a_3}_c(g_{a_2, a_3}, d_{a_3})$ do not intersect. Thus, if $\kappa^{a_2} = h_{a_1, a_2} \tailhead g_{a_2, a_3}$ or $\kappa^{a_2} = h_{a_1, a_2} \headtail g_{a_2, a_3}$, then $\tilde{\pi}^{a_1, a_3}_c$ contains every vertex at most once and hence is a path. Now assume that $\kappa^{a_2} = h_{a_1, a_2} \headtail u_{a_2} \tailhead g_{a_2, a_3}$. In this case, we show that $u_{a_2}$ is neither on $\tilde{\pi}^{a_1, a_2}_c(d_{a_1}, h_{a_1, a_2})$ nor $\tilde{\pi}^{a_2, a_3}_c(g_{a_2, a_3}, d_{a_3})$:
\begin{enumerate}
\item First, because as shown in step 7 all vertices on $\tilde{\pi}^{a_2, a_3}_c(g_{a_2, a_3}, d_{a_3})$ are observable, $u_{a_2}$ cannot be on $\tilde{\pi}^{a_2, a_3}_c(g_{a_2, a_3}, d_{a_3})$.
\item Second, if $u_{a_2}$ is on $\tilde{\pi}^{a_1, a_2}_c(d_{a_1}, h_{a_1, a_2})$, then it is a non end-point vertex of this path. Since every unobservable vertex in $\Dc(\Mtaumax(\D))$ is adjacent two exactly vertices, which for $u_{a_2}$ are $h_{a_1, a_2}$ and $g_{a_2, a_3}$, we thus find that $g_{a_2, a_3}$ is on $\tilde{\pi}^{a_1, a_2}_c(d_{a_1}, h_{a_1, a_2})$. This observation contradicts the fact that $\tilde{\pi}^{a_1, a_2}_c(d_{a_1}, h_{a_1, a_2})$ and $\tilde{\pi}^{a_2, a_3}_c(g_{a_2, a_3}, d_{a_3})$ have no common vertex.
\end{enumerate}

For reference below we note that by construction $\tilde{\pi}^{a_1, a_2}_c(d_{a_1}, h_{a_1, a_2}) = \tilde{\pi}^{a_1, a_3}_c(d_{a_1}, h_{a_1, a_2})$ and $\tilde{\pi}^{a_2, a_3}_c(g_{a_2, a_3}, d_{a_3}) = \tilde{\pi}^{a_1, a_3}_c(g_{a_2, a_3}, d_{a_3})$.

\underline{Step 10: $\tilde{\pi}^{a_1, a_3}_c$ is collider-free}\\
Since the three constituting subpaths $\tilde{\pi}^{a_1, a_2}_c(d_{a_1}, h_{a_1, a_2})$, $\kappa^{a_2}$ and $\tilde{\pi}^{a_2, a_3}_c(g_{a_2, a_3}, d_{a_3})$ are collider-free, only $h_{a_1, a_2}$ or $g_{a_2, a_3}$ can potentially be colliders on $\pi^{a_1, a_3}_c$.

First, because $\tilde{\pi}^{a_2, a_3}_c(g_{a_2, a_3}, d_{a_3})$ is a trivial path or a non-trivial path and out of $g_{a_2, a_3}$, see step 7, $g_{a_2, a_3}$ is a non-collider on $\tilde{\pi}^{a_1, a_3}_c$.

Second, assume that $h_{a_1, a_2}$ is a collider on $\tilde{\pi}^{a_1, a_3}_c$. This premise requires that $\kappa^{a_2}$ is $h_{a_1, a_2} \headtail g_{a_2, a_3}$ or $h_{a_1, a_2} \headtail u_{a_2} \tailhead g_{a_2, a_3}$ and that $\tilde{\pi}^{a_1, a_2}_c(d_{a_1}, h_{a_1, a_2})$ is non-trivial and into $h_{a_1, a_2}$. In combination with the facts that $\tilde{\pi}^{a_1, a_2}_c$ is collider-free and that $h_{a_1, a_2} \neq c_{a_2}$ this form of $\tilde{\pi}^{a_1, a_2}_c(d_{a_1}, h_{a_1, a_2})$ requires that $\tilde{\pi}^{a_1, a_2}_c(h_{a_1, a_2}, c_{a_2})$ is a non-trivial directed path from $h_{a_1, a_2}$ to $c_{a_2}$. Hence, $h_{a_1, a_2}$ is an ancestor of $c_{a_2}$ in $\Dc(\Mtaumax(\D))$ and, thus, also an ancestor in $\D$. Together with step 2 this ancestral relationship shows that $h_{a_1, a_2}$ is in $\D$ an ancestor of $g_{a_2, a_3}$. Lemma~\ref{lemmaapp:at_most_almost_adjacent_2} in combination with the definition of $\kappa^{a_2}$ then requires $\kappa^{a_2} = h_{a_1, a_2} \tailhead g_{a_2, a_3}$, a contradiction.

\underline{Step 11: All observable vertices on $\tilde{\pi}^{a_1, a_3}_c$ are also on $\tilde{\pi}^{a_1, a_3}$}\\
Recall from above that $\tilde{\pi}^{a_1, a_3}_c = \tilde{\pi}^{a_1, a_2}_c(d_{a_1}, h_{a_1, a_2}) \oplus \kappa^{a_2} \oplus \tilde{\pi}^{a_2, a_3}_c(g_{a_2, a_3}, d_{a_3})$ and $\tilde{\pi}^{a_1, a_3} = \rho^{a_1}(d_{a_1}, c_{a_1}) \oplus \pi^{a_1, a_2} \oplus \pi^{a_2, a_3} \oplus \rho^{a_3}$ (the latter because $\pi^{a_1, a_2} \oplus \pi^{a_2, a_3} = \pi^{a_1, a_3}$). According to step 7, every vertex on $\tilde{\pi}^{a_2, a_3}_c(g_{a_2, a_3}, d_{a_3})$ is on $\pi^{a_2, a_3} \oplus \rho^{a_3}$ and hence on $\tilde{\pi}^{a_1, a_3}$. According to step 8, every observable vertex on $\tilde{\pi}^{a_1, a_2}_c(d_{a_1}, h_{a_1, a_2})$ is on $ \rho^{a_1}(d_{a_1}, c_{a_1}) \oplus \pi^{a_1, a_2}$  and hence on $\tilde{\pi}^{a_1, a_3}$. Lastly, due to the three particular forms that the path $\kappa^{a_2}$ may have, see step 9, every observable vertex on $\kappa^{a_2}$ is on $\tilde{\pi}^{a_2, a_3}_c(g_{a_2, a_3}, d_{a_3})$ or $\tilde{\pi}^{a_1, a_2}_c(d_{a_1}, h_{a_1, a_2})$ and hence on $\tilde{\pi}^{a_1, a_3}$.

\underline{Step 12: $\tilde{\pi}^{a_1, a_3}_c$ fulfills point 2.~in Def.~\ref{defapp:canonical_paths}}\\
For reference below we note the following results:
\begin{enumerate}
\item If $w_1$ and $w_2$ are on $\tilde{\pi}^{a_1, a_2}(d_{a_1}, h_{a_1, a_2})$ or on $\tilde{\pi}^{a_1, a_3}(d_{a_1}, h_{a_1, a_2})$, then $\tilde{\pi}^{a_1, a_2}(w_1 , w_2) = \tilde{\pi}^{a_1, a_3}(w_1 , w_2)$. This equality follows because $h_{a_1, a_2}$ and thus also $w_1$ and $w_2$ are on $\rho^{a_1}(d_{a_1}, c_{a_1}) \oplus \pi^{a_1, a_2}$.
\item If $w_1$ and $w_2$ are observable vertices on $\tilde{\pi}^{a_1, a_2}_c(d_{a_1}, h_{a_1, a_2})$, then $\tilde{\pi}^{a_1, a_2}(w_1 , w_2) = \tilde{\pi}^{a_1, a_3}(w_1 , w_2)$. This equality follows from the previous result because $w_1$ and $w_2$ are on $\tilde{\pi}^{a_1, a_2}(d_{a_1}, h_{a_1, a_2})$ by means of $\tilde{\pi}^{a_1, a_2}_c$ being canonical with respect to $\tilde{\pi}^{a_1, a_2}$.
\item If $w_1$ and $w_2$ are on $\tilde{\pi}^{a_2, a_3}(g_{a_2, a_3}, d_{a_3})$ or on $\tilde{\pi}^{a_1, a_3}(g_{a_2, a_3}, d_{a_3})$, then $\tilde{\pi}^{a_2, a_3}(w_1 , w_2) = \tilde{\pi}^{a_1, a_3}(w_1 , w_2)$. This equality follows because $g_{a_2, a_3}$ and thus also $w_1$ and $w_2$ are on $\pi^{a_2, a_3} \oplus \rho^{a_3}$.
\item If $w_1$ and $w_2$ are observable vertices on $\tilde{\pi}^{a_2, a_3}_c(g_{a_1, a_2}, d_{a_3})$, then $\tilde{\pi}^{a_2, a_3}(w_1 , w_2) = \tilde{\pi}^{a_1, a_3}(w_1 , w_2)$. This equality follows from the previous result because $w_1$ and $w_2$ are on $\tilde{\pi}^{a_2, a_3}(g_{a_2, a_3}, d_{a_3})$ by means of $\tilde{\pi}^{a_2, a_3}_c$ being canonical with respect to $\tilde{\pi}^{a_2, a_3}$.
\end{enumerate}

Let $v_1$, $v_2$ and $v_3$ be distinct observable vertices on $\tilde{\pi}^{a_1, a_3}_c$ such that $v_2$ is on $\tilde{\pi}^{a_1, a_3}_c(v_1, v_3)$ and, without loss of generality, $v_1$ is closer to $d_{a_1}$ on $\tilde{\pi}^{a_1, a_3}_c$ than $v_3$ is to $d_{a_1}$ on $\tilde{\pi}^{a_1, a_3}_c$. We distinguish several collectively exhaustive cases:
\begin{enumerate}
\item First, assume $v_3$ is on $\tilde{\pi}^{a_1, a_3}_c(d_{a_1}, h_{a_1, a_2})$. This premise implies that also both $v_1$ and $v_2$ are on $\tilde{\pi}^{a_1, a_3}_c(d_{a_1}, h_{a_1, a_2}) = \tilde{\pi}^{a_1, a_2}_c(d_{a_1}, h_{a_1, a_2})$. Hence, $v_2$ is on $\tilde{\pi}^{a_1, a_2}_c(v_1, v_3) = \tilde{\pi}^{a_1, a_3}_c(v_1, v_3)$ and thus, using that $\tilde{\pi}^{a_1, a_2}_c$ is canonical with respect to $\tilde{\pi}^{a_1, a_2}$, also on $\tilde{\pi}^{a_1, a_2}(v_1, v_3)$. Moreover, using the second result at the beginning of this step, $\tilde{\pi}^{a_1, a_2}(v_1, v_3) = \tilde{\pi}^{a_1, a_3}(v_1, v_3)$. Hence, $v_2$ is on $\tilde{\pi}^{a_1, a_3}(v_1, v_3)$.
\item Second, assume that $v_2$ is on $\tilde{\pi}^{a_1, a_3}_c(d_{a_1}, h_{a_1, a_2})$ excluding $h_{a_1, a_2}$ and that $v_3$ is not on $\tilde{\pi}^{a_1, a_3}_c(d_{a_1}, h_{a_1, a_2})$. This premise implies that also $v_1$ is on $\tilde{\pi}^{a_1, a_3}_c(d_{a_1}, h_{a_1, a_2}) = \tilde{\pi}^{a_1, a_2}_c(d_{a_1}, h_{a_1, a_2})$. Following the same steps as in the previous case with $v_3$ replaced by $h_{a_1, a_2}$, we get that $v_2$ is on $\tilde{\pi}^{a_1, a_3}(v_1, h_{a_1, a_2}) = \tilde{\pi}^{a_1, a_2}(v_1, h_{a_1, a_2})$. Moreover, $v_3$ is on $\tilde{\pi}^{a_2, a_3}_c(g_{a_2, a_3}, d_{a_3}) = \tilde{\pi}^{a_1, a_3}_c(g_{a_2, a_3}, d_{a_3})$ and hence, using that $\tilde{\pi}^{a_2, a_3}_c$ is canonical with respect to $\tilde{\pi}^{a_2, a_3}$, on $\tilde{\pi}^{a_2, a_3}(g_{a_2, a_3}, d_{a_3})$. This observation shows that $\tilde{\pi}^{a_1, a_3}(v_1, h_{a_1, a_2})$ is a subpath of $\tilde{\pi}^{a_1, a_3}(v_1, v_3)$ and hence that $v_2$ is on $\tilde{\pi}^{a_1, a_3}(v_1, v_3)$.
\item Third, assume $v_1$ is on $\tilde{\pi}^{a_1, a_3}_c(g_{a_2, a_3}, d_{a_3})$. This premise implies that also both $v_2$ and $v_3$ are on $\tilde{\pi}^{a_1, a_3}_c(g_{a_2, a_3}, d_{a_3}) = \tilde{\pi}^{a_2, a_3}_c(g_{a_2, a_3}, d_{a_3})$. Hence, $v_2$ is on $\tilde{\pi}^{a_2, a_3}_c(v_1, v_3) = \tilde{\pi}^{a_1, a_3}_c(v_1, v_3)$ and thus, using that $\tilde{\pi}^{a_2, a_3}_c$ is canonical with respect to $\tilde{\pi}^{a_2, a_3}$, also on $\tilde{\pi}^{a_2, a_3}(v_1, v_3)$. Moreover, using the fourth result at the beginning of this step, $\tilde{\pi}^{a_2, a_3}(v_1, v_3) = \tilde{\pi}^{a_1, a_3}(v_1, v_3)$. Hence, $v_2$ is on $\tilde{\pi}^{a_1, a_3}(v_1, v_3)$.
\item Fourth, assume that $v_2$ is on $\tilde{\pi}^{a_1, a_3}_c(g_{a_2, a_3}, d_{a_3})$ excluding $g_{a_2, a_3}$ and that $v_1$ is not on $\tilde{\pi}^{a_1, a_3}_c(g_{a_2, a_3}, d_{a_3})$. This premise implies that also $v_3$ is on $\tilde{\pi}^{a_1, a_3}_c(g_{a_2, a_3}, d_{a_3}) = \tilde{\pi}^{a_1, a_2}_c(g_{a_2, a_3}, d_{a_3})$. Following the same steps as in the previous case with $v_1$ replaced by $g_{a_2, a_3}$, we get that $v_2$ is on $\tilde{\pi}^{a_1, a_3}(g_{a_2, a_3}, v_3) = \tilde{\pi}^{a_2, a_3}(g_{a_2, a_3}, v_3)$. Moreover, $v_1$ is on $\tilde{\pi}^{a_1, a_2}_c(d_{a_1}, h_{a_1, a_2}) = \tilde{\pi}^{a_1, a_3}_c(d_{a_1}, h_{a_1, a_2})$ and hence, using that $\tilde{\pi}^{a_1, a_2}_c$ is canonical with respect to $\tilde{\pi}^{a_1, a_2}$, on $\tilde{\pi}^{a_1, a_2}(d_{a_1}, h_{a_1, a_2})$. This observation shows that $\tilde{\pi}^{a_1, a_3}(g_{a_2, a_3}, v_3)$ is a subpath of $\tilde{\pi}^{a_1, a_3}(v_1, v_3)$ and hence that $v_2$ is on $\tilde{\pi}^{a_1, a_3}(v_1, v_3)$.
\item Fifth, assume that $v_2$ is $h_{a_1, a_2}$ or $g_{a_2, a_3}$. This premise implies that $v_1$ is on the path $\tilde{\pi}^{a_1, a_2}_c(d_{a_1}, h_{a_1, a_2}) = \tilde{\pi}^{a_1, a_3}_c(d_{a_1}, h_{a_1, a_2})$ and hence, because $\tilde{\pi}^{a_1, a_2}_c$ is canonical with respect to $\tilde{\pi}^{a_1, a_2}$, on $\tilde{\pi}^{a_1, a_3}(d_{a_1}, h_{a_1, a_2}) = \tilde{\pi}^{a_1, a_2}(d_{a_1}, h_{a_1, a_2})$, where the latter equality follows by the first result at the beginning of this step. Moreover, the premise implies that $v_3$ is on $\tilde{\pi}^{a_2, a_3}_c(g_{a_1, a_2}, d_{a_2}) = \tilde{\pi}^{a_1, a_3}_c(g_{a_1, a_2}, d_{a_2})$ and hence, because $\tilde{\pi}^{a_2, a_3}_c$ is canonical with respect to $\tilde{\pi}^{a_2, a_3}$, on $\tilde{\pi}^{a_1, a_3}(g_{a_1, a_2}, d_{a_2}) = \tilde{\pi}^{a_2, a_3}(g_{a_1, a_2}, d_{a_2})$, where the latter equality follows by the third result at the beginning of this step. These considerations show that $\tilde{\pi}^{a_1, a_3}(v_1, v_3)$ decomposes as $\tilde{\pi}^{a_1, a_3}(v_1, h_{a_1, a_2}) \oplus \tilde{\pi}^{a_1, a_3}(h_{a_1, a_2}, g_{a_2, a_3}) \oplus \tilde{\pi}^{a_1, a_3}(g_{a_2, a_3}, v_3)$. Hence, $v_2$ is on $\tilde{\pi}^{a_1, a_3}(v_1, v_3)$ irrespective of whether $v_2 = h_{a_1, a_2}$ or $v_2 = g_{a_2, a_3}$.
\end{enumerate}

\underline{Step 13: $\tilde{\pi}^{a_1, a_3}_c$ fulfills point 4.~in Def.~\ref{defapp:canonical_paths}}\\
Let $v_1$ and $v_2$ be two distinct observable vertices on $\tilde{\pi}^{a_1, a_3}_c$ such that, without loss of generality, $v_1$ is closer to $d_{a_1}$ on $\tilde{\pi}^{a_1, a_3}_c$ than $v_2$ is to $d_{a_1}$ on $\tilde{\pi}^{a_1, a_3}_c$. We distinguish three collectively exhaustive cases:
\begin{enumerate}
\item First, assume $v_2$ is on $\tilde{\pi}^{a_1, a_3}_c(d_{a_1}, h_{a_1, a_2})$. Then, both $v_1$ and $v_2$ are on $\tilde{\pi}^{a_1, a_2}_c(d_{a_1}, h_{a_1, a_2}) = \tilde{\pi}^{a_1, a_3}_c(d_{a_1}, h_{a_1, a_2})$ and hence $\tilde{\pi}^{a_1, a_2}_c(v_1, v_2) = \tilde{\pi}^{a_1, a_3}_c(v_1, v_2)$. Since $\tilde{\pi}^{a_1, a_2}_c$ is canonical with respect to $\tilde{\pi}^{a_1, a_2}$, we thus find that $\tilde{\pi}^{a_1, a_2}(v_1, v_2)$ is an inducing path relative to $\mathbf{O}(\operatorname{max}(t(v_1), t(v_2)) - \taumax, t)[\tilde{\pi}^{a_1, a_3}_c(v_1, v_2)]$. The desired results follows since $\tilde{\pi}^{a_1, a_2}(v_1, v_2) = \tilde{\pi}^{a_1, a_3}(v_1, v_2)$ according to the second result at the beginning of step 12.

\item Second, assume $v_1$ is on $\tilde{\pi}^{a_2, a_3}_c(g_{a_2, a_3}, d_{a_3})$. Then, both $v_1$ and $v_2$ are on $\tilde{\pi}^{a_2, a_3}_c(g_{a_2, a_3}, d_{a_3})$ and hence $\tilde{\pi}^{a_1, a_3}_c(v_1, v_2) = \tilde{\pi}^{a_2, a_3}_c(g_{a_2, a_3}, d_{a_3})$. Since $\tilde{\pi}^{a_2, a_3}_c$ is canonical with respect to $\tilde{\pi}^{a_2, a_3}$, we find that $\tilde{\pi}^{a_2, a_3}(v_1, v_2)$ is an inducing path relative to $\mathbf{O}(\operatorname{max}(t(v_1), t(v_2)) - \taumax, t)[\tilde{\pi}^{a_1, a_3}_c(v_1, v_2)]$. The desired results follows since $\tilde{\pi}^{a_2, a_3}(v_1, v_2) = \tilde{\pi}^{a_1, a_3}(v_1, v_2)$ according to the fourth result at the beginning of step 12.

\item Third, assume that neither of the previous two cases applies. Then, using that every observable vertex on $\kappa^{a_2}$ is on $\tilde{\pi}^{a_1, a_2}_c(d_{a_1}, h_{a_1, a_2})$ or $\tilde{\pi}^{a_2, a_3}_c(g_{a_2, a_3}, d_{a_3})$, the vertex $v_1$ is on $\tilde{\pi}^{a_1, a_3}_c(d_{a_1}, h_{a_1, a_2})$ and $v_2$ is on $\tilde{\pi}^{a_1, a_3}_c(g_{a_2, a_3}, d_{a_3})$. Thus, both $h_{a_1, a_2}$ and $g_{a_2, a_3}$ are on $\tilde{\pi}^{a_1, a_3}_c(v_1, v_2)$. Since $\tilde{\pi}^{a_1, a_3}_c$ is collider-free, we thus find that both $h_{a_1, a_2}$ and $g_{a_2, a_3}$ are ancestors of $v_1$ or $v_2$ in $\Dc(\Mtaumax(\D))$ and, thus, ancestors in $\D$. Moreover, since $\tilde{\pi}^{a_1, a_3}_c$ fulfills point~2 in Def.~\ref{defapp:canonical_paths}, $v_1$ is on $\tilde{\pi}^{a_1, a_3}(d_{a_1}, h_{a_1, a_2})$ and $v_2$ is on $\tilde{\pi}^{a_1, a_3}(g_{a_2, a_3}, d_{a_3})$. Consequently, the path of interest $\tilde{\pi}^{a_1, a_3}(v_1, v_2)$ decomposes as $\tilde{\pi}^{a_1, a_3}(v_1, h_{a_1, a_2}) \oplus \tilde{\pi}^{a_1, a_3}(h_{a_1, a_2}, g_{a_1, a_2}) \oplus \tilde{\pi}^{a_1, a_3}(g_{a_1, a_2}, v_2)$. We now individually look at the three constituting subpaths:
\begin{enumerate}
\item By following the same steps as in the first case of this enumeration with $v_2$ replaced by $h_{a_1, a_2}$, we get that $\tilde{\pi}^{a_1, a_3}(v_1, h_{a_1, a_2})$ is an inducing path relative to $\mathbf{O}(t_{v_1h}-\taumax, t)[\tilde{\pi}_c^{a_1, a_3}(v_1, h_{a_1, a_2})]$, where $t_{v_1h} = \operatorname{max}(t(v_1), t(h_{a_1, a_2}))$. Moreover, note that $\tilde{\pi}_c^{a_1, a_3}(v_1, h_{a_1, a_2})$ is a subpath of $\tilde{\pi}_c^{a_1, a_3}(v_1, v_2)$. Since an inducing path relative to some set $\mathbf{O}_1$ of observed vertices is also an inducing path relative to another set $\mathbf{O}_2$ of observed with $\mathbf{O}_2 \subseteq \mathbf{O}_1$, we get that $\tilde{\pi}^{a_1, a_3}(v_1, h_{a_1, a_2})$ is an inducing path relative to $\mathbf{O}(t_{v_1h}-\taumax, t)[\tilde{\pi}_c^{a_1, a_3}(v_1, v_2)]$.

\item Because $h_{a_1, a_2}$ is on $\rho^{a_1}(d_{a_1}, c_{a_1}) \oplus \pi^{a_1, a_2}$ and $g_{a_2, a_3}$ is on $\pi^{a_2, a_3} \oplus \rho^{a_3}$, see steps 5 and 3, the path $\tilde{\pi}^{a_1, a_3}(h_{a_1, a_2}, g_{a_1, a_2})$ decomposes as $\tilde{\pi}^{a_1, a_2}(h_{a_1, a_2}, c_{a_2}) \oplus \tilde{\pi}^{a_2, a_3}(c_{a_2}, g_{a_1, a_2})$. This decomposition shows that $\tilde{\pi}^{a_1, a_3}(h_{a_1, a_2}, g_{a_1, a_2}) = \psi$, with $\psi$ as considered in step 6. Hence, $\tilde{\pi}^{a_1, a_3}(h_{a_1, a_2}, g_{a_1, a_2})$ is an inducing path relative to $\mathbf{O}(t_{hg} - \taumax, t)$ and, thus, an inducing path relative to $\mathbf{O}(t_{hg} - \taumax, t)[\tilde{\pi}_c^{a_1, a_3}(v_1, v_2)]$.

\item By following the same steps as in the second case of this enumeration with $v_1$ replaced by $g_{a_1, a_2}$, we get that $\tilde{\pi}^{a_2, a_3}(g_{a_2, a_3}, v_2)$ is an inducing path relative to $\mathbf{O}(t_{gv_2}-\taumax, t)[\tilde{\pi}_c^{a_1, a_3}(g_{a_2, a_3}, v_2)]$, where $t_{gv_2} = \operatorname{max}(t(g_{a_1, a_2}), t(v))$. Moreover, since $\tilde{\pi}_c^{a_1, a_3}(g_{a_2, a_3}, v_3)$ is a subpath of $\tilde{\pi}_c^{a_1, a_3}(v_1, v_2)$, we get that $\tilde{\pi}^{a_2, a_3}(g_{a_2, a_3}, v_2)$ is inducing path relative $\mathbf{O}(t_{gv_2}-\taumax, t)[\tilde{\pi}_c^{a_1, a_3}(v_1, v_2)]$.
\end{enumerate}
To proof the desired inducing path property of $\tilde{\pi}^{a_1, a_3}(v_1, v_2)$, we now separately consider its colliders and observable non end-point non-collider:
\begin{enumerate}
\item First, let $v$ be a collider on $\tilde{\pi}^{a_1, a_3}(v_1, v_2)$. Then, $v$ is a collider on one of the three constituting subpaths or is $h_{a_1, a_2}$ or $g_{a_2, a_3}$. In all cases, using the inducing path properties of the constituting subpaths, $v$ is in $\D$ an ancestor of $v_1$ or $h_{a_1, a_2}$ or $g_{a_2, a_3}$ or $v_2$. Since, as shown above in this step, both $h_{a_1, a_2}$ and $g_{a_2, a_3}$ are in $\D$ ancestors of $v_1$ or $v_2$, we get that $v$ is in $\D$ an ancestor of $v_1$ or $v_2$.
\item Second, let $v$ be an observable non end-point non-collider on $\tilde{\pi}^{a_1, a_3}(v_1, v_2)$. Since $h_{a_1, a_2}$ and $g_{a_2, a_3}$ are in $\D$ ancestors of $v_1$ or $v_2$, time order of $\D$ guarantees that $t(h_{a_1, a_2}) \leq t_{v_1v_2}$ and $t(g_{a_2, a_3}) \leq t_{v_1v_2}$, where $t_{v_1v_2} = \operatorname{max}(t(v_1), t(v_2))$. These inequalities imply $t_{v_1h} \leq t_{v_1v_2}$, $t_{hg} \leq t_{v_1v_2}$, and $t_{gv_2} \leq t_{v_1v_2}$. Thus, using the inducing path properties of the constituting subpaths in combination with $t(v) \leq t$, see step 1, we find that $t(v) < t_{v_1v_2} - \taumax$ if $v$ is not a non end-point vertex of $\tilde{\pi}_c^{a_1, a_3}(v_1, v_2)$
\end{enumerate}
Thus, $\tilde{\pi}^{a_1, a_3}(v_1, v_2)$ is an inducing path relative to $\mathbf{O}(t_{v_1v_2}-\taumax, t)[\tilde{\pi}_c^{a_1, a_3}(v_1, v_2)]$.
\end{enumerate}
The combination of the four steps 11, 12, 10 and 13 shows that $\tilde{\pi}_c^{a_1, a_3}$ is canonical with respect to $\tilde{\pi}^{a_1, a_3}$.

\underline{Step 14: $\tilde{\pi}_c^{a_1, a_3}$ is active given $\mathbf{S} \setminus \{d_{a_1}, d_{a_3}\}$}\\
Recall that $\tilde{\pi}_c^{a_1, a_3} = \tilde{\pi}^{a_1, a_2}_c(d_{a_1}, h_{a_1, a_2}) \oplus \kappa^{a_2} \oplus \tilde{\pi}^{a_2, a_3}_c(g_{a_2, a_3}, d_{a_3})$ and assume $\tilde{\pi}_c^{a_1, a_3}$ is blocked given $\mathbf{S} \setminus \{d_{a_1}, d_{a_3}\}$. Since $\tilde{\pi}_c^{a_1, a_3}$ is collider-free, see step 10, this premise means there is a non end-point vertex $v$ on $\tilde{\pi}_c^{a_1, a_3}$ with $v \in \mathbf{S} \setminus \{d_{a_1}, d_{a_3}\}$. There are three collectively exhaustive cases:
\begin{enumerate}
\item First, assume $v$ is on $\tilde{\pi}^{a_1, a_2}_c(d_{a_1}, h_{a_1, a_2})$. Then, because $\tilde{\pi}^{a_1, a_2}_c$ is canonical with respect to $\tilde{\pi}^{a_1, a_2}$, the vertex $v$ is on $\tilde{\pi}^{a_1, a_2}(d_{a_1}, h_{a_1, a_2})$. Since $\tilde{\pi}^{a_1, a_2}(d_{a_1}, h_{a_1, a_2})$ is a subpath of $\rho^{a_1}(d_{a_1}, c_{a_1}) \oplus \pi^{a_1, a_2}$ excluding $c_{a_2}$, see step 5, $v$ is not on $\rho^{a_2}$. In particular, $v \neq d_{a_2}$ and thus $v \in \mathbf{S} \setminus \{d_{a_1}, d_{a_2}\}$. Moreover, since $v \neq d_{a_1}$ and $h_{a_1, a_2} \neq d_{a_2}$ by the respective definitions of $v$ and $h_{a_1, a_2}$, the vertex $v$ is a non end-point vertex on $\tilde{\pi}^{a_1, a_2}_c$. Since $\tilde{\pi}^{a_1, a_2}_c$ is collider-free by means of being canonical with respect to $\tilde{\pi}^{a_1, a_2}$ and since $v \in \mathbf{S} \setminus \{d_{a_1}, d_{a_2}\}$, we arrive at a contradiction to the assumption that $\tilde{\pi}^{a_1, a_2}_c$ is active given $\mathbf{S} \setminus \{d_{a_1}, d_{a_2}\}$.

\item Second, assume $v$ is $\tilde{\pi}^{a_2, a_3}_c(g_{a_2, a_3}, d_{a_3})$. Then, because $\tilde{\pi}^{a_2, a_3}_c$ is canonical with respect to $\tilde{\pi}^{a_2, a_3}$, the vertex $v$ is on $\tilde{\pi}^{a_2, a_3}(g_{a_2, a_3}, d_{a_3})$. Since $\tilde{\pi}^{a_2, a_3}(g_{a_2, a_3}, d_{a_3})$ is a subpath of $\pi^{a_2, a_3} \oplus \rho^{a_2}$ excluding $c_{a_2}$, see step 3, $v$ is not on $\rho^{a_2}$. In particular, $v \neq d_{a_2}$ and thus $v \in \mathbf{S} \setminus \{d_{a_2}, d_{a_3}\}$. Moreover, since $v \neq d_{a_3}$ and $g_{a_1, a_3} \neq d_{a_2}$ by the respective definitions of $v$ and $g_{a_2, a_3}$, the vertex $v$ is a non end-point vertex on $\tilde{\pi}^{a_2, a_3}_c$. Since $\tilde{\pi}^{a_2, a_3}_c$ is collider-free by means of being canonical with respect to $\tilde{\pi}^{a_2, a_3}$ and since $v \in \mathbf{S} \setminus \{d_{a_2}, d_{a_3}\}$, we arrive at a contradiction to the assumption that $\tilde{\pi}^{a_2, a_3}_c$ is active given $\mathbf{S} \setminus \{d_{a_2}, d_{a_3}\}$.

\item Third, assume $v$ is on $\kappa^{a_2}$. Then, since every element of $\mathbf{S}$ is observable, $v$ is $h_{a_1, a_2}$ or $g_{a_2, a_3}$ and thus on $\tilde{\pi}^{a_1, a_2}_c(d_{a_1}, h_{a_1, a_2})$ or $\tilde{\pi}^{a_2, a_3}_c(g_{a_2, a_3}, d_{a_3})$. These cases have already been covered by the previous two points.
\end{enumerate}
We have thus shown that $\tilde{\pi}_c^{a_1, a_3}$ is active given $\mathbf{S} \setminus \{d_{a_1}, d_{a_3}\}$, which completes the proof.
\end{proof}

In case at least one of $\tilde{\pi}^{a_1,a_2}_{c} = \tilde{\pi}^{a-1,a}_{ci}$ and $\tilde{\pi}^{a_2,a_3}_{c} = \tilde{\pi}^{a,a+1}_{ci}$ is out of $d_{a} = d_{a_2}$, we can thus collectively replace them by a path $\tilde{\pi}^{a_1,a_3}_{c}$ between $d_{a_1} = d_{a-1}$ and $d_{a_3} = d_{a+1}$ in $\Dc(\Mtaumax(\D))$ that is active given $\mathbf{S} \setminus \{d_{a_1}, d_{a_3}\}$. Moreover, since $\tilde{\pi}^{a_1,a_3}_{c}$ is canonical with respect to $\tilde{\pi}^{a_1,a_3}$ and since Lemma~\ref{lemmaapp:glueing_induced_paths} only used that $\tilde{\pi}^{a_1,a_2}_{c}$ and $\tilde{\pi}^{a_2,a_3}_{c}$ are, respectively, canonical with respect to $\tilde{\pi}^{a_1,a_2}$ and $\tilde{\pi}^{a_2,a_3}$ as well as , respectively, active given $\mathbf{S} \setminus \{d_{a_1}, d_{a_2}\}$ and $\mathbf{S} \setminus \{d_{a_2}, d_{a_3}\}$, this procedure can be repeated in case, for example, $\tilde{\pi}^{a_1,a_3}_{c}$ or $\tilde{\pi}^{a_3,a_4}_{ci}$ is out of $d_{a_3}$, and so on.

\begin{mylemma}\label{lemmaapp:MG_subgraph_of_MGc_reason_part_2}
Let $(i, t_i)$ and $(j, t_j)$ with $t - \taumax \leq t_i, t_j \leq t$ be distinct observable vertices in $\D$ and let $\mathbf{S} \subseteq \mathbf{O}(t-\taumax, t) \setminus\{(i, t_i), (j, t_j)\}$. Then: If $(i, t_i)$ and $(j, t_j)$ are $d$-connected given $\mathbf{S}$ in $\D$, then $(i, t_i)$ and $(j, t_j)$ are $d$-connected given $\mathbf{S}$ in $\Dc(\Mtaumax(\D))$.
\end{mylemma}

\begin{proof}[Proof of Lemma~\ref{lemmaapp:MG_subgraph_of_MGc_reason_part_2}]
Let $(i, t_i)$ and $(j, t_j)$ be $d$-connected given $\mathbf{S} \subseteq \mathbf{O}(t-\taumax, t) \setminus\{(i, t_i), (j, t_j)\}$. Then, in $\D$ there is path $\pi$ between $(i, t_i)$ and $(j, t_j)$ that is active given $\mathbf{S}$. The proof is by induction over $n_c$, the number of colliders on $\pi$.

\textit{Induction base case: $n_c = 0$}\\
In this case, $(i, t_i)$ and $(j, t_j)$ are $d$-connected given $\mathbf{S}$ in $\Dc(\Mtaumax(\D))$ according to Lemma~\ref{lemmaapp:MG_subgraph_of_MGc_reason_part_1}.

\textit{Induction step: $n_c \to n_c+1$}\\
In this case, $\pi$ has $n_c + 1 \geq 1$ colliders and according to the assumption of induction we have already proven the statement for paths that have at most $n_c$ colliders.

We may without loss of generality assume that $\pi$ has a collider extension structure, because if not then according to Lemma~\ref{lemmaapp:extension_by_rho^a_paths} there is path $\pi^\prime$ between $(i, t_i)$ and $(j, t_j)$ in $\D$ with at most $n_c$ colliders that is active given $\mathbf{S}$ and hence, by assumption of induction, $(i, t_i)$ and $(j, t_j)$ are $d$-connected given $\mathbf{S}$ in $\Dc(\Mtaumax(\D))$. Therefore, the assumptions and notation of Def.~\ref{defapp:collider_extended_paths} apply.

Consider the following algorithmic procedure:
\begin{enumerate}
  \item For all $0 \leq a \leq n_c$ let $\tilde{\pi}^{a, a+1}_c$ be $\tilde{\pi}^{a, a+1}_{ci}$, that is, let $\tilde{\pi}^{a, a+1}_c$ be the canonically induced path of $\tilde{\pi}^{a, a+1}$.
  \item Set $m$ to $n_c$.
  \item Let $\sigma: \{0, 1, \dots , m+1\} \mapsto \{0, 1, \dots , n_c+1\}$ be the identity map.
  \item While there is an integer $a$ with $1 \leq a \leq m$ such that at least one of $\tilde{\pi}_c^{\sigma(a-1), \sigma(a)}$ and $\tilde{\pi}_c^{\sigma(a), \sigma(a+1)}$ is out of $d_{\sigma(a)}$:
  \begin{enumerate}
    \item Let $b$ be the smallest such $a$.
    \item Let $\tilde{\pi}_c^{\sigma(b-1), \sigma(b+1)} = \tilde{\pi}_c^{a_1, a_3}$ be a path as in Lemma~\ref{lemmaapp:glueing_induced_paths} applied to $\tilde{\pi}_c^{\sigma(b-1), \sigma(b)} = \tilde{\pi}_c^{a_1, a_2}$ and $\tilde{\pi}_c^{\sigma(b), \sigma(b+1)} = \tilde{\pi}_c^{a_2, a_3}$.
    \item If $\sigma(b - 1) = 0$ and $\sigma(b + 1) = n_c+1$, then return $\tilde{\pi}_c^{\sigma(b-1), \sigma(b+1)} = \tilde{\pi}_c^{0, n_c+1}$ and terminate.
    \item Decrease $m$ by one.
    \item Let $\sigma^\prime: \{0, 1, \dots , m+1\} \mapsto \{0, 1, \dots , n_c+1\}$ be such that $\sigma^\prime(a) = \sigma(a)$ for all $1 \leq a < b$ and $\sigma^\prime(a) = \sigma(a + 1)$ for all $b \leq a \leq m+1$.
    \item Replace $\sigma$ by $\sigma^\prime$.
  \end{enumerate}
  \item Return the ordered sequence of paths $\tilde{\pi}_c^{\sigma(0), \sigma(1)}, \dots , \tilde{\pi}_c^{\sigma(m), \sigma(m+1)}$ and terminate.
\end{enumerate}
According to Lemma~\ref{lemmaapp:collider_extended_paths_active} and Lemma~\ref{lemmaapp:canonical_path_properties}, the canonically induced paths $\tilde{\pi}^{a, a+1}_{ci}$ are for all $0 \leq a \leq n_c$ active given $\mathbf{S} \setminus \{d_{a}, d_{a+1}\}$ and canonical with respect to $\tilde{\pi}^{a, a+1}$. Lemma~\ref{lemmaapp:glueing_induced_paths} thus guarantees that all paths $\tilde{\pi}_c^{a_1, a_3}$ with $0 \leq a_1 < a_3 \leq n_c+1$ constructed in step 4(b) of the above procedure are active given $\mathbf{S} \setminus \{d_{a_1}, d_{a_3}\}$ and canonical with respect to $\tilde{\pi}^{a_1, a_3}$. Thus, if the algorithm terminates in step 4(c), it returns a path $\tilde{\pi}_c^{0, n_c+1}$ between $d_0 = (i, t_i)$ and $d_{n_c+1} = (j, t_j)$ in $\Dc(\Mtaumax(\D))$ that is active given $\mathbf{S} = \mathbf{S} \setminus \{d_{0}, d_{n_c+1}\}$. If the algorithm terminates in step 5, then the following is true:
\begin{enumerate}
\item $\tilde{\pi}_c^{\sigma(0), \sigma(1)}$ is a path in $\Dc(\Mtaumax(\D))$ between $d_{\sigma(0)} = d_0 = (i, t_i)$ and $d_{\sigma(1)}$ that is into $d_{\sigma(1)}$ and active given $\mathbf{S} \setminus \{d_{\sigma(0)}, d_{\sigma(1)}\}$.
\item For all $1 \leq a \leq m-1$ the path $\tilde{\pi}_c^{\sigma(a), \sigma(a+1)}$ is a path in $\Dc(\Mtaumax(\D))$ between $d_{\sigma(a)}$ and $d_{\sigma(a+1)}$ that is into both $d_{\sigma(a)}$ and $d_{\sigma(a+1)}$ and active given $\mathbf{S} \setminus \{d_{\sigma(a)}, d_{\sigma(a+1)}\}$.
\item $\tilde{\pi}_c^{\sigma(m), \sigma(m+1)}$ is a path in $\Dc(\Mtaumax(\D))$ between $d_{\sigma(m)}$ and $d_{\sigma(m+1)} = d_{n_c+1} = (j,t_j)$ that is into $d_{\sigma(m)}$ and active given $\mathbf{S} \setminus \{d_{\sigma(m)}, d_{\sigma(m+1)}\}$.
\end{enumerate}
Further, by definition of collider extension structures, $d_a$ is an ancestor of $\mathbf{S}$ for all $1 \leq a \leq n_c$. Lemma 3.3.1 in \citet{Spirtes2000} thus applies to the ordered sequence of paths $\tilde{\pi}_c^{\sigma(0), \sigma(1)}, \dots , \tilde{\pi}_c^{\sigma(m), \sigma(m+1)}$ and guarantees the existence of a path between $(i, t_i)$ and $(j, t_j)$ in $\Dc(\Mtaumax(\D))$ that is active given $\mathbf{S}$.

Hence, $(i, t_i)$ and $(j, t_j)$ are $d$-connected given $\mathbf{S}$ in $\Dc(\Mtaumax(\D))$.
\end{proof}

\begin{mylemma}\label{lemmaapp:MtaumaxD_subgraph_MtaumaxDc}
$\Mtaumax(\D)$ is a subgraph of $\Mtaumax(\Dc(\Mtaumax(\D)))$.
\end{mylemma}

\begin{proof}[Proof of Lemma~\ref{lemmaapp:MtaumaxD_subgraph_MtaumaxDc}]
As an immediate consequence of Lemma~\ref{lemmaapp:MG_subgraph_of_MGc_reason_part_2}, every adjacency in $\Mtaumax(\D)$ is also in $\Mtaumax(\Dc(\Mtaumax(\D)))$. The statement follows with Lemma~\ref{lemmaapp:same_ancestral_main_proposition} because the orientation of edges are uniquely determined by the ancestral relationships.
\end{proof}

\begin{proof}[Proof of Lemma~4.14]
First, $\Mtaumax(\Dc(\Mtaumax(\D)))$ is a subgraph of $\Mtaumax(\D)$ according to Lemma~\ref{lemmaapp:MtaumaxDc_subgraph_MtaumaxD}. Second, $\Mtaumax(\D))$ is a subgraph of $\Mtaumax(\Dc(\Mtaumax(\D)))$ according to Lemma~\ref{lemmaapp:MtaumaxD_subgraph_MtaumaxDc}. Hence, $\Mtaumax(\Dc(\Mtaumax(\D)))= \Mtaumax(\D)$.
\end{proof}

\subsection{Proofs for Sec.~4.6}

\begin{mylemma}\label{lemmaapp:canonical_ts-DAG}
The canonical ts-DAG $\Dc(\G)$ of an acyclic directed mixed graph $\G$ with time series structure is a ts-DAG.
\end{mylemma}

\begin{proof}[Proof of Lemma~\ref{lemmaapp:canonical_ts-DAG}]
The time series structure of $\Dc(\G)$ with $\Tindex = \mathbb{Z}$ is apparent from the first point of Def.~4.13, the repeating edges property is enforced explicitly in the second point of Def.~4.13, and time order is enforced explicitly in the second point of Def.~4.13 by only considering edges in $\E^{\stat}_{\tailhead}$ of the form $((i, t-\tau), (j, t))$ in the first and second point of Def.~4.13. Assume there is a directed cycle in $\Dc(\G)$. Because as apparent from the second point of Def.~4.13 there are no edges into unobservable vertices, all vertices on the directed cycle are observable. Moreover, due to time order all vertices on the directed cycle must be at a single time step. Due to repeating edges there thus is a directed cycle at time $t$ in $\Dc(\G)$. Since the second point of Def.~4.13 further shows that all edges between observable vertices at time $t$ in $\Dc(\G)$ are also in $\stat(\G)$, which is a subgraph of $\G$, we get a contradiction to the acyclicity of $\G$.
\end{proof}

\begin{proof}[Proof of Theorem~1]
\textbf{If.} The premise is $\M = \Mtaumax(\Dc(\M))$. Since according to Lemma~\ref{lemmaapp:canonical_ts-DAG} the canonical ts-DAG $\Dc(\M)$ is a ts-DAG, we can choose $\D = \Dc(\M)$ and get $\M = \Mtaumax(\D)$.

\textbf{Only if.} The premise is $\M = \Mtaumax(\D)$. Together with Lemma~4.14, which says $\Mtaumax(\D) = \Mtaumax(\Dc(\Mtaumax(\D)))$, then $\M  = \Mtaumax(\D) = \Mtaumax(\Dc(\Mtaumax(\D))) = \Mtaumax(\Dc(\M))$.
\end{proof}

\begin{proof}[Proof of Theorem~2]
\textbf{If.} The premise is $\G = \Mtaumax(\Dc(\G))$ with $\G$ acyclic. Since according to Lemma~\ref{lemmaapp:canonical_ts-DAG} the canonical ts-DAG $\Dc(\G)$ is a ts-DAG, we can choose $\D = \Dc(\G)$ and get $\G = \Mtaumax(\D)$.

\textbf{Only if.} The premise is $\G = \Mtaumax(\D)$. Together with Lemma~4.14, which says $\Mtaumax(\D) = \Mtaumax(\Dc(\Mtaumax(\D)))$, then $\G =  \Mtaumax(\D) = \Mtaumax(\Dc(\Mtaumax(\D))) = \Mtaumax(\Dc(\G))$. Moreover, $\G$ is acyclic because it equals the DMAG $\Mtaumax(\D)$ and DMAGs are acyclic.
\end{proof}

\subsection{Proofs for Sec.~4.7 and of Lemma~\ref{lemmaapp:different_DMAGs_not_both_implied}}\label{secapp:contains_proof_not_same_stationarification}

\begin{proof}[Proof of Lemma~4.20]
Combine Lemma~\ref{lemmaapp:canonical_ts-DAG_of_stat}, according to which $\Dc(\Mtaumaxstat(\D)) = \Dc(\Mtaumax(\D))$, with Lemma~4.14, according to which $\Mtaumax(\D) = \Mtaumax(\Dc(\Mtaumax(\D)))$.
\end{proof} 

\begin{proof}[Proof of Lemma~\ref{lemmaapp:different_DMAGs_not_both_implied}]
Assume that both $\M_1$ and $\M_2$ are ts-DMAGs, i.e., $\M_1 = \Mtaumax(\D_1)$ and $\M_2 = \Mtaumax(\D_2)$ for some ts-DAGs $\D_1$ and $\D_2$. Then, combining the premise $\stat(\M_1) = \stat(\M_2)$ with Lemma~4.20 leads to the contradiction $\M_1 = \M_2$.
\end{proof}

\begin{proof}[Proof of Lemma~4.21]
\textbf{If.} The premise is $\M = \Mtaumaxstat(\Dc(\M))$. Since according to Lemma~\ref{lemmaapp:canonical_ts-DAG} the canonical ts-DAG $\Dc(\M)$ is a ts-DAG, we can choose $\D = \Dc(\M)$ and get $\M = \Mtaumaxstat(\D)$.

\textbf{Only if.} The premise is $\M = \Mtaumaxstat(\D)$. We then get $\Mtaumax(\D) = \Mtaumax(\Dc(\M))$ according to Lemma~4.20, which by applying the operation of stationarification to both sides of this equality gives that $\M = \Mtaumaxstat(\Dc(\M))$.
\end{proof}

\begin{proof}[Proof of Lemma~4.22]
\textbf{If.} The premise is $\G = \Mtaumaxstat(\Dc(\G))$ with $\G$ acyclic. Since according to Lemma~\ref{lemmaapp:canonical_ts-DAG} the canonical ts-DAG $\Dc(\G)$ is a ts-DAG, we can choose $\D = \Dc(\G)$ and get $\G = \Mtaumaxstat(\G)$.

\textbf{Only if.} The premise is $\G = \Mtaumaxstat(\D)$. We then get $\Mtaumax(\D) = \Mtaumax(\Dc(\G))$ according to Lemma~4.20, which by applying the operation of stationarification to both sides of this equality gives that $\G = \Mtaumaxstat(\Dc(\G))$. Moreover, $\G$ is acyclic because it equals the DMAG $\Mtaumaxstat(\D)$ and DMAGs are acyclic.
\end{proof}

\section{Proofs for Sec.~5}

\subsection{Proofs for Sec.~5.3}

\begin{mylemma}\label{lemmaapp:ro_of_DPAG}
Let $\D$ be a ts-DAG and $\BR \in \{\BRtoro, \,\BRtora, \,\BRtsDAG\}$. Then, $\PAG(\Mtaumax(\D), \BR)$ has repeating orientations.
\end{mylemma}

\begin{proof}[Proof of Lemma~\ref{lemmaapp:ro_of_DPAG}]
Note that $\BRtsDAG$ is stronger than $\BRtora$ and that $\BRtora$ is stronger than $\BRtoro$. Thus, every graph consistent with $\BR$ has repeating orientations. The statement then follows from the definition of m.i.~DPAGs because every element in $[\M]_{\BR}$ has repeating orientations.
\end{proof}

\begin{proof}[Proof of Lemma~5.5]
From Lemma~\ref{lemmaapp:ro_of_DPAG} we know that $\PAG(\Mtaumax(\D), \BR)$ has repeating orientations. Moreover, $\PAG(\Mtaumax(\D), \BR)$ has past-repeating adjacencies because its by definition equals the skeleton of $\Mtaumax(\D)$.

Let $(i, t_i)$ and $(j, t_j)$ with $t-\taumax \leq t_i \leq t_j \leq t$ be distinct observable vertices. According to Lemma~\ref{lemmaapp:stat_of_implied_DMAG}, these vertices are adjacent in $\stat(\PAG(\Mtaumax(\D), \BR))$ if and only if $(i, t - (t_j - t_i))$ and $(j, t)$ are adjacent in $\PAG(\Mtaumax(\D), \BR)$. Because the skeletons of $\PAG(\Mtaumax(\D), \BR)$ and $\Mtaumax(\D)$ are the same, Lemma~4.7 gives that $(i, t - (t_j - t_i))$ and $(j, t)$ are adjacent in $\PAG(\Mtaumax(\D), \BR)$ if and only if $(i, t_i)$ and $(j, t_j)$ are adjacent in $\Mtaumaxstat(\D)$. Consequently, $\stat(\PAG(\Mtaumax(\D), \BR))$ and $\Mtaumaxstat(\D)$ have the same skeleton.

Moreover, consider an unambiguous edge mark in $\stat(\PAG(\Mtaumax(\D), \BR))$. This edge mark is also in $\PAG(\Mtaumax(\D), \BR)$ and therefore corresponds to an ancestral relationship in $\Mtaumax(\D)$. Because according to Lemma~4.10 the graphs $\Mtaumax(\D)$ and $\Mtaumaxstat(\D)$ have the same ancestral relationships, the same unambiguous edge mark is then also in $\Mtaumaxstat(\D)$.
\end{proof}

\begin{mylemma}\label{lemmaapp:stat_of_markov_equivalent_are_markov_equivalent_reason}
Let $\M$ be a DMAG with time series structure that has repeating orientations and past-repeating adjacencies and for part 2 in addition is time ordered. Then:
\begin{enumerate}
\item Let $(i, t_i) \astast (j, t_j) \astast (k, t_k)$ with $t_j \leq \max(t_i, t_k)$ be an unshielded triple in $\stat(\M)$ and let $\Delta t = t_j - \max(t_i, t_k)$. Then:
\begin{enumerate}
\item $(i, t_i + \Delta t) \astast (j, t_j + \Delta t) \astast (k, t_k + \Delta t)$ is an unshielded triple in $\M$.
\item $(i, t_i + \Delta t) \astast (j, t_j + \Delta t) \astast (k, t_k + \Delta t)$ is oriented as a collider in $\M$ if and only if $(i, t_i) \astast (j, t_j) \astast (k, t_k)$ is oriented as a collider in $\stat(\M)$.
\end{enumerate}
\item Let $\pi = (l, t_l) \dots  \asthead (i, t_i) \headast (j, t_j) \asthead (k, t_k)$ with $t_j \leq \max(t_l, t_k)$ be a discriminating path for $(j, t_j)$ in $\stat(\M)$ and let $\Delta t = t_j - \max(t_l, t_k)$. Then:
\begin{enumerate}
\item $\pi_{\Delta t}$, the copy of $\pi$ shifted forward in time by $\Delta t$ time steps, is a discriminating path for $(j, t_j + \Delta t)$ in $\M$.
\item $(i, t_i + \Delta t) \astast (j, t_j + \Delta t) \astast (k, t_k + \Delta t)$ is oriented as a collider in $\M$ if and only if $(i, t_i) \astast (j, t_j) \astast (k, t_k)$ is oriented as a collider in $\stat(\M)$.
\end{enumerate}
\end{enumerate}
\end{mylemma}

\begin{proof}[Proof of Lemma~\ref{lemmaapp:stat_of_markov_equivalent_are_markov_equivalent_reason}]
\textbf{1(a)}
The repeating edges property of $\stat(\M)$ together with $t_j \leq \max(t_i, t_k)$ implies that $(i, t_i + \Delta t) \astast (j, t_j + \Delta t) \astast (k, t_k + \Delta t)$ is an unshielded triple in $\stat(\M)$. Lemma~\ref{lemmaapp:stat_of_implied_DMAG} then guarantees that $(i,t_i+\Delta t) = (i, t-(\max(t_i, t_k)-t_i))$ and $(k,t_k + \Delta t)= (k, t-(\max(t_i, t_k)-t_i))$ are non-adjacent in $\M$, because else they would be adjacent in $\stat(\M)$ too and hence $(i, t_i + \Delta t) \astast (j, t_j + \Delta t) \astast (k, t_k + \Delta t)$ would not be unshielded in $\stat(\M)$. Since $\stat(\M)$ is a subgraph of $\M$, we thus get that $(i, t_i + \Delta t) \astast (j, t_j + \Delta t) \astast (k, t_k + \Delta t)$ is an unshielded triple in $\M$.

\textbf{2(a)}
By the definition of discriminating paths, all vertices on $\pi$ other than, perhaps, $(j, t_j)$ and/or $(l, t_l)$ are ancestors of $(k, t_k)$. Time order of $\M$ together with $t_j \leq \max(t_l, t_k)$ thus guarantees that all vertices on $\pi$ are within $[t-\taumax, \max(t_l, t_k)]$. In combination with the repeating edges property of $\stat(\M)$ we thus see that $\pi_{\Delta t}$ is a discriminating path for $(j, t_j + \Delta t)$ in $\stat(\M)$. Lemma~\ref{lemmaapp:stat_of_implied_DMAG} then guarantees that $(l,t_l+\Delta t) = (l, t - (\max(t_l, t_k)-t_l))$ and $(k,t_k+\Delta t) = (k, t- (\max(t_l, t_k)-t_k))$ are non-adjacent in $\M$ because else they would be adjacent in $\stat(\M)$ too and hence $\pi_{\Delta t}$ would not be a discriminating path in $\stat(\M)$. Consequently, $\pi_{\Delta t}$ is a discriminating path in $\M$ because $\stat(\M)$ is a subgraph of $\M$.

\textbf{1(b) and 2(b)}
Because $\stat(\M)$ has repeating edges, the triple $(i, t_i) \astast (j, t_j) \astast (k, t_k)$ is oriented as a collider in $\stat(\M)$ if and only if $(i, t_i + \Delta t) \astast (j, t_j + \Delta t) \astast (k, t_k + \Delta t)$ is oriented as a collider in $\stat(\M)$. Moreover, since $\stat(\M)$ is a subgraph of $\M$ and $(i, t_i + \Delta t) \astast (j, t_j + \Delta j) \astast (k, t_k+\Delta t)$ is in $\M$, the triple $(i, t_i + \Delta t) \astast (j, t_j + \Delta j) \astast (k, t_k+\Delta t)$ is oriented as collider in $\stat(\M)$ if and only if $(i, t_i + \Delta t) \astast (j, t_j + \Delta j) \astast (k, t_k+\Delta t)$ is oriented as a collider in $\M$.
\end{proof}

\begin{mylemma}\label{lemmaapp:stat_of_markov_equivalent_are_markov_equivalent}
Let $\M_1$ and $\M_2$ be Markov equivalent DMAGs with time series structure that are time ordered and have repeating orientations and past-repeating adjacencies. Then, $\stat(\M_1)$ and $\stat(\M_2)$ are Markov equivalent DMAGs.
\end{mylemma}

\begin{proof}[Proof of Lemma~\ref{lemmaapp:stat_of_markov_equivalent_are_markov_equivalent}]
Both $\stat(\M_1)$ and $\stat(\M_2)$ are DMAGs according to Lemma~\ref{lemmaapp:stat_of_implied_DMAG_is_DMAG}. Next, we show that $\stat(\M_1)$ and $\stat(\M_2)$ are Markov equivalent. For this purpose, assume the opposite. Then, according to the characterizing of Markov equivalence of MAGs in \citet{MAGs_equivalence}, at least one of the following statements is true:
\begin{enumerate}
  \item The skeletons of $\stat(\M_1)$ and $\stat(\M_2)$ differ.
  \item There is an unshielded triple $(i, t_i) \astast (j, t_j) \astast (k, t_k)$ in both $\stat(\M_1)$ and $\stat(\M_1)$ that is oriented as a collider in $\stat(\M_a)$ with $a \in \{1, 2\}$ and oriented as a non-collider in $\stat(\M_{\bar{a}})$ with $\bar{a} = 3-a$.
  \item There is a path $\pi$ that is in both $\stat(\M_1)$ and $\stat(\M_1)$ a discriminating path for $(j, t_j)$ such that $(j, t_j)$ is a collider on $\pi$ in $\stat(\M_a)$ with $a \in \{1, 2\}$ and a non-collider in $\stat(\M_{\bar{a}})$ with $\bar{a} = 3-a$.
\end{enumerate}
We now show that $\stat(\M_1)$ and $\stat(\M_2)$ do have the same skeleton and that both the second and third statement contradict Markov equivalence of $\M_1$ and $\M_2$.

\textit{Case 1: Skeleton.} 
According to Lemma~\ref{lemmaapp:stat_of_implied_DMAG}, the skeletons of $\stat(\M_1)$ and $\stat(\M_2)$ are determined uniquely by, respectively, the skeletons of $\M_1$ and $\M_2$. Thus, since $\M_1$ and $\M_2$ have the same skeleton due to being Markov equivalent, also $\stat(\M_1)$ and $\stat(\M_2)$ have the same skeleton.

\textit{Case 2: Unshielded colliders}. Since $(i, t_i) \astast (j, t_j) \astast (k, t_k)$ is oriented as a non-collider in $\stat(\M_{\bar{a}})$, the vertex $(j, t_j)$ is in $\stat(\M_{\bar{a}})$ an ancestor (parent, in fact) of $(i, t_i)$ or $(k, t_k)$. Time order of $\stat(\M_{\bar{a}})$ thus implies $t_j \leq \operatorname{max}(t_i, t_k)$. Hence, we can apply part 1 of Lemma~\ref{lemmaapp:stat_of_markov_equivalent_are_markov_equivalent_reason} to both $\stat(\M_{\bar{a}})$ and $\stat(\M_{a})$, which gives that $(i, t_i + \Delta t) \astast (j, t_j + \Delta t) \astast (k, t_k + \Delta t)$ with $\Delta t_j = t - \operatorname{max}(t_i, t_k)$ is an unshielded collider in $\M_a$ and an unshielded non-collider in $\M_{\bar{a}}$. This observation contradicts the assumptiont that $\M_1$ and $\M_2$ are Markov equivalent.

\textit{Case 3: Discriminating paths}. By definition of discriminating paths, $\pi$ takes the form $\dots  \asthead (i, t_i) \headast (j, t_j) \asthead (k, t_k)$. Moreover, as follows from the definition of discriminating paths together with the absence of almost directed cycles, $(i, t_i) \headhead (j, t_j)$ if $(j, t_j) \headhead (k, t_k)$. In combination with the fact that $(j, t_j)$ is in $\stat(\M_{\bar{a}})$ an ancestor (parent, in fact) of $(i, t_i)$ or $(k, t_k)$ by means of $(j, t_j)$ being a non-collider on $\pi$ in $\stat(\M_{\bar{a}})$, we thus find that $(j, t_j)$ is in $\stat(\M_{\bar{a}})$ an ancestor (parent, in fact) of $(k, t_k)$. Time order of $\stat(\M_{\bar{a}})$ thus implies $t_j \leq t_k \leq \operatorname{max}(t_l, t_k)$. Hence, we can apply part 2 of Lemma~\ref{lemmaapp:stat_of_markov_equivalent_are_markov_equivalent_reason} to both $\stat(\M_{\bar{a}})$ and $\stat(\M_{a})$, which gives that $\pi_{\Delta t}$, the copy of $\pi$ that is shifted forward in time by $\Delta t = t - \operatorname{max}(t_l, t_k)$ time steps, is a discriminating path for $(j, t_j + \Delta t)$ in both $\M_a$ and $\M_{\bar{a}}$ and that $(j, t_j + \Delta t)$ is a collider on $\pi_{\Delta t}$ in $\M_a$ whereas $(j, t_j + \Delta t)$ is a non-collider on $\pi_{\Delta t}$ in $\M_{\bar{a}}$. This observation contradicts Markov equivalence of $\M_1$ and $\M_2$.
\end{proof}

\begin{proof}[Proof of Theorem~3]
\textbf{1.}
Note that $\stat(\PAG(\Mtaumax(\D), \BR))$ and $\PAG(\Mtaumaxstat(\D), \BRstat)$ have the same skeleton because both of them are DPAGs for $\Mtaumaxstat(\D)$, as follows from Lemma~5.5. We prove the statement by showing that $(i, t_i) \oast (j, t_j)$ in $\stat(\PAG(\Mtaumax(\D), \BR))$ implies $(i, t_i) \oast (j, t_j)$ in $\PAG(\Mtaumaxstat(\D), \BRstat)$.

Let the edge $(i, t_i) \oast (j, t_j)$ be in $\stat(\PAG(\Mtaumax(\D), \BR))$. Then, $(i, t_i) \oast (j, t_j)$ is also in $\PAG(\Mtaumax(\D), \BR)$ because $\PAG(\Mtaumax(\D), \BR)$ is a supergraph of $\stat(\PAG(\Mtaumax(\D), \BR))$. By definition of m.i.~DPAGs, there thus are DMAGs $\M_1$ and $\M_2$ in $[\Mtaumax(\D)]_{\BR}$ such that $(i, t_i) \tailhead (j, t_j)$ in $\M_1$ and $(i, t_i) \headast (j, t_j)$ in $\M_2$. Without loss of generality we may assume that $\M_1$ or $\M_2$ is $\Mtaumax(\D)$.

Since i) according to Lemma~\ref{lemmaapp:stat_of_markov_equivalent_are_markov_equivalent} $\stat(\M_1)$ and $\stat(\M_2)$ are Markov equivalent DMAGs and since ii) either $\M_1$ or $\M_2$ is $\Mtaumax(\D)$ and hence either $\stat(\M_1) = \Mtaumaxstat(\D)$ or $\stat(\M_2) = \Mtaumaxstat(\D)$, we thus get that both $\stat(\M_1)$ and $\stat(\M_2)$ are in $[\Mtaumaxstat(\D)]$. Recall that $\stat(\M)$ always has repeating ancestral relationships (and, hence, also repeating orientations), that $\stat(\M)$ is time ordered if $\M$ is time ordered, and that $\stat(\M)$ is a stationarified ts-DMAG if $\M$ is a ts-DMAG. Hence, both $\stat(\M_1)$ and $\stat(\M_2)$ are in the Markov equivalence class $[\Mtaumaxstat(\D)]_{\BRstat}$.

Since i) both $\stat(\M_1)$ and $\stat(\M_2)$ are in $[\Mtaumaxstat(\D)]_{\BRstat}$ and since ii) $(i, t_i)$ and $(j, t_j)$ are adjacent in $\Mtaumaxstat(\D)$, the vertices $(i, t_i)$ and $(j, t_j)$ are also adjacent in both $\stat(\M_1)$ and $\stat(\M_2)$. Since $\stat(\M_i)$ is a subgraph of $\M_i$ for $i = 1, 2$, we conclude that $(i, t_i) \tailhead (j, t_j)$ in $\stat(\M_1)$ and $(i, t_i) \headast (j, t_j)$ in $\stat(\M_2)$. Hence, $(i, t_i) \oast (j, t_j)$ in $\PAG(\Mtaumaxstat(\D), \BRstat)$.

\textbf{2.}
This claim immediately follows from part 1 of Theorem~3 because $\stat(\PAG(\Mtaumax(\D), \BR))$ is a subgraph of $\PAG(\Mtaumax(\D), \BR)$.

\textbf{3.}
The three graphs obtained by applying stationarification $\stat(\cdot)$ to the graph in parts c), d) and e) of Fig.~9 in the main text provide such examples, see also the discussion in Example~5.6 in the main text.

\textbf{4.}
This claim immediately follows from part 3 of Theorem~3 because $\stat(\PAG(\Mtaumax(\D), \BR))$ is a subgraph of $\PAG(\Mtaumax(\D), \BR)$ and because $\stat(\PAG(\Mtaumax(\D), \BR))$ and $\PAG(\Mtaumaxstat(\D), \BR)$ have the same skeleton. 
\end{proof}

\subsection{Proofs for Sec.~5.4}

\begin{proof}[Proof of Lemma~5.9]
The premise that $(i, t_i) \oast (j, t_j)$ is in the ts-DPAG $\PAGtaumax(\D)$ by definition of m.i.~DPAGs and the background knowledge $\BRtsDAG$ means: There are ts-DAGs $\D_1$ and $\D_2$ such that both $\Mtaumax(\D_1)$ and $\Mtaumax(\D_2)$ are Markov equivalent to $\Mtaumax(\D)$ and $(i, t_i) \tailhead (j, t_j)$ in $\Mtaumax(\D_1)$ and $(i, t_i) \headast (j, t_j)$ in $\Mtaumax(\D_2)$. Consequently, $(i, t_i) \in \an((j, t_j), \Mtaumax(\D_1))$ and $(i, t_i) \notin \an((j, t_j), \Mtaumax(\D_2))$ and thus, using Lemma~\ref{lemmaapp:dag_and_mag_same_ancestral}, $(i, t_i) \in \an((j, t_j), \D_1)$ and $(i, t_i) \notin \an((j, t_j), \D_2)$.
\end{proof}

\section{Proofs for Sec.~\ref{secapp:increasing_taumax}}\label{secapp:proofs_for_increasing}

\subsection{Proofs for Sec.~\ref{secapp:different_time_windows}}

\begin{mylemma}\label{lemmaapp:future_does_not_matter}
Let $(i, t_i)$ and $(j, t_j)$ with $t - \taumax \leq t_i, t_j \leq t$ be distinct observable vertices in a ts-DAG $\D$ and let $\Delta t > 0$. Then: There is $\mathbf{S} \subseteq \mathbf{O}(t-\taumax, t) \setminus \{(i, t_i), (j, t_j)\}$ such that $(i, t_i) \ci (j, t_j) ~|~ \mathbf{S}$ if and only if there is $\mathbf{S}^\prime \subseteq \mathbf{O}(t-\taumax, t + \Delta t) \setminus \{(i, t_i), (j, t_j)\}$ such that $(i, t_i) \ci (j, t_j) ~|~ \mathbf{S}^\prime$.
\end{mylemma}

\begin{proof}[Prof of Lemma~\ref{lemmaapp:future_does_not_matter}]
\textbf{If.} The premise is $(i, t_i) \ci (j, t_j) ~|~ \mathbf{S}^\prime$ with $\mathbf{S}^\prime \subseteq \mathbf{O}(t-\taumax, t + \Delta t) \setminus \{(i, t_i), (j, t_j)\}$ for some $\Delta t > 0$. According to Lemma S5 in the supplementary material of \citet{LPCMCI} this premise implies $(i, t_i) \ci (j, t_j) ~|~ \mathbf{S}$, where $\mathbf{S}$ is the restriction of $\mathbf{S}^\prime$ to ancestors of $(i, t_i)$ and $(j, t_j)$, i.e., where $\mathbf{S} = \mathbf{S}^\prime \cap \left(\an((i, t_i), \D) \cup \an((j, t_j), \D)\right)$. By time order of $\D$, no element of $\an((i, t_i), \D) \cup \an((j, t_j), \D)$ is after $\max(t_i, t_j) \leq t$ and hence $\mathbf{S} \subseteq \mathbf{O}(t-\taumax, t) \setminus \{(i, t_i), (j, t_j)\}$.

\textbf{Only if.} Take $\mathbf{S} = \mathbf{S}^\prime$.
\end{proof}

\begin{proof}[Proof of Lemma~\ref{lemmaapp:DMAGs_different_taumax_2}]
\textbf{1.}
The combination of past-repeating adjacencies and repeating orientations implies that $\M^{\taumaxtilde, [t-\taumax-\Delta t, t - \Delta t]}(\D)$ with $0 \leq \Delta t < \taumaxtilde - \taumax$ is a subgraph of $\M^{\taumaxtilde, [t-\taumaxtilde, t - \taumaxtilde + \taumax]}(\D)$. The statement then follows because by part 2 of Lemma~\ref{lemmaapp:DMAGs_different_taumax_2} the graphs $\M^{\taumaxtilde, [t-\taumaxtilde, t - \taumaxtilde + \taumax]}(\D)$ and $\Mtaumax(\D)$ are equal up to relabeling vertices.

\textbf{2.}
We first show that $\Mtaumax(\D)$ and $\M^{\taumaxtilde, [t-\taumaxtilde, t - \taumaxtilde + \taumax]}(\D)$ have the same skeleton up to relabeling vertices. To this end, consider two distinct observable vertices $(i, t_i)$ and $(j, t_j)$ with $t-\taumax \leq t_i, t_j \leq t$ and let $\Delta t = \taumaxtilde - \taumax$. According to Lemma~\ref{lemmaapp:future_does_not_matter} there is $\mathbf{S}^\prime \subseteq \mathbf{O}(t-\taumaxtilde, t) \setminus \{(i, t_i -\Delta t), (j, t_j - \Delta t)\}$ such that $(i, t_i - \Delta t) \ci (j, t_j - \Delta t) ~|~ \mathbf{S}^\prime$ if and only if there is $\mathbf{S} \subseteq \mathbf{O}(t-\taumaxtilde, t - \Delta t) \setminus \{(i, t_i -\Delta t), (j, t_j - \Delta t)\}$ such that $(i, t_i - \Delta t) \ci (j, t_j - \Delta t) ~|~ \mathbf{S}$. By the repeating separating sets property of $\D$, the existence of $\mathbf{S} \subseteq \mathbf{O}(t-\taumaxtilde, t - \Delta t) \setminus \{(i, t_i -\Delta t), (j, t_j - \Delta t)\}$ such that $(i, t_i - \Delta t) \ci (j, t_j - \Delta t) ~|~ \mathbf{S}$ is in turn equivalent to the existence of $\mathbf{S}_{\Delta t} \subseteq \mathbf{O}(t-\taumax, t) \setminus \{(i, t_i), (j, t_j)\}$ such that $(i, t_i) \ci (j, t_j) ~|~ \mathbf{S}_{\Delta t}$. Hence, $(i, t_i)$ and $(j, t_j)$ are adjacent in $\Mtaumax(\D)$ if and only if $(i, t_i - \Delta t)$ and $(j, t_j - \Delta t)$ are adjacent in $\Mtaumaxtilde(\D)$. This equivalence shows that $\Mtaumax(\D)$ and $\M^{\taumaxtilde, [t-\taumaxtilde, t - \taumaxtilde + \taumax]}(\D)$ have the same skeleton up to relabeling vertices.

Next, let $(i, t_i - \Delta t) \astast (j, t_j - \Delta t)$ be an edge in $\M^{\taumaxtilde, [t-\taumaxtilde, t - \taumaxtilde + \taumax]}(\D)$. Then, this edge $(i, t_i - \Delta t) \astast (j, t_j - \Delta t)$ in $\M^{\taumaxtilde, [t-\taumaxtilde, t - \taumaxtilde + \taumax]}(\D)$ and the edge $(i, t_i) \astast (j, t_j)$ in $\Mtaumax(\D)$ have the same orientation because in both graphs the edge orientations signify ancestral relationships according to $\D$ and $\D$ has repeating ancestral relationships.

\textbf{3.}
See Example~\ref{myexampleapp:different_taumax_DMAG}.
\end{proof}

\begin{mylemma}\label{lemmaapp:DMAGs_different_taumax_canonical_tsDAG}
Let $\D$ be a ts-DAG and $\taumaxtilde > \taumax \geq 0$. Then:
\begin{enumerate}
\item $\Mtaumax(\D) = \Mtaumax(\Dc(\Mtaumaxtilde(\D)))$.
\item There are cases in which $\Mtaumaxtilde(\D) \neq \Mtaumaxtilde(\Dc(\Mtaumax(\D)))$.
\end{enumerate}
\end{mylemma}

\begin{proof}[Proof of Lemma~\ref{lemmaapp:DMAGs_different_taumax_canonical_tsDAG}]
\textbf{1.}
Let $\D_1 = \D$ and $\D_2 = \Dc(\Mtaumaxtilde(\D))$. We then get the equality $\Mtaumaxtilde(\D_1) = \Mtaumaxtilde(\D) = \Mtaumaxtilde(\Dc(\Mtaumaxtilde(\D))) = \Mtaumaxtilde(\D_2)$, where the second equality follows from Lemma~4.14. Thus, $\Mtaumax(\D_1) = \Mtaumax(\D_2)$ according to part 1 of Lemma~\ref{lemmaapp:DMAGs_different_taumax_inference}.

\textbf{2.}
Consider the ts-DAG $\D$ in part a) of Fig.~\ref{figapp:DMAG_larger_taumax}, which respectively implies the ts-DMAGs $\M^1(\D)$ and $\M^2(\D)$ in parts b) of c) of the same figure. Part a) of Fig.~\ref{fig:same_tsDMAGs_for_smaller_taumax} shows the canonical ts-DAG $\Dc(\M^1(\D))$ of $\M^1(\D)$, which in turn marginalizes to the ts-DMAG $\M^2(\Dc(\M^1(\D))) \neq \M^2(\D)$ shown in part b) of the same figure.
\end{proof}

\renewcommand{\thefigure}{F}
\begin{figure}[tb]
\centering
\includegraphics[width=0.65\linewidth, page = 1]{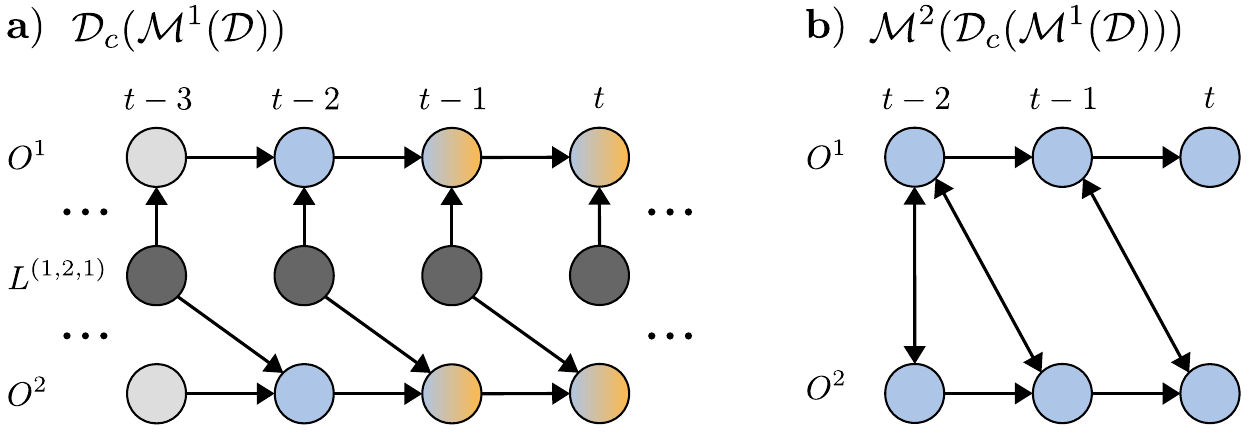}
\caption{Example for proving part 2 of Lemma~\ref{lemmaapp:DMAGs_different_taumax_canonical_tsDAG}. \textbf{a)} The canonical ts-DAG $\Dc(\M^1(\D))$ of the ts-DMAG $\M^1(\D)$ in part b) of Fig.~\ref{figapp:DMAG_larger_taumax} . \textbf{b)} The ts-DMAG $\M^2(\Dc(\M^1(\D)))$ implied by the canonical ts-DAG in part a).}
\label{fig:same_tsDMAGs_for_smaller_taumax}
\end{figure}

\begin{proof}[Proof of Lemma~\ref{lemmaapp:DMAGs_different_taumax_inference}]
\textbf{1.}
This claim follows from the commutativity of the marginalization process as stated by Theorem 4.20 in \citet{richardson2002}.

\textbf{2.}
According to part 2 of Lemma~\ref{lemmaapp:DMAGs_different_taumax_canonical_tsDAG}, there is a ts-DAG $\D$ and $\taumaxtilde > \taumax \geq 0$ such $\Mtaumaxtilde(\D) \neq \Mtaumaxtilde(\Dc(\Mtaumax(\D)))$.\footnote{Note that the proof of Lemma~\ref{lemmaapp:DMAGs_different_taumax_canonical_tsDAG} uses part~1 of Lemma~\ref{lemmaapp:DMAGs_different_taumax_inference} but not part~2 of Lemma~\ref{lemmaapp:DMAGs_different_taumax_inference}, such that the proofs of these two lemmas are \emph{not} circular.} Moreover, Lemma~4.14 implies the equality $\Mtaumax(\D) = \Mtaumax(\Dc(\Mtaumax(\D)))$. Take $\D_1 = \D$ and $\D_2 = \Dc(\Mtaumax(\D))$.
\end{proof}

\begin{proof}[Proof of Lemma~\ref{lemmaapp:DPAGs_different_taumax}]
\textbf{1.}
We prove the contraposition: Let there be a circle mark on $(i, t_i) \astast (j, t_j)$ in $\PAGtaumaxtilde(\D)$. Without loss of generality we may assume this edge to be of the form $(i, t_i) \oast (j, t_j)$. Thus, by Lemma~5.9, there are ts-DAGs $\D_1$ and $\D_2$---one of which without loss of generality is $\D$---such that the ts-DMAGs $\Mtaumaxtilde(\D_1)$ and $\Mtaumaxtilde(\D_2)$ are Markov equivalent and that $(i, t_i) \in \an((j, t_j), \D_1)$ and $(i, t_i) \notin \an((j, t_j), \D_2)$. According to commutativity of the marginalization process as stated in Theorem 4.20 in \citet{richardson2002}, the ts-DMAGs $\Mtaumax(\D_1)$ and $\Mtaumax(\D_2)$ are respectively obtained by marginalizing $\Mtaumaxtilde(\D_1)$ and $\Mtaumaxtilde(\D_2)$ over the vertices within $[t-\taumaxtilde, t-\taumax - 1]$. Hence, given that $\Mtaumaxtilde(\D_1)$ and $\Mtaumaxtilde(\D_2)$ are Markov equivalent, so are $\Mtaumax(\D_1)$ and $\Mtaumax(\D_2)$. These considerations show that $(i, t_i) \oast (j, t_j)$ in $\PAGtaumax(\D)$.

\textbf{2.}
See Example~\ref{myexampleapp:DPAG_larger_taumax}.
\end{proof}

\begin{proof}[Proof of Lemma~\ref{lemmaapp:DPAGs_different_taumax_2}]
\textbf{1.}
Let $(i, t_i) \oast (j, t_j)$ be an edge in $\PAG^{\taumaxtilde, [t-\taumax, t]}(\D)$. Because $\PAGtaumaxtilde(\D)$ is a supergraph of $\PAG^{\taumaxtilde, [t-\taumax, t]}(\D)$, the edge $(i, t_i) \oast (j, t_j)$ is then also in $\PAGtaumaxtilde(\D)$. Moreover, according to part 1 of Lemma~\ref{lemmaapp:DMAGs_different_taumax_2} (for $\Delta t = 0$) in combination with the definition of DPAGs, $(i, t_i)$  and $(j, t_j)$ are adjacent in $\PAGtaumax(\D)$. We conclude that the edge $(i, t_i) \oast (j, t_j)$ is $\PAGtaumax(\D)$ too because the opposite would contradict part 1 of Lemma~\ref{lemmaapp:DPAGs_different_taumax}.

\textbf{2.}
See Example~\ref{myexampleapp:DPAG_larger_taumax}.
\end{proof}

\subsection{Proofs for Sec.~\ref{secapp:limiting}}

\begin{proof}[Proof of Lemma~\ref{lemmaapp:existence_limiting}]
\textbf{1.}
Assume the opposite. Then, there is a strictly monotonically increasing sequence $a_n$ of positive integers such that $\M^{\taumax + a_n ,[t-\taumax,t]}(\D) \neq \M^{\taumax + a_{n+1} ,[t-\taumax,t]}(\D)$ for all $n \in \mathbb{N}$. Using part~1 of Lemma~\ref{lemmaapp:DMAGs_different_taumax_2} with $(\taumax,\taumaxtilde, \Delta) \mapsto (\taumax +a_n, \taumax+a_{n+1}, 0)$, we see that $\M^{\taumax + a_n}(\D)$ is a subgraph of $\M^{\taumax + a_{n+1} ,[t-(\taumax+a_n),t]}(\D)$ and hence $\M^{\taumax + a_n ,[t-\taumax,t]}(\D)$ is a subgraph of $\M^{\taumax + a_{n+1},[t-\taumax,t]}$. In combination with $\M^{\taumax + a_n ,[t-\taumax,t]}(\D) \neq \M^{\taumax + a_{n+1} ,[t-\taumax,t]}(\D)$ we thus find that $\M^{\taumax + a_{n+1} ,[t-\taumax,t]}(\D)$ is a \emph{proper} subgraph of $\M^{\taumax + a_{n} ,[t-\taumax,t]}(\D)$. Moreover, using part~1 of Lemma~\ref{lemmaapp:DMAGs_different_taumax_2} for $(\taumax,\taumaxtilde, \Delta) \mapsto (\taumax, \taumax+a_{n}, 0)$ and for $(\taumax,\taumaxtilde, \Delta) \mapsto (\taumax, \taumax+a_{n+1}, 0)$, we learn that both $\M^{\taumax + a_n ,[t-\taumax,t]}(\D)$ and $\M^{\taumax + a_{n+1} ,[t-\taumax,t]}(\D)$ are subgraphs of $\Mtaumax(\D)$. By combining these observations we arrive at a contradiction since there are only finitely many edges between the finitely many vertices of $\Mtaumax(\D)$.

\textbf{2.}
Assume the opposite. Then, means there is a strictly monotonically increasing sequence $a_n$ of positive integers such that $\PAG^{\taumax + a_n ,[t-\taumax,t]}(\D) \neq \PAG^{\taumax + a_{n+1} ,[t-\taumax,t]}(\D)$ for all $n \in \mathbb{N}$. Let $m$ be such that $\M^{\taumax + a_{m} ,[t-\taumax,t]}(\D) = \Mtaumaxlim(\D)$, which exists as a result of part~1 of Lemma~\ref{lemmaapp:existence_limiting}. Then, for all $n \geq m$ the skeletons of $\PAG^{\taumax + a_n ,[t-\taumax,t]}(\D)$ and $\PAG^{\taumax + a_{n+1} ,[t-\taumax,t]}(\D)$ are equal. Using part~1 of Lemma~\ref{lemmaapp:DPAGs_different_taumax} we thus learn that for all $n \geq m$ there is a non-circle mark in $\PAG^{\taumax + a_{n+1} ,[t-\taumax,t]}(\D)$ that is not in $\PAG^{\taumax + a_n ,[t-\taumax,t]}(\D)$. Since there are only finitely many circle marks on the finitely many edges in $\Mtaumaxlim(\D)$, this observation is a contradiction.
\end{proof}

\begin{mylemma}\label{lemmaapp:induced_sub_MAG}
Let $\mathcal{M}$ be a DMAG with vertex set $\mathbf{V}$ and let $\mathcal{M}[{\mathbf{O}}]$ be the induced subgraph of $\mathcal{M}$ on the subset of vertices $\mathbf{O} \subseteq \mathbf{V}$. Then, $\mathcal{M}[\mathbf{O}]$ is a DMAG.
\end{mylemma}

\begin{myremark}[on Lemma~\ref{lemmaapp:induced_sub_MAG}]
In particular, the graphs $\M^{\taumaxtilde, [t_1, t_2]}(\D)$ defined in Def.~\ref{defapp:sub_ts_DMAG/DPAG} are DMAGs.
\end{myremark}

\begin{proof}[Proof of Lemma~\ref{lemmaapp:induced_sub_MAG}]
We have to show that $\mathcal{M}[\mathbf{O}]$ does not have directed cycles, does not have almost directed cycles, and is maximal.

\emph{No (almost) directed cycles}:
Since $\mathcal{M}[\mathbf{O}]$ is a subgraph of $\mathcal{M}$ and $\mathcal{M}$ does neither have a directed nor an almost directed cycle, also $\mathcal{M}[\mathbf{O}]$ does neither have a directed nor an almost directed cycle.

\emph{Maximality:} Assume the opposite, i.e., assume in $\mathcal{M}[\mathbf{O}]$ there are non-adjacent vertices $i$ and $j$ between which there is an inducing path $\pi$. Since $\mathcal{M}[\mathbf{O}]$ is a subgraph of $\mathcal{M}$, the inducing path $\pi$ between $i$ and $j$ is also in $\mathcal{M}$. Maximality of $\mathcal{M}$ thus implies that $i$ and $j$ are adjacent in $\mathcal{M}$. By the definition of induced subgraphs the nodes $i$ and $j$ would then also be adjacent in $\mathcal{M}[\mathbf{O}]$, a contradiction.
\end{proof}

\begin{proof}[Proof of Lemma~\ref{lemmaapp:limiting_DMAGs_properties}]
\textbf{1.}
According to part~1 of Lemma~\ref{lemmaapp:existence_limiting}, there is $\Delta \taumax$ such that $\Mtaumaxlim(\D) = \M^{\taumax + \Delta \taumax^\prime, [t-\taumax, t]}(\D) = \M^{\taumax + \Delta \taumax, [t-\taumax, t]}(\D)$ for all $\Delta \taumax^\prime \geq \Delta \taumax$. Thus, since the ts-DMAG $\M^{\taumax + \Delta \taumax}(\D)$ has repeating orientations and past-repeating adjacencies, also $\Mtaumaxlim(\D) = \M^{\taumax + \Delta \taumax, [t-\taumax, t]}(\D)$ has both these properties.

To complete the proof, we need to show that $\Mtaumaxlim(\D)$ has repeating adjacencies. To this end, assume the opposite. Since $\Mtaumaxlim(\D)$ has past-repeating adjacencies, this assumption means in $\M^{\taumax + \Delta \taumax}(\D)$ there is an edge $(i, t_i - \Delta t) \astast (j, t_j - \Delta t)$ with $t-\taumax \leq t_i, t_j \leq t$ and $\Delta t > 0$ such that $(i, t_i)$ and $(j, t_j)$ are non-adjacent in $\M^{\taumax + \Delta \taumax}(\D)$. That $(i, t_i)$ and $(j, t_j)$ are non-adjacent in $\M^{\taumax + \Delta \taumax}(\D)$ shows the existence of $\mathbf{S} \subseteq \mathbf{O}(t-\taumax - \Delta \taumax, t) \setminus \{(i, t_i), (j, t_j)\}$ with $(i, t_i) \ci (j, t_j) ~|~ \mathbf{S}$ in $\D$. Due to the repeating separating sets property of $\D$, then $(i, t_i - \Delta t) \ci (j, t_j- \Delta t) ~|~ \mathbf{S}_{-\Delta t}$ where $\mathbf{S}_{-\Delta t}$ is obtained by shifting all vertices in $\mathbf{S}$ backward in time by $\Delta t$ steps. The vertices $(i, t_i - \Delta t)$ and $(j, t_j - \Delta t)$ are thus non-adjacent in $\M^{\taumax + \Delta \taumax + \Delta t}(\D)$ and hence also non-adjacent in $\Mtaumaxlim(\D)$. This observation is in contradiction to the equality $\Mtaumaxlim(\D) = \M^{\taumax + \Delta \taumax, [t-\taumax, t]}(\D)$.

\textbf{2.}
Let $(i, t_i)$ and $(j, t_j)$ with $\tau = t_j - t_i \geq 0$ be distinct non-adjacent vertices in $\Mtaumaxstat(\D)$. Then, the vertices $(i, t - \tau)$ and $(j, t)$ are non-adjacent in $\Mtaumaxstat(\D)$ due to the repeating edges property of $\Mtaumaxstat(\D)$ and thus, using Lemma~4.7, also non-adjacent in $\Mtaumax(\D)$. Hence, there is $\mathbf{S} \subseteq \mathbf{O}(t-\taumax, t) \setminus \{(i, t-\tau), (j, t)\}$ such that $(i, t-\tau) \ci (j, t) ~|~ \mathbf{S}$ in $\D$. Due to the repeating separating sets property of $\D$, we thus get that $(i, t_i)$ and $(j, t_j)$ are non-adjacent in $\M^{\taumax + (t-t_j)}(\D)$ and hence also non-adjacent in $\Mtaumaxlim(\D)$. Consequently, the skeleton of $\Mtaumaxlim(\D)$ is a subgraph of the skeleton of $\Mtaumaxstat(\D)$.

Next, let $(i, t_i) \astast (j, t_j)$ be an edge in $\Mtaumaxlim(\D)$. Then, since the skeleton of $\Mtaumaxlim(\D)$ is a subgraph of the skeleton of $\Mtaumaxstat(\D)$, the vertices $(i, t_i)$ and $(j, t_j)$ are also adjacent in $\Mtaumaxstat(\D)$. Note that, since $\Mtaumaxlim(\D) = \M^{\taumaxtilde, [t-\taumax, t]}(\D)$ for some $\taumaxtilde > \taumax$ according to part~1 of Lemma~\ref{lemmaapp:existence_limiting}, the orientation of $(i, t_i) \astast (j, t_j)$ in $\Mtaumaxlim(\D)$ conveys an ancestral relationships according to $\D$. Since also in $\Mtaumaxstat(\D)$ the orientations of edges convey ancestral relationships according to $\D$, we finally get that the edge $(i, t_i) \astast (j, t_j)$ in $\Mtaumaxstat(\D)$ has the same orientation as in $\Mtaumaxlim(\D)$.

\textbf{3.}
Take $n \geq 0$ such that $\Mtaumaxlim(\D) = \M^{\taumax + n, [t-\taumax, t]}(\D)$, which exists according to part~1 of Lemma~\ref{lemmaapp:existence_limiting}. The statement now follows by applying Lemma~\ref{lemmaapp:induced_sub_MAG} with $\M \mapsto \M^{\taumax + n}(\D)$ and $\mathbf{O}$ the set of observable vertices within $[t-\taumax, t]$.

\textbf{4.}
The limiting ts-DPAG $\PAGtaumaxlim(\D)$ has repeating adjacencies because according to part~3 of Lemma~\ref{lemmaapp:limiting_DMAGs_properties} it is a DPAG for the limiting ts-DMAG $\Mtaumaxlim(\D)$, which according to part~1 of Lemma~\ref{lemmaapp:limiting_DMAGs_properties} has repeating adjacencies. Let $\Delta \taumax$ be such that $\PAGtaumaxlim(\D) = \PAG^{\taumax + \Delta \taumax, [t-\taumax, t]}(\D)$, which exists according to part~2 of Lemma~\ref{lemmaapp:existence_limiting}. Because $\PAG^{\taumax + \Delta \taumax}(\D)$ has repeating orientations according to Lemma~\ref{lemmaapp:ro_of_DPAG}, also $\PAGtaumaxlim(\D)$ has repeating orientations. Now use part~1 of Lemma~4.3.

\textbf{5.}
Part~1 of Lemma~\ref{lemmaapp:existence_limiting} gives the existence of an integer $n \geq 0$ such that $\Mtaumaxlim(\D) = \M^{\taumax + n^\prime, [t-\taumax, t]}(\D)$ for all $n^\prime \geq n$, and part~2 of the same lemma gives the existence of an integer $m \geq 0$ such that $\PAGtaumaxlim(\D) = \PAG^{\taumax + m^\prime, [t-\taumax, t]}(\D)$ for all $m^\prime \geq m$. Thus, $\Mtaumaxlim(\D) = \M^{\taumax + k, [t-\taumax, t]}(\D)$ and $\PAGtaumaxlim(\D) = \PAG^{\taumax + k, [t-\taumax, t]}(\D)$ for $k = \max(n, m)$. The statement now follows because $\PAG^{\taumax+k}(\D)$ is a DPAG for $\M^{\taumax + k}(\D)$.
\end{proof}

\bibliographystyle{imsart-nameyear} 
\bibliography{library_annals}       
